\documentclass[11pt]{article}
\usepackage{ifthen}
\usepackage{mdwlist}
\usepackage{amsmath,amssymb,amsfonts,amsthm}
\usepackage{bm}
\usepackage{mathtools}
\usepackage{url}
\usepackage{hyperref}
\usepackage{enumitem}
\usepackage{graphicx}
\usepackage{xspace}
\usepackage{verbatim}
\usepackage[margin=1in]{geometry}
\usepackage{color}
\usepackage{latexsym}
\usepackage{epsfig}
\usepackage{bm}
\usepackage[noend]{algpseudocode}
\usepackage{algorithm}

\def\a{\alpha}

\def\d{\delta}

\def\eps{\epsilon}
\def\ve{\varepsilon}

\def\g{\gamma}

\def\m{\mu}

\def\p{\pi}

\def\s{\sigma}
\def\S{\Sigma}
\def\t{\tau}

\def\to{\rightarrow}

\newcommand{\prob}[2][]{\text{\bf Pr}\ifthenelse{\not\equal{}{#1}}{_{#1}}{}\!\left[#2\right]}
\newcommand{\expect}[2][]{\text{\bf E}\ifthenelse{\not\equal{}{#1}}{_{#1}}{}\!\left[#2\right]}

\newcommand{\muhat}{\widehat{\mu}}

\newcommand{\dtv}{d_{\mathrm {TV}}}

\newcommand{\dkl}{d_{\mathrm {KL}}}
\newcommand{\Tr}{{\mathrm {Tr}}}

\newcommand{\mixpdf}{\ensuremath{\mathcal{M}}}

\newcommand{\tail}{\mathrm{Tail}}

\newtheorem{theorem}{Theorem}[section]
\newtheorem{lemma}[theorem]{Lemma}

\newtheorem{proposition}[theorem]{Proposition}
\newtheorem{corollary}[theorem]{Corollary}
\newtheorem{claim}[theorem]{Claim}
\newtheorem{remark}[theorem]{Remark}
\newtheorem{definition}[theorem]{Definition}
\newcommand{\littlesum}{\mathop{\textstyle \sum}}

\newtheorem{fact}[theorem]{Fact}
\newtheorem{question}[theorem]{Question}

\newcommand{\ignore}[1]{}

\newenvironment{prevproof}[2]{\noindent {\em {Proof of {#1}~\ref{#2}:}}}{$\hfill\qed$\vskip \belowdisplayskip}

\newcommand{\bg}[1]{\medskip\noindent{\bf #1}}
\definecolor{Red}{rgb}{1,0,0}

\newcommand{\oldbound}[1]{{}}

\newcommand\numberthis{\addtocounter{equation}{1}\tag{\theequation}}

\newcommand{\normalpdf}{\ensuremath{\mathcal{N}}}

\newcommand{\Sgood}{G}
\newcommand{\Sbad}{E}
\newcommand{\Ssym}{\mathcal{S}_{\mathrm{sym}}}

\newcommand{\Sigmahat}{\widehat{\Sigma}}

\newcommand{\Cov}{\mbox{Cov}}
\newcommand{\diag}{\mbox{diag}}
\newcommand{\Otilde}{\widetilde{O}}

\renewcommand{\epsilon}{\varepsilon}
\newcommand{\tr}{\mathrm{tr}}

\DeclareMathOperator{\R}{\mathbb{R}}

\DeclareMathOperator{\Z}{\mathbb{Z}}

\DeclareMathOperator{\sS}{{\mathcal{S}}}

\DeclareMathOperator*{\Var}{Var}
\DeclareMathOperator*{\var}{Var}
\DeclareMathOperator*{\E}{\mathbb{E}}

\DeclareMathOperator{\poly}{poly}

\DeclareMathOperator{\rank}{rank}

\DeclareMathOperator{\normal}{\mathcal{N}}
\renewcommand{\bar}{\overline}

\renewcommand{\[ }{\begin{eqnarray*}}
\renewcommand{\]}{\end{eqnarray*}}

\renewcommand{\emptyset}{\varnothing}

\definecolor{darkpastelred}{rgb}{0.76, 0.23, 0.13}

\algnewcommand\INPUT{\item[{\textbf {input:}}]}
\algnewcommand\OUTPUT{\item[{\textbf{output:}}]}

\newcommand{\Brac}[1]{\left[#1\right]}

\def\colorful{0}

\usepackage{times}
\usepackage{amsthm,amsfonts,amsmath,amssymb,epsfig,color,float,graphicx,verbatim, enumitem}
\usepackage{multirow}

\newif\ifhyper\IfFileExists{hyperref.sty}{\hypertrue}{\hyperfalse}
\hypertrue
\ifhyper\usepackage{hyperref}\fi

\newcommand{\supp}{\mathrm{supp}}
\newcommand{\eqdef}{\stackrel{{\mathrm {\footnotesize def}}}{=}}

\ifnum\colorful=1
\newcommand{\new}[1]{{\color{red} #1}}

\else
\newcommand{\new}[1]{{#1}}

\fi

\newcommand{\wt}[1]{{\widetilde{#1}}}

\title{{\LARGE Robust Estimators in High Dimensions \\ without the Computational Intractability}}

\author {
Ilias Diakonikolas\thanks{University of Southern California. 
Supported by NSF Award CCF-1652862 (CAREER) and a Sloan Research Fellowship.
Part of this work was performed while the author was at the University of Edinburgh,
supported in part by EPSRC grant EP/L021749/1 and a Marie Curie Career Integration Grant.
\href{mailto:diakonik@usc.edu}{\texttt{diakonik@usc.edu}}} 
\and
Gautam Kamath\thanks{Simons Institute for the Theory of Computing. 
Supported by NSF Award CCF-0953960 (CAREER) and ONR grant N00014-12-1-0999. 
This work was done in part while the author was an intern at Microsoft Research Cambridge, visiting the Simons Institute for the Theory of Computing, and a graduate student at MIT.
\href{mailto:g@csail.mit.edu}{\texttt{g@csail.mit.edu}}}
\and
Daniel M. Kane\thanks{University of California, San Diego. 
Part of this work was performed while visiting the University of Edinburgh.
\href{mailto:dakane@cs.ucsd.edu}{\texttt{dakane@cs.ucsd.edu}}} 
\and
Jerry Li 
\thanks{
Microsoft Research AI.
Supposed by NSF CAREER Award CCF-1453261, a Google Faculty Research Award, and an NSF Fellowship. 
This work was done in part while the author was an intern at Microsoft Research Cambridge and a graduate student at MIT.
\href{mailto:jerrl@microsoft.com}{\texttt{jerrl@microsoft.com}}}
\and
Ankur Moitra
\thanks{
Massachusetts Institute of Technology.
Supported by NSF CAREER Award CCF-1453261, a grant from the MIT NEC Corporation, and a Google Faculty Research Award.
\href{mailto:moitra@mit.edu}{\texttt{moitra@mit.edu}}}
\and
Alistair Stewart
\thanks{Web3 Foundation. 
Part of this work was performed while the author was at the University of Edinburgh and the University of Southern California.
Research supported in part by EPSRC grant EP/L021749/1.
\href{mailto:stewart.al@gmail.com}{\texttt{stewart.al@gmail.com}}}
}

\begin{document}
\maketitle

\setcounter{page}{-2}

\thispagestyle{empty}

\begin{abstract}
We study high-dimensional distribution learning in an agnostic setting where an adversary is allowed to arbitrarily corrupt an $\ve$-fraction of the samples. 
Such questions have a rich history spanning statistics, machine learning and theoretical computer science. 
Even in the most basic settings, the only known approaches are either computationally inefficient or lose dimension-dependent factors in their error guarantees. 
This raises the following question: Is high-dimensional agnostic distribution learning even possible, algorithmically? 

In this work, we obtain the first computationally efficient algorithms with dimension-independent error guarantees for agnostically learning several fundamental classes of high-dimensional distributions: 
(1) a single Gaussian, 
(2) a product distribution on the hypercube, 
(3) mixtures of two product distributions (under a natural balancedness condition),
and (4) mixtures of spherical Gaussians. 
Our algorithms achieve error that is independent of the dimension, and in many cases scales nearly-linearly with the fraction of adversarially corrupted samples. 
Moreover, we develop a general recipe for detecting and correcting corruptions in high-dimensions that may be applicable to many other problems. 
\end{abstract}

\newpage
\tableofcontents
\thispagestyle{empty}
\addtocontents{toc}{\protect\thispagestyle{empty}}
\newpage
\section{Introduction} \label{sec:intro}

\subsection{Background} \label{ssec:background}

A central goal of machine learning is to design efficient algorithms for fitting a model to a collection of observations.
In recent years, there has been considerable progress on a variety of problems in this domain,
including algorithms with provable guarantees for learning mixture models~\cite{FOS:05focs, KMV:10, MoitraValiant:10, BelkinSinha:10, HK},
phylogenetic trees~\cite{CGG:02, MosselRoch:05}, HMMs~\cite{AHK12}, topic models~\cite{AGM, AGHK},
and independent component analysis \cite{AGMS}.
These algorithms crucially rely on the assumption that the observations
were actually generated by a model in the family.
However, this simplifying assumption is not meant to be exactly true,
and it is an important direction to explore what happens when it holds only in an approximate sense.
In this work, we study the following family of questions:

\begin{question} \label{question}
Let $\mathcal{D}$ be a family of distributions on $\R^d$.
Suppose we are given samples generated from the following process:
First, $m$ samples are drawn from some unknown distribution $P$ in $\mathcal{D}$.
Then, an adversary
is allowed to arbitrarily corrupt an $\eps$-fraction of the samples.
Can we efficiently find a distribution $P'$ in $\mathcal{D}$
that is $f(\eps, d)$-close, in total variation distance, to $P$?
\end{question}

This is a natural formalization of the problem
of designing robust and efficient algorithms for distribution estimation.
We refer to it as {\em (proper) agnostic distribution learning} and we refer to the samples as being \emph{$\ve$-corrupted}.
This family of problems has its roots in many fields,
including statistics, machine learning, and theoretical computer science.
Within computational learning theory, it is related to the agnostic learning
model of Haussler~\cite{Haussler:92} and Kearns, Schapire, and Sellie~\cite{KSS:94},
where the goal is to learn a labeling function whose agreement
with some underlying target function is close to the best possible,
among all functions in some given class.
In the even more %powerful
challenging malicious noise model~\cite{Valiant:85short, keali93},
an adversary is allowed to corrupt both the labels and the samples. A major difference with our setting is that these
models apply to supervised learning problems, while here we {will work} in an unsupervised setting.

Within statistics and machine learning, inference problems like Question~\ref{question}
are often termed ``estimation under model misspecification.''
The usual prescription is to use the maximum likelihood estimator \cite{Huber:67, White:82},
which is unfortunately hard to compute in general.
Even ignoring computational considerations,
the maximum likelihood estimator is only guaranteed to converge to the distribution $P'$
in $\mathcal{D}$ that is closest (in Kullback-Leibler divergence) to the distribution
from which the observations are generated.
This is problematic because such a distribution is not necessarily close to $P$ at all.

A branch of statistics \--- called robust statistics \cite{Huber09, HampelEtalBook86} \--- aims
to tackle questions like the one above. The usual formalization is in terms of \emph{breakdown point},
which (informally) is the fraction of observations that an adversary would need to control
to be able to completely corrupt an estimator. In low-dimensions, this leads to the
prescription that one should use the empirical median instead of the empirical mean
to robustly estimate the mean of a distribution,
and interquartile range for robust estimates of the variance.
In high-dimensions, the Tukey median~\cite{Tukey75}
is a high-dimensional analogue of the median that, although provably robust, is hard to compute~\cite{JP:78}.
Similar hardness results have been shown~\cite{Bernholt, HardtM13} for
essentially all known estimators in robust statistics.

{\em Is high-dimensional agnostic distribution learning even possible, algorithmically?}
The difficulty is that corruptions are often hard to detect in high dimensions,
and could bias the natural estimator by dimension-dependent factors.
In this work, we study agnostic distribution learning for a number of fundamental classes {of distributions}:
$(1)$ a single Gaussian, $(2)$ a product distribution on the hypercube $\{0,1\}^d$,
$(3)$ mixtures of two product distributions (under a natural balancedness condition),
and $(4)$ mixtures of $k$ Gaussians with spherical covariances.
Prior to our work, all known efficient algorithms (e.g., \cite{LT15, BS15}) for these classes
required the error guarantee, $f(\eps, d)$, to depend {\em polynomially}
in the dimension $d$. Hence, previous efficient estimators could only tolerate
at most a $1/\poly(d)$ fraction of errors.
In this work, we obtain {\em the first} efficient algorithms
for the aforementioned problems, where $f(\eps, d)$ is {\em completely independent of $d$}
and  depends polynomially (often, nearly linearly) in the fraction $\eps$ of corrupted samples.
Our work is just a first step in this direction, and there are many exciting questions left to explore.

\subsection{Our Techniques} \label{ssec:tech}

All of our algorithms are based on a common recipe.
The first question to address is the following:
Even if we were given a candidate hypothesis $P'$,
how could we test if it is $\eps$-close in total variation distance to $P$?
The usual way to certify {closeness}
is to exhibit a coupling between {$P$ and $P'$} that marginally samples from both distributions,
{where} the samples produced from each agree with probability $1-\eps$.
However, we have no control over the process by which samples are generated from $P$,
in order to produce such a coupling. And even then, the way that an adversary
decides to corrupt samples can introduce complex statistical dependencies.

We circumvent this issue by working with an appropriate notion of parameter distance,
which we use as a proxy for the total variation distance between two distributions in the class $\mathcal{D}$.
Various notions of parameter distance
underly several efficient algorithms for distribution learning in the following sense.
If $\theta$ and $\theta'$ are two sets of parameters that define
distributions $P_\theta$ and $P_{\theta'}$ in a given class $\mathcal{D}$,
a learning algorithm {often} relies on establishing
the following type of relation\footnote{For example, the work of Kalai, Moitra, and Valiant~\cite{KMV:10}
%the approach was through the method of moments. In their setting,
can be reformulated as showing that for any pair of mixtures of two Gaussians (with suitably bounded parameters),
the following quantities are polynomially related: $(1)$ discrepancy
in their low-order moments, $(2)$ their parameter distance, and $(3)$ their total variation distance.
This ensures that any candidate set of parameters
that produce almost identical moments must itself result
in a distribution that is close in total variation distance. } between $\dtv(P_\theta, P_{\theta'})$ and the parameter distance $d_p(\theta, \theta')$:
\begin{equation} \label{eqn:parameter-weak}
\poly(d_p(\theta, \theta'), 1/d) \leq \dtv(P_\theta, P_{\theta'}) \leq \poly(d_p(\theta, \theta'), d) \;.
\end{equation}
Unfortunately, in our agnostic setting, we cannot afford for
(\ref{eqn:parameter-weak})  to depend on the dimension $d$ at all.
Any such dependence
would appear in the error guarantee of our algorithm.
Instead, the starting point of our algorithms is a notion
of parameter distance that satisfies
\begin{equation} \label{eqn:parameter-strong}
\poly(d_p(\theta, \theta')) \leq \dtv(P_{\theta}, P_{\theta'}) \leq \poly(d_p(\theta, \theta'))
\end{equation}
which
allows us to reformulate our goal of designing robust estimators, with distribution-independent error guarantees,
%that can tolerate a constant fraction of noise, independent of the dimension,
as the goal of robustly estimating $\theta$ according to $d_p$.
In several settings, the choice of the parameter distance is rather straightforward.
It is often the case that some variant of the $\ell_2$-distance
between the parameters works.\footnote{This discussion already points to why it may be challenging
to design agnostic algorithms for mixtures  of arbitrary Gaussians or arbitrary product distributions:
It is not clear what notion of parameter distance is polynomially related
to the total variation distance between two such mixtures, without any dependence on $d$.}

Given our notion of parameter distance satisfying (\ref{eqn:parameter-strong}),
our main  ingredient is an efficient method for robustly estimating the parameters.
We provide two algorithmic approaches which are based on similar principles.
Our first approach is faster,
requiring only approximate eigenvalue computations.
Our second approach relies on convex programming
and achieves slightly better sample complexity, in some cases matching the information-theoretic limit.
Notably, either approach can be used to give all of our concrete learning applications
with nearly identical sample complexity and error guarantees.
In what follows, we specialize to the problem of robustly learning
the mean $\mu$ of a Gaussian whose covariance is promised to be the identity,
which we will use to illustrate how both approaches operate.
We emphasize that what is needed to learn the parameters
in more general settings requires many additional ideas.

Our first algorithmic approach is an iterative greedy method that, in each iteration,
filters out {some of the} corrupted samples. Given a set of samples $S'$
that {contains} a set $S$ of uncorrupted samples,
an iteration of our algorithm either returns the sample mean of $S'$ or finds a {\em filter}
that allows us to efficiently compute a set $S'' \subset S'$ that is much closer to $S$.
%Hence, it is plausible that this algorithm performs well in practice.
%To motivate our filter approach, we start by noting that
Note the sample mean $\widehat{\mu} = \sum_{i=1}^N (1/N) X_i$
(even after we remove points that are obviously outliers)
can be $\Omega(\eps \sqrt{d})$-far  from the true mean in $\ell_2$-distance.
The filter approach shows that  either the sample mean is already a good estimate for $\mu$
or else there is an elementary spectral test that rejects some of the corrupted points
and almost none of the uncorrupted ones.
%How might the adversary bias this natural estimate by dimension-dependent factors?
%He could place a single point very far away from $\mu$ and all the other samples,
%but this is very easy to detect.  Alternatively, he could distribute $\eps m$ points
%at an appropriate distance from $\mu$, but align them in one direction,
%such that $\widehat{\mu}$ is off by an additive $\Omega(\eps \sqrt{d})$.
The crucial observation is that if a small number of corrupted points are responsible
for a large change in the sample mean, it must be the case
that many of the error points are very far from the mean in some particular direction.
Thus, we obtain our filter by
computing the top absolute eigenvalue of a modified sample covariance matrix.

% We can use your intuitive nice prose that I commented out for space saving for the main part of the paper.

Our second algorithmic approach relies on convex programming.
Here, instead of rejecting corrupted samples, we compute appropriate {\em weights}
$w_i$ for the samples $X_i$, such that the weighted empirical average
$\widehat{\mu}_w = \littlesum_{i=1}^N w_i X_i$ is close to $\mu$.
We work with the convex set:
$$\mathcal{C}_\delta = \left \{ w_i \mid
0 \leq w_i \leq 1/((1-\epsilon)N), \littlesum_{i=1}^N w_i = 1, \left\| \littlesum_{i = 1}^N w_i (X_i - \mu) (X_i - \mu)^T - I \right\|_2 \leq \delta \right \} \;.$$
We prove that {\em any} set of weights in $\mathcal{C}_\delta$ yields a good estimate $\widehat{\mu}_w = \littlesum_{i=1}^N w_i X_i$
in the obvious way. The catch is that the set $\mathcal{C}_\delta$ is defined based on $\mu$, {\em which is unknown}.
Nevertheless, it turns out that we can use the same type of spectral arguments
that underlie the filtering approach
to design an approximate separation oracle for $\mathcal{C}_\delta$.
Combined with standard results in convex optimization, this yields an algorithm for robustly estimating $\mu$.

The third and final ingredient is some new concentration bounds.
In both of the approaches above, at best we are hoping
that we can remove all of the corrupted points and be left
with only the uncorrupted ones, and then use standard estimators (e.g., the empirical average) on them.
However, an adversary could have removed an $\eps$-fraction of the samples
in a way that biases the empirical average of the remaining uncorrupted samples.
What we need are concentration bounds that show for sufficiently large $N$,
for samples $X_1, X_2, \ldots, X_N$ from a Gaussian with mean $\mu$
and identity covariance, that every set of $(1-\eps)N$ samples
produces a good estimate for $\mu$. In some cases, we can derive such concentration bounds
by appealing to known concentration inequalities and taking a union bound.
However, in other cases (e.g., concentration bounds for degree-two polynomials of Gaussian random variables)
the existing concentration bounds are not strong enough,
and we need other arguments to prove that every set of $(1-\eps)N$ samples produces a good estimate.
%Also in Section~\ref{sec:fail} we explain why some other natural strategies for robust distribution learning obtain poor guarantees in high-dimensions.

\subsection{Our Results} \label{ssec:results}

We give the first efficient algorithms for agnostically learning several important distribution classes
with dimension-independent error guarantees.
Our first main result is for a single arbitrary Gaussian with mean $\mu$ and covariance $\Sigma$, which we denote by $\normal(\mu, \Sigma)$.
In the previous subsection, we described our convex programming approach
for learning the mean vector when the covariance is promised to be the identity.
A technically more involved version of the technique
can handle the case of zero mean and unknown covariance.
More specifically, consider the following convex set, where $\Sigma$ is the unknown covariance matrix and $\left\|\cdot\right\|_F$ is the Frobenius norm:
$$\mathcal{C}_\delta = \left \{ w_i \mid
0 \leq w_i \leq 1/((1-\epsilon)N), \littlesum_{i=1}^N w_i = 1, \left\| \Sigma^{-1/2} \left(\littlesum_{i = 1}^N w_i X_i X_i^T\right) \Sigma^{-1/2} - I \right\|_F \leq \delta \right \} \;.$$
We design an approximate separation oracle for this unknown convex set,
by analyzing the spectral properties of the fourth moment tensor of a Gaussian. Combining these two intermediate results, we obtain our first main result (below). Throughout this paper, we will abuse notation and write $N \geq \widetilde{\Omega}(f(d, \ve, \tau))$ when referring to our sample complexity, to signify that our algorithm works if $N \geq C f(d, \ve, \tau) \mbox{polylog}(f(d, \ve, \tau))$ for a large enough universal constant $C$.

\begin{theorem} \label{thm:arbitrary-gaussian}
Let $\mu, \Sigma$ be arbitrary and unknown, and let $\ve, \tau > 0$.
There is a polynomial time algorithm which given $\ve, \tau,$ and an $\ve$-corrupted set of $N$ samples from $\normal (\mu, \Sigma)$ with $N \geq \widetilde{\Omega} \left( \frac{d^2 \log^5 (1/ \tau)}{\epsilon^2} \right)$, produces $\muhat$ and $\Sigmahat$ such that with probability $1 - \tau$ we have $\dtv (\normal (\mu, \Sigma), \normal (\muhat, \Sigmahat)) \leq O(\ve \log^{3/2} (1 / \ve) )$.
\end{theorem}

\noindent We can alternatively establish Theorem~\ref{thm:arbitrary-gaussian} via our filtering technique. See Section~\ref{sec:filterGaussian}. In the first version of our paper, our analysis
required $N \gtrsim d^3 \log^2 (1/\tau)/\epsilon^2$ samples. In \cite{DiakonikolasKKLMS17}, we showed that a simple adaptation of our algorithm and analysis 
achieves the improved sample complexity above, which is information-theoretically optimal up to logarithmic factors. 
We have incorporated this modification (along with the analysis) into this version of the paper, for the sake of completeness.

Our second agnostic learning result is for a product distribution on the hypercube --
arguably the most fundamental discrete high-dimensional distribution.
We solve this problem using our filter technique, though our convex programming approach
would also yield similar results.
We start by analyzing the balanced case, when no coordinate is very close to being deterministic.
This special case is interesting in its own right and captures the essential ideas of our more involved analysis for the general case.
The reason is that, for two balanced product distributions, the $\ell_2$-distance between their means
is equivalent to their total variation distance (up to a constant factor).
This leads to a clean and elegant presentation of our spectral arguments.
For an arbitrary product distribution, we handle the coordinates that are essentially deterministic separately.
Moreover, we use the $\chi^2$-distance between the means as the parameter distance and, as a consequence,
we need to apply the appropriate corrections to the covariance matrix. Formally, we prove:

\begin{theorem}
Let $\Pi$ be an unknown binary product distribution, and let $\ve, \tau > 0$.
There is a polynomial time algorithm which given $\ve, \tau,$ and an $\ve$-corrupted set of $N$ samples from $\Pi$  
with $N \geq \Omega \left( \frac{d^6\log(1/\tau)}{\eps^3} \right)$, produces a binary product distribution $\widetilde{\Pi}$ 
such that with probability $1 - \tau$, we have $\dtv (\Pi, \widetilde{\Pi}) \leq O(\sqrt{\ve \log (1 / \ve)})$.
\end{theorem}

For the sake of simplicity in the presentation, we did not make an effort to optimize 
the sample complexity of our robust estimators in the above setting. We note that methods
similar to the analysis of the Gaussian setting can lead 
to near-optimal sample complexity in this setting as well.
We also remark that for the case of balanced binary product distributions, our algorithm achieves an error of 
$O(\ve \sqrt{\log (1 / \ve)}).$

Interestingly enough, the above two distribution classes
are trivial to learn in the noiseless case, but in the agnostic setting
the learning problem turns out to be surprisingly challenging.
Using additional ideas,  we are able to generalize our agnostic learning algorithms
to {\em mixtures} of the above classes under some natural conditions.
We note that even in the noiseless case, learning mixtures of the above families
is non-trivial.
%obtain agnostic distribution learning analogues of other results,
%where even the algorithm in the presence of no corruptions requires some effort.
First, we study $2$-mixtures of $c$-balanced products, 
which stipulates that the coordinates of the mean vector of each 
component are in the range $(c, 1-c)$. We prove:

\begin{theorem}[informal]
Let $\Pi$ be an unknown mixture of two $c$-balanced binary product distributions, and let $\ve, \tau > 0$.
There is a polynomial time algorithm which given $\ve, \tau,$ and an $\ve$-corrupted set of $N$ samples from $\Pi$ 
with $N \geq \tilde{\Omega} \left( \frac{d^4 \log (1 / \tau)}{\ve^{13/6}} \right)$, 
produces a mixture of two binary product distributions $\widetilde{\Pi}$ such that with probability $1 - \tau$, 
we have $\dtv (\Pi, \widetilde{\Pi}) \leq O_c(\ve^{1/6})$, where the notation $O_c(\cdot)$ suppresses dependence on $c$.
\end{theorem}

\noindent This generalizes the algorithm of Freund and Mansour~\cite{FreundMansour:99short} to the agnostic setting.
An interesting open question is to improve the $\eps$-dependence in the above bound
to (nearly) linear, or to remove the assumption of balancedness
and obtain an agnostic algorithm for mixtures of two arbitrary product distributions.

%It is the first of our algorithms whose dependence on $\epsilon$ is not nearly linear,
%and a challenging open question is to improve the dependence on $\epsilon$.

Finally, we give an agnostic learning algorithm for mixtures of spherical Gaussians.

\begin{theorem}[informal]
Let $\mixpdf$ be a mixture of $k$ Gaussians with spherical covariances, and let $\ve, \tau > 0$ and $k$ be a constant.
There is a polynomial time algorithm which given $\ve, \tau$, and an $\ve$-corrupted set of $N$ samples 
from $\mixpdf$ with $N \geq \poly(k, d, 1/\ve, \log(1/\tau))$, outputs an $\mixpdf'$ such that with probability $1 - \tau$, 
we have $\dtv (\mixpdf, \mixpdf') \leq \tilde O(\poly(k) \cdot \sqrt{\ve})$.
\end{theorem}

\noindent 
%We believe that we can strengthen this result to work for mixtures of arbitrary spherical Gaussians, 
%but for simplicity of exposition we only prove this special case here.
Our agnostic algorithms for (mixtures of) balanced product distributions
and for (mixtures of) spherical Gaussians are conceptually related,
since in both cases the goal is to robustly learn the means of each component with respect to $\ell_2$-distance.

In total, these results give new robust and computationally
efficient estimators for several well-studied
distribution learning problems that can tolerate a constant fraction of errors
independent of the dimension.
This points to an interesting new direction of making %as much of
robust statistics algorithmic.
%as possible.
The general recipe we have developed
here gives us reason to be optimistic about many other problems in this domain.

\subsection{Discussion and Related Work} \label{ssec:discuss}

Our results fit in the framework of {\em density estimation} and {\em parameter learning}
which are both classical problems in statistics with a rich history (see e.g.,~\cite{BBBB:72, DG85, Silverman:86,Scott:92,DL:01}).
While these problems have been studied for several decades by different communities, 
the computational complexity of learning is still not well understood, even
for some surprisingly simple distribution families.
Most textbook estimators are hard to compute in general, 
especially in high-dimensional settings. In the past few decades, a rich body of work 
within theoretical computer science has focused on designing computationally efficient 
distribution learning algorithms. 
In a seminal work, Kearns, Mansour, Ron, Rubinfeld, Schapire, and Sellie~\cite{KMR+:94short} initiated a systematic
investigation of the computational complexity of distribution learning.
Since then, efficient learning algorithms have been developed 
for a wide range of distributions in both low and high-dimensions
\cite{Dasgupta:99, FreundMansour:99short, AroraKannan:01, VempalaWang:02, CGG:02, MosselRoch:05, BV:08, KMV:10, MoitraValiant:10, BelkinSinha:10, DDS12soda, CDSS13, DDOST13focs, CDSS14, CDSS14b, HP14, ADLS15, DDS15, DDKT15, DKS15, DKS16}.

%problems such as learning decision trees \cite{EhrenfeuchtHaussler:89}, juntas \cite{mosodoser03, Valiant15}, 
%phylogenetic trees \cite{MosselRoch:05} and many other natural classes of distributions \cite{DDS12soda, DDS12stoc, CDSS13, DDOST13focs, CDSS14, CDSS14b, ADLS15, %DDS15, DDKT15, DKS15, DKS16}. 

We will be particularly interested in efficient learning algorithms for mixtures of high-dimensional 
Gaussians and mixtures of product distributions, as this is the focus of our algorithmic results in the agnostic setting. 
In a pioneering work, Dasgupta~\cite{Dasgupta:99} introduced the problem of parameter estimation of a Gaussian mixture
to theoretical computer science, and gave the first provably efficient algorithms under the assumption that 
the components are suitably well-separated.  Subsequently, a number of works improved 
these separation conditions~\cite{AroraKannan:01, VempalaWang:02, BV:08} 
and ultimately removing them entirely \cite{KMV:10, MoitraValiant:10, BelkinSinha:10}. 
%HK, GeHK15.
In another line of work, Freund and Mansour~\cite{FreundMansour:99short} gave the first polynomial time 
algorithm  for properly learning mixtures of two binary product distributions. This algorithm was substantially 
generalized to phylogenetic trees~\cite{CGG:02} and 
to mixtures of any constant number of discrete product distributions~\cite{FOS:05focs}.
Given the vast body of work on high-dimensional distribution learning, 
there is a plethora of problems where one could hope to reconcile robustness and computational efficiency. 
Thus far, the only setting where robust and efficient algorithms are known is 
on one-dimensional distribution families, where brute-force search or some form of polynomial regression often works. 
In contrast, essentially nothing is known about efficient agnostic distribution learning 
in the high-dimensional setting that we study here.

Question~\ref{question} also resembles learning in the presence of malicious errors~\cite{Valiant:85short, keali93}.
There, an algorithm is given samples from a distribution along
with their labels according to an unknown target function.
The adversary is allowed to corrupt an $\eps$-fraction of both the samples and their labels.
A sequence of works studied the problem of learning a homogeneous halfspace with malicious noise
in the setting where the underlying distribution is a Gaussian \cite{Servedio:01lnh, Servedio:03jmlrshort, KLS09},
culminating in the work of Awasthi, Balcan, and Long~\cite{ABL14},
who gave an efficient algorithm that finds a halfspace with agreement $O(\eps)$.
There is no direct connection between their problem and ours, especially since 
one is a supervised learning problem and the other is unsupervised. 
We note however that there is an interesting technical parallel in that the work \cite{KLS09} 
also uses spectral methods to detect outliers. Both their work and our algorithm for agnostically 
learning the mean are based on the intuition that an adversary 
can only substantially bias the empirical mean if the corruptions are correlated along some direction. 
More specifically, \cite{KLS09} produce a ``hard'' filter which leads to errors that scale 
logarithmically with the dimension, even in a weaker corruption model than ours.
Our algorithms need to handle many significant conceptual and technical 
complications that arise when working with higher moments or other distribution families. 
%while theirs works in a supervised setting.
%Moreover, their algorithms need to assume
%that the underlying Gaussian distribution is in isotropic position.
%In fact, our results are complementary to theirs:
%One could use our algorithms (on the unlabeled examples)
%to learn an affine transformation that puts the underlying Gaussian distribution in approximately isotropic position,
%even in the presence of malicious errors, such that one can then directly apply the~\cite{ABL14} algorithm.

Another connection is to the work on robust principal component analysis (PCA).
PCA is a transformation that (among other things) is often justified as being able
to find the affine transformation $Y = \Sigma^{-1/2}(X - \mu)$ that would place
a collection of Gaussian random variables in isotropic position.
One can think of our results on agnostically learning a Gaussian
as a type of robust PCA that tolerates gross corruptions,
where entire samples are corrupted. This is different than other variants
of the problem where random sets of coordinates of the points
are corrupted~\cite{CLMW11}, or where the uncorrupted points
were assumed to lie in a low-dimensional subspace to begin with~\cite{ZL, LMTZ}.
Finally, Brubaker~\cite{Brubaker09} studied
the problem of clustering samples from a {\em well-separated} mixture of Gaussians in the presence of adversarial noise.
The goal of~\cite{Brubaker09} was to separate the Gaussian components from each other,
while the adversarial points are allowed to end up in any of clusters.
Our work is orthogonal to~\cite{Brubaker09},
since even if such a clustering is given,
the problem still remains
to estimate the parameters of each component.%, which is where our work begins.

\subsection{Concurrent and Subsequent Work}
In concurrent and independent work, Lai, Rao, and Vempala \cite{LaiRV16} also study high-dimensional agnostic learning. 
Their results were shown to apply for more general types of distributions, 
but our guarantees are stronger when learning a Gaussian. 
Our results are qualitatively similar when the mean is unknown and the covariance is promised to be the identity. 
But when the covariance is also unknown, their algorithm estimates the mean and covariance to within 
error $O(\sqrt{\ve \|\S\|_2 \log d})$ and $O(\sqrt{\ve \log d} \|\S\|_2)$, measured in $\ell_2$-norm and Frobenius norm respectively. 
However, such guarantees do not directly imply bounds on the total variation distance (which is our main focus), 
because one needs to estimate the parameters with respect to Mahalanobis distance. 
In contrast, by virtue of being close in total variation distance, our estimates for the mean and covariance 
are within $\widetilde O(\ve \sqrt{\|\S\|_2})$ and $\widetilde O(\ve \|\S\|_2)$ of the true values, 
again measured in $\ell_2$ norm and Frobenius norm respectively. 
An interesting open question is to bridge these two works \---- what are the most general families of distributions 
for which one can obtain nearly optimal agnostic learning guarantees?

After the initial publication of our results~\cite{DiakonikolasKKLMS16}, there has been a flurry of recent work on robust high-dimensional estimation.
Diakonikolas, Kane, and Stewart \cite{DiakonikolasKS16b} studied the problem of learning the parameters of a graphical model in the presence of noise, when given its graph theoretic structure. Charikar, Steinhardt, and Valiant \cite{CharikarSV16} developed algorithms that can tolerate a fraction of corruptions greater than a half, %approaching one
under the weaker goal of outputting a small list of candidate hypotheses that contains a parameter set close to the true values. 
Balakrishnan, Du, Li, and Singh \cite{Li17, DBS17, BalakrishnanDLS17} 
studied sparse mean and covariance estimation in the presence of noise obtaining
computationally efficient robust algorithms with sample complexity sublinear in the dimension.
Diakonikolas, Kane, and Stewart \cite{DiakonikolasKS16c} proved statistical query lower bounds
providing evidence that the error guarantees of our robust mean and covariance estimation
algorithms are best possible, within constant factors, for efficient algorithms. 
%any statistical query learning algorithm 
%that works in our corruption model and achieves an error guarantee asymptotically better than $O(\epsilon \sqrt{\log1/\epsilon})$ must make a super-polynomial number of statistical queries. 
In a subsequent paper~\cite{DiakonikolasKKLMS17}, we obtained improved bounds on the sample complexity of our algorithms, 
which are optimal up to polylogarithmic factors. For the sake of completeness, we include these improved sample bounds in the present version of this paper.
In the same work~\cite{DiakonikolasKKLMS17}, we showed that our algorithmic approach easily extends to obtain 
dimension-independent robustness guarantees under much weaker distributional assumptions, 
and gave a practical demonstration of the efficacy of our robust algorithms on both real and synthetic data. 

Since the initial submission of the journal version of this paper, there has been a substantial amount of work on robust high-dimensional 
estimation in a variety of settings. Diakonikolas, Kane, and Stewart~\cite{DKS17-nasty} studied 
PAC learning of geometric concept classes (including low-degree polynomial threshold functions and intersections of halfspaces) 
in the same corruption model as ours, obtaining the first dimension-independent error guarantees for these classes.
Steinhardt, Charikar, and Valiant~\cite{SteinhardtCV18} focused on deterministic conditions of a dataset which allow robust estimation to be possible.
In our initial publication, we gave explicit deterministic conditions in various settings; by focusing directly on this goal,~\cite{SteinhardtCV18} somewhat relaxed 
some of these assumptions. Meister and Valiant~\cite{meister2017data} studied learning in a crowdsourcing model, 
where the fraction of honest workers may be very small (similar to~\cite{CharikarSV16}). Qiao and Valiant~\cite{QiaoV18} considered 
robust estimation of discrete distributions in a setting where we have several sources (a fraction of which are adversarial) who each provide a batch of samples.
A number of simultaneous works~\cite{KothariSS18, HopkinsL18, DiakonikolasKS18} investigated robust mean estimation 
in even more general settings, and apply their techniques to learning mixtures of spherical 
Gaussians under minimal separation conditions.
Finally, several concurrent results study robustness in supervised learning tasks~\cite{PSBR18, klivans2018efficient, DiakonikolasKKLSS18}, including regression and SVM problems.
%Some of these works operate on the insight that finding a point which minimizes the sum of gradients can be viewed as a mean-estimation problem in the gradient space.
Despite all of this rapid progress, there are still many interesting theoretical and practical questions left to explore. 

\subsection{Organization}
The structure of this paper is as follows:
In Section \ref{sec:preliminaries}, we introduce basic notation and a number of useful facts that will be required throughout the paper, 
as well as the formal definition of our adversary model. 
In Section \ref{sec:dontwork}, we discuss several natural approaches to high-dimensional agnostic learning, 
all of which lose polynomial factors that depend on the dimension, in terms of their error guarantee.

The main body of the paper is in Sections~\ref{sec:sepGaussian}--\ref{sec:filter-mixtures}.
Sections~\ref{sec:sepGaussian} and~\ref{sec:sepGMM} illustrate our convex programming framework,
while Sections~\ref{sec:filterGaussian}, \ref{sec:filterProduct}, and \ref{sec:filter-mixtures} illustrate our filter framework.
More specifically, in Sections~\ref{sec:sepGaussian} and \ref{sec:filterGaussian}, we analyze the setting of a single Gaussian with unknown mean and unknown covariance, using our convex programming and filter frameworks, respectively.
In Section \ref{sec:sepGMM}, we generalize the convex programming method to obtain an agnostic algorithm for mixtures of spherical Gaussians with unknown means.
In Section \ref{sec:filterProduct}, we apply our filter techniques to a binary product distribution, and generalize these in Section~\ref{sec:filter-mixtures} to obtain 
an agnostic learning algorithm for a mixture of two balanced binary product distributions. 
%This last example serves to better explain the parallel between our two sets of techniques, and how they relate to each other. 

We note that for some of the more advanced applications of our frameworks, the technical details can get in the way of the fundamental ideas.
For the reader who is interested in seeing the details of our most basic application of the convex programming framework, we recommend reading the case a Gaussian with unknown mean, in Section~\ref{sec:UnknownMeanConvex}.
Similarly, for the filter framework, we suggest either the Gaussian with unknown mean in Section~\ref{sec:filter-gaussian-mean} or the balanced product distribution in Section~\ref{sec:bal-prod}.

%!TEX root = ./main.tex

\section{Preliminaries}
\label{sec:preliminaries}

\subsection{Basic Notation}
Throughout this paper, if $v$ is a vector, we will let $\| v \|_2$ denote its Euclidean norm.
If $M$ is a matrix, we will let $\| M \|_2$ denote its spectral norm, and $\| M \|_F$ denote its Frobenius norm.
We will also let $\preceq$ and $\succeq$ denote the PSD ordering on matrices.
For a discrete distribution $P$, we will denote by $P(x)$ the probability mass at point $x$. For a continuous distribution, let it denote the probability density function at $x$.
Let $S$ be a multiset over $\{0,1\}^d$ .
We will write $X \in_u S$ to denote that $X$ is drawn from the empirical distribution defined by $S$.
Throughout the paper, we let $\otimes$ denote the Kronecker product of matrices.

As a measure of distance between distributions, we will use the notion of total variation distance:
\begin{definition}
Let $P, Q$ be two probability distributions on $\mathbb{R}^d$.
Then the \emph{total variation distance} between $P$ and $Q$, denoted $\dtv(P, Q)$, is defined as
\[
\dtv (P, Q) = \sup_{A \subseteq \mathbb{R}^d} |P(A) - Q(A)| \;.
\]
\end{definition}

\subsection{Types of Adversaries}

In this paper, we will consider a powerful model for agnostic distribution learning
that generalizes many other existing models. The standard setup involves an {\em oblivious adversary}
who chooses a distribution that is close in total variation distance to an unknown distribution in some class $\mathcal{D}$.

\begin{definition}
Given $\ve > 0$ and a class of distributions $\mathcal{D}$,
the \emph{oblivious adversary} chooses a distribution $P$
such that there is an unknown distribution $D \in \mathcal{D}$ with $\dtv (P, D) \leq \ve$.
An algorithm is then given $m$ independent samples $X_1, X_2, \ldots, X_m$ from $P$.
\end{definition}

The goal of the algorithm is to return the parameters of a distribution $\widehat{D}$ in $\mathcal{D}$,
where $\dtv (D, \widehat{D})$ is small.
%In general, the best that we can hope for even from an information theoretic standpoint is that $\dtv(D, \widehat{D}) \leq 2\ve$ because the adversary could choose a distribution $P \in \mathcal{D}$ where there is more than one choice of a distribution $D, D' \in \mathcal{D}$ with $\dtv(P, D) = \dtv(P, D') = \ve$ and there is no basis to choose between $D$ and $D'$.
We refer to the above adversary as oblivious
because it fixes the model for noise before seeing any of the samples.
In contrast, a more powerful adversary is allowed to inspect the samples before corrupting them,
both by adding corrupted points and deleting uncorrupted points.
%, however only in a limited fashion: it cannot delete any samples from the true distribution $D$.
We refer to this as the {\em full adversary}:

\begin{definition}
Given $\ve > 0$, and a class of distributions $\mathcal{D}$, the \emph{full adversary} operates as follows:
The algorithm specifies some number of samples $m$.
The adversary generates $m$ samples $X_1, X_2, \ldots, X_m$ from some (unknown) distribution $D \in \mathcal{D}$.
It then draws $m'$ from an appropriate distribution.
This distribution is allowed to depend on $X_1, X_2, \ldots, X_m$,
but when marginalized over the $m$ samples satisfies $m' \sim \mbox{Bin} (m, \ve)$.
The adversary is allowed to inspect the samples, removes $m'$ of them, and replaces them with arbitrary points.
The set of $m$ points is given (in any order) to the algorithm.
\end{definition}

%\begin{definition}
%Given $\ve > 0$ and a class of distributions $\mathcal{D}$, the \emph{additive adversary} operates as follows.
%The algorithm specifies some number of samples $m$.
%The adversary generates $m' \sim \mbox{Bin} (\ve, m)$.
%It then draws $m - m'$ independent samples from some (unknown) $D \in \mathcal{D}$.
%It is then allowed to inspect them, and chooses an additional $m'$ points arbitrarily.
%The aggregate set of $m$ points is given (in any order) to the algorithm.
%\end{definition}

We remark that there are no computational restrictions on the adversary.
As before, the goal is to return the parameters of a distribution $\widehat{D}$ in $\mathcal{D}$,
where $\dtv (D, \widehat{D})$ is small. The reason we allow the draw $m'$ to depend
on the samples $X_1, X_2, \ldots, X_m$ is because our algorithms will tolerate this extra generality,
and it will allow us to show that the full adversary is at least as strong as the oblivious adversary
(this would not necessarily be true if $m'$ were sampled independently from $ \mbox{Bin} (m, \ve)$).

We rely on the following well-known fact:

\begin{fact}
\label{fact:coupling}
Let $P, D$ be two distributions such that $\dtv (P, D) = \ve$.
Then there are distributions $N_1$ and $N_2$ such that
$(1-\ve_1)P + \ve_1 N_1 = (1 - \ve_2) D + \ve_2 N_2$, where $\ve_1 + \ve_2 = \ve$.
\end{fact}

Now we can describe how the full adversary can corrupt samples from $D$
to get samples distributed according to $P$.

\begin{claim}
\label{claim:full-to-oblivious}
The full adversary can simulate any oblivious adversary.
\end{claim}

\begin{proof}
We draw $m$ samples $X_1, X_2, \ldots, X_m$ from $D$.
We delete each sample $X_i$ independently with probability $\ve_2$
and replace it with an independent sample from $N_2$.
This gives a set of samples $Y_1, Y_2, \ldots, Y_m$ that are independently
sampled from $(1 - \ve_2) D + \ve_2 N_2$. Since the distributions
$(1-\ve_1)P + \ve_1 N_1$ and $(1 - \ve_2) D + \ve_2 N_2$ are identical,
we can couple them to independent samples $Z_1, Z_2, \ldots, Z_m$
from $(1-\ve_1)P + \ve_1 N_1$. Now each sample $Z_i$ that came from $N_1$,
we can delete and replace with an independent sample from $P$. The result is
a set of samples that are independently sampled from $P$ where we have
made $m'$ edits and marginally $m' \sim \mbox{Bin} (m, \ve_1 + \ve_2)$,
although $m'$ has and needs to have some dependence on the original samples from $D$.
\end{proof}

The challenge in working with the full adversary is that even the samples
that came from $D$ can have biases. The adversary can now choose how
to remove uncorrupted points in a careful way so as to compensate
for certain other biases that he introduces using the corrupted points.

\noindent
Throughout this paper, we will make use of the following notation and terminology:

\begin{definition}
We say a set of samples $X_1, X_2, \ldots, X_m$ is an \emph{$\eps$-corrupted} set of samples
generated by the oblivious (resp. full) adversary if it is generated by the process
described above in the definition of the oblivious (resp. full) adversary.
If it was generated by the full adversary, we let $\Sgood \subseteq [m]$ denote the indices
of the uncorrupted samples, and we let $\Sbad \subseteq [m]$ denote the indices of the corrupted samples.
\end{definition}

%\paragraph{The Number of Corrupted Points}

In this paper, we will give a number of algorithms for agnostic distribution learning that work in the full adversary model.
In our analysis, we will identify a set of events that ensure the algorithm succeeds
and will bound the probability that any of these events does not occur when $m$ is suitably large.
We will often explicitly invoke the assumption that $|\Sbad| \leq 2 \ve m$.
We can do this even though the number of points that are corrupted is itself a random variable,
because by the Chernoff bound, as long as $m \geq O \left( \frac{\log 1 / \tau}{\ve} \right)$,
we know that  $|\Sbad| \leq 2 \ve m$ holds with probability at least $1 - O(\tau)$.
Thus, making the assumption that $|\Sbad| \leq 2 \ve m$ costs us an additional additive
$O(\tau)$ term in our union bound, when bounding the failure probability of our algorithms.

%\noindent Of particular interest to our work will be the Schatten top-2 norm, $\|\cdot\|_{T_2}$.
%The intuition behind why this norm will be useful is that if we are given samples
%$X_1, \ldots, X_m$ from a Gaussian distribution $\normal (\mu, I)$
%whose covariance $I$ is known, then
%$$ \sum_{i =1}^m X_i X_i^T - I $$
%will have its top eigenvector close to $\mu \mu^T$
%with eigenvalue $\|\mu\|_2^2 - 1$, and its remaining eigenvalues would be small.

\subsection{Distributions of Interest}
One object of study in this paper is the Gaussian (or Normal) distribution.
\begin{definition}
A \emph{Gaussian distribution} $\mathcal{N}(\m,\S)$ with mean $\m$ and covariance $\S$ is the distribution with probability density function
$$f(x) = (2\p)^{-d/2}|\S|^{-1/2}\exp\left(-\frac12 (x - \m)^T\S^{-1}(x - \m)\right).$$
\end{definition}

We will also be interested in binary product distributions.
\begin{definition}
A \emph{(binary) product distribution} is a probability distribution over $\{0, 1\}^d$
whose coordinate random variables are independent.
Note that a binary product distribution is completely determined by its mean vector.
\end{definition}

We will also be interested in mixtures of such distributions.
\begin{definition}
A \emph{mixture} $P$ of distributions $P_1, \dots, P_k$ with mixing weights $\a_1, \dots, \a_k$ is the distribution defined by
$$P(x) = \sum_{j \in [k]} \a_j P_k(x),$$
where $\a_j \geq 0$ for all $j$ and $\sum_{j \in [k]} \a_j = 1$.
\end{definition}

\subsection{Bounds on TV Distance}
The \emph{Kullback-Liebler divergence} (also known as \emph{relative entropy},
\emph{information gain}, or \emph{information divergence})
is a well-known measure of distance between two distributions.
\begin{definition}
Let $P, Q$ be two probability distributions on $\mathbb{R}^d$.
Then the \emph{KL divergence} between $P$ and $Q$, denoted $\dkl(P \| Q)$, is defined as
\[
\dkl (P \| Q) = \int_{\mathbb{R}^d} \log \frac{d P}{d Q} dP \;.
\]
\end{definition}
The primary interest we have in this quantity is the fact that
(1) the KL divergence between two Gaussians has a closed form expression,
and (2) it can be related (often with little loss) to the total variation distance between the Gaussians.
The first statement is expressed in the fact below:
\begin{fact}
Let $\normal (\mu_1, \Sigma_1)$ and $\normal(\mu_2, \Sigma_2)$ be two Gaussians such that $\det(\Sigma_1), \det(\Sigma_2) \neq 0$.
Then
\begin{equation}
\label{eq:kl}
\dkl\left(\mathcal{N}(\mu_1,\S_1) \| \mathcal{N}(\mu_2,\S_2) \right) =
\frac{1}{2}\left(\tr(\S_2^{-1}\S_1) + (\m_2 - \m_1)^T\S_2^{-1}(\m_2 - \m_1) - d - \ln\left(\frac{\det (\S_1)}{\det (\S_2)}\right)\right).
\end{equation}
\end{fact}
The second statement is encapsulated in the well-known \emph{Pinsker's inequality}:
\begin{theorem}[Pinsker's inequality]
\label{thm:pinskers}
Let $P, Q$ be two probability distributions over $\R^d$.
Then $$\dtv (P, Q) \leq \sqrt{\frac{1}{2} \dkl (P \| Q)} \;.$$
\end{theorem}
With this we can show the following two useful lemmas,
which allow us to relate parameter distance between two Gaussians to their total variation distance.
The first bounds the total variation distance between two Gaussians
with identity covariance in terms of the Euclidean distance between the means:
\begin{corollary}
\label{cor:kl-to-means}
Let $\mu_1, \mu_2 \in \R^d$ be arbitrary.
Then
$\dtv\left(\normalpdf(\m_1, I), \normalpdf(\m_2, I)\right) \leq \frac{1}{\sqrt{2}} \|\m_2 - \m_1\|_2$.
\end{corollary}
\begin{proof}
In the case where $\S_1 = \S_2 = I$, (\ref{eq:kl}) simplifies to
\begin{equation*}
\dkl\left(\normalpdf(\m_1, I) \| \normalpdf(\m_2, I)\right) = \frac{1}{2} \|\m_2 - \m_1\|_2^2.
\end{equation*}
Pinsker's inequality (Theorem \ref{thm:pinskers}) then implies that
\begin{equation*}
\dtv\left(\normalpdf(\m_1, I), \normalpdf(\m_2, I)\right) \leq \sqrt{\frac12\dkl\left(\normalpdf(\m_1, I) \| \normalpdf(\m_2, I)\right)} = \frac{1}{\sqrt{2}}\|\m_2 - \m_1\|_2,
\end{equation*}
as desired.
\end{proof}
The second bounds the total variation distance between two
mean zero Gaussians in terms of the Frobenius norm of the difference
between their covariance matrices:
\begin{corollary}
\label{cor:kl-to-cov}
Let $\delta > 0$ be sufficiently small.
Let $\Sigma_1, \Sigma_2$ such that $\| I - \Sigma_2^{-1/2} \Sigma_1 \Sigma_2^{-1/2} \|_F = \delta$.
Then,
\[
\dtv (\normal (0, \Sigma_1) || \normal (0, \Sigma_2)) \leq O(\delta) \; .
\]
\end{corollary}
\begin{proof}
Let $M = \Sigma_2^{-1/2} \Sigma_1 \Sigma_2^{-1/2}$.
Then (\ref{eq:kl}) simplifies to
\[
\dkl \left(\mathcal{N}(\mu_1,\S_1)\| \mathcal{N}(\mu_2,\S_2) \right) = \frac{1}{2} \left( \tr (M) - d - \ln \det(M) \right) \;.
\]
Since both terms in the last line are rotationally invariant,
we may assume without loss of generality that $M$ is diagonal.
Let $M = \mbox{diag} (1 + \lambda_1, \ldots, 1 + \lambda_d)$.
Thus, the KL divergence between the two distributions is given exactly by $\frac{1}{2} \sum_{i = 1}^d \left( \lambda_i - \log (1 + \lambda_i) \right),$
where we are guaranteed that $( \sum_{i = 1}^d \lambda_i^2 )^{1/2} = \delta$.
By the second order Taylor approximation to $\ln (1 + x)$, for $x$ small, we have that for $\delta$ sufficiently small,
\[\sum_{i = 1}^d \lambda_i - \log (1 + \lambda_i)  = \Theta \left( \sum_{i = 1}^d \lambda_i^2 \right) = \Theta (\delta^2) \; .\]
Thus, we have shown that for $\delta$ sufficiently small, $\dkl \left(\mathcal{N}(\mu_1,\S_1)\| \mathcal{N}(\mu_2,\S_2) \right) \leq O(\delta^2)$.
The result now follows by an application of Pinsker's inequality (Theorem \ref{thm:pinskers}).
\end{proof}

Our algorithm for agnostically learning an arbitrary Gaussian
will be based on solving two intermediate problems:
(1) We are given samples from $\normal(\mu, I)$ and our goal is to learn $\mu$,
and (2) We are given samples from $\normal(0, \Sigma)$ and our goal is to learn $\Sigma$.
The above bounds on total variation distance will allow us to conclude that our estimate is close in total variation
distance to the unknown Gaussian distribution in each of the two settings.

We note the following folklore sample complexity bounds for learning a Gaussian in the non-agnostic setting.
% The upper bounds can be derived by combining Lemma~\ref{lem:vershynin} and Corollary~\ref{cor:frob-conv} with Corollaries~\ref{cor:kl-to-means} and \ref{cor:kl-to-cov}.
\begin{theorem}
$N = \Theta\left(\frac{d + \log(1/\t)}{\ve^2}\right)$ samples are both necessary and sufficient to learn a $d$-dimensional Gaussian with unknown mean and known covariance to total variation distance $\ve$ with probability $1-\t$.
\end{theorem}
\begin{theorem}
$N = \Theta\left(\frac{d^2 + \log(1/\t)}{\ve^2}\right)$ samples are both necessary and sufficient to learn a $d$-dimensional Gaussian with unknown mean and covariance to total variation distance $\ve$ with probability $1-\t$.
\end{theorem}

We will also need the following lemma
bounding the total variation distance between two product distributions:
\begin{lemma} \label{lem:prod-dtv-chi2}
Let $P, Q$ be binary product distributions with mean vectors $p, q \in (0, 1)^d.$
We have that
$$\dtv^2(P, Q) \leq 2 \sum_{i=1}^d \frac{(p_i-q_i)^2}{(p_i+q_i)(2-p_i-q_i)} \;.$$
\end{lemma}
\begin{proof}
We include the simple proof for completeness.
By Kraft's inequality (see e.g., Theorem 5.2.1 in \cite{CoverThomas:06}), for any pair of distributions,
we have that $\dtv^2(P, Q) \leq 2 H^2(P, Q),$
where $H(P, Q)$ denotes the Hellinger distance between $P, Q.$
Since $P, Q$ are product measures, we have that
$$1 - H^2(P, Q) = \prod_{i=1}^d (1 - H^2(P_i, Q_i)) = \prod_{i=1}^d (\sqrt{p_i q_i} + \sqrt{(1-p_i)(1-q_i)}) \;.$$
The elementary inequality $2\sqrt{ab} = a+b - (\sqrt{a}-\sqrt{b})^2,$ $a, b>0,$ gives that
$$\sqrt{p_i q_i} + \sqrt{(1-p_i)(1-q_i)} \geq 1-  \frac{(p_i-q_i)^2}{(p_i+q_i)(2-p_i-q_i)} \;.$$
Let 
\[
z_i = \frac{(p_i-q_i)^2}{(p_i+q_i)(2-p_i-q_i)} \; .
\]
We have
$$ \dtv^2(P, Q) \leq 2 \cdot (1- \prod_{i=1}^d (1-z_i)) \leq 2 \sum_{i=1}^d z_i \;,$$
where the last inequality follows from the union bound.
\end{proof}

\subsection{Additional Concentration Lemmata}
In this section, we list a number of standard concentration inequalities 
for nice random variables which we will frequently use throughout this paper.
The proofs of these results are standard and omitted, see e.g.,~\cite{Vershynin} 
for a more thorough treatment of these results.

The first is a Chernoff bound for bounded random variables.
\begin{theorem}
Let $Z_1,\ldots,Z_d$ be independent random variables with $Z_i$ supported on $[a_i,b_i].$
Let $Z=\sum_{i=1}^d Z_i$. Then for any $T>0$,
$$
\Pr(|Z-\E[Z]|>T) \leq 2 \exp\left( \frac{-2T^2}{\sum_{i=1}^d (b_i-a_i)^2}\right).
$$
\end{theorem}
\noindent
We will also require the following tail bounds for Gaussians and quadratic forms of Gaussians:
\begin{lemma}
\label{lem:mean-chernoff}
Let $n$ be a positive integer.
Let $D$ be a sub-gaussian distribution with mean $0$ and covariance $I$.
Let $Y_i \sim D$ be independent, for $i = 1, \ldots, n$.
Let $v \in \R^d$ be an arbitrary unit vector.
Then, there exist a universal constant $B > 0$ so that for all $T > 0$, we have
\[
\Pr \Brac{ \left|\frac{1}{n} \sum_{i = 1}^n \langle v, Y_i \rangle \right|  > T } \leq 4 \exp \left( -B n T^2 \right) \; .
\]
\end{lemma}

\begin{lemma}[Hanson-Wright]
\label{lem:hansonwright}
Let $n$ be a positive integer.
Let $D$ be a sub-gaussian distribution with mean $0$ and covariance $\Sigma \preceq I$.
Let $Y_i \sim D$ be independent, for $i = 1, \ldots, n$.
Let $U \in \R^{d \times d}$ satisfy $U \succeq 0$ and $\| U \|_F = 1$.
Then, there exists a universal constant $B > 0$ so that for all $T > 0$, we have
\[
\Pr \Brac{ \left| \frac{1}{n} \sum_{i = 1}^n \tr (X_i X_i^\top U) - \tr (U) \right|  > T } \leq 4 \exp \left( -B n \min (T, T^2) \right) \; .
\]
\end{lemma}

\noindent
By standard union bound arguments (see e.g.~\cite{Vershynin}), we obtain the following concentration results for the empirical mean and covariance of a set of Gaussian vectors:

\begin{lemma}
\label{lem:mean1}
Let $n$ be a positive integer.
Let $D$ be a sub-gaussian distribution with mean $0$ and covariance $I$.
Let $Y_i \sim D$ be independent, for $i = 1, \ldots, n$.
Then, there exist universal constants $A, B > 0$ so that for all $t > 0$, we have
\[
\Pr \Brac{\left\| \frac{1}{n} \sum_{i = 1}^n Y_i \right\|_2 > t} \leq 4 \exp \left( A d - B n t^2 \right) \; .
\]
\end{lemma}

\begin{lemma}
\label{lem:vershynin}
With the same setup as in Lemma~\ref{lem:mean1}, there exist universal constants $A, B > 0$ so that for all $t > 0$, we have
\[
\Pr \Brac{\left\|  \frac{1}{n} \sum_{i = 1}^n Y_i Y_i^\top - I \right\|_2 > t} \leq 4 \exp \left( A d - B n \min (t, t^2) \right) \; .
\]
\end{lemma}

\subsection{Agnostic Hypothesis Selection}

Several of our algorithms will return a polynomial-sized list of hypotheses at least one of which is guaranteed to be close to the target distribution. 
Usually (e.g., in a non-agnostic setting), one could use a polynomial number of additional samples to run a tournament to identify the candidate hypothesis that is (roughly) the closest to the target distribution. In the discussion that follows, we will refer to these additional samples as test samples. 
Such {\em hypothesis selection} algorithms have been extensively studied \cite{Yatracos85, DL96, DL97, DL:01, DK14, AcharyaJOS14, SOAJ14, DDS15-journal, DDS15}.
Unfortunately, against a strong adversary we run into a serious technical complication: the training samples and test samples are not necessarily independent. Moreover even if we randomly partition our samples in training and test, a priori there are an unbounded set of possible hypotheses that the training phase could output, and when we analyze the tournament we cannot condition on the list of hypotheses and assume that the test samples are sampled anew. Our approach is to require our original algorithm to return only hypotheses from some finite set of possibilities, and as we will see this mitigates the problem. 

\begin{lemma}\label{tournamentLem}
Let $\mathcal{C}$ be a class of probability distributions. Suppose that for some $N,\eps,\tau>0$ there exists a polynomial time algorithm that given $N$ independent samples from some $\Pi\in\mathcal{C}$, of which up to a $2\epsilon$-fraction have been arbitrarily corrupted, returns a list $\mathcal{L}$ of $M$ distributions whose probability density functions are explicitly computable and which can be effectively sampled from such that with $1-\tau/2$ probability there exists a $\Pi'\in \mathcal{L}$ with $\dtv(\Pi',\Pi)<\delta$. Suppose furthermore that the distributions returned by this algorithm are all in some fixed set $\mathcal{M}$. Then there exists another polynomial time algorithm, which given $O(N+(\log(|\mathcal{M}|)+\log(1/\tau))/\epsilon^2)$ samples from $\Pi$, an $\eps$-fraction of which have been arbitrarily corrupted, returns a single distribution $\Pi'$ such that with $1-\tau$ probability $\dtv(\Pi',\Pi)<O(\delta+\eps)$.
\end{lemma}

\begin{remark}
As a simple corollary of the agnostic tournament, observe that this allows us to do agnostic learning without knowing the precise error rate $\eps$.
Throughout the paper, we assume the algorithm knows $\eps$, and guarantees that the output will have error which is at most $O(f(\eps))$. 
However, if the algorithm is not given this information, and instead given an $\eta$ and asked to return something with error at most $O(f(\eps + \eta))$, we may simply grid over $\{ \eta, (1 + \gamma) \eta, (1 + \gamma)^2 \eta, \ldots, 1 \}$ (here $\gamma$ is some arbitrary constant that governs a tradeoff between runtime and accuracy), run our algorithm with $\eps$ set to each element in this set, and perform hypothesis selection via \textsc{Tournament}.
Then it is not hard to see that we are guaranteed to output something which has error at most $O(f(\eps + (1 + \gamma) \eta))$.
\end{remark}

\begin{proof}
Firstly, we randomly choose a subset of $N$ of our samples and a disjoint subset of $C(\log(|\mathcal{M}|)+\log(1/\tau))/\eps^2$ of our samples for some sufficiently large $C$. Note that with high probability over our randomization, at most a $2\epsilon$-fraction of samples from each subset are corrupted. Thus, we may instead consider the stronger adversary who sees a set $S_1$ of $N$ independent samples from $\Pi$ and another set, $S_2$, of $C(\log(|\mathcal{M}|)+\log(1/\tau))/\eps^2$ samples from $\Pi$ and can arbitrary corrupt a $2\eps$-fraction of each, giving sets $S'_1$,$S'_2$.

With probability at least $1-\tau/2$ over $S_1$, the original algorithm run on $S'_1$ returns a set $\mathcal{L}$ satisfying the desired properties.

For two distributions $P$ and $Q$ in $\mathcal{M}$ we let $A_{PQ}$ be the set of inputs $x$ where $\Pr_P(x)>\Pr_Q(x)$. We note that we can test membership in $A_{PQ}$ as, by assumption, the probability density functions are computable. We also note that $\dtv(P,Q)=\Pr_P(A_{PQ})-\Pr_Q(A_{PQ})$. Our tournament will depend on the fact that if $P$ is close to the target and $Q$ is far away, that many samples will necessarily lie in $A_{PQ}$.

We claim that with probability at least $1-\tau/2$ over the choice of $S_2$, we have for any $P,Q\in \mathcal{M}$:
$$
\Pr_{x\in_u S_2}(x\in A_{PQ}) = \Pr_{x\sim \Pi}(x\in A_{PQ})+O(\eps).
$$
This follows by Chernoff bounds and a union bound over the $|\mathcal{M}|^2$ possibilities for $P$ and $Q$. Since the total variation distance between the uniform distributions over $S_2$ and $S'_2$ is at most $2\eps$, we also have for $S'_2$ that
$$
\Pr_{x\in_u S'_2}(x\in A_{PQ}) = \Pr_{x\sim \Pi}(x\in A_{PQ})+O(\eps).
$$

Suppose that $\dtv(P,\Pi)<\delta$ and $\dtv(Q,\Pi)>5\delta+C\eps$. We then have that
$$
\Pr_{x\in_u S_2'}(x\in A_{PQ}) = \Pr_{x\sim \Pi}(x\in A_{PQ})+O(\eps) \geq \Pr_{x\sim P}(x\in A_{PQ})+O(\eps)-\delta \geq \Pr_{x\sim Q}(x\in A_{PQ})+\delta+C\eps/5.
$$
On the other hand, if $\dtv(\Pi,Q)<\delta$ then
$$
\Pr_{x\in_u S_2'}(x\in A_{PQ})= \Pr_{x\sim \Pi}(x\in A_{PQ})+O(\eps) < \Pr_{x\sim Q}(x\in A_{PQ})+\delta+C\eps/5.
$$
Therefore, if we throw away any $Q$ in our list for which there is a $P$ in our list such that
$$
\Pr_{x\in_u S_2'}(x\in A_{PQ}) \geq \Pr_{x\sim Q}(x\in A_{PQ})+\delta+C\eps/5,
$$
we have thrown away all the $Q$ with $\dtv(Q,\Pi)>5\delta+C\eps$, but none of the $Q$ with $\dtv(Q,\Pi)<\delta$. Therefore, there will be a $Q$ remaining, and returning it will yield an appropriate $\Pi'$.
\end{proof}

%!TEX root = ./main.tex

\section{Some Natural Approaches, and Why They Fail}
\label{sec:dontwork}
Many of the agnostic distribution learning problems that we study are so natural that one would immediately wonder why simpler approaches do not work. Here we detail some other plausible approaches, and what causes them to lose dimension-dependent factors (if they have any guarantees at all!). For the discussion that follows, we note that by Fact~\ref{cor:kl-to-means} in order to achieve an estimate that is $O(\ve)$-close in total variation distance (for a Gaussian when $\mu$ is unknown and $\Sigma =I$) it is necessary and sufficient that $\|\hat{\mu} - \mu\|_2 = O(\ve)$. 

\subsection*{Learn Each Coordinate Separately}
One plausible approach for robust mean estimation in high dimensions 
is to agnostically learn along each coordinate separately.
For instance, if our goal is to agnostically learn the mean of a Gaussian with known covariance $I$, 
we could try to learn each coordinate of the mean separately.
But since an $\ve$-fraction of the samples are corrupted, 
our estimate can be off by $\ve$ in each coordinate and would be off by $\ve \sqrt{d}$ in high dimensions. 

\subsection*{Maximum Likelihood}
Given a set of samples $X_1, \ldots, X_N$, and a class of distributions $\mathcal{D}$, 
the maximum likelihood estimator (MLE) is the distribution $ F \in \mathcal{D}$ 
that maximizes $\prod_{i = 1}^N F(X_i)$.
Equivalently, $F$ minimizes the negative log likelihood (NLL), which is given by
\[
\textrm{NLL}(F, X_1, \ldots, X_N) = -\sum_{i = 1}^N \log F(X_i) \; .
\]
In particular, if $\mathcal{D} = \{\normal (\mu, I) : \mu \in \R^d \}$ is the set of Gaussians with unknown mean and identity covariance, we see that for any $\mu \in \R^{d}$, the NLL of the set of samples is given by
\begin{align*}
\textrm{NLL} (\normal (\mu, I),  X_1, \ldots, X_N) &= - \sum_{i = 1}^N \log \left( \frac{1}{\sqrt{2 \pi}} e^{- \| X_i - \mu \|_2^2 / 2} \right) \\
&= N \log {\sqrt{2 \pi}} + \frac{1}{2} \sum_{i = 1}^N \| X_i - \mu \|_2^2 \; ,
\end{align*}
and so the $\mu$ which minimizes $\textrm{NLL} (\normal (\mu, I),  X_1, \ldots, X_N)$ is the mean of the samples $X_i$, since for any set of vectors $v_1, \ldots, v_N$, the average of the $v_i$'s is the minimizer of the function $h(x) = \sum_{i = 1}^N \| v_i - x \|_2^2$.
Hence, if an adversary places an $\ve$-fraction of the points at some very large distance, 
then the estimate for the mean would need to move considerably in that direction. By placing the corruptions 
further and further away, the MLE can be an arbitrarily bad estimate. That is, even though it is well known 
\cite{Huber:67, White:82} that the MLE converges to the distribution $F \in \mathcal{D}$ that is closest 
in KL-divergence to the distribution from which our samples were generated (i.e., after the adversary has added corruptions), 
$F$ is not necessarily close to the uncorrupted distribution.

\subsection*{Geometric Median}
In one dimension, it is well-known that the median provides a provably robust estimate for the mean in a number of settings.
The mean of a set of points $a_1, \ldots, a_N$ is the minimizer of the function $f(x) = \sum_{i = 1}^N (a_i - x)^2$, and in contrast the median is the minimizer of the function $f(x) = \sum_{i = 1}^N  |a_i - x|$.
In higher dimensions, there are many natural definitions for the median that generalize the one-dimensional case.
The \emph{Tukey median} is one such notion, but as we discussed it is hard to compute \cite{JP:78} and the best known algorithms run in time exponential in $d$. 
Motivated by this, the geometric median is another high-dimensional notion of a median.
It often achieves better robustness than the mean, and can be computed quickly \cite{CohenLMPS16}.
The formal definition is:
\[
\textrm{geomed}(S) \triangleq \min_v \sum_{x \in S} \| x - v \|_2 \; .
\]
\noindent Unfortunately, this notion of median still incurs an error containing a factor of $O(\sqrt{d})$:
\begin{proposition}[Proposition 2.1 of~\cite{LaiRV16}]
Given a set $S$ of $N = \Omega\left(\frac{d + \log (1/\tau)}{\eps^2}\right)$ samples from $\normal (0, I)$, then with probability at least $1 - \tau$, there exists a corruption $S'$ of $S$, such that:
\[
\mathrm{geomed}(S') = \Omega(\eps \sqrt{d}).
\]
%This implies that $\dtv(\normal (0, I), \normal(\mathrm{geomed}(S'), I)) = \Omega(\eps \sqrt{d})$.
\end{proposition}

%!TEX root = ./main.tex

\section{Agnostically Learning a Gaussian, via Convex Programming}
\label{sec:sepGaussian}
In this section we give a polynomial time algorithm to agnostically learn a single Gaussian up to error $\Otilde (\ve)$.
Our approach is based on the following ingredients:
First, in Section \ref{sec:sne}, we define the set $S_{N, \ve}$, which will be a key algorithmic object in our framework.
In Section \ref{sec:GaussianConc} we give key, new concentration bounds on certain statistics of Gaussians.  We will make crucial use of these concentration bounds throughout this section.
In Section \ref{sec:UnknownMeanConvex} we give an algorithm to agnostically learn a Gaussian with unknown mean and whose covariance is promised to be the identity via convex programming.
This will be an important subroutine in our overall algorithm, and it also helps to illustrate our algorithmic approach without many of the additional complications that arise in our later applications.
In Section \ref{sec:UnknownCovarianceConvex} we show how to robustly learn a Gaussian with mean zero and unknown covariance again via convex programming.
Finally, in Section \ref{sec:UnknownGaussianWhole} we show how to combine these two intermediate results to get our overall algorithm. 

\subsection{The Set $S_{N, \ve}$}
\label{sec:sne}
An important algorithmic object for us will be the following set:
\begin{definition}
For any $\frac{1}{2} > \ve > 0$ and any integer $N$, let
\[S_{N, \ve} = \left\{(w_1, \ldots, w_N) : \sum_{i = 1}^N w_i = 1, \mbox{ and } 0 \leq w_i \leq \frac{1}{(1 - 2\ve) N}, \forall i \right\} \;.\]
\end{definition}
Next, we motivate this definition.
For any $J \subseteq [N]$, let $w^J \in \R^N$ be the vector
which is given by $w^J_i = \frac{1}{|J|}$ for $i \in J$ and $w_i^J = 0$ otherwise.
Then, observe that
\[
S_{N, \ve} =  \mbox{conv} \left\{ w^J: |J| = (1 - 2\ve) N) \right\} \;,
\]
and so we see that this set is designed to capture the notion of selecting a set of $(1 - 2\ve) N$ samples from $N$ samples.

Given $w \in S_{N, \ve}$ we will use the following notation
$$w_{g} = \sum_{i \in \Sgood} w_i \mbox{ and } w_{b} = \sum_{i \in \Sbad} w_i$$
to denote the total weight on good and bad points respectively.
The following facts are immediate from $|\Sbad| \leq 2 \ve N$ and the properties of $S_{N, \ve}$.

\begin{fact}\label{fact:weights}
If $w \in S_{N, \ve}$ and $|\Sbad| \leq 2 \ve N$, then $w_b \leq \frac{2\ve}{1-2\ve}$.
Moreover, the renormalized weights $w'$ on good points given
by $w'_i = \frac{w_i}{w_g}$ for all $i \in \Sgood$, and $w'_i = 0$ otherwise,
satisfy $w' \in S_{N, 4 \ve}$.
\end{fact}

%!TEX root = ./main.tex

\subsection{Concentration Inequalities}
\label{sec:GaussianConc}
Throughout this section and in Section \ref{sec:sepGMM} we will make use of various concentration bounds on low moments of Gaussian random variables. 
Some are well-known, and others are new but follow from known bounds and appropriate union bound arguments. 

\subsubsection{Empirical Estimates of First and Second Moments of Large Subsets}
We will also be interested in how well various statistics of Gaussians concentrate around their expectation, when we take the worst-case set of weights in $S_{N, \ve}$. 
This is more subtle than standard settings such as Lemma~\ref{lem:mean1} or Lemma~\ref{lem:vershynin} because as we take more samples, any fixed statistic (e.g. taking the uniform distribution over the samples) concentrates better but the size of $S_{N, \ve}$ (e.g. the number of sets of $(1-2\ve) N$ samples) grows too. 
We defer the proofs to Appendix \ref{sec:concAppendix}.
The first concerns the behavior of the empirical covariance:
\begin{lemma}
\label{lem:union-bound}
Fix $\epsilon \leq 1/2$ and $ \tau \leq 1$. There is a $\delta_1 =  O(\epsilon \log 1 / \epsilon)$ such that if
$Y_1, \ldots, Y_N$ are independent samples from $\normal (0, I)$ and $N = \Omega \left( \frac{d + \log (1 / \tau)}{\delta_1^2} \right)$, then
\begin{equation}
\label{eqn:cov-conv}
\Pr \left[\exists w \in S_{N, \ve} : \left \| \sum_{i = 1}^N w_i Y_i Y_i^T - I \right\|_2 \geq \delta_1 \right] \leq \tau \; .
\end{equation}
\end{lemma}
\noindent
A nearly identical argument (Using Hoeffding instead of Bernstein in the proof of Theorem 5.50 in \cite{Vershynin}) yields:

\begin{lemma}
\label{lem:union-bound-cross}
Fix $\epsilon$ and $\tau$ as above. There is a $\delta_2 =  O(\epsilon \sqrt{\log 1 / \epsilon})$ such that if
$Y_1, \ldots, Y_N$ are independent samples from $\normal (0, I)$ and $N = \Omega \left( \frac{d + \log (1 / \tau)}{\delta_2^2} \right)$, then
\begin{equation}
\label{eqn:mean-conv}
\Pr \left[\exists w \in S_{N, \ve} :  \left \| \sum_{i = 1}^N w_i Y_i  \right\|_2 \geq \delta_2 \right] \leq \tau \; .
\end{equation}
\end{lemma}
Note that by Cauchy-Schwarz,  this implies:
\begin{corollary}
Fix $\epsilon$ and $\tau$ as above. There is a $\delta_2 =  O(\epsilon \sqrt{\log 1 / \epsilon})$ such that if
$Y_1, \ldots, Y_N$ are independent samples from $\normal (0, I)$ and $N = \Omega \left( \frac{d + \log (1 / \tau)}{\delta_2^2} \right)$, then
\begin{equation}
\label{eqn:mean-conv1}
\Pr \left[\exists v \in \mathbb{R}^d, \exists w \in S_{N, \ve} :  \left \| \left( \sum_{i = 1}^N w_i Y_i \right) v^T \right\|_2 \geq \delta_2 \| v \|_2 \right] \leq \tau \; .
\end{equation}
\end{corollary}

We will also require the following, well-known concentration, which says that no sample from a Gaussian deviates too far from its mean in $\ell_2$-distance.
\begin{fact}
\label{fact:gaussians-are-not-large}
Fix $\tau > 0$.
Let $X_1, \ldots, X_N \sim \normal (0, I)$.
Then, with probability $1 - \tau$, we have that $\| X_i \|_2 \leq O \left( \sqrt{d \log (N / \tau)}\right)$ for all $i = 1, \ldots, N$.
\end{fact}

\subsubsection{Estimation Error in the Frobenius Norm}
Let $X_1, ..., X_N$ be $N$ i.i.d. samples from $\normal (0, I)$. 
In this section we demonstrate a tight bound on how many samples are necessary such that the sample covariance is close to $I$ in Frobenius norm.
Let $\widehat{\Sigma}$ denote the empirical covariance, defined to be
\[\widehat{\Sigma} = \frac{1}{N} \sum_{i = 1}^N X_i X_i^T \; .\]
By self-duality of the Frobenius norm, we know that
\begin{eqnarray*}
\| \widehat{\Sigma} - I \|_F &=& \sup_{\| U \|_F = 1} \left| \left\langle \widehat{\Sigma} -  I, U \right\rangle \right| \\
&=& \sup_{\| U \|_F = 1} \left| \frac{1}{N} \sum_{i = 1}^N \tr (X_i X_i^T U) - \tr (U) \right| \; .
\end{eqnarray*}

Since there is a $1/4$-net over all PSD matrices with Frobenius norm $1$ of size $9^{d^2}$ (see e.g. Lemma 1.18 in~\cite{rigolletH2017}), the Vershynin-type union bound argument combined with Lemma~\ref{lem:hansonwright} immediately gives us the following:
\begin{corollary}
\label{cor:frob-conv}
There exist universal constants $A, B > 0$ so that for all $t > 0$, we have
\[
\Pr \left[ \left\| \frac{1}{N} \sum_{i = 1}^N X_i X_i^\top - I \right\|_F > t \right] \leq 4 \exp \left(A d^2 - B N \min (t, t^2) \right) \; .
\]
\end{corollary}
\noindent
By the argument as used in the proof of Lemma \ref{lem:union-bound}, we obtain:
\begin{corollary}
\label{cor:unknown-covariance-union-bound}
Fix $\eps, \tau > 0$.
There is a $\delta_1 = O(\eps \log 1 / \eps)$ such that if $X_1, \ldots, X_N$ are independent samples from $\normal (0, I)$, with 
\[
N = \Omega \left( \frac{d^2 + \log 1 / \tau}{\delta_1^2} \right) \; ,
\] 
then
\[
\Pr \left[\exists w \in S_{N,\ve} : \left \| \sum_{i = 1}^N w_i X_i X_i^\top - I \right\|_F \geq \delta_1 \right] \leq \tau \; .
\]
\end{corollary}
Since the proof is essentially identical to the proof of Lemma~\ref{lem:union-bound}, we omit the proof.
However, we note that in fact, the proof technique there can be used to show something slightly stronger, which we will require later.
The technique actually shows that if we take any set of size at most $\eps N$, and take the uniform weights over that set, then the empirical covariance is not too far away from the truth.
More formally:
\begin{corollary}
\label{cor:unknown-covariance-deviation}
Fix $\eps, \tau > 0$.
There is a $\delta_2 = O(\eps \log 1 / \eps)$ such that if $X_1, \ldots, X_N$ are independent samples from $\normal (0, I)$, with 
\[
N = \Omega \left( \frac{d^2 + \log 1 / \tau}{\delta_2^2} \right) \; ,
\]
then
\[
\Pr \left[ \exists T \subseteq [N] : |T| \leq \ve N \mbox{ and } \left\| \sum_{i \in T} \frac{1}{|T|} X_i X_i^\top - I \right\|_F \geq O \left( \delta_2 \frac{N}{|T|} \right) \right] \leq \tau \; .
\]
\end{corollary}
\noindent We prove this corollary in the Appendix.

\subsubsection{Understanding the Fourth Moment Tensor}
Our algorithms will be based on understanding the behavior of the fourth moment tensor of a Gaussian when restricted to various subspaces.
Let $\otimes$ denote the Kronecker product on matrices.
We will make crucial use of the following definition:
\begin{definition}
For any matrix $M \in \R^{d \times d}$, let $M^\flat \in \R^{d^2}$ denote its canonical flattening into a vector in $\R^{d^2}$,
and for any vector $v \in \R^{d^2}$, let $v^\sharp$ denote the unique matrix $M \in \R^{d \times d}$ such that $M^\flat = v$.
\end{definition}

We will also require the following definitions:
\begin{definition}
Let $\Ssym = \{M^\flat \in \R^{d^2}: \mbox{$M$ is symmetric}\}$, let $\sS \subseteq \Ssym$ be the subspace given by
\[
\sS = \{v \in \Ssym : \tr(v^\sharp) = 0 \} \; ,
\]
and let $\Pi_S$ and $\Pi_{\sS^\perp}$ denote the projection operators onto $\sS$ and $\sS^\perp$ respectively. Finally let
\[
\|v\|_{\sS} = \| \Pi_{\sS} v \|_2 \mbox{ and } \| v \|_{\sS^\perp} = \| \Pi_{\sS^\perp} v\|_2 \; .
\]
Moreover, for any $M \in \R^{d^2 \times d^2}$, let
\[
\| M \|_{\sS} = \sup_{v \in \sS - \{ 0 \}} \frac{v^T M v}{\| v \|_2^2} \; .
\]
\end{definition} 

In fact, the projection of $v = M^\flat$ onto $\sS$ where $M$ is symmetric can be written out explicitly. 
Namely, it is given by
\[
M = \left( M - \frac{\tr (M)}{d} I \right) +  \frac{\tr (M)}{d} I \; .
\]
By construction the flattening of the first term is in $\sS$ and the flattening of the second term is in $\sS^\perp$.
The expression above immediately implies that $\| v\|_{\sS^\perp} = \frac{|\tr (M)|}{\sqrt{d}}$.

The key result in this section is the following:

\begin{theorem}\label{thm:fourth-order}
Let $X \sim \normal (0, \Sigma)$. 
Let $M$ be the $d^2 \times d^2$ matrix given by $M = \E[ (X \otimes X) (X \otimes X)^T]$.
Then, as an operator on $\Ssym$, we have
\[ M = 2 \Sigma^{\otimes 2} +  \left(\Sigma^\flat \right) \left( \Sigma^\flat \right)^T \; .\]
\end{theorem}

%\noindent Such results are well-known, and we include a proof of it for completeness. Let $M$ be a $d^2 \times d^2$ matrix whose entries are given by $M_{ij,k \ell} = \E[X_i X_j X_k X_\ell]$.
%Let $U D U^T$ be the diagonalization of $\Sigma$ such that if we let $B = U D^{1/2}$, we have $X = B Y$ for $Y \sim \normal (0, I)$.
%The first step is to show that we can understand the structure of $M$ by restricting to its action on flattenings of symmetric matrices. 
%\begin{lemma}
%For all $v \in \R^{d^2}$ such that $v \perp (x \otimes x)^\flat$ for all $x \in \R^{d}$, we have that $v \in ker (M)$.
%\end{lemma}
%\begin{proof}
%It suffices to show that $v^T M v = 0$.
%We expand:
%\begin{align*}
%v^T M v &= \langle v, \E[X^{\otimes 4}] v \rangle \\
%&= \E \left[ \langle v, ((B \otimes B) (Y \otimes Y))^{\otimes 2} v \rangle \right] \\
%&= \E \left ((B \otimes B) (Y \otimes Y) v)^2 \right]  \\
%&= 0
%\end{align*}
%by assumption.
%\end{proof}
%
%Since the span of $x \otimes x$ is the set of symmetric matrices, this shows that the only non-trivial action of $M$ is on symmetric matrices.
%We now explicitly compute the quadratic form that $M$ induces restricted to this subspace.

%\begin{lemma}
%\label{lem:fourth-order}
%Let $v \in \R^{d^2}$ be the canonical flattening of a symmetric matrix in $R^{d \times d}$.
%Then
%\[\langle v, M v \rangle = 2 v^T \left(\Sigma^{\otimes 2} \right) v + v^T \left(\Sigma^\flat \right) \left( \Sigma^\flat \right)^T v \; , \]
%where $\Sigma^{\otimes 2}$ is the Kronecker product of $\Sigma$ with itself.
%\end{lemma}
It is important to note that the two terms above are {\em not} the same; the first term is high rank, but the second term is rank one.
The proof of this theorem will require Isserlis' theorem, and is deferred to Appendix \ref{sec:concAppendix}.

\subsubsection{Concentration of the Fourth Moment Tensor}
We also need to show that the fourth moment tensor concentrates:
\begin{theorem}
\label{thm:fourth-moment-union-bound}
Fix $\epsilon, \tau > 0$.
Let $Y_i \sim \normal (0, I)$ be independent, for $i = 1, \ldots, N$, where we set 
\[
N = \widetilde{\Omega} \left(  \frac{d^2 \log^5 1 / \tau}{\delta_3^2} \right) \; ,
\]
Let $Z_i = Y_i^{\otimes 2}$. 
Let $M_4 = \E [Z_i Z_i^T]$ be the canonical flattening of the true fourth moment tensor.
There is a $\delta_3 = O(\ve \log^2 1 / \ve)$ such that if $Y_1, \ldots, Y_N$, and $Z_1, \ldots, Z_m$ are as above, then we have
\[
\Pr \left[\exists w \in S_{N, \ve} : \left \| \sum_{i = 1}^N w_i Z_i Z_i^T - M_4 \right\|_{\sS} \geq \delta_3 \right] \leq \tau \; .
\]
\end{theorem}

To do so will require somewhat more sophisticated techniques than the ones used so far to bound spectral deviations.
At a high level, this is because fourth moments of Gaussians have a sufficiently larger variance that the union bound techniques used so far are insufficient.
However, we will show that the tails of degree four polynomials of Gaussians still sufficiently concentrate such that removing points cannot change the mean by too much.
The proof requires slightly fancy machinery and appears in Appendix \ref{sec:filterGaussianAppendix}.

%!TEX root = ./main.tex

\subsection{Finding the Mean, Using a Separation Oracle}
\label{sec:UnknownMeanConvex}
In this section, we consider the problem of approximating $\mu$ given $N$ samples from $\normal(\mu, I)$ in the full adversary model. 
Our algorithm will be based on working with the following convex set:
\[
\mathcal{C}_\delta = \left\{ w \in S_{N, \ve}: \left\| \sum_{i = 1}^N w_i (X_i - \mu)(X_i - \mu)^T - I \right\|_2 \leq \delta \right\} \; .
\]
It is not hard to show that $\mathcal{C}_\delta$ is non-empty for reasonable values of $\delta$ (and we will show this later). Moreover we will show that for any set of weights $w$ in $\mathcal{C}_\delta$, the empirical average 
$$ \muhat = \sum_{i =1}^N w_i X_i$$
will be a good estimate for $\mu$. The challenge is that since $\mu$ itself is unknown, there is not an obvious way to design a separation oracle for $\mathcal{C}_\delta$ even though it is convex. 
Our algorithm will run in two basic steps.
First, it will run a very naive outlier detection to remove any points which are more than $O(\sqrt{d})$ away from the good points.
These points are sufficiently far away that a very basic test can detect them.
Then, with the remaining points, it will use the approximate separation oracle given below to approximately optimize with respect to $C_{\delta}$.
It will then take the outputted set of weights and output the empirical mean with these weights.
We will explain these steps in detail below.

Our results will hold under the following deterministic conditions:
\begin{align}
\| X_i - \mu \|_2 &\leq O \left( \sqrt{d \log (N / \tau)} \right), \mbox{$\forall i \in \Sgood$} \; ,  \label{eqn:sepconds1} \\
\left \| \sum_{i \in \Sgood} w_i (X_i - \mu)(X_i - \mu)^T - w_g I \right\|_2 &\leq \delta_1  \mbox{ $\forall w \in S_{N, 4\ve}$, and} \label{eqn:sepconds2}  \\
\left \| \sum_{i \in \Sgood} w_i (X_i - \mu) \right\|_2 &\leq \delta_2   \mbox{ $\forall w \in S_{N, 4\ve}$ } \; . \label{eqn:sepconds3} 
\end{align}
The concentration bounds we gave earlier were exactly bounds on the failure probability of either of these conditions, albeit for  $S_{N, \ve}$ instead of $S_{N, 4\ve}$. 

\subsubsection{Naive Pruning}
The first step of our algorithm will be to remove points which have distance which is much larger than $O(\sqrt{d})$ from the mean.
Our algorithm is very naive: it computes all pairwise distances between points, and throws away all points which have distance more than $O(\sqrt{d})$ from more than a $2 \ve$-fraction of the remaining points.

\begin{algorithm}[htb]
\begin{algorithmic}[1]
\Function{NaivePrune}{$X_1, \ldots, X_N$}
\State For $i, j = 1, \ldots, N$, define $\delta_{i, j} =  \| X_i - X_j \|_2$.
\For{$i = 1, \ldots, j$}
	\State Let $A_i = \{j \in [N]: \delta_{i, j } > \Omega (\sqrt{d \log (N / \tau)}) \}$
	\If{$|A_i| > 2 \ve N$}
		\State Remove $X_i$ from the set.
	\EndIf
\EndFor
\State \textbf{return} the pruned set of samples.
\EndFunction
\end{algorithmic}
\caption{Naive Pruning}
\label{alg:prune-mean}
\end{algorithm}

Then we have the following fact:
\begin{fact}
\label{fact:prune}
Suppose that (\ref{eqn:sepconds1}) holds.
Then \textsc{NaivePrune} removes no uncorrupted points, and moreover, if $X_i$ is not removed by \textsc{NaivePrune}, we have $\| X_i - \mu \|_2 \leq  O \left( \sqrt{d \log (N / \tau)} \right)$.
\end{fact}
\begin{proof}
That no uncorrupted point is removed follows directly from (\ref{eqn:sepconds1}) and the fact that there can be at most $2 \ve N$ corrupted points.
Similarly, if $X_i$ is not removed by \textsc{NaivePrune}, that means there must be an uncorrupted $X_j$ such that $\| X_i - X_j \|_2 \leq O(\sqrt{d \log (N / \tau)})$.
Then the desired property follows from (\ref{eqn:sepconds1}) and a triangle inequality.
\end{proof}

Henceforth, for simplicity we shall assume that no point was removed by \textsc{NaivePrune}, and that for all $i = 1, \ldots, N$, we have $\| X_i - \mu \|_2 < O(\sqrt{d \log (N / \tau)})$.
Otherwise, we can simply work with the pruned set, and it is evident that nothing changes.

\subsubsection{The Separation Oracle}
Our main result in this section is an approximate separation oracle for $\mathcal{C}_\delta$. 
Throughout this section, let $w \in S_{N, \ve}$ and set $\muhat = \sum_{i =1}^N w_i X_i$. 
Moreover, let $ \Delta = \mu - \muhat$.
Our first step is to show that any set of weights that does not yield a good estimate for $\mu$ cannot be in the set $\mathcal{C}_\delta$: 

\begin{lemma}
\label{lem:key}
Suppose that (\ref{eqn:sepconds2})-(\ref{eqn:sepconds3}) holds. Suppose that $\|\Delta\|_2 = \Omega (\sqrt{\eps \delta_1}) = \Omega(\eps \log 1/\eps)$. 
Then
\[
\left\| \sum_{i = 1}^N w_i (X_i - \mu) (X_i - \mu)^T - I \right\|_2 \geq \Omega \left( \frac{\| \Delta \|_2^2}{\epsilon} \right) \; .
\]
\end{lemma}
\begin{proof}
By Fact~\ref{fact:weights} and (\ref{eqn:sepconds3}) we have $\| \sum_{i \in \Sgood} \frac{w_i}{w_g} X_i - \mu \|_2 \leq \delta_2$.
Now by the triangle inequality we have 
$$\left \| \sum_{i \in \Sbad} w_i (X_i - \mu) \right \|_2 \geq \|\Delta\|_2 - \left \| \sum_{i \in \Sgood} w_i (X_i - \mu) - w_g \mu \right \|_2 \geq \Omega (\|\Delta\|_2)$$
Using the fact that the variance is nonnegative we have
\[
\sum_{i \in \Sbad} \frac{w_i}{w_{b}} (X_i - \mu) (X_i - \mu)^T \succeq \left( \sum_{i \in \Sbad} \frac{w_i}{w_{b}} \left( X_i - \mu \right) \right) \left( \sum_{i \in \Sbad} \frac{w_i}{w_{b}} \left( X_i - \mu \right) \right)^T \; ,
\]
and therefore 
\[
\left\| \sum_{i \in \Sbad} w_i (X_i - \mu) (X_i - \mu)^T \right\|_2 \geq \Omega \left( \frac{\| \Delta \|_2^2}{w_b} \right) \geq \Omega \left( \frac{\| \Delta \|_2^2}{\epsilon} \right).
\]

On the other hand, 
\begin{align*}
\left\| \sum_{i \in \Sgood} w_i (X_i - \mu) (X_i - \mu)^T - I \right\|_2 &\leq \left\| \sum_{i \in \Sgood} w_i (X_i - \mu) (X_i - \mu)^T - w_g I \right\|_2 + w_b \leq \delta_1 + w_b .
\end{align*}
where in the last inequality we have used Fact~\ref{fact:weights} and (\ref{eqn:sepconds2}). 
Hence altogether this implies that
\begin{align*}
 \left\| \sum_{i = 1}^N w_i (X_i - \mu) (X_i - \mu)^T - I \right\|_2 &\geq \Omega \left( \frac{\|\Delta\|_2^2}{\ve} \right) -w_b - \delta_1 \geq \Omega \left( \frac{\|\Delta\|_2^2}{\ve} \right) \; ,
\end{align*}
as claimed.
\end{proof}

As a corollary, we find that {\em any} set of weights in $\mathcal{C}_\delta$ immediately yields a good estimate for $\mu$:

\begin{corollary}
\label{cor:close}
Suppose that (\ref{eqn:sepconds2}) and (\ref{eqn:sepconds3}) hold. Let $w \in \mathcal{C}_\delta$ for $\delta = O(\ve \log 1 / \ve)$.
Then $$ \|\Delta\|_2 \leq O(\ve \sqrt{\log 1 / \ve})$$
\end{corollary}

Our main result in this section is an approximate separation oracle for $\mathcal{C}_\delta$ with $\delta = O(\ve \log 1 / \ve)$. 

\begin{theorem}
\label{thm:separation}
Fix $\ve > 0$, and let $\delta = O (\epsilon \log 1 / \epsilon).$ Suppose that (\ref{eqn:sepconds2}) and (\ref{eqn:sepconds3}) hold.
Let $w^\ast$ denote the weights which are uniform on the uncorrupted points.
Then there is a constant $c$ and an algorithm such that:
\begin{enumerate}
\item (Completeness)
If $w = w^\ast$, then it outputs ``YES''.
\item (Soundness)
If $w \not\in \mathcal{C}_{c \delta}$, the algorithm outputs a hyperplane $\ell : \R^N \to \R$ such that $\ell(w) \geq 0$ but $\ell (w^\ast) < 0$.
Moreover, if the algorithm ever outputs a hyperplane $\ell$, then $\ell(w^*) < 0$.

\end{enumerate}
\end{theorem}

\noindent
We remark that these two facts imply that for any $\tau > 0$, the ellipsoid method with this separation oracle will output a $w'$ such that $\| w - w' \|_\infty <\ve / (N \sqrt{d \log (N / \tau)})$, for some $w \in \mathcal{C}_{c \delta} $ in $\poly (d, 1 / \epsilon, \log 1 / \tau)$ steps.

\begin{remark}
The conditions that the separation oracle given here satisfy are slightly weaker than the traditional guarantees, given, for instance, in \cite{GLS:88}.
However, the correctness of the ellipsoid algorithm with this separation oracle follows because outside $C_{c \delta}$, the separation oracle acts exactly as a separation oracle for $w^*$.
Thus, as long as the algorithm continues to query points outside of $C_{c \delta}$, the action of the algorithm is equivalent to one with a separation oracle for $w^*$.
Moreover, the behavior of the algorithm is such that it will never exclude $w^*$, even if queries are made within $C_{c \delta}$.
From these two conditions, it is clear from the classical theory presented in \cite{GLS:88} that the ellipsoid method satisfies the guarantees given above.
\end{remark}

\noindent The separation oracle is given in Algorithm \ref{alg:sepmean}.
Next, we prove correctness for our approximate separation oracle: 

\begin{algorithm}[htb]
\begin{algorithmic}[1]
\Function{SeparationOracleUnknownMean}{$w, \eps, X_1, \ldots, X_N$}
\State Let $\muhat = \sum_{i = 1}^N w_i X_i$.
\State Let $\delta = O(\eps \log 1 / \eps)$.
\State For $i = 1, \ldots, N$, define $Y_i = X_i - \muhat$.
\State Let $\lambda$ be the eigenvalue of largest magnitude of $M = \sum_{i = 1}^N w_i Y_i Y_i^T - I$.
\State Let $v$ be its associated eigenvector.
\If{$| \lambda | \leq \frac{c}{2} \delta$}
	\State \textbf{return} ``YES".
\ElsIf{$\lambda > \frac{c }{2} \delta$}
	\State \textbf{return} the hyperplane $\ell (u) = \left( \sum_{i = 1}^N u_i \langle Y_i, v \rangle^2 - 1 \right) - \lambda$.
\Else 
	\State \textbf{return} the hyperplane $\ell (u) = \lambda - \left( \sum_{i = 1}^N u_i \langle Y_i, v \rangle^2 - 1 \right) $.
\EndIf
\EndFunction
\end{algorithmic}
\caption{Separation oracle sub-procedure for agnostically learning the mean.}
\label{alg:sepmean}
\end{algorithm}

\begin{proof}[Proof of Theorem \ref{thm:separation}]
Again, let $\Delta = \mu - \muhat$, and let $M = \sum_{i = 1}^N w_i Y_i Y_i^T - I$.
By expanding out the formula for $M$, we get:
\begin{align*}
\sum_{i = 1}^N w_i Y_i Y_i^T - I &= \sum_{i = 1}^N w_i (X_i - \mu + \Delta) (X_i - \mu + \Delta)^T - I \\
&= \sum_{i = 1}^N w_i (X_i - \mu) (X_i - \mu)^T - I + \sum_{i = 1}^N w_i (X_i - \mu) \Delta^T + \Delta \sum_{i = 1}^N  w_i (X_i - \mu)^T + \Delta \Delta^T \\
&= \sum_{i = 1}^N  w_i (X_i - \mu) (X_i - \mu)^T - I - \Delta \Delta^T \; .
\end{align*}
Let us now prove completeness. 
\begin{claim}
Suppose $w = w^\ast$. Then $\|M\|_2 < \frac{c}{2} \delta$.
\end{claim}
\begin{proof}
Recall that $w^\ast$ are the weights that are uniform on the uncorrupted points. Because $|\Sbad| \leq 2 \ve N$ we have that $w^\ast \in S_{N, \ve}$. We can now use (\ref{eqn:sepconds2}) to conclude that $w^\ast \in \mathcal{C}_{\delta_1}$. Now by Corollary \ref{cor:close} we have that $\| \Delta \|_2 \leq O(\ve \sqrt{\log 1 / \ve})$. Thus
\begin{align*}
\left\| \sum_{i = 1}^N w^{\ast}_i (X_i - \mu) (X_i - \mu)^T - I - \Delta \Delta^T \right\|_2 &\leq \left\| \sum_{i = 1}^N  w^{\ast}_i (X_i - \mu) (X_i - \mu)^T - I \right\|_2 + \| \Delta \Delta^T \|_2 \\
&\leq \delta_1 + O(\ve^2 \log 1/\ve) < \frac{c \delta}{2} \; .
\end{align*}
\end{proof}
We now turn our attention to soundness. 
\begin{claim}
Suppose that $w \not\in C_{c \delta}$. Then $|\lambda| > \frac{c}{2} \delta$.
\end{claim}
\begin{proof}
By the triangle inequality, we have
\begin{align*}
\left\| \sum_{i = 1}^N w_i (X_i - \mu) (X_i - \mu)^T - I - \Delta \Delta^T \right\|_2 &\geq \left\| \sum_{i = 1}^N w_i (X_i - \mu) (X_i - \mu)^T - I \right\|_2 - \left\| \Delta \Delta^T \right\|_2 \; .
\end{align*}
 Let us now split into two cases. 
If $\| \Delta \|_2 \leq \sqrt{c \delta / 10}$, then the first term above is at least $c\delta$ by definition and we can conclude that $|\lambda| > c \delta / 2$. 
On the other hand, if $\| \Delta \|_2 \geq \sqrt{c \delta / 10}$, by Lemma \ref{lem:key}, we have that
\begin{align*}
\left\| \sum_{i = 1}^N w_i (X_i - \mu) (X_i - \mu)^T - I - \Delta \Delta^T \right\|_2 &\geq \Omega \left( \frac{\| \Delta \|_2^2}{\ve} \right) - \| \Delta \|_2^2 =  \Omega \left( \frac{\| \Delta \|_2^2}{\epsilon} \right) \; . \numberthis \label{eq:delta-large}
\end{align*}
which for sufficiently small $\ve$ also yields $|\lambda| > c \delta / 2$. 
\end{proof}
Now by construction $\ell(w) \geq 0$.  All that remains is to show that $\ell (w^\ast) < 0$ always holds. 
We will only consider the case where the top eigenvalue $\lambda$ of $M$ is positive.
The other case (when $\lambda < - \frac{c}{2} \delta$) is symmetric.
We will split the analysis into two parts. 
%Note that because we are in the ``NO" case, $\| \Delta \|_2 \geq \sqrt{c \delta / 10}$. 
\begin{align*}
&\left\| \frac{1}{|\Sgood|} \sum_{i \in \Sgood} (X_i - \muhat) (X_i - \muhat)^T - I \right\|_2 = \left\| \frac{1}{|\Sgood|} \sum_{i \in \Sgood} (X_i - \mu + \Delta) (X_i - \mu + \Delta)^T - I \right\|_2 \\
& \qquad \qquad \qquad \leq  \underbrace{\left \| \frac{1}{|\Sgood|} \sum_{i \in \Sgood} (X_i - \mu) (X_i - \mu)^T - I  \right \|_2}_{\leq  \delta_1} + \underbrace{2 \| \Delta \|_2 \left\| \frac{1}{|\Sgood|} \sum_{i \in \Sgood} (X_i - \mu) \right\|_2}_{ \leq 2 \delta_2 \|\Delta\|_2~\mbox{\footnotesize{since $w^* \in C_{\delta_2}$}}} + \| \Delta \|_2^2 \;  \numberthis \label{eq:w*-cond}
\end{align*}

Suppose $\| \Delta \|_2 \leq \sqrt{c \delta / 10}$.
By (\ref{eq:w*-cond}) we immediately have:
\begin{align*}
\ell (w^\ast) &\leq  \delta_1 + 2\delta_2 \|\Delta\|_2 + \|\Delta\|_2^2 - \lambda \leq \frac{c \delta}{5} - \lambda < 0 \; ,
\end{align*}
since $\lambda > c \delta / 2$. 
%A symmetric argument shows that the same holds for the hyperplane output when $\lambda < - c \delta / 2$.
On the other hand, if $\| \Delta \|_2 \geq \sqrt{c \delta / 10}$ then by (\ref{eq:delta-large}) we have $\lambda = \Omega \left( \frac{\| \Delta \|_2^2}{\ve} \right) $.
Putting it all together we have:
\begin{align*}
\ell (w^\ast) &\leq \underbrace{ \left\| \frac{1}{|\Sgood|} \sum_{i \in \Sgood} (X_i - \muhat) (X_i - \muhat)^T - I \right\|_2}_{\leq  \delta_1 + 2\delta_2 \|\Delta\|_2 + \|\Delta\|_2^2}  - \lambda  \; ,
\end{align*}
where in the last line we used the fact that $\lambda > \Omega \left( \frac{\| \Delta \|_2^2}{\epsilon}  \right)$, and $\| \Delta \|_2^2 \geq \Omega(\ve^2 \log 1 / \ve)$.
This now completes the proof.
\end{proof}

\subsubsection{The Full Algorithm}
This separation oracle, along with the classical theory of convex optimization \cite{GLS:88}, implies that we have shown the following:
\begin{corollary}
Fix $\ve, \tau > 0$, and let $\delta = O(\ve \sqrt{\log 1 / \ve})$. 
Let $X_1, \ldots, X_N$ be an $\ve$-corrupted set of points satisfying (\ref{eqn:sepconds2})-(\ref{eqn:sepconds3}), for $\delta_1 \leq \delta$ and $\delta_2 \leq \delta \sqrt{\log 1 / \eps}$.
Let $c$ be a sufficiently large constant.
Then, there is an algorithm $\textsc{LearnApproxMean}(\ve, \tau, X_1, \ldots, X_N)$ which runs in time $\poly (N, d, 1 / \ve, \log 1 / \tau)$, and outputs a set of weights $w' \in S_{N, \ve}$ such that there is a $w \in C_{c \delta}$ such that $\| w - w' \|_\infty \leq \ve /(N \sqrt{d \log (N / \tau)})$.
\end{corollary}

This algorithm, while an extremely powerful primitive, is technically not sufficient.
However, given this, the full algorithm is not too difficult to state: simply run \textsc{NaivePrune}, then optimize over $C_{c \delta}$ using this separation oracle, and get some $w$ which is approximately in $C_{c \delta}$.\
Then, output $\sum_{i = 1}^N w_i X_i$.
For completeness, the pseudocode for the algorithm is given below.
In the pseudocode, we assume that $\textsc{Ellipsoid} (\textsc{SeparationOracleUnknownMean}, \ve')$ is a convex optimization routine, which given the \textsc{SeparationOracleUnknownMean} separation oracle and a target error $\ve'$, outputs a $w'$ such that $\| w - w' \|_\infty \leq \ve'$.
From the classical theory of optimization, we know such a routine exists and runs in polynomial time.

\begin{algorithm}[htb]
\begin{algorithmic}[1]
\Function{LearnMean}{$\ve, \tau, X_1, \ldots, X_N$}
\State Run $\textsc{NaivePrune}(X_1, \ldots, X_N)$. Let $\{X_i\}_{i \in I}$ be the pruned set of samples. 

{\em /* For simplicity assume $I = [N]$ */}

\State Let $w' \gets \textsc{LearnApproxMean} (\ve, \tau, X_1, \ldots, X_N)$.
\State \textbf{return} $\sum_{i = 1}^N w'_i X_i$.
\EndFunction
\end{algorithmic}
\caption{Convex programming algorithm for agnostically learning the mean.}
\label{alg:conv-mean}
\end{algorithm}

We have:
\begin{theorem}
Fix $\ve, \tau > 0$, and let $\delta = O(\ve \sqrt{\log 1 / \ve})$. 
Let $X_1, \ldots, X_N$ be an $\ve$-corrupted set of samples, where 
\[
N = \Omega \left( \frac{d + \log 1 / \tau}{\delta^2} \right) \; .
\]
Let $\muhat$ be the output of $\textsc{LearnMean}(\eps, \tau, X_1, \ldots, X_N)$.
Then with probability $1-\tau$, we have $\| \muhat - \mu \|_2 \leq \delta$.
\end{theorem}
\begin{proof}
By Fact \ref{fact:gaussians-are-not-large}, Lemma \ref{lem:union-bound}, and Lemma \ref{lem:union-bound-cross}, we know that (\ref{eqn:sepconds1})-(\ref{eqn:sepconds3}) hold with probability $1 - \tau$, with $\delta_1, \delta_2 \leq \delta$.
Condition on the event that this event holds.
After \textsc{NaivePrune}, by Fact \ref{fact:prune} we may assume that no uncorrupted points are removed, and all points satisfy $\| X_i - \mu \|_2 \leq O(\sqrt{d \log (N / \tau)})$.
Let $w'$ be the output of the algorithm, and let $w \in C_{c \delta}$ be such that $\| w - w' \|_\infty < \ve / (N \sqrt{d \log (N / \tau)})$.
By Corollary \ref{cor:close}, we know that $\| \sum_{i = 1}^N w_i X_i - \mu \|_2 \leq O(\delta)$.
Hence, we have
\[ 
\left\| \sum_{i = 1}^N w_i' X_i - \mu \right\|_2 \leq \left\| \sum_{i = 1}^N w_i X_i - \mu \right\|_2 + \sum_{i = 1}^N |w_i - w_i'| \cdot \| X_i - \mu \|_2 \leq O(\delta) + \ve \; ,
\]
so the entire error is at most $O(\delta)$, as claimed.
\end{proof}

%!TEX root = ./main.tex

\subsection{Finding the Covariance, Using a Separation Oracle}
\label{sec:UnknownCovarianceConvex}
In this section, we consider the problem of approximating $\S$ given $N$ samples from $\normal(0, \Sigma)$ in the full adversary model.
 Let $U_i = \Sigma^{-1/2}X_i $ such that if $X_i \sim \normal(0, \Sigma)$ then $U_i \sim \normal (0, I)$. Moreover let $Z_i = U_i^{\otimes 2}$.
Our approach will parallel the one given earlier in Section \ref{sec:UnknownMeanConvex}. Again, we will work with a convex set
\[
C_\delta = \left\{ w \in S_{N, \ve} : \left\| \Sigma^{-1/2} \left( \sum_{i = 1}^m w_i X_i X_i^T \right) \Sigma^{-1/2}  - I \right\|_F \leq \delta \right\} \; .
\]
and our goal is to design an approximate separation oracle. Our results in this section will rely on the following deterministic conditions:
\begin{align}
\| U_i \|_2^2 &\leq O \left( d \log (N / \tau) \right) \; \mbox{, $\forall i \in \Sgood$} \label{eq:cov-conc-0} \\
 \left\| \sum_{i \in \Sgood} w_i U_i U_i^T - w_g I \right\|_F &\leq \delta_1 \; , \label{eq:cov-conc-1} \\
\left\| \sum_{i \in T} \frac{1}{|T|} U_i U_i^T - I \right\|_F &\leq O \left( \delta_2 \frac{N}{|T|} \right) \; \mbox{, and} \label{eq:cov-conc-2} \\ 
\left\| \sum_{i \in \Sgood} w_i Z_i Z_i^T - w_g M_4 \right\|_{\sS} &\leq \delta_3 \label{eq:cov-conc-3} \; , 
\end{align}
for all $w \in S_{N, \ve}$, and all sets $T \subseteq \Sgood$ of size $|T| \leq 2 \ve N$.
As before, by Fact~\ref{fact:weights}, the renormalized weights over the uncorrupted points are in $S_{N, 4\ve}$. 
Hence, we can appeal to Fact \ref{fact:gaussians-are-not-large}, Corollary~\ref{cor:unknown-covariance-union-bound}, Corollary~\ref{cor:unknown-covariance-deviation}, and Theorem~\ref{thm:fourth-moment-union-bound} with $S_{N, 4\ve}$ instead of $S_{N, \ve}$ to bound the probability that this event does not hold.
Let $w^\ast$ be the set of weights which are uniform over the uncorrupted points; by (\ref{eq:cov-conc-1}) for $\delta \geq \Omega (\epsilon \sqrt{\log 1 / \ve})$ we have that $w^\ast \in C_\delta$.
\begin{theorem}
\label{thm:covariance-separation}
Let $\delta = O(\ve \log 1 / \ve)$. Suppose that (\ref{eq:cov-conc-1}), (\ref{eq:cov-conc-2}), and \ref{eq:cov-conc-3} hold for $\delta_1, \delta_2 \leq O(\delta)$ and $\delta_3 \leq O(\delta \log 1 / \ve)$.
Then, there is a constant $c$ and an algorithm such that, given any input $w \in S_{N, \ve}$ we have:
\begin{enumerate}
\item (Completeness)
If $w = w^\ast$, the algorithm outputs ``YES''.
\item (Soundness)
If $w \not\in C_{c \delta}$, the algorithm outputs a hyperplane $\ell : \R^m \to \R$ such that $\ell(w) \geq 0$ but we have $\ell (w^\ast) < 0$.
Moreover, if the algorithm ever outputs a hyperplane $\ell$, then $\ell (w^*) < 0$.
\end{enumerate}
\end{theorem}

\noindent
As in the case of learning an unknown mean, by the classical theory of convex optimization this implies that we will find a point $w$ such that $\| w - w' \|_\infty \leq \frac{\epsilon}{\poly(N)}$ for some $w' \in C_{c \delta}$, using polynomially many calls to this oracle.
We make this more precise in the following subsubsection.

The pseudocode for the (approximate) separation oracle is given in Algorithm \ref{alg:conv-cov}.
Observe briefly that this algorithm does indeed run in polynomial time.
Lines \ref{line:cov-init}-\ref{line:cov-sS} require only taking top eigenvalues and eigenvectors, and so can be done in polynomial time.
For any $\xi \in \{-1, +1\}$, line \ref{line:cov-big-frob-test} can be run by sorting the samples by $w_i \left( \frac{\| Y_i \|^2}{\sqrt{d}} - \sqrt{d} \right)$ and seeing if there is a subset of the top $2 \ve N$ samples satisfying the desired condition, and line \ref{line:cov-big-frob-test2} can be executed similarly.

\begin{algorithm}[htb]
\begin{algorithmic}[1]
\Function{SeparationOracleUnknownCovariance}{$w$}
\State Let $\Sigmahat = \sum_{i = 1}^N w_i X_i X_i^T$. \label{line:cov-init}
\State For $i = 1, \ldots, N$, let $Y_i = \Sigmahat^{-1/2} X_i$ and let $Z_i = \left(Y_i \right)^{\otimes 2}$.
\State Let $v$ be the top eigenvector of $M = \sum_{i = 1}^N w_i Z_i Z_i^T - 2 I$ restricted to $\sS$, and let $\lambda$ be its associated eigenvalue. \label{line:eigenvalue}
\If{$|\lambda| > \Omega (\ve \log^2 1 / \ve)$}
	\State Let $\xi = \mbox{sgn} (\lambda)$.
	\State \textbf{return} the hyperplane 
	\[
	\ell (u) = \xi \left( \sum_{i = 1}^N u_i \langle v, Z_i \rangle^2 - 2 - \lambda \right) \; .
	\] \label{line:cov-sS}
\ElsIf{there exists a sign $\xi \in \{-1, 1\}$ and a set $T$ of samples of size at most $2 \epsilon N$ such that 
\[
\alpha = \xi \sum_{i \in T} w_i \left( \frac{\| Y_i \|_2^2}{\sqrt{d}} - \sqrt{d} \right) > \frac{(1 - \epsilon) \alpha \delta}{2}\; ,
\]
~~~~} \label{line:cov-big-frob-test}
	\State \textbf{return} the hyperplane
\[
\ell (u) = \xi \sum_{i \in T} u_i \left( \frac{\| Y_i \|_2^2}{\sqrt{d}} - \sqrt{d} \right) - \alpha \; ,
\]  \label{line:cov-big-frob-test2}
	\State \textbf{return} ``YES''.
\EndIf
\EndFunction
\end{algorithmic}
\caption{Convex programming algorithm for agnostically learning the covariance.}
\label{alg:conv-cov}
\end{algorithm}
We now turn our attention to proving the correctness of this separation oracle.
We require the following technical lemmata.
\begin{claim}
\label{lem:sos}
Let $w_i$ for $i = 1, \ldots, N$ be a set of non-negative weights such that $\sum_{i = 1}^N w_i = 1$, and let $a_i \in \R$ be arbitrary.
Then
\[
\sum_{i = 1}^N a_i^2 w_i \geq \left( \sum_{i = 1}^N a_i w_i \right)^2 \; .
\]
\end{claim}
\begin{proof}
Let $P$ be the distribution where $a_i$ is chosen with probability $w_i$.
Then $\E_{X \sim P} [X] = \sum_{i = 1}^N a_i w_i$ and $\E_{X \sim P} [X^2] = \sum_{i= 1}^N a_i w_i^2$.
Since $\Var_{X \sim P} [X] = \E_{X \sim P} [X^2] - \E_{X \sim P} [X]^2$ is always a non-negative quantity, by rearranging the desired conclusion follows.
\end{proof}

\begin{lemma}
\label{lem:inverse}
Fix $ \delta < 1$ and suppose that
$M$ is symmetric.  If $\| M - I \|_F \geq \delta$ then $\| M^{-1} - I \|_F \geq  \frac{\delta}{2}$.
\end{lemma}
\begin{proof}
We will prove this lemma in the contrapositive, by showing that if $\|M^{-1} - I \|_F < \frac{\delta}{2}$ then $\|M - I\|_F < \delta$. Since the Frobenius norm is rotationally invariant, we may assume that $M^{-1} =  \diag (1 + \nu_1, \ldots, 1 + \nu_d)$, where by assumption $\sum \nu_i^2 < \delta^2/4$. By our assumption that $\delta < 1$, we have $|\nu_i| \leq 1/2$ for all $i$. Thus
\begin{align*}
\sum_{i = 1}^d \left(1 - \frac{1}{1 + \nu_i} \right)^2 &\leq \sum_{i = 1}^d 4\nu_i^2 < \delta  \; ,
\end{align*}
where we have used the inequality $|1 - \frac{1}{1+x}| \leq |2x|$ which holds for all $|x| \leq 1/2$. This completes the proof. 
\end{proof}

\begin{lemma}
\label{lem:frob-bound}
Let $M, N \in \R^{d \times d}$ be arbitrary matrices. Then $\| M N \|_F \leq \| M \|_2 \| N \|_F$.
\end{lemma}
\begin{proof}
Let $N_1, \ldots, N_d$ be the columns of $N$.
Then
\[\| M N \|_F^2 = \sum_{i = 1}^d \| M N \|_2^2 \leq \| M \|_2^2 \sum_{i = 1}^d \| N_i \|_2^2 = \| M \|_2^2 \| N \|_F^2 \; ,\]
 so the desired result follows by taking square roots of both sides.
\end{proof}

\begin{lemma}
Let $M \in \R^{d \times d}$. Then, $\left\| \left( M^\flat \right)  \left( M^\flat \right)^T \right\|_{\sS}  \leq \| M - I \|_F^2$.
\end{lemma}
\begin{proof}
By the definition of $\| \cdot \|_{\sS}$, we have
\begin{align*}
\left\| \left( M^\flat \right)  \left( M^\flat \right)^T \right\|_{\sS} &= \sup_{\substack{A^\flat \in \sS \\ \| A \|_F = 1}} \left( A^\flat \right)^T \left( M^\flat \right) \left( M^\flat \right)^T A^\flat =  \sup_{\substack{A \in \sS \\ \| A \|_F = 1}} \langle A, M \rangle^2 \; .
\end{align*}
By self duality of the Frobenius norm, we know that 
\[ \langle A, M \rangle = \langle A, M - I \rangle \leq \| M - I \|_F\; , \]
since $I^\flat \in \sS^\perp$.
The result now follows.
\end{proof}

\begin{proof}[Proof of Theorem \ref{thm:covariance-separation}]
Let us first prove completeness.
Observe that by Theorem \ref{thm:fourth-order}, we know that restricted to $\sS$, we have that $M_4 = 2 I$.
Therefore, by (\ref{eq:cov-conc-3}) we will not output a hyperplane in line \ref{line:cov-sS}.
Moreover, by (\ref{eq:cov-conc-2}), we will not output a hyperplane in line \ref{line:cov-big-frob-test}.
This proves completeness.

Thus it suffices to show soundness.
Suppose that $w \not\in \mathcal{C}_{c \delta}$. We will make use of the following elementary fact:
\begin{fact}\label{fact:tr}
Let $A = \Sigma^{-1/2} \Sigmahat \Sigma^{-1/2}$ and $B = \Sigmahat^{-1/2} \Sigma \Sigmahat^{-1/2}$. Then
$$\|A^{-1} - I\|_F = \|B - I\|_F$$
\end{fact}
\begin{proof}
In particular $A^{-1} = \Sigma^{1/2} \Sigmahat^{-1} \Sigma^{1/2}$. Using this expression and the fact that all the matrices involved are symmetric, we can write
\begin{align*}
\|A^{-1} -I\|_F^2 &= \tr\left ((A^{-1} - I)^T (A^{-1} - I) \right) \\
&= \tr\left(\Sigma^{1/2} \Sigmahat^{-1} \Sigma \Sigmahat^{-1} \Sigma^{1/2} - 2 \Sigma^{1/2} \Sigmahat^{-1} \Sigma^{1/2} - I  \right) \\
&= \tr \left( \Sigmahat^{-1/2} \Sigma \Sigmahat^{-1} \Sigma \Sigmahat^{-1/2} - 2  \Sigmahat^{-1/2} \Sigma  \Sigmahat^{-1/2} - I \right ) \\
&= \tr\left((B - I)^T (B - I)\right) = \|B -I\|_F^2
\end{align*}
where in the third line we have used the fact that the trace of a product of matrices is preserved under cyclic shifts. 
\end{proof}
This allows us to show:
\begin{claim}
\label{claim:cov:something-is-big}
Assume (\ref{eq:cov-conc-1}) holds with $\delta_1 \leq O(\delta)$ and assume furthermore that $\| A - I \|_F \geq c \delta$.
Then, if we let $\delta' = \frac{(1 - \ve) c}{2} \delta = \Theta (\delta)$, we have
\begin{equation}
\label{eq:cov-split}
\left\| \sum_{i \in \Sbad} w_i Z_i - w_b I^\flat \right\|_{\sS} + \left\| \sum_{i \in \Sbad} w_i Z_i - w_b I^\flat \right\|_{\sS^\perp} \geq \delta' \; .
\end{equation}
\end{claim}

\begin{proof}
Let $A, B$ be as in Fact~\ref{fact:tr}.
Combining Lemma \ref{lem:inverse} and Fact~\ref{fact:tr} we have
\begin{equation}
\label{eq:cov-frob-norm-is-large}
\| A - I \|_F \geq c \delta \Rightarrow \|B - I\|_F \geq \frac{c\delta}{2} \; .
\end{equation}
%Letting $\Sigmahat = \sum_{i = 1}^m w_i X_i X_i^T$, by definition this means that
%\[
%\left\| \Sigma^{-1/2} \Sigmahat \Sigma^{-1/2} - I \right\|_F \geq \alpha \delta \; .
%\]
%Define $\delta'$ to be
%\[
%\delta = \left\| \Sigmahat^{-1/2} \Sigma \Sigmahat^{-1/2} - I \right\|_F \; .
%\]
%By Lemma \ref{lem:inverse}, we know that for $\delta$ sufficiently small, $\delta' \geq \frac{\alpha}{2} \delta$.
%As before we split the samples into $\Sgood$ and $\Sbad$ depending whether or not the sample was corrupted, and we let $\Mgood = \sum_{i \in \Sgood} w_i$ and $\Mbad = \sum_{i \in \Sbad} w_i$.
We can rewrite (\ref{eq:cov-conc-1}) as the expression
$\sum_{i \in \Sgood} w_i X_i X_i^T = w_g \Sigma^{1/2} (I + R) \Sigma^{1/2}$
where $R$ is symmetric and satisfies $\|R\|_F \leq \delta_1$. 
By the definition of $\Sigmahat$ we have that $\sum_{i = 1}^N w_i Y_i Y_i^T = I$, and so
\begin{align*}
\left\| \sum_{i \in \Sbad} w_i Y_i Y_i^T - w_b I \right\|_F &= \left\| \sum_{i \in \Sgood} w_i Y_i Y_i^T - w_g I \right\|_F = w_g \left\| \Sigmahat^{-1/2} \Sigma^{1/2} (I + R) \Sigma^{1/2} \Sigmahat^{-1/2} - I \right\|_F 
\end{align*}
Furthermore we have
\begin{align*}
 \left\|  \Sigmahat^{-1/2} \Sigma^{1/2} R \Sigma^{1/2} \Sigmahat^{-1/2} \right\|_F  \leq \delta_1 \left \|\Sigmahat^{-1/2} \Sigma \Sigmahat^{-1/2} \right\|_2  \; ,
\end{align*}
by applying Lemma \ref{lem:frob-bound}. And putting it all together we have
\begin{align*}
\left\| \sum_{i \in \Sbad} w_i Y_i Y_i^T - w_b I \right\|_F  \geq w_g \left ( \left \| \Sigmahat^{-1/2} \Sigma \Sigmahat^{-1/2}  - I\right\|_F - \delta_1 \left \|\Sigmahat^{-1/2} \Sigma \Sigmahat^{-1/2} \right\|_2 \right )
\end{align*}
It is easily verified that for $c > 10$, we have that for all $\delta$, if $\|  \Sigmahat^{-1/2} \Sigma \Sigmahat^{-1/2} - I \|_F \geq c \delta$, then
\[
\|  \Sigmahat^{-1/2} \Sigma \Sigmahat^{-1/2} - I \|_F \geq 2 \delta \| \Sigmahat^{-1/2} \Sigma \Sigmahat^{-1/2} \|_2 \; .
\]
Hence all this implies that
\[
\left\| \sum_{i \in \Sbad} w_i Y_i Y_i^T - w_b I \right\|_F \geq \delta' \; ,
\]
where $\delta' =  \frac{c (1 - \epsilon)}{2} \delta = \Theta (\delta)$.
The desired result then follows from the Pythagorean theorem.
\end{proof}

Claim \ref{claim:cov:something-is-big} tells us that if $w \not\in C_{c \delta},$ we know that one of the terms in (\ref{eq:cov-frob-norm-is-large}) must be at least $ \frac{1}{2} \delta'$.
We first show that if the first term is large, then the algorithm outputs a separating hyperplane:
\begin{claim}
Assume that (\ref{eq:cov-conc-1})-(\ref{eq:cov-conc-3}) hold with $\delta_1, \delta_2 \leq O(\delta)$ and $\delta_3 \leq O(\delta \log 1 / \ve)$.
Moreover, suppose that
\[
\left\| \sum_{i \in \Sbad} w_i Z_i - w_b I^\flat \right\|_{\sS} \geq \frac{1}{2} \delta' \; .
\]
Then the algorithm outputs a hyperplane in line \ref{line:cov-sS}, and moreover, it is a separating hyperplane.
\end{claim}

\begin{proof}
Let us first show that given these conditions, then the algorithm indeed outputs a hyperplane in line \ref{line:cov-sS}.
Since $I^\flat \in S^\perp$, the first term is just equal to $\left\| \sum_{i \in \Sbad} w_i Z_i \right\|_S$.
But this implies that there is some $M^\flat \in S$ such that $\| M^\flat \|_2 = \| M \|_F = 1$ and such that 
\[
\sum_{i \in \Sbad} w_i \langle M^\flat, Z_i \rangle \geq \frac{1}{2} \delta' \; ,
\]
which implies that
\begin{align*}
\sum_{i \in \Sbad} \frac{w_i}{w_b} \langle M^\flat, Z_i \rangle &\geq \frac{1}{2} \frac{\delta'}{w_b} \; .
\end{align*}
The $w_i / w_b$ are a set of weights satisfying the conditions of Claim \ref{lem:sos} and so this implies that
\begin{align*}
\sum_{i \in \Sbad} w_i \langle M^\flat, Z_i \rangle^2 &\geq O \left( \frac{{\delta'}^2}{w_b} \right) \\
&\geq O \left( \frac{{\delta'}^2}{\epsilon} \right) \; \numberthis \label{eq:cov-1}
\end{align*}
Let $\widetilde{\Sigma} = \Sigmahat^{-1} \Sigma$.
By Theorem \ref{thm:fourth-order} and (\ref{eq:cov-conc-3}), we have that
\[
\sum_{i \in \Sgood} w_i Z_i Z_i^T = w_g \left( \left( \widetilde{\Sigma}^\flat \right) \left( \widetilde{\Sigma}^\flat \right)^T  + 2 \widetilde{\Sigma}^{\otimes 2} + \left( \widetilde{\Sigma}^{1/2} \right)^{\otimes 2} R  \left( \widetilde{\Sigma}^{1/2} \right)^{\otimes 2} \right) \; ,
\]
where $\| R \|_2 \leq \delta_3$.
Hence,
\begin{align*}
\left\| \sum_{i \in \Sgood} w_i Z_i Z_i^T - 2 I \right\|_S &= w_g \left\| \left( \widetilde{\Sigma}^\flat \right) \left( \widetilde{\Sigma}^\flat \right)^T  + 2 \left(\widetilde{\Sigma}^{\otimes 2} - I \right) + (1 - w_g) I + \left( \widetilde{\Sigma}^{1/2} \right)^{\otimes 2} R  \left( \widetilde{\Sigma}^{1/2} \right)^{\otimes 2} \right\|_S \\
&\leq \| \widetilde{\Sigma} - I \|_F^2 + 2 \| \widetilde{\Sigma} - I \|_2 + (1 - w_g) + \| R \| \| \widetilde{\Sigma} \|^2 \\
&\leq 3 \| \widetilde{\Sigma} - I \|_F^2  + \delta \| \widetilde{\Sigma} \|^2  + O(\ve) \;  .\\
&\leq O \left( {\delta'}^2 + \delta' \right) \; , \numberthis \label{eq:cov-2}
\end{align*}
since it is easily verified that $\delta \| \widetilde{\Sigma} \|^2 \leq O (\| \widetilde{\Sigma} - I \|_F)$ as long as $\| \widetilde{\Sigma} - I \|_F \geq \Omega (\delta)$, which it is by (\ref{eq:cov-frob-norm-is-large}).

Equations \ref{eq:cov-1} and \ref{eq:cov-2} then together imply that
\[
\sum_{i = 1}^N w_i (M^\flat)^T Z_i Z_i^T (M^\flat) - (M^\flat)^T I M^\flat \geq O \left( \frac{\delta^2}{\epsilon} \right) \; ,
\]
and so the top eigenvalue of $M$ is greater in magnitude than $\lambda$, and so the algorithm will output a hyperplane in line \ref{line:cov-sS}.
Letting $\ell$ denote the hyperplane output by the algorithm, by the same calculation as for (\ref{eq:cov-2}), we must have $\ell (w^\ast) < 0$, so this is indeed a separating hyperplane.
Hence in this case, the algorithm correctly operates.
\end{proof}

Moreover, observe that from the calculations in (\ref{eq:cov-2}), we know that if we ever output a hyperplane in line \ref{line:cov-sS}, which implies that $\lambda \geq \Omega (\ve \log^2 1 / \ve)$, then we must have that $\ell(w^*) < 0$.

Now let us assume that the first term on the LHS is less than $ \frac{1}{2} \delta'$, such that the algorithm does not necessarily output a hyperplane in line \ref{line:cov-sS}.
Thus, the second term on the LHS of Equation \ref{eq:cov-split} is at least $ \frac{1}{2} \delta'$.
We now show that this implies that this implies that the algorithm will output a separating hyperplane in line \ref{line:cov-big-frob-test2}.
\begin{claim}
Assume that (\ref{eq:cov-conc-1})-(\ref{eq:cov-conc-3}) hold.
Moreover, suppose that
\[
\left\| \sum_{i \in \Sbad} w_i Z_i - w_b I^\flat \right\|_{\sS^\perp} \geq \frac{1}{2} \delta' \; .
\]
Then the algorithm outputs a hyperplane in line \ref{line:cov-big-frob-test2}, and moreover, it is a separating hyperplane.
\end{claim}
\begin{proof}
By the definition of $\sS^\perp$, the assumption implies that
\[
\left| \sum_{i \in \Sbad} w_i \frac{\tr(Z_i^\sharp)}{\sqrt{d}} - M_b \sqrt{d} \right| \geq \frac{1}{2} \delta' \; ,
\]
which is equivalent to the condition that
\begin{align*}
\xi \sum_{i \in \Sbad} w_i \left( \frac{\| Y_i \|_2^2}{\sqrt{d}} - \sqrt{d} \right) &\geq  \frac{(1 - \ve) \delta'}{2} \; ,
\end{align*}
for some $\xi \in \{-1, 1\}$.
In particular, the algorithm will output a hyperplane 
\[
\ell (w) = \xi \sum_{i \in S} w_i  \left( \frac{\| Y_i \|_2^2}{\sqrt{d}} - \sqrt{d} \right) - \lambda
\]
in Step \ref{line:cov-big-frob-test2}, where $S$ is some set of size at most $\epsilon N$, and $\lambda = O(\delta')$.
Since it will not affect anything, for without loss of generality let us assume that $\xi = 1$.
The other case is symmetrical.

It now suffices to show that $\ell(w^\ast) < 0$ always.
Let $T = S \cap \Sgood$.
By (\ref{eq:cov-conc-2}), we know that
\[
\sum_{i \in T} \frac{1}{|T|} Y_i Y_i^T - I = \widetilde{\Sigma}^{1/2} \left( I + A \right)  \widetilde{\Sigma}^{1/2} - I \; ,
\]
where $\| A \|_F = O \left( \delta \frac{N}{|T|} \right)$.				
Hence, 
\begin{align*}
\left\| \sum_{i \in T} \frac{1}{(1 - \epsilon) N} Y_i Y_i^T - \frac{|T|}{(1 - \epsilon) N } I \right\|_F &= \frac{|T|}{(1 - \epsilon) N} \left\|  \widetilde{\Sigma}^{1/2} \left( I + A \right)  \widetilde{\Sigma}^{1/2} - I \right\|_F \\
&\leq \frac{|T|}{(1 - \epsilon) N} \left( \| \widetilde{\Sigma} - I \|_F + \| A \|_F \| \widetilde{\Sigma} \|_2  \right) \\
&\leq \frac{|T|}{(1 - \epsilon) N} \| \widetilde{\Sigma} - I \|_F + O( \delta) \| \widetilde{\Sigma} \|_2 \\
&\leq O(\delta \delta' + \delta) \; ,
\end{align*}
as long as $\delta' \geq O (\delta)$.
By self-duality of the Frobenius norm, using the test matrix $\frac{1}{\sqrt{d}} I$, this implies that
\[
\left| \sum_{i \in T} \frac{1}{(1 - \epsilon) N} \left( \| Y_i \|^2 - \sqrt{d} \right) \right| \leq O(\delta \delta' + \delta) < \alpha
\]
and hence $\ell (w^\ast) < 0$, as claimed.
\end{proof}
\noindent These two claims in conjunction directly imply the correctness of the theorem.
\end{proof}
\subsubsection{The Full Algorithm}
As before, this separation oracle and the classical theory of convex optimization \cite{GLS:88} shows that we have demonstrated an algorithm \textsc{FindApproxCovariance} with the following properties:
\begin{theorem}
\label{thm:findapproxcov}
Fix $\ve, \tau > 0$, and let $\delta = O(\ve \log 1 / \ve)$.
Let $c > 0$ be a universal constant which is sufficiently large.
Let $X_1, \ldots, X_N$ be an $\ve$-corrupted set of points satisfying (\ref{eq:cov-conc-1}-(\ref{eq:cov-conc-3}), for $\delta_1, \delta_2 \leq O(\delta)$ and $\delta_3 \leq O(\delta \log 1 / \ve)$.
Then $\textsc{FindApproxCovariance} (\ve, \tau, X_1, \ldots, X_N)$ runs in time $\poly (N, d, 1 / \ve, \log 1 / \tau)$, and outputs a $u$ such that there is some $w \in C_{c \delta}$ such that $\| w - u \|_\infty \leq \ve / (N d \log (N / \tau))$.
\end{theorem}
As before, this is not quite sufficient to actually recover the covariance robustly.
Naively, we would just like to output $\sum_{i = 1}^N u_i X_i X_i^T$.
However, this can run into issues if there are points $X_i$ such that $\| \Sigma^{-1/2} X_i \|_2$ is extremely large.
We show here that we can postprocess the $u$ such that we can weed out these points.
First, observe that we have the following lemma:
\begin{lemma}
\label{lem:s1-preceq-s}
Assume $X_1, \ldots, X_N$ satisfy (\ref{eq:cov-conc-1}).
Let $w \in S_{N, \ve}$.
Then
\[
\sum_{i = 1}^N w_i X_i X_i^T \succeq (1 - O(\delta_1)) \Sigma \; . 
\]
\end{lemma}
\begin{proof}
This follows since by (\ref{eq:cov-conc-1}), we have that $\sum_{i \in \Sgood} w_i X_i X_i^T \succeq w_g (1 - \delta_1) \Sigma \succeq (1 - O(\delta_1)) \Sigma$.
The lemma then follows since $\sum_{i \in \Sbad} w_i X_i X_i^T \succeq 0$ always.
\end{proof}
Now, for any set of weights $w \in S_{N, \ve}$, let $\widetilde{w}^- \in \R^N$ be the vector given by $\widetilde{w}^-_i = \max (0, w_i - \ve / (N d \log (N / \tau)))$, and let $w^-$ be the set of weights given by renormalizing $\widetilde{w}^-$.
It is a straightforward calculation that for any $w \in S_{N, \ve}$, we have $w^- \in S_{N, 2 \ve}$.
In particular, this implies:
\begin{lemma}
\label{lem:s-preceq-s1}
Let $u$ be such that there is $w \in C_{c \delta}$ such that $\| u - w \|_\infty \leq  \ve / (N d \log (N / \tau))$.
Then, $\sum_{i = 1}^N u^-_i X_i X_i^T \preceq (1 + O(\delta)) \Sigma$.
\end{lemma}
\begin{proof}
By the definition of $C_{c \delta}$, we must have that $\sum_{i = 1}^N w_i X_i X_i^T \preceq (1 + c \delta) \Sigma$.
Moreover, we must have $\widetilde{u}^-_i \leq w_i$ for every index $i \in [N]$.
Thus we have that $\sum_{i = 1}^N \widetilde{u}^-_i w_i X_i X_i^T \preceq (1 + c \delta) \Sigma$, and hence $\sum_{i = 1}^N u^-_i w_i X_i X_i^T \preceq (1 + c \delta) \Sigma$, since $\sum_{i = 1}^N u^-_i w_i X_i X_i^T \preceq (1 + O(\ve)) \sum_{i = 1}^N \widetilde{u}^-_i w_i X_i X_i^T$.
\end{proof}
We now give the full algorithm.
The algorithm proceeds as follows: first run \textsc{FindApproxCovariance} to get some set of weights $u$ which is close to some element of $C_{c \delta}$.
We then compute the empirical covariance $\Sigma_1 = \sum_{i = 1}^N u_i X_i X_i^T$ with the weights $u$, and remove any points which have $\| \Sigma_1^{-1/2} X_i \|_2^2$ which are too large.
We shall show that this removes no good points, and removes all corrupted points which have $\| \Sigma^{-1/2} X_i \|_2^2$ which are absurdly large.
We then rerun \textsc{FindApproxCovariance} with this pruned set of points, and output the empirical covariance with the output of this second run.
Formally, we give the pseudocode for the algorithm in Algorithm \ref{alg:learn-cov}.
\begin{algorithm}[htb]
\begin{algorithmic}[1]
\Function{LearnCovariance}{$\ve, \tau, X_1, \ldots, X_N$}
\State Let $u \gets \textsc{FindApproxCovariance}(\ve, \tau, X_1, \ldots, X_N)$.
\State Let $\Sigma_1 = \sum_{i = 1}^N u^-_i X_i X_i^T$.
\For{$i = 1, \ldots, N$}
	\If{$\| \Sigma_1^{-1/2} X_i \|_2^2 \geq \Omega(d \log N / \tau)$}
		\State Remove $X_i$ from the set of samples \label{line:prune}
	\EndIf
\EndFor
\State Let $S'$ be the set of pruned samples.
\State Let $u' \gets \textsc{FindApproxCovariance}(\ve, \tau, \{ X_i\}_{i \in S'} )$.
\State \textbf{return} $\sum_{i = 1}^N u'_i X_i X_i^T$.
\EndFunction
\end{algorithmic}
\caption{Full algorithm for learning the covariance agnostically}
\label{alg:learn-cov}
\end{algorithm}

We now show that this algorithm is correct.
\begin{theorem}
Let $1 /2 \geq \ve > 0$, and $\tau > 0$.
Let $\delta = O(\ve \log 1 / \ve)$.
Let $X_1, \ldots, X_N$ be a $\ve$-corrupted set of samples from $\normal (0, \Sigma)$ 
where 
\[
N = \widetilde{\Omega} \left( \frac{d^2 \log^5 1 / \tau}{\ve^2} \right).
\]
Let $\widehat{\Sigma}$ be the output of $\textsc{LearnCovariance} (\ve, \tau, X_1, \ldots, X_N)$.
Then with probability $1 - \tau$, $\| \Sigma^{-1/2} \widehat{\Sigma} \Sigma^{-1/2} - I \|_F \leq O(\delta)$.
\end{theorem}
\begin{proof}
We first condition on the event that we satisfy (\ref{eq:cov-conc-0})-(\ref{eq:cov-conc-3}) with $\delta_1, \delta_2 \leq O(\delta)$ and $\delta_3 \leq O(\delta \log 1 / \ve)$.
By our choice of $N$, Fact \ref{fact:gaussians-are-not-large}, Corollary \ref{cor:frob-conv}, Corollary \ref{cor:unknown-covariance-deviation}, and Theorem \ref{thm:fourth-moment-union-bound}, and a union bound, we know that this event happens with probability $1 - \tau$.

By Theorem \ref{thm:findapproxcov}, Lemma \ref{lem:s1-preceq-s}, and Lemma \ref{lem:s-preceq-s1}, we have that since $\ve$ is sufficiently small,
\[
\frac{1}{2} \Sigma \preceq \Sigma_1 \preceq 2 \Sigma \; .
\]
In particular, this implies that for every vector $X_i$, we have
\[
\frac{1}{2} \| \Sigma^{-1/2} X_i \|_2^2 \leq \| \Sigma_1^{-1/2} X_i \|_2^2 \leq 2 \| \Sigma^{-1/2} X_i \|_2^2 \; . 
\]
Therefore, by (\ref{eq:cov-conc-0}), we know that in line \ref{line:prune}, we never throw out any uncorrupted points, and moreover, if $X_i$ is corrupted with $\| \Sigma^{-1/2} X_i \|_2^2 \geq \Omega(d \log N / \tau)$, then it is thrown out.
Thus, let $S'$ be the set of pruned points.
Because no uncorrupted point is thrown out, we have that $|S'| \geq (1 - 2 \ve) N$, and moreover, this set of points still satisfies (\ref{eq:cov-conc-1})-(\ref{eq:cov-conc-3})\footnote{Technically, the samples satisfy a slightly different set of conditions since we may have thrown out some corrupted points, and so in particular the number of samples may have changed, but the meaning should be clear.} and moreover, for ever $i \in S'$, we have $\| \Sigma^{-1/2} X_i \|_2^2 \leq O (d \log N / \tau)$.
Therefore, by Theorem \ref{thm:findapproxcov}, we have that there is some $u'' \in C_{c |I|}$ such that $\| u' - u'' \|_\infty < \ve / (N d \log (N / \tau))$. 
But now if $\Sigmahat = \sum_{i \in |I|} u'_i X_i X_i^T$, we have
\begin{align*}
\| \Sigma^{-1/2} \Sigmahat \Sigma^{-1/2} - I \|_F &\leq \left\| \sum_{i \in I} u''_i \Sigma^{-1/2} X_i X_i^T \Sigma^{-1/2} - I \right\|_F + \sum_{i \in I} |u'_i - u'_i| \| \Sigma^{-1/2} X_i \|_2^2 \\ 
&\leq c \delta + O(\ve) \leq O(\delta) \; , 
\end{align*}
which completes the proof.
\end{proof}

\subsection{Learning an Arbitrary Gaussian Agnostically}
\label{sec:UnknownGaussianWhole}
We have shown how to agnostically learn the mean of a Gaussian with known covariance, and we have shown how to agnostically learn the covariance of a mean zero Gaussian.
In this section, we show how to use these two in conjunction to agnostically learn an arbitrary Gaussian. 
Throughout, let $X_1, \ldots, X_{N}$ be an $\ve$-corrupted set of samples from $\normal (\mu, \Sigma)$, where both $\mu$ and $\Sigma$ are unknown.
We will set 
\[
\widetilde{\Omega} \left( \frac{d^2 \log^5 1 / \tau}{\ve^2} \right) \; .
\]

\subsubsection{From Unknown Mean, Unknown Covariance, to Zero Mean, Unknown Covariance}
We first show a simple trick which, at the price of doubling the amount of error, allows us to assume that the mean is zero, without changing the covariance.
We do so as follows: for each $i = 1, \ldots, N / 2$, let $X_i' = (X_i - X_{N / 2 + i}) / \sqrt{2}$.
Observe that if both $X_i$ and $X_{N / 2 + i}$ are uncorrupted, then $X_i' \sim \normal (0, \Sigma)$.
Moreover, observe that $X_i'$ is corrupted only if either $X_i$ or $X_{N / 2 + i}$ is corrupted.
Then we see that if $X_1, \ldots, X_N$ is $\ve$-corrupted, then the $X'_1, \ldots, X'_{N / 2}$ is a $N / 2$-sized set of samples which is $2 \ve$-corrupted.
Thus, by using the results from Section \ref{sec:UnknownCovarianceConvex}, with probability $1 - \tau$, we can recover a $\Sigmahat$ such that 
\begin{equation}
\label{eqn:full-cov}
\| \Sigma^{-1/2} \Sigmahat \Sigma^{-1/2} - I \|_F \leq O(\ve \log 1 / \ve) \; ,
\end{equation}
which in particular by Corollary \ref{cor:kl-to-cov}, implies that
\begin{equation}
\label{eqn:full-cov-tv}
\dtv (\normal (0, \Sigmahat), \normal (0, \Sigma)) \leq O(\ve \log 1 / \ve) \; .
\end{equation}

\subsubsection{From Unknown Mean, Approximate Covariance, to Approximate Recovery}
For each $X_i$, let $X''_i = \Sigmahat^{-1/2} X_i$.
Then, for $X_i$ which is not corrupted, we have that $X_i'' \sim \normal (\widehat{\Sigma}^{-1/2} \mu, \Sigma_1)$, where $\Sigma_1 =  \Sigmahat^{-1/2} \Sigma \Sigmahat^{-1/2}$.
By Corollary \ref{cor:kl-to-cov} and Lemma \ref{lem:inverse}, if (\ref{eqn:full-cov}) holds, then we have 
\[
\dtv (\normal (\widehat{\Sigma}^{-1/2} \mu, \Sigma_1), \normal (\widehat{\Sigma}^{-1/2} \mu, I)) \leq O(\ve \log 1 / \ve) \; .
\]
By Claim \ref{claim:full-to-oblivious}, this means that if (\ref{eqn:full-cov}) holds, the uncorrupted set of $X_i''$ can be treated as an $O(\ve \log 1 / \ve)$-corrupted set of samples from $\normal (\widehat{\Sigma}^{-1/2} \mu, I)$.
Thus, if (\ref{eqn:full-cov}) holds, the entire set of samples $X''_1, \ldots, X''_m$ is a $O(\ve \log 1 / \ve)$-corrupted set of samples from $\normal (\widehat{\Sigma}^{-1/2} \mu, I)$.
Then, by using results from Section \ref{sec:UnknownMeanConvex}, with probability $1 - \tau$, assuming that \ref{eqn:full-cov} holds, we can recover a $\muhat$ such that $\| \muhat - \Sigmahat^{-1/2} \mu \|_2 \leq O(\ve \log^{3/2} (1 / \ve) )$.
Thus, by Corollary \ref{cor:kl-to-means}, this implies that 
\[
\dtv (\normal (\muhat, I), \normal (\Sigmahat^{-1/2} \mu, I)) \leq O(\ve \log^{3/2} (1 / \ve) ) \; ,
\]
or equivalently,
\[
\dtv (\normal (\Sigmahat^{1/2} \muhat, \Sigmahat), \normal (\mu, \Sigmahat)) \leq O(\ve \log^{3/2} (1 / \ve) ) \; ,
\]
which in conjunction with (\ref{eqn:full-cov-tv}), implies that 
\[
\dtv (\normal (\Sigmahat^{1/2} \muhat, \Sigmahat), \normal (\mu, \Sigma)) \leq O(\ve \log^{3/2} (1 / \ve) ) \; ,
\]
and thus by following this procedure, whose formal pseudocode is given in Algorithm \ref{alg:gaussian}, we have shown the following:
\begin{algorithm}[htb]
\begin{algorithmic}[1]
\Function{RecoverRobustGuassian}{$\ve, \tau, X_1, \ldots, X_N$}
\State For $i = 1, \ldots, N / 2$, let $X_i' = (X_i - X_{N / 2 + i}) / \sqrt{2}$.
\State Let $\Sigmahat \gets \textsc{LearnCovariance}(\ve, \tau, X_1', \ldots, X_{N / 2}')$.
\State For $i = 1, \ldots, N$, let $X_i'' = \Sigmahat^{-1/2} X_i$.
\State Let $\muhat \gets \textsc{LearnMean}(\ve, \tau, X_1'', \ldots, X_N'')$.
\State \textbf{return} the Gaussian with mean $\Sigmahat^{1/2} \muhat$, and covariance $\Sigmahat$.
\EndFunction
\end{algorithmic}
\caption{Algorithm for learning an arbitrary Gaussian robustly}
\label{alg:gaussian}
\end{algorithm}

\begin{theorem}
Fix $\ve, \tau > 0$.
Let $X_1, \ldots, X_N$ be an $\ve$-corrupted set of samples from $\normal (\mu, \Sigma)$, where $\mu, \Sigma$ are both unknown, and 
\[
N = \widetilde{\Omega} \left( \frac{d^2 \log^5 1 / \tau}{\ve^2} \right) \; .
\]
There is a polynomial-time algorithm $\textsc{RecoverRobustGaussian} (\ve, \tau, X_1, \ldots, X_N)$ which with probability $1 - \tau$, outputs a $\Sigmahat, \muhat$ such that 
\[
\dtv (\normal (\Sigmahat^{1/2} \muhat, \Sigmahat), \normal (\mu, \Sigma)) \leq O(\ve \log^{3/2} (1 / \ve) ) \; .
\]
\end{theorem}

\section{Agnostically Learning a Gaussian, via Filters}
\label{sec:filterGaussian}
%!TEX root = ./main.tex

\subsection{Learning a Gaussian With Unknown Mean} \label{sec:filter-gaussian-mean}

In this section, we use our filter technique to give an agnostic learning algorithm for an unknown mean Gaussian
with known covariance matrix. More specifically, we prove:

\begin{theorem} \label{thm:filter-gaussian-mean}
Let $G$ be a Gaussian distribution on $\R^d$ with
mean $\mu^G$, covariance matrix $I$, and $\eps, \tau > 0$.
Let $S'$ be an $\eps$-corrupted set of samples from $G$ of size
\new{$\Omega((d/\eps^2) \poly\log(d/\eps\tau))$}. 
There exists an efficient algorithm that, on input $S'$ and $\eps>0$, returns a mean vector $\widehat{\mu}$
such that with probability at least $1-\tau$ we have $\|\widehat{\mu}-\mu^{G}\|_2 = O(\eps\sqrt{\log(1/\eps)}).$
\end{theorem}

\medskip

\noindent {\bf Notation.} We will denote $\mu^S = \frac{1}{|S|}\sum_{X\in S} X$ and 
$M_S=\frac{1}{|S|}\sum_{X\in S} (X-\mu^G)(X-\mu^G)^T$ for the sample mean and 
modified sample covariance matrix of the set $S$.

\medskip

We start by defining our notion of good sample, i.e, 
a set of conditions on the uncorrupted set of samples under which our algorithm will succeed.

\begin{definition} \label{def:good-set}
Let $G$ be an identity covariance Gaussian in $d$ dimensions 
with mean $\mu^G$ and covariance matrix $I$, and $\eps,\tau >0$.
We say that a multiset $S$ of elements in $\R^d$ is {\em $(\eps,\tau)$-good with respect to $G$}
if the following conditions are satisfied:
\begin{itemize}
\item[(i)] For all $x \in S$ we have $\|x-\mu^G\|_2 \leq O(\sqrt{d \log (|S|/\tau)})$.

\item[(ii)]  For every affine function $L:\R^d \to \R$ \new{such that $L(x) = v \cdot (x-\mu^G)-T$, $\|v\|_2=1$,}
we have that 
$\left|\Pr_{X \in_u S}[L(X) \ge 0] - \Pr_{X\sim G}[L(X) \ge 0] \right| \leq \frac{\eps}{\new{T^2\log\left(d \log(\frac{d}{\eps\tau})\right)}} \;.$

\item[(iii)]  We have that $\|\mu^S - \mu^G \|_2\leq \eps.$

\item[(iv)]  We have that $\left\|M_S- I \right\|_2 \leq \new{\eps}.$
\end{itemize}
\end{definition}

We show in Appendix~\ref{sec:filterGaussianAppendix} that a sufficiently large set of independent samples from $G$
is $(\eps,\tau)$-good (with respect to $G$) with high probability. Specifically, we prove:

\begin{lemma} \label{lem:random-good-gaussian-mean}
Let $G$ be a Gaussian distribution with identity covariance, and $\eps, \tau>0.$
If the multiset $S$ is obtained by taking \new{$\Omega( (d/\eps^2) \poly\log(d/\eps\tau))$} independent samples from $G,$
it is $(\eps,\tau)$-good with respect to $G$ with probability at least $1-\tau .$
\end{lemma}

We require the following definition that quantifies the extent to which a multiset has been corrupted:

\begin{definition} \label{def:Delta-G}
Given finite multisets $S$ and $S'$ we let $\Delta(S,S')$
be the size of the symmetric difference of $S$ and $S'$ divided by the cardinality of $S.$
\end{definition}

As in the convex program case, we will first use~\textsc{NaivePrune} to remove points which are far from the mean. 
Then, we iterate the algorithm whose performance guarantee is given by the following:

\begin{proposition} \label{prop:filter-gaussian-mean}
Let $G$ be a Gaussian distribution on $\R^d$ with mean $\mu^G$, covariance matrix $I$, $\eps>0$ be sufficiently  small and $\tau > 0$.
Let $S$ be an $(\eps,\tau)$-good set with respect to $G$.
Let $S'$ be any multiset with $\Delta(S,S') \leq 2\eps$ and for any $x,y \in S'$, $\|x-y\|_2 \leq  O(\sqrt{d \log(d/\eps\tau)})$.
There exists a polynomial time algorithm \textsc{Filter-Gaussian-Unknown-Mean}
that, given $S'$ and $\eps>0,$ returns one of the following:
\begin{enumerate}
\item[(i)]  A mean vector $\widehat{\mu}$ such that $\|\widehat{\mu}-\mu^{G}\|_2 = O(\eps\sqrt{\log(1/\eps)}).$
\item[(ii)] A multiset $S'' \subseteq S'$ such that $\Delta(S,S'') \leq \Delta(S,S') - \eps/\new{\alpha}$, where
\new{$\alpha \eqdef d\log\left( \frac{d}{\eps\tau} \right)\log \left(d \log(\frac{d}{\eps\tau}) \right)$.}
\end{enumerate}
\end{proposition}

We start by showing how Theorem~\ref{thm:filter-gaussian-mean}  follows
easily from Proposition~\ref{prop:filter-gaussian-mean}.

\begin{proof}[Proof of Theorem~\ref{thm:filter-gaussian-mean}]
By the definition of $\Delta(S, S'),$ since $S'$ has been obtained from $S$
by corrupting an $\eps$-fraction of the points in $S,$ we have that
$\Delta(S, S') \le 2\eps.$ By Lemma~\ref{lem:random-good-gaussian-mean}, the set $S$ of uncorrupted samples
is $(\eps,\tau)$-good with respect to $G$ with probability at least $1-\tau.$
We henceforth condition on this event.

Since $S$ is $(\eps,\tau)$-good, all $x \in S$ 
have $\|x-\mu^G\|_2 \leq O(\sqrt{d \log |S|/\tau})$. 
Thus, the \textsc{NaivePrune} procedure does not remove from $S'$ any member of $S$. 
Hence, its output, $S''$, has $\Delta(S, S'') \leq \Delta(S, S')$ 
and for any $x \in S''$, there is a $y \in S$ with $\|x-y\|_2  \leq O(\sqrt{d \log |S|/\tau})$. 
By the triangle inequality, for any $x,z  \in S''$, $\|x-z\|_2 \leq O(\sqrt{d \log |S|/\tau})= O(\sqrt{d \log (d/\eps\tau}))$.

Then, we iteratively apply the \textsc{Filter-Gaussian-Unknown-Mean} procedure 
of Proposition~\ref{prop:filter-gaussian-mean} until it terminates
returning a mean vector $\mu$ with $\|\widehat{\mu}-\mu^{G}\|_2 = O(\eps\sqrt{\log(1/\eps)}).$
We claim that we need at most \new{$O(\alpha)$} iterations for this to happen.
Indeed, the sequence of iterations results in a sequence of sets $S_i',$
such that $\Delta(S,S_i') \leq \Delta(S,S')  - i \cdot \eps/\new{\alpha}.$
Thus, if we do not output the empirical mean in the first \new{$2\alpha$} iterations, 
in the next iteration there are no outliers left.
Hence in the next iteration it is impossible for the algorithm to output a subset satisfying Condition (ii) of Proposition~\ref{prop:filter-gaussian-mean}, so it must output a mean vector satisfying (i), as desired.
\end{proof}

\subsubsection{Algorithm \textsc{Filter-Gaussian-Unknown-Mean}: Proof of Proposition~\ref{prop:filter-gaussian-mean}}

In this subsection, we describe the efficient algorithm establishing
Proposition~\ref{prop:filter-gaussian-mean} and prove its correctness.
Our algorithm calculates the empirical mean vector $\mu^{S'}$ and empirical covariance matrix $\Sigma$.
If the matrix $\Sigma$ has no large eigenvalues, it returns $\mu^{S'}.$
Otherwise, it uses the eigenvector $v^{\ast}$ corresponding to the maximum magnitude eigenvalue of $\Sigma$
and the mean vector $\mu^{S'}$ to define a filter.
Our efficient filtering procedure is presented in detailed pseudocode below.

\begin{algorithm}%[htb]
\begin{algorithmic}[1]
\Procedure{Filter-Gaussian-Unknown-Mean}{$S',\eps,\tau$}
\INPUT A multiset $S'$ such that there exists an $(\eps,\tau)$-good $S$ with $\Delta(S, S') \le 2\eps$
\OUTPUT Multiset $S''$ or mean vector $\widehat{\mu}$ satisfying Proposition~\ref{prop:filter-gaussian-mean}
\State Compute the sample mean $\mu^{S'}=\E_{X\in_u S'}[X]$ and the sample covariance matrix $\Sigma$ ,
i.e., $\Sigma = (\Sigma_{i, j})_{1 \le i, j \le d}$
with $\Sigma_{i,j} = \E_{X\in_u S'}[(X_i-\mu^{S'}_i) (X_j-\mu^{S'}_j)]$.
\State  Compute approximations for the largest absolute eigenvalue of $\Sigma - I$, $\lambda^{\ast} := \|\Sigma - I\|_2,$
and the associated unit eigenvector $v^{\ast}.$

\If {$\|\Sigma - I\|_2 \leq O(\eps \log (1/\eps)),$}
 \textbf{return} $\mu^{S'}.$ \label{step:bal-small-G}
\EndIf
\State \label{step:bal-large-G}  Let $\delta := 3 \sqrt{ \eps  \|\Sigma - I\|_2}.$ Find $T>0$ such that
$$
\Pr_{X\in_u S'}\left[|v^{\ast} \cdot (X-\mu^{S'})|>T+\delta \right] > 8\exp(-T^2/2)+8 \frac{\eps}{\new{T^2\log\left(d \log(\frac{d}{\eps\tau})\right)}}.
$$
\State \label{step:gaussian-mean-filter} \textbf{return} the multiset $S''=\{x\in S': |v^{\ast} \cdot (x-\mu^{S'}) | \leq T+\delta\}$.

\EndProcedure
\end{algorithmic}
\caption{Filter algorithm for a Gaussian with unknown mean and identity covariance}
\label{alg:filter-Gaussian-mean}
\end{algorithm}

\subsubsection{Proof of Correctness of \textsc{Filter-Gaussian-Unknown-Mean}} \label{ssec:L2-setup-G}
By definition, there exist disjoint multisets $L,E,$ of points in $\R^d$, where $L \subset S,$
such that $S' = (S\setminus L) \cup E.$ With this notation, we can write $\Delta(S,S')=\frac{|L|+|E|}{|S|}.$
Our assumption $\Delta(S,S') \le 2\eps$ is equivalent to $|L|+|E| \le 2\eps \cdot |S|,$ and the definition of $S'$
directly implies that $(1-2\eps)|S| \le  |S'| \le (1+2\eps) |S|.$ Throughout the proof, we assume that $\eps$
is a sufficiently small constant. %The assumption that $\eps \leq 1/2$ suffices for our purposes.

We define $\mu^G,\mu^S,\mu^{S'},\mu^L,$ and $\mu^E$ to be the means of $G,S,S',L,$ and $E$, respectively.

Our analysis will make essential use of the following matrices:
\begin{itemize}
\item $M_{S'}$ denotes $\E_{X\in_u S'}[(X-\mu^G)(X-\mu^G)^T]$,
\item $M_{S}$ denotes $\E_{X\in_u S}[(X-\mu^G)(X-\mu^G)^T]$,
\item $M_{L}$ denotes $\E_{X\in_u L}[(X-\mu^G)(X-\mu^G)^T]$, and
\item $M_{E}$ denotes $\E_{X\in_u E}[(X-\mu^G)(X-\mu^G)^T]$.
\end{itemize}

Our analysis will hinge on proving the important claim that $\Sigma-I$ is approximately $(|E|/|S'|)M_E$. 
This means two things for us. First, it means that if the positive errors align in some direction 
(causing $M_E$ to have a large eigenvalue), there will be a large eigenvalue in $\Sigma-I$. 
Second, it says that any large eigenvalue of $\Sigma-I$ will correspond to an eigenvalue of $M_E$, 
which will give an explicit direction in which many error points are far from the \new{empirical} mean.

\medskip

\noindent {\bf Useful Structural Lemmas.} 
We will use the following simple fact about the concentration of Gaussian random variables:

\begin{fact} \label{fact:tail-bound} 
If $G$ is Gaussian on $\R^d$ with mean vector $\mu$, 
then for any unit vector $v \in \R^d$ we have that $\Pr_{X \sim G}\left[|v \cdot (X-\mu)| \geq T \right] \leq \exp(-t^2/2)$.
\end{fact}

We begin by noting that we have concentration 
bounds on $G$ and therefore, on $S$ due to its goodness.
%\new{Note that with our new definition of a good set, the second statement below (for $S$) may not hold for all $T>0$, 
%but rather for all $ \sqrt{\log(1/\eps)} << T << \sqrt{d}$. This will come up when we have the corresponding lemma, but does not 
%seem to affect the proof, as we only integrate over this range.}
\begin{fact}\label{fact:conc}
Let $w \in \R^d$ be any unit vector, then for any $T>0$,
$\Pr_{X\sim G}\left[|w\cdot(X-\mu^{G})| > T\right] \leq 2 \exp(-T^2/2)$ and
$\Pr_{X\in_u S}\left[|w\cdot(X-\mu^{G})| > T\right] \leq 2 \exp(-T^2/2)+\frac{\eps}{\new{T^2\log\left(d \log(\frac{d}{\eps\tau})\right)}}$.
\end{fact}
\begin{proof}
The first line is Fact \ref{fact:tail-bound}, and the former follows from it using the goodness of $S$.
\end{proof}

By using the above fact, we obtain the following simple claim:

\begin{claim}\label{claim:conc}
Let $w \in \R^d$ be any unit vector, then for any $T>0$, we have that:
$$
\Pr_{X\sim G}[|w\cdot(X-\mu^{S'})| > T + \|\mu^{S'}-\mu^G\|_2] \leq 2 \exp(-T^2/2).
$$
and
$$
\Pr_{X\in_u S}[|w\cdot(X-\mu^{S'})| > T + \|\mu^{S'}-\mu^G\|_2] \leq 2 \exp(-T^2/2)+\frac{\eps}{\new{T^2\log\left(d \log(\frac{d}{\eps\tau})\right)}}.
$$
\end{claim}
\begin{proof}
This follows from Fact \ref{fact:conc} upon noting that $|w\cdot (X-\mu^{S'})| > T +\|\mu^{S'}-\mu^G\|_2$ only if $|w\cdot (X-\mu^G)|>T$.
\end{proof}

We can use the above facts to prove concentration bounds for $L$. In particular, 
we have the following lemma:

\begin{lemma}\label{LBoundCor}
We have that $\|M_L\|_2 = \new{O\left(\log(|S|/|L|) + \eps |S| / |L| \right)}.$
\end{lemma}
\begin{proof}
Since $L \subseteq S$, for any $x \in \R^d$, we have that
\begin{equation} \label{eqn:L-subset-S}
|S| \cdot \Pr_{X \in_u S}(X= x) \geq |L| \cdot \Pr_{X \in_u L}(X= x) \;.
\end{equation}
Since $M_L$ is a symmetric matrix, we have $\|M_L\|_2 = \max_{\|v\|_2=1} |v^T M_L v|.$ So,
to bound $\|M_L\|_2$ it suffices to bound $|v^T M_L v|$ for unit vectors $v.$
By definition of $M_L,$
for any $v \in \R^d$ we have that
$$|v^T M_L v| = \E_{X \in_u L}[|v\cdot (X-\mu^{G})|^2].$$
For unit vectors $v$, the RHS is bounded from above as follows:
\begin{align*}
\E_{X \in_u L}\left[|v\cdot (X-\mu^{G})|^2 \right] & =  2 \int_{0}^{\infty} \Pr_{X\in_u L}\left[|v\cdot(X-\mu^G)|>T\right] T dT\\
& = 2 \int_{0}^{O(\sqrt{d \log(d/\eps\tau)})} \Pr_{X\in_u L}[|v\cdot(X-\mu^G)|>T] T dT \\
& \leq 2 \int_0^{O(\sqrt{d \log(d/\eps\tau)})} \min \left\{ 1,  \frac{|S|}{|L|} \cdot \Pr_{X \in_u S}\left[|v \cdot(X-\mu^{G})|>T\right]  \right\} TdT\\
& \ll \int_0^{4\sqrt{\log(|S|/|L|)}} T dT \\
& + (|S|/|L|) \int_{4\sqrt{\log(|S|/|L|)}}^{O(\sqrt{d \log(d/\eps\tau)})}  \Big( \exp(-T^2/2)+\frac{\eps}{\new{T^2\log\left(d \log(\frac{d}{\eps\tau})\right)}} \Big)  T dT\\
& \ll  \log(|S|/|L|) + \eps \cdot |S|/|L| \;,
\end{align*}
where the second line follows from the fact that $\|v\|_2 = 1$, $L \subset S$, and $S$ satisfies condition (i) of Definition~\ref{def:good-set};
the third line follows from (\ref{eqn:L-subset-S}); and the fourth line follows from Fact~\ref{fact:conc}.
\end{proof}

As a corollary, we can relate the matrices $M_{S'}$ and $M_E$, in spectral norm:

\begin{corollary}\label{MSPrimeBound}
We have that 
$M_{S'} - I = (|E|/|S'|)M_E + O(\eps\log(1/\eps))$,
where the $O(\eps\log(1/\eps))$ term denotes a matrix of spectral norm $O(\eps\log(1/\eps))$.
\end{corollary}
\begin{proof}
By definition, we have that $|S'|M_{S'} = |S|M_S - |L|M_L + |E|M_E$.
Thus, we can write
\begin{align*}
M_{S'} & = (|S|/|S'|) M_S - (|L|/|S'|)M_L + (|E|/|S'|)M_E\\
& = I + O(\eps) + O(\eps \log(1/\eps)) + (|E|/|S'|)M_E \;,
\end{align*}
where the second line uses the fact that $1-2\eps \leq |S|/|S'| \leq 1+2\eps$, 
the goodness of $S$ (condition (iv) in Definition~\ref{def:good-set}), and Lemma~\ref{LBoundCor}.
Specifically, Lemma~\ref{LBoundCor} implies that $(|L|/|S'|) \|M_L\|_2 = O(\eps \log(1/\eps))$.
Therefore, we have that
$$
M_{S'} = I + (|E|/|S'|)M_E + O(\eps\log(1/\eps)) \;,
$$
as desired.
\end{proof}

We now establish a similarly useful bound on the difference between the mean vectors:

\begin{lemma}\label{meansBoundLem}
We have that $\mu^{S'}-\mu^G = (|E|/|S'|)(\mu^E-\mu^G) + O(\eps\sqrt{\log(1/\eps)})$,
where the $O(\eps\sqrt{\log(1/\eps)})$ term denotes a vector with $\ell_2$-norm at most $O(\eps\sqrt{\log(1/\eps)})$.
\end{lemma}
\begin{proof}
By definition, we have that
\begin{align*}
|S'|(\mu^{S'}-\mu^G) = |S|(\mu^S-\mu^G) - |L|(\mu^L-\mu^G) + |E|(\mu^E-\mu^G).
\end{align*}
Since $S$ is a good set, by condition (iii) of Definition~\ref{def:good-set}, we have $\|\mu^S-\mu^G\|_2 = O(\eps)$. 
Since $1-2\eps \leq |S|/|S'| \leq 1+2\eps$, it follows that $(|S|/|S'|)\|\mu^S-\mu^G\|_2 = O(\eps)$.
Using the valid inequality $\|M_L\|_2 \geq \|\mu^L-\mu^G\|_2^2$ and Lemma~\ref{LBoundCor}, we obtain that
$\|\mu^L-\mu^G\|_2 \leq O\left(\sqrt{\log(|S|/|L|)} + \sqrt{\eps |S| / |L|} \right)$. Therefore, 
$$(|L|/|S'|)  \|\mu^L-\mu^G\|_2 \leq O\left( (|L|/|S|) \sqrt{\log(|S|/|L|)} + \sqrt{\eps |L| / |S|} \right) = O(\eps\sqrt{\log(1/\eps)}) \;.$$ 
In summary,
$$
\mu^{S'}-\mu^G = (|E|/|S'|)(\mu^E-\mu^G) + O(\eps\sqrt{\log(1/\eps)}) \;,
$$
as desired. This completes the proof of the lemma.
\end{proof}

By combining the above, we can conclude that $\Sigma-I$ is approximately proportional to $M_E$.
More formally, we obtain the following corollary:

\begin{corollary}\label{MApproxCor-G}
We have 
$\Sigma-I = (|E|/|S'|)M_E + O(\eps\log(1/\eps))+O(|E|/|S'|)^2\|M_E\|_2$, where the additive terms denote matrices
of appropriately bounded spectral norm.
\end{corollary}
\begin{proof}
By definition, we can write
$\Sigma-I = M_{S'}-I-(\mu^{S'}-\mu^G)(\mu^{S'}-\mu^G)^T.$
Using Corollary~\ref{MSPrimeBound} and Lemma~\ref{meansBoundLem}, we obtain:
\begin{align*}
\Sigma-I  &= (|E|/|S'|)M_E + O(\eps\log(1/\eps)) +O((|E|/|S'|)^2\|\mu^E-\mu^G\|_2^2) + \new{O(\eps^2\log(1/\eps)) } \\ 
                 &= (|E|/|S'|)M_E + O(\eps\log(1/\eps)) + O(|E|/|S'|)^2 \|M_E\|_2 \;,
\end{align*}
where the second line follows from the valid inequality  $\|M_E\|_2 \geq \|\mu^E-\mu^G\|_2^2$.
This completes the proof.
\end{proof}

\medskip

\noindent {\bf Case of Small Spectral Norm.} 
We are now ready to analyze the case that the mean vector $\mu^{S'}$ is returned by the algorithm in Step \ref{step:bal-small-G}.
In this case, we have that $\lambda^{\ast} \eqdef \|\Sigma-I\|_2 = O(\eps\log(1/\eps))$. Hence, Corollary \ref{MApproxCor-G} yields that
$$
(|E|/|S'|) \|M_E\|_2 \leq \lambda^{\ast} + O(\eps\log(1/\eps)) + O(|E|/|S'|)^2\|M_E\|_2 \;,
$$
which in turns implies that 
$$
(|E|/|S'|) \|M_E\|_2 = O(\eps\log(1/\eps)) \;.
$$
On the other hand, since $\|M_E\|_2 \geq \|\mu^E-\mu^G\|_2^2$, Lemma \ref{meansBoundLem} gives that
$$
\|\mu^{S'}-\mu^G\|_2 \leq (|E|/|S'|) \sqrt{\|M_E\|_2} + O(\eps\sqrt{\log(1/\eps)}) = O(\eps\sqrt{\log(1/\eps)}).
$$
\new{This proves part (i) of Proposition~\ref{prop:filter-gaussian-mean}.}

\medskip

\noindent {\bf Case of Large Spectral Norm.} 
We next show the correctness of the algorithm when it returns a filter in Step \ref{step:bal-large-G}.

We start by proving that if $\lambda^{\ast} \eqdef \|\Sigma-I\|_2  > C\eps\log(1/\eps)$, for a sufficiently large universal constant $C$,
then a value $T$ satisfying the condition in Step \ref{step:bal-large-G} exists. We first note that $\|M_E\|_2$ is 
appropriately large. Indeed, by Corollary \ref{MApproxCor-G} and the assumption that $\lambda^{\ast}  > C\eps\log(1/\eps)$
we deduce that
\begin{equation} \label{eqn:large-me}
(|E|/|S'|) \|M_E\|_2  = \Omega (\lambda^{\ast}) \;.
\end{equation}
Moreover, using the inequality $\|M_E\|_2 \geq \|\mu^E-\mu^G\|_2^2$ and Lemma \ref{meansBoundLem} as above,
we get that
\begin{equation} \label{eqn:delta}
\|\mu^{S'}-\mu^G\|_2 \leq (|E|/|S'|) \sqrt{\|M_E\|_2} + O(\eps\sqrt{\log(1/\eps)}) \leq \delta/2 \;, 
\end{equation}
where we used the fact that $\delta \eqdef \sqrt{\eps \lambda^{\ast}} > C' \eps \sqrt{\log(1/\eps)}.$

Suppose for the sake of contradiction that for all $T>0$ we have that
\begin{equation*} 
\Pr_{X\in_u S'}\left[|v^{\ast} \cdot (X-\mu^{S'})|>T+\delta \right] \leq 8\exp(-T^2/2)+8 \frac{\eps}{\new{T^2\log\left(d \log(\frac{d}{\eps\tau})\right)}} \;.
\end{equation*}
Using (\ref{eqn:delta}), we obtain that for all $T>0$ we have that
\begin{equation} \label{eqn:contradiction}
\Pr_{X\in_u S'}\left[|v^{\ast} \cdot (X-\mu^{G})|>T+\delta/2 \right] \leq 8\exp(-T^2/2)+8 \frac{\eps}{\new{T^2\log\left(d \log(\frac{d}{\eps\tau})\right)}} \;.
\end{equation}
Since $E \subseteq S',$ for all $x \in \R^d$ we have that
$|S'|\Pr_{X\in_u S'}[X=x] \geq |E| \Pr_{Y\in_u E}[Y=x].$
This fact combined with (\ref{eqn:contradiction}) implies that for all $T>0$
\begin{equation} \label{eqn:contradiction2}
\Pr_{X\in_u E}\left[|v^{\ast} \cdot (X-\mu^{G})|>T+\delta/2 \right]  \leq C (|S'|/|E|)\left(\exp(-T^2/2)+ \frac{\eps}{\new{T^2\log\left(d \log(\frac{d}{\eps\tau})\right)}} \right) \; ,
\end{equation}
for some universal constant $C''$.

We now have the following sequence of
inequalities:
\begin{align*}
\|M_E\|_2 &=  \E_{X\in_u E}\left[|v^{\ast}\cdot (X-\mu^G)|^2\right]  = 2 \int_{0}^{\infty} \Pr_{X\in_u E}\left[|v^{\ast}\cdot(X-\mu^G)|>T\right] T dT\\
& = 2 \int_{0}^{O(\sqrt{d \log(d/\eps\tau)})} \Pr_{X\in_u E}\left[|v^{\ast} \cdot(X-\mu^G)|>T\right] T dT \\
& \leq 2 \int_0^{O(\sqrt{d \log(d/\eps\tau)})} \min \left\{ 1,  \frac{|S'|}{|E|} \cdot \Pr_{X \in_u S'}\left[|v^{\ast} \cdot(X-\mu^{G})|>T\right]  \right\} TdT\\
&\leq \int_0^{4\sqrt{\log(|S'|/|E|)}+\delta} T dT + C''  \frac{|S'|}{|E|} \int_{4\sqrt{\log(|S'|/|E|)} + \delta}^{O(\sqrt{d \log(d/\eps\tau)})}  
\Big( \exp(-T^2/2)+\frac{\eps}{\new{T^2\log\left(d \log(\frac{d}{\eps\tau})\right)}} \Big)  T dT\\
&\leq \int_0^{4\sqrt{\log(|S'|/|E|)}+\delta} T dT + C'' \frac{|S'|}{|E|} \left( \int_{4\sqrt{\log(|S'|/|E|)} + \delta}^{\infty}  
\Big( \exp(-T^2/2) \Big)  T dT + O(\eps) \right)\\
& \leq \log(|S'|/|E|) + \delta^2 + O(1) + O(\eps) \cdot |S'|/|E| \\
& \leq \log(|S'|/|E|) + \eps \lambda^{\ast} + O(\eps) \cdot |S'|/|E| \; ,
\end{align*}
for 

Rearranging the above, we get that
$$(|E|/|S'|) \|M_E\|_2  \ll (|E|/|S'|)\log(|S'|/|E|) + (|E|/|S'|) \eps \lambda^{\ast} + O(\eps) = O(\eps \log(1/\eps) + \eps^2 \lambda^{\ast}).$$
Combined with (\ref{eqn:large-me}), we obtain 
$\lambda^{\ast} = O(\eps \log(1/\eps))$, which is a contradiction if $C$ is sufficiently large.
Therefore, it must be the case that for some value of $T$ the condition in Step \ref{step:bal-large-G} is satisfied.

The following claim completes the proof:
\begin{claim} \label{claim:filter}
Fix $\alpha \eqdef \new{d\log(d/\eps\tau)\log(d \log(\frac{d}{\eps\tau}))}.$
We have that
$
\Delta(S,S'') \leq \Delta(S,S') - 2\eps/\new{\alpha} \;.
$
\end{claim}
\begin{proof}
Recall that $S' = (S\setminus L) \cup E,$ with $E$ and $L$ disjoint multisets such that $L \subset S.$
We can similarly write $S''=(S \setminus L') \cup E',$ with $L'\supseteq L$ and $E'\subset E.$
Since $$\Delta(S,S') - \Delta(S,S'')  = \frac{|E \setminus E'| - |L' \setminus L| }{|S|},$$
it suffices to show that $|E \setminus E'| \geq |L' \setminus L| + \eps|S|/\new{\alpha}.$
Note that $|L' \setminus L|$ is the number of points rejected by the filter that lie in $S \cap S'.$
%By Claim~\ref{cernCor-G} and Claim~\ref{claim:mean-l2-delta-G}, it follows that 
Note that the fraction of elements of $S$
that are removed to produce $S''$ (i.e., satisfy $|v^{\ast}\cdot(x-\mu^{S'})|>T+\delta$) is at most $2\exp(-T^2/2) + \eps/\new{\alpha}$.
\new{This follows from Claim~\ref{claim:conc} and the fact that $T  = O(\sqrt{ d\log(d/\eps\tau)})$.}

Hence, it holds that $|L' \setminus L| \leq (2\exp(-T^2/2) + \eps/\new{\alpha}) |S|.$
On the other hand, Step~\ref{step:bal-large-G} of the algorithm ensures that the fraction of elements of $S'$ that are rejected
by the filter is at least $8\exp(-T^2/2)+8\eps/\new{\alpha})$. Note that
$|E \setminus E'|$ is the number of points rejected by the filter that lie in $S' \setminus S.$
Therefore, we can write:
\begin{align*}
|E\setminus E'| & \geq (8\exp(-T^2/2)+8\eps/\new{\alpha})|S'| - (2\exp(-T^2/2) + \eps/\new{\alpha})  |S| \\
				& \geq (8\exp(-T^2/2)+8\eps/\new{\alpha})|S|/2 - (2\exp(-T^2/2) + \eps/\new{\alpha}) |S| \\
				& \geq  (2\exp(-T^2/2) + 3\eps/\new{\alpha}) |S| \\
				& \geq |L' \setminus L| + 2\eps|S|/\new{\alpha}  \;,
\end{align*}
where the second line uses the fact that $|S'| \ge |S|/2$
and the last line uses the fact that $|L' \setminus L| / |S| \leq 2\exp(-T^2/2) + \eps/\new{\alpha}.$
Noting that $\log(d/\eps\tau) \geq 1$, this completes the proof of the claim.
\end{proof}

%!TEX root = ./main.tex

\subsection{Learning a Gaussian With Unknown Covariance} \label{sec:filter-gaussian-cov}

In this subsection, we use our filter technique to agnostically learn a Gaussian with zero mean vector
and unknown covariance. By combining the algorithms of the current and the previous subsections, as in our convex programming
approach (Section~\ref{sec:UnknownGaussianWhole}), 
we obtain a filter-based algorithm to agnostically learn an arbitrary unknown Gaussian.

The main result of this subsection is the following theorem:

\begin{theorem}\label{unknownCovarianceTheorem}
Let $G \sim \normal(0, \Sigma) $ be a Gaussian in $d$ dimensions 
with mean $0$ and unknown covariance, and let $\epsilon,\tau>0$. 
Let $S$ be an $\eps$-corrupted set of samples from $G$ of size \new{$\Omega((d^2/\eps^2) \poly\log(d/\eps\new{\tau}))$}.
There exists an efficient algorithm that, given $S$ and $\eps$, 
returns the parameters of a Gaussian distribution $G'  \sim \normal(0, \widehat{\Sigma})$ 
such that with probability at least $\new{1-\tau}$, 
it holds $\|I  - \Sigma^{-1/2} \widehat{\Sigma} \Sigma^{-1/2} \|_F  = O(\eps\log(1/\eps)).$
\end{theorem}

As in the previous subsection, we will need a condition on $S$ under which our algorithm will succeed.
\begin{definition} \label{def:good-sample-cov}
Let $G$ be a Gaussian in $\R^d$ with mean $0$ and covariance $\Sigma$. 
Let $\epsilon>0$ be sufficiently small. We say that a multiset $S$ of points in $\R^d$ is \emph{$(\eps,\tau)$-good with respect to $G$} if the following hold:
\begin{enumerate}
\item \label{farPoints} For all $x\in S$, $x^T \Sigma^{-1} x < O(d \log (|S|/\tau))$.
\item \label{covariance} We have that $\|\Sigma^{-1/2} \Cov(S) \Sigma^{-1/2} - I\|_F = O(\eps)$.
\item \label{cocovariance} For all even degree-$2$ polynomials $p$, we have that $\Var(p(S)) = \Var(p(G))(1+O(\eps))$.
\item \label{tails} For $p$ an even degree-$2$ polynomial with $\E[p(G)]=0$ and $\Var(p(G))=1$, and for any $T> 10\ln(1/\eps)$ we have that
$$
\Pr_{x\in_u S}(|p(x)| > T) \leq \eps/(T^2 \log^2(T)).
$$
\end{enumerate}
\end{definition}

Let us first note some basic properties of such polynomials on a normal distribution.
The proof of this lemma is deferred to Section \ref{sec:filterGaussianAppendix}.
\begin{lemma} \label{lem:evenp}
For any even degree-$2$ polynomial $p:\R^d\rightarrow \R$, 
we can write $p(x)=(\Sigma^{-1/2}x)^T P_2 (\Sigma^{-1/2}x) + p_0$, 
for a $d \times d$ symmetric matrix $P_2$ and $p_0 \in \R$. Then, for $X \sim G$, we have
\begin{enumerate}
\item $\E[p(X)] = p_0 + \tr(P_2) ,$
\item $\Var[p(X)] = 2 \|P_2\|_F^2$ and
\item For all $T > 1$, $\Pr(|p(X) - \E[p(X)]| \geq T) \leq 2 e^{1/3-2T/3\Var[p(X)]}$.
\item For all  $\delta > 0$, $\Pr(|p(X)| \leq \delta^2 ) \leq O(\delta).$
\end{enumerate}
\end{lemma}

We note that, if $S$ is obtained by taking random samples from $G$, 
then $S$ is good with high probability.
The proof of this lemma is also deferred to Section \ref{sec:filterGaussianAppendix}.
\begin{lemma}\label{GoodSamplesLemma}
Let $G$ be a $d$-dimensional Gaussian with mean $0$ and let $\epsilon, \tau>0$. 
Let $N$ be a sufficiently large constant multiple of $d^2 \log^5(d/\eps \tau)/\eps^2$. 
Then a set $S$ of $N$ independent samples from $G$ is $(\eps,\tau)$-good with respect to $G$ with probability at least \new{$1-\tau$}.
\end{lemma}

As in Definition \ref{def:Delta-G}, $\Delta(S,S')$ is the size of the symmetric difference of $S$ and $S'$ divided by $|S|$.

%\blue{I added the sample complexity bounds, but this section still needs some editing to fix the integrals
%with the new definition of good set and an overall pass for presentation.}

The basic thrust of our algorithm is as follows: By Lemma \ref{GoodSamplesLemma}, with high probability 
we have that $S$ is $(\eps, \tau)$-good with respect to $G$. The algorithm is then handed a new set $S'$ such that $\Delta(S,S')\leq 2\epsilon|S|$. 
The algorithm will run in stages. In each stage, the algorithm will either output $G'$ or 
will return a new set $S''$ such that $\Delta(S,S'') < \Delta(S,S')$. 
In the latter case, the algorithm will recurse on $S''$. We formalize this idea below:

\begin{proposition}\label{stepProposition}
There is an algorithm that given a finite set $S'\subset \R^d$, 
such that there is a mean $0$ Gaussian $G$ 
and a set $S$ that is $(\epsilon,\tau)$-good with respect to $G$ 
with $\Delta(S,S') \leq 2\epsilon|S|$, runs in time $\poly(d\log(1/\tau)/\eps)$ 
and returns either the parameters of a Gaussian $G'$ with $\dtv(G,G')\leq O(\epsilon\log(1/\epsilon))$ 
or a subset $S''$ of $\R^d$ with $\Delta(S,S'')<\Delta(S,S')$.
\end{proposition}

Given Proposition \ref{stepProposition}, the proof of Theorem \ref{unknownCovarianceTheorem} is straightforward. 
By Lemma \ref{GoodSamplesLemma} the original set $S$ is $(\epsilon,\tau)$-good with respect to $G$ with probability at least $1-\tau$. 
Then, $S'$ satisfies the hypotheses of Proposition \ref{stepProposition}. 
We then repeatedly iterate the algorithm from Proposition \ref{stepProposition} 
until it outputs a distribution $G'$ close to $G$. 
This must eventually happen because at every step the distance between $S$ 
and the set returned by the algorithm decreases by at least $1$.

\subsubsection{Analysis of Filter-based Algorithm: Proof of Proposition~\ref{stepProposition}}

We now turn our attention to the proof of Proposition \ref{stepProposition}. 
We first define the matrix $\Sigma'$ to be $\E_{X\in S'}[XX^T]$, 
and let $G'$ be the mean $0$ Gaussian with covariance matrix $\Sigma'$. 
Our goal will be to either obtain a certificate that $G'$ is close to $G$ 
or to devise a filter that allows us to clean up $S'$ by removing some elements, 
most of which are not in $S$. 
The idea here is the following: We know by Corollary \ref{cor:kl-to-cov} that $G$ and $G'$ 
are close unless $I-\Sigma^{-1/2} \Sigma'\Sigma^{-1/2}$ has large Frobenius norm. 
This happens if and only if there is some matrix $M$ with $\|M\|_F=1$ 
such that $$\tr(M\Sigma^{-1/2}\Sigma'\Sigma^{-1/2}-M) = \E_{X\in_u S'} [(\Sigma^{-1/2}X)^T M (\Sigma^{-1/2}X) - \tr(M)]$$ is far from $0$.
On the other hand, we know that the distribution of $p(X)=(\Sigma^{-1/2}X)^T M (\Sigma^{-1/2}X) - \tr(M)$ for $X\in_u S$ is approximately that of $p(G)$, 
which is a variance $O(1)$ polynomial of Gaussians with mean $0$. 
In order to substantially change the mean of this function, 
while only changing $S$ at a few points, 
one must have several points in $S'$ for which $p(X)$ is abnormally large. 
This in turn will imply that the variance of $p(X)$ for $X$ from $S'$ must be large. 
This phenomenon will be detectable as a large eigenvalue of the matrix 
of fourth moments of $X\in S'$ (thought of as a matrix over the space of second moments). 
If such a large eigenvalue is detected, 
we will have a $p$ with $p(X)$ having large variance. 
By throwing away from $S'$ elements for which $|p|$ is too large, 
we will return a cleaner version of $S'$. The algorithm is as follows:

\begin{algorithm}[htb]
\begin{algorithmic}[1]
\Procedure{Filter-Gaussian-Unknown-Covariance}{$S',\eps,\tau$}
\INPUT A multiset $S'$ such that there exists an $(\eps,\tau)$-good $S$ with $\Delta(S, S') \le 2\eps$
\OUTPUT Either a set $S''$ with $\Delta(S,S'') < \Delta(S,S')$ or the parameters of a Gaussian $G'$ with $\dtv(G,G') = O(\epsilon\log(1/\epsilon)).$

Let $C>0$ be a sufficiently large universal constant.

\State Let $\Sigma'$ be the matrix $\E_{X\in_u S'}[XX^T]$ and let $G'$ be the mean $0$ Gaussian with covariance matrix $\Sigma'$.
\If {there is any $x\in S'$ such that $x^T(\Sigma')^{-1} x \geq Cd\log(|S'|/\tau)$}
\State \textbf{return} $S''=S'  \setminus \{x:x^T(\Sigma')^{-1} x \geq Cd\log(|S'|/\tau)\}$.\label{removeOutlierStep}
\EndIf
\State Let $L$ be the space of even degree-$2$ polynomials $p$ such that $\E_{X \sim G'}[p(X)]=0$.
\State Define two quadratic forms on $L$
\begin{enumerate}
\item[(i)] $Q_{G'}(p) = \E[p^2(G')] \;,$
\item[(ii)] $Q_{S'}(p) = \E_{X\in_u S'}[p^2(X)] \;.$
\end{enumerate}
\State Computing $\max_{p \in L \setminus \{0\}} Q_{S'}(p)/Q_{G'}(p)$ and the associated polynomial $p^{\ast}(x)$ normalized such that $Q_{G'}(p^{\ast})=1$ 
using \textsc{Find-max-poly} below.
%by finding the top eigenvalue and eigenvector of an appropriate matrix using the second and fourth moments of $S'$.
\If{$Q_{S'}(p^{\ast}) \leq (1+C\epsilon\log^2(1/\epsilon))Q_{G'}(p^{\ast})$}
\State \textbf{return} $G'$ \label{returnGStep}
\EndIf
\State  Let $\mu$ be the median value of $p^{\ast}(X)$ over $X\in S'$.
\State Find a $T \geq C'$ such that\new{
$$
\Pr_{X\in_u S'} (|p^{\ast}(X)-\mu| \geq T + 3) \geq \tail(T, d, \eps, \tau) \;,
$$
where $\tail(T, d, \eps, \tau)= 3\eps/(T^2 \log^2(T))$
when $ T \geq 10\ln(1/\eps)$, and $\tail(T, d, \eps, \tau)=1$  when $T < 10\log(1/\eps)$.
\label{thresholdStep}}
\State \textbf{return} $S'' = \{X\in S': |p^{\ast}(X)-\mu| < T\}.$ \label{filterStep}

\EndProcedure
\end{algorithmic}
\caption{Filter algorithm for a Gaussian with unknown covariance matrix.}
\label{alg:filter-Gaussian-covariance}
\end{algorithm}

\begin{algorithm}%[htb]
\begin{algorithmic}[1]
\Function{Find-max-poly}{$S',\Sigma'$}
\INPUT A multiset $S'$ and a Gaussian $G'=\normalpdf(0,\Sigma')$
\OUTPUT The even degree-$2$ polynomial $p^{\ast}(x)$ with $\E_{X \sim G'}[p^{\ast}(X)] \approx 0$ 
and $Q_{G'}(p^{\ast}) \approx 1$ that approximately maximizes $Q_{S'}(p^{\ast})$ and this maximum $\lambda^{\ast} = Q_{S'}(p^{\ast})$.
\State Compute an approximate eigen-decomposition of $\Sigma'$ and use it to compute $\Sigma'^{-1/2}$
\State Let $x_{(1)}, \ldots, x_{(|S'|)}$ be the elements of $S'$.
\State For $i = 1, \ldots, |S'|$, let $y_{(i)} = \Sigma'^{-1/2} x_{(i)}$ and $z_{(i)} = y_{(i)}^{\otimes 2}$.
\State Let $T_{S'}= \new{-I^\flat I^{\flat T}} + (1/|S'|)\sum_{i=1}^{|S'|} z_{(i)} z_{(i)}^T.$ 
\State Approximate the top eigenvalue $\lambda^{\ast}$ and corresponding unit eigenvector $v^{\ast}$ of $T_{S'}.$
\State Let $p^{\ast}(x) = \frac{1}{\sqrt{2}} ((\Sigma'^{-1/2} x)^T v^{\ast \sharp} (\Sigma'^{-1/2} x) - \tr(v^{\ast \sharp})).$
\State \textbf{return} $p^{\ast}$ and $\lambda^{\ast}\new{/2}.$
\EndFunction
\end{algorithmic}
\caption{Algorithm for maximizing $Q_{S'}(p)/Q_{G'}(p).$}
\label{alg:subroutine-find-poly}
\end{algorithm}

The function \textsc{Find-max-poly} uses similar notation to \textsc{SeparationOracleUnknownCovariance},
such that \textsc{Filter-Gaussian-Unknown-Covariance} 
and \textsc{SeparationOracleUnknownCovariance} can be more easily compared.

Let us first show that \textsc{Find-max-poly} is correct.
\begin{claim} 
Algorithm \textsc{Find-max-poly} is correct and \textsc{Filter-Gaussian-Unknown-Covariance} runs time $\poly(d \log \tau/\eps).$
\end{claim}
\begin{proof}
First, assume that we can compute all eigenvalues and eigenvectors exactly. 
By Lemma \ref{lem:evenp} all even polynomials with degree-$2$ 
that have $\E_{X \sim G}[p(X)]=0$ can be written as 
$p(x)=(\Sigma'^{-1/2} x)^T P_2 (\Sigma'^{-1/2} x) - \tr(P_2)$ for a symmetric matrix $P_2$. 
If we take $P_2= v^{\sharp}/\sqrt{2}$ for a unit vector $v$ such that $v^{\sharp}$ is symmetric, then $\Var_{X \sim G'}[p(X)]= 2 \|P_2\|_F= \|v_2\|=1$.

Note that since the covariance matrix of $S'$ is $\Sigma' \;,$ we have 
\begin{align*}
\E_{X \sim S'}[p(X)] & = \E_{X \sim S'}[(\Sigma'^{-1/2} X)^T P_2 (\Sigma'^{-1/2} X) - \tr(P_2)] \\
& = \E_{X \sim S'}[\tr( (X X^T) \Sigma'^{-1/2} P_2 \Sigma'^{-1/2}) ] - \tr(P_2) \\
& = \tr(\E_{X \sim S'}[(X X^T)] \Sigma'^{-1/2} P_2 \Sigma'^{-1/2}) - \tr(P_2) \\
& = \tr(\Sigma' \Sigma'^{-1/2} P_2 \Sigma'^{-1/2}) - \tr(P_2) = 0 \;.
\end{align*}
We let $T'$ be the multiset of $y = \Sigma^{-1/2} x$ for $x \in S'$ and $U'$ the multiset of $z=y^{\otimes 2}$ for $y$ in $T'$. Recall that $P_2^\flat = \sqrt{2} v$.
We thus have
\begin{align*}
Q_{S'}(p) & := \E_{X \in_u S'}[p(X)^2]  =  \E_{Y \sim T'}[(Y^T P_2 Y - \tr(P_2))^2]  \\
& = \E_{Y \in_u T'}[(Y^T P_2 Y)^2] + \tr(P_2)^2 - 2\tr(P_2))^2] \\
& = \E_{Y \in_u T'}[\tr(((Y Y^T) P_2)^2] \new{-} \tr(P_2 I)^2 - 0 \\
& = \E_{Z \in_u U'}[( Z^T v)^2\new{/}2] \new{-} (v^T I^\flat )^2\new{/}2 \\
&= \E_{Z \in_u U'}[  v^T (Z Z^T) v\new{/}2] \new{-}  2 v^T (I^\flat I^{\flat T}) v\new{/}2\\
&= v^T T_{S'} v\new{/}2 \;.
\end{align*}
Thus, the $p(x)$ that maximizes $Q_{S'}(p)$ is given by the unit vector $v$ that maximizes $v^T T_{S'} v$ subject to $v^{\sharp}$ being symmetric.

Let $v'=v^{\sharp T \flat}$. Note that $v^T T_{S'} v= v'^T T_{S'} v'$ by symmetries of $T_{S'}$. 
Thus, by linearity, $v''=v/2+v'/2$ 
also has  $v''^T T_{S'} v''= v^T T_{S'} v$. 
However, if $v^{\sharp}$ is not symmetric, 
$v''$ has $\|v''\|_2 < 1$. 
Thus, the unit vector $v''/\|v''\|_2$ achieves a higher value 
of the bilinear form. Consequently, $v^{\ast \sharp}$ is symmetric.

Now we have that $p^{\ast}(x)$ that maximizes $Q_{S'}(p)$ is given 
by the unit vector $v$ that maximizes $v^T T_{S'} v$. 
Since $Q_{G'}(p) := \E_{X \sim G'}[p(X)^2]= 2\|P_2\|_F = \|v\|_2=1$, 
this also maximizes $Q_{S'}(p)/Q_{G'}(p)$.

We note that we can achieve  $\E_{X \sim G'}[p^{\ast}(X)]=O(\eps^2)$ and 
$\E_{X \sim G'}[(p^{\ast}(X))^2] = 1 + O(\eps^2)$ in time $\poly(\eps/d)$ 
using standard algorithms to compute the eigen-decomposition of a symmetric matrix. 
This suffices for the correctness of the remaining part of \textsc{Filter-Gaussian-Unknown-Covariance} 
The other steps in  \textsc{Filter-Gaussian-Unknown-Covariance} can be easily done in $\poly(|S'|d\log(1\tau)/\eps)$ time.
\end{proof}

In order to analyze algorithm \textsc{Filter-Gaussian-Unknown-Covariance}, 
we note that we can write $S'=(S \setminus L) \cup E$ where $L= S \setminus S'$ and $L= S' \setminus S$. 
It is then the case that $\Delta(S,S')=(|L|+|E|)/|S|.$ 
Since this is small we have that $|L|,|E|= O(\epsilon|S'|).$ 
We can also write $\Sigma'$ and 
$\Sigma_{S \setminus L}((|S|-|L|)/|S'|) + \Sigma_E(|E|/|S'|) = \Sigma_{S \setminus L} + O(\epsilon) (\Sigma_E-\Sigma_{S \setminus L})$, 
where $\Sigma_{S \setminus L} = \E_{X\in_u S \setminus L}[XX^T], \Sigma_E=\E_{X\in_u E}[XX^T].$ 
A critical part of our analysis will be to note that $\Sigma_{S \setminus L}$ is very close to $\Sigma$, 
and thus that either $\Sigma'$ is very close to $\Sigma$ or else $\Sigma_E$ is very large in some direction.

\begin{lemma}\label{ACloseLem}
We have that
$$\|I -  \Sigma^{-1/2}\Sigma_{S \setminus L} \Sigma^{-1/2}\|_F = O(\epsilon \log(1/\epsilon).$$
\end{lemma}

\new{To prove Lemma~\ref{ACloseLem}, we will require the following:
\begin{lemma} \label{lem:L-cov}
Let $p(x)$ be an even degree-$2$ polynomial with 
$\E_{X \sim G}[p(X)] = 0$ and $\var_{X \sim G}[p(X)] = 1$.
Then, we have that $|L|\E_{X \in_u L}[p(X)^2]= O(\eps \log^2(1/\eps) |S|)$ 
and $|L||\E_{X \in_u L}[p(X)]|= O(\eps \log(1/\eps) |S|)$.
\end{lemma}
\begin{proof}
This holds essentially because the distribution of $p(X)$ for $X \in S$ is close to that for $p(G)$, 
which has rapidly decaying tails. Therefore, throwing away an $\eps$-fraction of the mass 
cannot change the value of the variance by very much. In particular, we have that
\begin{align*}
|L| \E_{X \in_u L}[p(X)^2]
& \leq \int_0^\infty |L| \Pr_{X\in_u L} (|p(X)|>T) 2T dT \\
& \leq \int_0^\infty |S| \min(2\epsilon,\Pr_{X\in_u S}(|p(X)|>T)) 2T dT\\
& \leq \int_0^{10\ln(1/\eps)} 4 \eps |S| T dT + \int_{10\ln(1/\eps)}^{\infty} 6|S| \eps T /(T^2 \log^2(T)) dT\\
& \leq O(\eps|S|\log^2(1/\epsilon)) + \int_{10\ln(1/\eps)}^{\infty} 6|S| \eps/(T \log^2(T)) dT\\
& = O(\eps|S|\log^2(1/\epsilon)) +6 \eps |S| / \ln(10\ln(1/\eps)) \\
& = O(\epsilon \log^2(1/\epsilon) |S|) \;.
\end{align*}
By the Cauchy-Schwarz inequality, 
we have $(|L|/|S|) |\E_{x \in_u L}[p(X)]| \leq (|L|/|S|)  \sqrt{\E_{x \in_u L}[p(X)^2]} \leq \sqrt{|L|/|S|} \cdot \sqrt{O(\eps \log^2(1/\eps))}=O(\eps\log(1/\eps)$.
\end{proof}
Now we can prove Lemma \ref{ACloseLem}.}

\begin{proof}[Proof of Lemma \ref{ACloseLem}]
Note that, since the matrix inner product is an inner product,
$$
\|I - \Sigma^{-1/2}\Sigma_{S \setminus L}\Sigma^{-1/2}\|_F = \sup_{\|M\|_F =1}\left( \tr(M\Sigma^{-1/2}\Sigma_{S \setminus L}\Sigma^{-1/2}) - \tr(M)\right).
$$
We need to show that for any $M$ with $\|M\|_F=1$ that $\tr(M\Sigma^{-1/2}\Sigma_{S \setminus L}\Sigma^{-1/2}) - \tr(M)$ is small.

Since $\tr(M\Sigma^{-1/2}\Sigma_{S \setminus L}\Sigma^{-1/2})=\tr(M^T \Sigma^{-1/2}\Sigma_{S \setminus L}\Sigma^{-1/2})=\tr(\frac12(M+M^T)\Sigma^{-1/2}\Sigma_{S \setminus L}\Sigma^{-1/2})$ and $\|\frac12 (M+M^T)\|_F \leq \frac12 (\|M\|_F+\|M^T\|_F) = 1$, we may assume WLOG that $M$ is symmetric.

Consider such an $M$. We note that
$$
\tr(M\Sigma^{-1/2}\Sigma_{S \setminus L}\Sigma^{-1/2}) = \E_{X\in_u S \setminus L} [\tr(M\Sigma^{-1/2}XX^T\Sigma^{-1/2})] = \E_{X\in_u S \setminus L}[(\Sigma^{-1/2}X)^TM(\Sigma^{-1/2}X)].
$$
Let $p(x)$ denote the quadratic polynomial
$$
p(x) = (\Sigma^{-1/2}x)^T M (\Sigma^{-1/2}x) - \tr(M).
$$

By Lemma \ref{lem:evenp}, $\E_{X \sim G}[p(X)] = 0$ and $\var_{X \sim G}[p(X)] = 2\|M\|_F^2 = 2$.

Since $S$ is $(\epsilon,\tau)$-good with respect to $G$, we have that 
$\E_{X\in S}[p(X)]=\epsilon\sqrt{\E_{X \sim G}[p^2(X)]} = O(\eps)$. 
Therefore, it suffices to show that the contribution from $L$ is small. 
\new{In particular, it will be enough to show that $(|L|/|S|) |\E_{x \in_u L}[p(X)]| \leq O(\eps \log(1/\eps)).$ 
This follows from Lemma \ref{lem:L-cov}, which completes the proof.}
\end{proof}

As a corollary of this we note that $\Sigma'$ cannot be too much smaller than $\Sigma$.
\begin{corollary}\label{notTooSmallCor}
$$\Sigma' \succeq (1-O(\epsilon \log(1/\epsilon)))\Sigma.$$
\end{corollary}
\begin{proof}
Lemma \ref{ACloseLem} implies that $\Sigma^{-1/2}\Sigma_{S \setminus L} \Sigma^{1/2}$ has all eigenvalues in the range $1\pm O(\epsilon\log(1/\epsilon)$. Therefore, $\Sigma_{S \setminus L} \succeq (1+O(\epsilon\log(1/\epsilon)))\Sigma$. Our result now follows from noting that
$\Sigma'=\Sigma_{S \setminus L}((|S|-|L|)/|S'|) + \Sigma_E(|E|/|S'|)$, and $\Sigma_E = \E_{X\in_u E}[XX^T] \geq 0$.
\end{proof}

The first step in verifying correctness is to note that if our algorithm returns on Step \ref{removeOutlierStep} that it does so correctly.
\begin{claim}
If our algorithm returns on Step \ref{removeOutlierStep}, then $\Delta(S,S'')<\Delta(S,S').$
\end{claim}
\begin{proof}
This is clearly true if we can show that all $x$ removed have $x \not\in S$. However, this follows because $(\Sigma')^{-1} \leq 2\Sigma^{-1}$, and therefore, by $(\epsilon,\tau)$-goodness, all $x \in S$ satisfy
$$
x^T(\Sigma')^{-1}x \leq 2x^T\Sigma^{-1}x  < Cd\log(N/\tau)
$$
for $C$ sufficiently large.
\end{proof}

Next, we need to show that if our algorithm returns a $G'$ in Step \ref{returnGStep} that $\dtv(G,G')$ is small.
\begin{claim}
If our algorithm returns in Step \ref{returnGStep}, then $\dtv(G,G')=O(\epsilon\log(1/\epsilon)).$
\end{claim}
\begin{proof}
By Corollary \ref{cor:kl-to-cov}, it suffices to show that
$$
\|I - \Sigma^{-1/2}\Sigma'\Sigma^{-1/2}\|_F = O(\epsilon \log(1/\epsilon)).
$$
However, we note that
\begin{align*}
 & \|I - \Sigma^{-1/2}\Sigma'\Sigma^{-1/2}\|_F \leq \|I-\Sigma^{-1/2}\Sigma_{S \setminus L} \Sigma^{-1/2}\|_F +(|E|/|S'|)\|I-\Sigma^{-1/2}\Sigma_E  \Sigma^{-1/2}\|_F \\
& \leq O(\epsilon\log(1/\epsilon)) + (|E|/|S'|)\|I-\Sigma^{-1/2}\Sigma_E\Sigma^{-1/2}\|_F.
\end{align*}
Therefore, we will have an appropriate bound unless $\|I-\Sigma^{-1/2}\Sigma_E\Sigma^{-1/2}\|_F = \Omega(\log(1/\epsilon)).$

Next, note that there is a matrix $M$ with $\|M\|_F=1$ such that
$$
\|I-\Sigma^{-1/2}\Sigma_E \Sigma^{-1/2}\|_F = \tr(M\Sigma^{-1/2}\Sigma_E\Sigma^{-1/2}-M)=\E_{X\in_u E}[(\Sigma^{-1/2}X)^T M (\Sigma^{-1/2}X)-\tr(M)] .
$$
Indeed we can take $M=(I-\Sigma^{-1/2}\Sigma_E \Sigma^{-1/2})/\|I-\Sigma^{-1/2}\Sigma_E \Sigma^{-1/2}\|_F$. Thus, there is a symmetric $M$ such that this holds.
%Note that $M$ is symmetric since otherwise $M'= (M+M^T)/2$ also has $\|I-\Sigma^{-1/2}\Sigma_E \Sigma^{-1/2}\|_F = \tr(M'\Sigma^{-1/2}\Sigma_E\Sigma^{-1/2}-M')$ but has $\|M'_F\|_2 < \|M_F\|_2=1$ which contradicts Cauchy-Schwarz.

Letting $p(X)$ be the polynomial
$$
p(X)=(\Sigma^{-1/2}X)^TM(\Sigma^{-1/2}X)-\tr(M),
$$
Using Lemma \ref{lem:evenp}, $\E_{X \sim G}[p(X)]=0$ and $\var_{X \sim G}[p(X)]=2$. Therefore, $p\in L$ and $Q_{G'}(p)=2$. We now compare this to the size of $Q_{S'}(p)$. On the one hand, we note that using methodology similar to that used in Lemma \ref{ACloseLem} we can show that $\E_{X\in_u S \setminus L}[p^2(X)]$ is not much less than 2. In particular,
$$
\E_{X\in_u S \setminus L}[p^2(X)] \geq \left(\E_{X\in_u S}[p^2(X)]-\frac{\sum_{X\in L}p^2(X)}{|S|}\right).
$$
On the one hand, we have that $$\E_{X\in_u S}[p^2(X)] \leq \E[p^2(G)](1+\epsilon) = 2 + O(\epsilon) \;,$$ by assumption.
On the other hand, \new{by Lemma \ref{lem:L-cov}, we have $|L| \E_{X\in_u L}[p^2(X)]/|S| \leq O(\epsilon \log^2(1/\epsilon))$.}

Therefore, we have that $\E_{X\in_u S \setminus L}[p^2(X)] = 2 + O(\epsilon\log^2(1/\epsilon))$. Since, by assumption $Q_{S'}(p)\leq 2+O(\epsilon \log^2(1/\epsilon))$, this implies that $(|E|/|S'|)\E_{X\in_u E}[p^2(X)] = O(\epsilon\log^2(1/\epsilon))$. By Cauchy-Schwartz, this implies that
$$
(|E|/|S'|)\E_{X\in_u E}[p(X)] \leq \sqrt{(|E|/|S'|)}\sqrt{(|E|/|S'|)\E_{X\in_u E}[p^2(X)] } = O(\epsilon \log(1/\epsilon)).
$$
Thus,
$$
(|E|/|S'|)\|I-\Sigma^{-1/2}\Sigma_E\Sigma^{-1/2}\|_F = O(\epsilon\log(1/\epsilon)).
$$
This shows that if the algorithm returns in this step, it does so correctly.
\end{proof}

Next, we need to show that if the algorithm reaches Step \ref{thresholdStep} that such a $T$ exists.
\begin{claim}
If the algorithm reaches Step \ref{thresholdStep}, then there exists a $T > 1$ such that
$$
\Pr_{X\in_u S'} (|p(X)-\mu| \geq T) \geq 12 \exp(- (T-1)/3) + 3\epsilon/(d \log(N/\tau))^2.
$$
\end{claim}
\begin{proof}
Before we begin, we will need the following critical Lemma:
\begin{lemma}
$$\var_{X \sim G}[p(X)] \leq 1+O(\epsilon\log(1/\epsilon)).$$
\end{lemma}
\begin{proof}
We note that since $\var_{X\sim G'}(p(G'))=Q_{G'}(p)=1$, we just need to show that the variance with respect to $G$ instead of $G'$ is not too much larger. This will essentially be because the covariance matrix of $G$ cannot be much bigger than the covariance matrix of $G'$ by Corollary \ref{notTooSmallCor}. 

Using Lemma \ref{lem:evenp}, we can write $$p(x)= (\Sigma'^{-1/2} x)^T P_2 (\Sigma'^{-1/2} x) +p_0 \;,$$ 
where $\|P_2\|_F = \frac12 \var_{X\sim G'}(p(G')) = \frac12$ and $p_0=\mu + \tr(P_2)$. 
We can also express $p(x)$ in terms of $G$ as $p(x)= (\Sigma^{-1/2} x)^T M (\Sigma^{-1/2} x) +p_0,$ 
and have $\var_{X \sim G}[p(X)] = \|M\|_F$. 
Here, $M$ is the matrix $\Sigma^{1/2} \Sigma'^{-1/2} P_2 \Sigma'^{-1/2} \Sigma^{1/2} $. 
By  Corollary \ref{notTooSmallCor}, it holds 
$\Sigma' \geq (1-O(\epsilon\log(1/\epsilon))) \Sigma$. 
Consequently, $\Sigma^{1/2} \Sigma'^{-1/2} \leq (1 + O(\epsilon\log(1/\epsilon))) I$, 
and so $\|\Sigma^{1/2} \Sigma'^{-1/2}\|_2 \leq 1 + O(\epsilon\log(1/\epsilon))$. 
Similarly, $\|\Sigma'^{-1/2} \Sigma^{1/2}\|_2 \leq 1 + O(\epsilon\log(1/\epsilon))$.

We claim that if $A,B$ are matrices, 
then $\|AB\|_F \leq \|A\|_2 \|B\|_F$. If $B_j$ are the columns of $B$, 
then we have $\|AB\|_F^2= \sum_j \|A B_j\|_2^2 \leq \|A\|_2^2 \sum_j \|B_j\|_2^2 = (\|A\|_2 \|B\|_F)^2$. 
Similarly for rows, we have $\|AB\|_F \leq \|A\|_F \|B\|_2$.

Thus, we have 
$$\var_{X \sim G}[p(X)] = 2\|M\|_F \leq 2 \|\Sigma^{1/2} \Sigma'^{-1/2}\|_2  \|P_2\|_F \|\Sigma'^{-1/2} \Sigma^{1/2}\|_2 \leq 1 + O(\epsilon\log(1/\epsilon)) \;.$$
%\begin{comment}
%By making a change of variables, we can guarantee that $\Sigma'=I$ and $\Sigma \leq I(1+O(\epsilon\log(1/\epsilon)))$. Furthermore, by applying an additional rotation, we can %ensure that $\Sigma$ is diagonal. Writing $p$ in the basis of Hermite polynomials, we find that we can write
%$$
%p(x) = \sum_{i<j} a_{i,j}x_ix_j + \sum_i a_{i,i}(x_i^2-1)/\sqrt{2},
%$$
%where
%$$
%\sum_{i\leq j} a_{i,j}^2 = Q_{G'}(p) = 1.
%$$
%However, we note that with respect to $G$, the polynomials $x_ix_j$ and $(x_i^2-1)/\sqrt{2}$ pairwise have $0$ covariance and each has variance at most $1+O(\epsilon\log%(1/\epsilon))$. Therefore,
%$$
%\var_{X \sim G}[p(X)] \leq \sum_{i<j} 2a_{i,j}^2+\sum_i 2a_{i,i}^2 = 2.
%$$
%This completes the proof.
%\end{comment}
\end{proof}
Next, we need to consider $\mu$. 
In particular, we note that by the similarity of $S$ and $S'$, 
$\mu$ must be between the $40$ and $60$ percentiles of values of $p(X)$ for $X\in S$. 
However, since $S$ is $(\epsilon,\tau)$-good, this must be between the $30$ and $70$ percentiles of $p(G)$. 
Therefore, by Cantelli' s inequality, 
\begin{equation} \label{eq:cantelli}
|\mu-\hat{\mu}|\leq 2\sqrt{\var_{X \sim G}[p(X)]} \leq 3 \;,
\end{equation}
\new{where $\hat{\mu}=\E_{X \sim G}[p(X)]$.}
We are now ready to proceed. Our argument will follow by noting that while $Q_{S'}(p)$ is much larger than expected, 
very little of this discrepancy can be due to points in $S \setminus L$. 
Therefore, the points of $E$ must provide a large contribution. 
Given that there are few points in $E$, much of this contribution must come from there 
being many points near the tails, and this will guarantee that some valid threshold $T$ exists.

In particular, we have that $\var_{X\in_u S'}(p(X))= Q_{S'}(p) \geq 1 + C\epsilon\ln^2(1/\epsilon)$, which means that
$$
\frac{\sum_{X\in S'} |p(X)-\hat{\mu}|^2}{|S'|} \geq 1 + C\epsilon\ln^2(1/\epsilon).
$$
Now, because $S$ is good, we know that
$$
\frac{\sum_{X\in S} |p(X)-\hat{\mu}|^2}{|S|} = \E[|p(G)-\hat{\mu}|^2](1+O(\epsilon)) = \var_{X \sim G}[p(X)](1+O(\epsilon)) \leq 1+O(\epsilon\log(1/\epsilon)).
$$
Therefore, \new{using (\ref{eq:cantelli}),} we have that
$$
\frac{\sum_{X\in S \setminus L} |p(X)-\hat{\mu}|^2}{|S'|} \leq 1+O(\epsilon\log(1/\epsilon)).
$$
Hence, for $C$ sufficiently large, it must be the case that
$$
\sum_{X\in E}|p(X)-\hat{\mu}|^2 \geq (C/2)\epsilon\ln^2(1/\epsilon)|S'| \;,
$$
and therefore,
$$
\sum_{X\in E}|p(X)-\mu|^2 \geq (C/3)\epsilon\ln^2(1/\epsilon)|S'| \;.
$$
On the other hand, we have that
\begin{align*}
\sum_{X\in E}|p(X)-\mu|^2 & = \int_0^\infty \{X\in E: |p(X)-\mu| > t\}2tdt\\
& \leq \int_0^{C^{1/4}\ln(1/\epsilon)} O(t\epsilon|S'|)dt+ \int_{C^{1/4}\ln(1/\epsilon)}^{\infty}\{X\in E: |p(X)-\mu| > t\}2tdt\\
& \leq O(C^{1/2}\epsilon\log^2(1/\epsilon)|S'|) + |S'|\int_{C^{1/4}\ln(1/\epsilon)}^{\infty}\Pr_{X\in_u S'}( |p(X)-\mu| > t)2tdt \;.
\end{align*}
Therefore, we have that
\begin{equation}\label{tailEqn}
\int_{C^{1/4}\ln(1/\epsilon)}^{\infty}\Pr_{X\in_u S'}( |p(X)-\mu| > t)2tdt \geq (C/4)\epsilon\log^2(1/\epsilon) \;.
\end{equation}
Assume for sake of contradiction that 
\new{
$$
\Pr_{X\in_u S'} (|p(X)-\mu| \geq T + 3) \leq \tail(T, d, \eps, \tau) \;,
$$
for all $T > 1$. 

Thus, we have that
\begin{align*}
\int_{10 \ln(1/\epsilon)+3}^{\infty}\Pr_{X\in_u S'}( |p(X)-\mu| > T)2TdT & \leq \int_{10 \ln(1/\epsilon)}^{\infty} 6(T+3) \eps /(T^2 \log^2 T) dT \\
& =  \int_{10 \ln(1/\epsilon)}^{\infty} 8\eps /(T \log^2 T) dT \\
& = 8\eps/\ln(10 \ln(1/\eps)) \;.
\end{align*}
For a sufficiently large $C$, this contradicts Equation (\ref{tailEqn}).}
\end{proof}

Finally, we need to verify that if our algorithm returns output in Step \ref{filterStep}, that it is correct.
\begin{claim}
If the algorithm returns during Step \ref{filterStep}, then $\Delta(S,S'')\leq \Delta(S,S') - \epsilon/(d \log(N/\tau))^2.$
\end{claim}
\begin{proof}
We note that it is sufficient to show that $|E \setminus S''| > |(S \setminus L) \setminus S''|$. 
\new{In particular, it suffices to show that
$$
|\{X\in E: |p(X)-\mu|>T+3\}| > |\{X\in S \setminus L: |p(X)-\mu|>T+3\}| \;.
$$
For this, it suffices to show that
$$
|\{X\in S': |p(X)-\mu|>T+3\}| > 2|\{X\in S \setminus L: |p(X)-\mu|>T+3\}| \;,
$$
or that
$$
|\{X\in S': |p(X)-\mu|>T+3\}| > 2|\{X\in S: |p(X)-\mu|>T+3\}| \;.
$$
By assumption, we have that
$$
|\{X\in S': |p(X)-\mu|>T+3\}| > 3|S'|\eps/(T^2 \log^2 T) \;.
$$
On the other hand, using (\ref{eq:cantelli}) and the $\eps$-goodness of $S$, we have that
\begin{align*}
|\{X\in S: |p(X)-\mu|>T+3 \}| & \leq |\{X\in S: |p(X)-\hat{\mu}|>T \}|\\
&  \leq |S| \eps/(T^2 \log^2T) \;.
\end{align*}
This completes our proof.}
\end{proof}

%!TEX root = ./main.tex

\section{Agnostically Learning a Mixture of Spherical Gaussians, via Convex Programming}
\label{sec:sepGMM}
In this section, we give an algorithm to agnostically learn a mixture of $k$ Gaussians with identical spherical covariance matrices up to error $\widetilde{O} (\poly(k) \cdot \sqrt{\ve})$.
Let $\mixpdf = \sum_{j \in [k]} \a_j \mathcal{N}(\m_j, \s^2 I)$ be the unknown $k$-GMM each of whose components are spherical. For $X \sim \mixpdf$, we write $X \sim_j \mixpdf$ if $X$ was drawn from the $j$th component of $\mixpdf$.

Our main result of this section is the following theorem:
\begin{theorem}
  Fix $\ve, \tau > 0$, and $k \in \mathbb{N}$.
Let $X_1, \ldots, X_N$ be an $\ve$-corrupted set of samples from a $k$-GMM $\mixpdf = \sum_{j \in [k]} \a_j \mathcal{N}(\m_j, \s_j^2 I)$, where all $\a_j, \m_j$, and $\s_j^2$ are unknown, and 
\[
  N = \widetilde{\Omega} \left(\poly\left(d, k, 1/\ve, \log (1/\tau)\right)\right) \; .
\]
There is an algorithm which with probability $1 - \tau$, outputs a distribution $\mixpdf'$ such that 
\[
\dtv (\mixpdf, \mixpdf') \leq  \widetilde O(\poly(k) \cdot \sqrt{\ve}) \; .
\]
The running time of the algorithm is $\poly(d, 1/\ve, \log(1/\tau))^{k^2}$.
\end{theorem}

Our overall approach will be a combination of our method for agnostically learning a single Gaussian and recent work on properly learning mixtures of multivariate spherical Gaussians \cite{SOAJ14, LiS15a}. 
At a high level, this recent work relies upon the empirical covariance matrix giving an accurate estimate of the overall covariance matrix in order to locate the subspace in which the component mean lie.
However, as we have observed already, the empirical moments do not necessarily give good approximations of the true moments in the agnostic setting.
Therefore, we will use our separation oracle framework to approximate the covariance matrix, and the rest of the arguments follow similarly as previous methods.

The organization of this section will be as follows.
We define some of the notation we will be using and the Schatten top-$k$ norm in Section \ref{sec:gmmnotation}.
Section \ref{sec:gmmconcentration} states the various concentration inequalities we require.
In Section \ref{sec:gmmalgorithm}, we go over our overall algorithm in more detail.
Section \ref{sec:gmmnaivecluster} describes a first naive clustering step, which deals with components which are very well separated.
Section \ref{sec:gmmsep} contains details on our separation oracle approach, allowing us to approximate the true covariance.
Section \ref{sec:gmmspectralcluster} describes our spectral clustering approach to cluster components with means separated more than $\Omega_k(\log 1/\ve)$.
In Section \ref{sec:gmmsearch}, we describe how to exhaustively search over a particular subspace to obtain a good estimate for the component means.
In Section \ref{sec:gmmtournament}, we go over how to limit the set of hypotheses in order to satisfy the conditions of Lemma \ref{tournamentLem}.
For clarity of exposition, all of the above describe the algorithm assuming all $\s_j^2$ are equal.
In Section \ref{sec:gmmdiffvar}, we discuss the changes to algorithm which are required to handle unequal variances.

For conciseness, many of the proofs are deferred to Section \ref{sec:sepGMMAppendix}.
\subsection{Notation and Norms}
\label{sec:gmmnotation}
Recall the definition of $S_{N,\ve}$ from Section \ref{sec:sne}, which we will use extensively in this section.
We will use the notation $\mu = \sum_{j \in [k]} \a_j \m_j$ to denote the mean of the unknown GMM. 
Also, we define parameters $\g_j = \a_j \|\m_j - \m\|_2^2$ and let $\g = \max_j \g_j$.
And for ease of notation, let
$$f(k, \g, \ve) = k^{1/2}\ve + k \g^{1/2}\ve + k\ve^2 \mbox{ and } h(k, \g, \ve) = k^{1/2}\ve + k \g^{1/2}\ve + k \g \ve + k\ve^2 = f(k, \g, \ve) + k \g \ve. $$
Finally, we use the notation
\begin{equation}
    Q = \sum_{j \in [k]} \a_j (\m_j - \m) (\m_j -\m)^T. \label{eq:gmmcorr}
\end{equation}
to denote the covariance of the unknown GMM. Our algorithm for learning spherical $k$-GMMs will rely heavily on the following, non-standard norm:

\begin{definition}
For any symmetric matrix $M \in \mathbb{R}^{d \times d}$ 
with singular values $\s_1 \geq \s_2 \geq \ldots \s_d$,
let the \emph{Schatten top-$k$ norm} be defined as 
$$\| M \|_{T_k} = \sum_{i = 1}^k \s_i \;,$$ 
i.e., it is the sum of the top-$k$ singular values of $M$.
\end{definition}

It is easily verified that $\| \cdot \|_{T_k}$ has a dual characterization
\[ \| M \|_{T_k} = \max_{X \in \R^{d \times k}} \Tr (X^T \sqrt{M^T M} X) \; ,\]
where the maxima is taken over all $X$ with orthonormal columns. 
From this, it is easy to see that the Schatten top-$k$ norm 
is indeed a norm, as its name suggests: 

\begin{fact}
$\| M \|_{T_k}$  is a norm on symmetric matrices. 
\end{fact}

\subsection{Concentration Inequalities}
\label{sec:gmmconcentration}
In this section, we will establish some concentration inequalities that we will need for our algorithm for agnostically learning mixtures of spherical Gaussians.
Recall the notation as described in Section \ref{sec:gmmnotation}.
The following two concentration lemmata follow from the same proofs as for Lemmata 42 and 44 in \cite{LiS15a}.
\begin{lemma}
  \label{lem:gmmconc}
  Fix $\ve, \d > 0$.
  If $Y_1, \dots, Y_N$ are independent samples from the GMM with PDF $\sum_{j \in [k]} \a_j \mathcal{N}(\m_j, \S_j)$  where $\alpha_j \geq \Omega (\ve)$ for all $j$, and $N = \Omega\left(\frac{d + \log{(k/\d)}}{\ve^2}\right)$ then with probability at least $1 - O(\d)$,
  $$\left\|\frac{1}{N} \sum_{i=1}^N (Y_i - \m) (Y_i - \m)^T - I - Q \right\|_2 \leq O\left(f(k, \g, \ve)\right),$$
  where $Q$ is defined as in equation (\ref{eq:gmmcorr}).
\end{lemma}

\begin{lemma}
  \label{lem:gmmconc2}
  Fix $\ve, \d > 0$.
  If $Y_1, \dots, Y_N$ are independent samples from the GMM with PDF $\sum_{j \in [k]} \a_j \mathcal{N}(\m_j, \S_j)$  where $\alpha_j \geq \Omega (\ve)$ for all $j$, and $N = \Omega\left(\frac{d + \log{(k/\d)}}{\ve^2}\right)$ then with probability at least $1 - O(\d)$,
  $$\left\|\frac{1}{N} \sum_{i=1}^N Y_i - \m \right\|_2 \leq O\left(k^{1/2} \ve\right).$$
\end{lemma}

From the same techniques as before, we get the same sort of union bounds as usual over the weight vectors:
\begin{lemma}
\label{lem:gmmunion-bound}
Fix $\ve \leq 1/2$ and $\tau \leq 1$. 
There is a $\d = O(\ve \sqrt{\log 1/\ve})$ such that if $Y_1, \dots, Y_N$ are independent samples from the GMM with PDF $\sum_{j \in [k]} \a_j \mathcal{N}(\m_j, \S_j)$  where $\alpha_j \geq \Omega (\ve)$ for all $j$, and $N = \Omega\left(\frac{d + \log{(k/\t)}}{\d_1^2}\right)$, then
\begin{equation}
\Pr \left[\exists w \in S_{N, \ve} : \left \| \sum_{i = 1}^N w_i (Y_i - \m) (Y_i - \m)^T - I - Q \right\|_2 \geq f(k, \g, \d_1) \right] \leq \tau \;, 
\end{equation}
  where $Q$ is defined as in equation (\ref{eq:gmmcorr}).
\end{lemma}

\begin{lemma}
\label{lem:gmmunion-bound2}
Fix $\ve \leq 1/2$ and $\tau \leq 1$. 
There is a $\d = O(\ve \sqrt{\log 1/\ve})$ such that if $Y_1, \dots, Y_N$ are independent samples from the GMM with PDF $\sum_{j \in [k]} \a_j \mathcal{N}(\m_j, \S_j)$  where $\alpha_j \geq \Omega (\ve)$ for all $j$, and $N = \Omega\left(\frac{d + \log{(k/\t)}}{\d_2^2}\right)$, then
\begin{equation}
\Pr \left[\exists w \in S_{N, \ve} : \left \| \sum_{i = 1}^N w_i Y_i - \m \right\|_2 \geq k^{1/2}\d_2 \right] \leq \tau.
\end{equation}
\end{lemma}

\subsection{Algorithm}
\label{sec:gmmalgorithm}

Our approach is based on a {\em tournament}, as used in several recent works~\cite{DK14, SOAJ14, DDS15-journal, DDS15, DKT15,DDKT15}.
We will generate a list $\mathcal{S}$ of candidate hypotheses (i.e., of $k$-GMMs) of size $|\mathcal{S}| = \poly(d,1/\ve, \log(1/\tau))^{k^2}$ with the guarantee that there is some $\mixpdf^* \in \mathcal{S}$ such that $\dtv(\mixpdf, \mixpdf^*) \leq \tilde O(\poly(k) \cdot \sqrt{\ve})$.
We then find (roughly) the best candidate hypothesis on the list. It is most natural to describe the algorithm as performing several layers of {\em guessing}. %The first step is to guess the mixing weights. This is somewhat standard, so we will ignore it for the time being. Next we need to guess the variance $\s$ of the components.To accomplish this, we will take $k+1$ samples (hoping to find only uncorrupted ones) and compute the minimum distance between any pair of them. If none of the $k+1$ samples are corrupted, then at least two of them came from the same component, and in our high-dimensional setting the distance between any pair of samples from the same component concentrates around $\sqrt{2d}\s$. We can then multiplicatively enumerate around this value to get an estimate for $\s$ that is sufficiently good for our purposes. For the rest of this section, we will assume we have $\s^2$ exactly and that $\s^2 = 1$. What remains is to guess the centers $\m_j$, and this is the main step. 
%Note that at this point, the only remaining parameters we have not guessed are the $\m_j$'s. 
We will focus our discussion on the main steps in our analysis, and defer a discussion of guessing the mixing weights, the variance $\s^2$ and performing naive clustering until later. 
For reasons we justify in Section \ref{sec:gmmtournament}, we may assume that the mixing weights and the variance are known exactly, and that the variance $\s^2 = 1$.
%First, we perform a naive clustering step, in which we separate components which are very far away (i.e. further than $\Omega_{k,\ve}(d)$).
%This is standard, and we give further details in Section \ref{sec:gmmnaivecluster}. What remains is to guess the values of all $\g_j =  \a_j \|\m_j - \m\|_2^2$. These values will be used by Algorithm \ref{alg:sepgmm}.
%As the accuracy required for $\sum \g_j$ in Algorithm \ref{alg:sepgmm} is an additive $f(k,\g,\d_1)$ (which is also the best we could hope for due to the limitations of the types of concentration inequalities we are using, see Lemma \ref{lem:gmmconc}), it suffices to guess the values of $\g_j$ to within this range. 
%Suppose we have $\g_j \leq O_{k,\ve}(d)$.
%Observe that if $\sum_{j \in [k]} \g_j$ were polynomially small in $\ve/k$, we would be close to a single Gaussian. Hence we can assume $\sum_{j \in [k]} \g_j$ is sufficiently large.
%We can again multiplicatively enumerate over this range with granularity $(1 + \poly(\ve/kd))$ to get an accurate relative estimate of the sum.
%Again, for clarity of exposition, we will assume that $\sum_{j \in [k]} \g_j$ is known exactly.

Our algorithm is based on the following deterministic conditions:
\begin{align}
\frac{|\{ X_i \in \Sgood, X_i \sim_j \mixpdf: \| X_i - \mu_j \|_2^2 \geq \Omega (d \log k / \ve) \} |}{|\{ X_i \in \Sgood, X_i \sim_j \mixpdf \}|} &\leq \ve / k \;, \mbox{$\forall j = 1, \ldots, N$} \label{eqn:gmmsepconds1}\\
  \left \| \sum_{i \in \Sgood} w_i (X_i - \m)(X_i - \m)^T - w_g I - w_g Q \right\|_2 &\leq  f(k,\g,\d_1)  \mbox{ $\forall w \in S_{N, 4\ve}$, and} \label{eqn:gmmsepconds2}  \\ 
                             \left \| \sum_{i \in \Sgood} w_i (X_i - \mu) \right\|_2 &\leq k^{1/2} \delta_2   \mbox{ $\forall w \in S_{N, 4\ve}$ } \; . \label{eqn:gmmsepconds3} 
\end{align}
(\ref{eqn:gmmsepconds1}) follows from basic Gaussian concentration, and (\ref{eqn:gmmsepconds2}) and (\ref{eqn:gmmsepconds3}) follow from the results in Section \ref{sec:gmmconcentration} for $N$ sufficiently large.
Note that these trivially imply similar conditions for the Schatten top-$k$ norm, at the cost of an additional factor of $k$ on the right-hand side of the inequalities.
For the rest of this section, let $\d = \max(\d_1, \d_2)$.

At this point, we are ready to apply our separation oracle framework.
In particular, we will find a weight vector $w$ over the points such that 
$$\left\| \sum_{i=1}^N w_i (X_i -\m) (X_i - \m)^T - I - \sum_{j \in [k]} \a_j (\m_j - \m) (\m_j - \m)^T \right\|_2 \leq \eta,$$
for some choice of $\eta$.
The set of such weights is convex, and concentration implies that the true weight vector will have this property.
Furthermore, we can describe a separation oracle given any weight vector not contained in this set (as long as $\eta$ is not too small).
At this point, we use classical convex programming methods to find a vector which satisfies these conditions.
Further details are provided in Section \ref{sec:gmmsep}.

After this procedure, Lemma \ref{lem:gmmcovest} shows that the weighted empirical covariance is spectrally close to the true covariance matrix.
We are now in the same regime as \cite{SOAJ14}, which obtains their results as a consequence of the empirical covariance concentrating about the true covariance matrix. 
Thus, we will appeal to their analysis, highlighting the differences between our approach and theirs.
We note that \cite{LiS15a} also follows a similar approach and the interested reader may also adapt their arguments instead.

First, if $\g$ is sufficiently large (i.e., $\Omega_k(\log (1/\ve))$), this implies a separation condition between some component mean and the mixture's mean.
This allows us to cluster the points further, using a spectral method.
We take the top eigenvector of the weighted empirical covariance matrix and project in this direction, using the sign of the result as a classifier.
In contrast to previous work, which requires that no points are misclassified, we can tolerate $\poly(\ve/k)$ misclassifications, since our algorithms are agnostic.
This crucially allows us to avoid a dependence on $d$ in our overall agnostic learning guarantee.
Further details are provided in Section \ref{sec:gmmspectralcluster}.

Finally, if $\g$ is sufficiently small, we may perform an exhaustive search.
The span of the means is in the span of the top $k-1$ eigenvectors of the true covariance matrix, which we can approximate with our weighted empirical covariance matrix.
Since $\g$ is small, by trying all points within a sufficiently tight mesh, we can guess a set of candidate means which are sufficiently close to the true means.
Combining the approximations to the means with Corollary \ref{cor:kl-to-means} and the triangle inequality, we can guarantee that at least one of our guesses is sufficiently close to the true distribution.
Additional details are provided in Section \ref{sec:gmmsearch}.

To conclude our algorithm, we can apply Lemma \ref{tournamentLem}.
We note that this hypothesis selection problem has been studied before (see, e.g., \cite{DL:01,DK14}), but we must adapt it for our agnostic setting. 
This allows us to select a hypothesis which is sufficiently close to the true distribution, thus concluding the proof.
We note that the statement of Lemma \ref{tournamentLem} requires the hypotheses to come from some fixed finite set, while there are an infinite number of Gaussian mixture models.
In Section \ref{sec:gmmtournament}, we discuss how to limit the number of hypotheses based on the set of uncorrupted samples in order to satisfy the conditions of Lemma \ref{tournamentLem}.

\subsection{Naive Clustering}
\label{sec:gmmnaivecluster}
We give a very naive clustering algorithm, the generalization of $\textsc{NaivePrune}$, which recursively allows us to cluster components if they are extremely far away.
The algorithm is very simple: for each $X_i$, add all points within distance $O(d \log (k / \ve))$ to a cluster $S_i$.
Let $\mathcal{C}$ be the set of clusters which contain at least $4 \ve N$ points, and let the final clustering be $C_1, \ldots, C_{k'}$ be formed by merging clusters in $\mathcal{C}$ if they overlap, and stopping if no clusters overlap.
We give the pseudocode in Algorithm \ref{alg:clustergmm}.

\begin{algorithm}[htb]
\begin{algorithmic}[1]
\Function{NaiveClusterGMM}{$X_1, \ldots, X_n$}
\For{$i = 1, \ldots, N$}
	\State Let $S_i = \{i': \| X_i - X_{i'} \|_2^2 \leq \Theta (d k \log 1 / \ve)\}$.
\EndFor
Let $\mathcal{C} = \{S_i : |S_i| \geq 4 \ve N\}$.
\While{$\exists C, C' \in \mathcal{C}$ such that $C \neq C'$ and $C \cup C' \neq \emptyset$}
	\State Remove $C, C'$ from $\mathcal{C}$
	\State Add $C \cup C'$ to $\mathcal{C}$
\EndWhile
\State \textbf{return} the set of clusters $\mathcal{C}$
\EndFunction
\end{algorithmic}
\caption{Naive clustering algorithm for spherical GMMs.}
\label{alg:clustergmm}
\end{algorithm}

We prove here that this process (which may throw away points) throws away only at most a $\ve$ fraction of good points, and moreover, the resulting clustering only misclassifies at most an $O(\ve)$-fraction of the good points, assuming (\ref{eqn:gmmsepconds1}).

\begin{theorem}
\label{thm:naive-cluster}
Let $X_1, \ldots, X_m$ be a set of samples satisfying (\ref{eqn:gmmsepconds1}).
Let $C_1, \ldots, C_{k'}$ be the set of clusters returned.
For each component $j$, let $\ell(j)$ be the $\ell$ such that $C_\ell$ contains the most points from $j$.
Then:
\begin{enumerate}
\item
Then, for each $\ell$, there is some $j$ such that $\ell(j) = \ell$.
\item
For all $j$, we have
\[
|\{X_i \in \Sgood, X_i \sim_j \mixpdf \}| - |\{X_i \in \Sgood, X_i \sim_j \mixpdf, X_i \in C_{\ell(j)} \}| \leq O \left( \frac{\ve}{k} |\{X_i \in \Sgood, X_i \sim_j \mixpdf \}| \right) \; .
\]
\item
For all $j, j'$, we have that if $\ell(j) = \ell (j')$, then $\| \mu_j - \mu_{j'} \|_2^2 \leq O(d k \log k / \ve)$
\item
If $X_i, X_j \in C_\ell$, then $\| X_i - X_j \|_2^2 \leq d k \log 1 / \ve$.
\end{enumerate}
\end{theorem}

Thus, we have that by applying this algorithm, given an $\ve$-corrupted set of samples from $\mixpdf$, we may cluster them in a way which misclassifies at most an $\ve / k$ fraction of the samples from any component, and such that within each cluster, the means of the associated components differ by at most $d k \log k / \ve$.
Thus, each separate cluster is simply a $\ve$-corrupted set of samples from the mixture restricted to the components within that cluster; moreover, the number of components in each cluster must be strictly smaller than $k$.
Therefore, we may simply recursively apply our algorithm on these clusters to agnostically learn the mixture for each cluster, since if $k = 1$, this is a single Gaussian, which we know how to learn agnostically.

Thus, for the remainder of this section, let us assume that for all $j, j'$, we have $\| \mu_j - \mu_j' \|_2^2 \leq O(d k \log 1 / \ve)$.
Moreover, we may assume that there are no points $j, j'$ (corrupted or uncorrupted), such that $\| X_j - X_{j'} \|_2^2 \geq \Omega (d k \log 1 / \ve)$.

\subsection{Estimating the Covariance Using Convex Programming}
\label{sec:gmmsep}
In this section, we will apply our separation oracle framework to estimate the covariance matrix.
While in the non-agnostic case, the empirical covariance will approximate the actual covariance, this is not necessarily true in our case.
As such, we will focus on determining a weight vector over the samples such that the weighted empirical covariance \emph{is} a good estimate for the true covariance.

We first define the convex set for which we want an interior point:
$$\mathcal{C}_\eta = \left\{ w \in S_{N,\ve} : \left\| \sum_{i=1}^N w_i (X_i -\m) (X_i - \m)^T - I - \sum_{j \in [k]} \a_j (\m_j - \m) (\m_j - \m)^T \right\|_2 \leq \eta \right\}.$$

In Section \ref{sec:gmmproperties}, we prove lemmata indicating important properties of this set.
In Section \ref{sec:gmmoracle}, we give a separation oracle for this convex set.
We conclude with Lemma \ref{lem:gmmcovest}, which shows that we have obtained an accurate estimate of the true covariance.

\subsubsection{Properties of Our Convex Set}
\label{sec:gmmproperties}
We start by proving the following lemma, which states that for any weight vector which is not in our set, the weighted empirical covariance matrix is noticeably larger than it should be (in Schatten top-$k$ norm).
\begin{lemma}
  \label{lem:gmmtklb}
  Suppose that (\ref{eqn:gmmsepconds2}) holds, and $w \not \in \mathcal{C}_{ckh(k,\g,\d)}$.
  Then
  $$\left\|\sum_{i = 1}^N w_i (X_i - \m)(X_i - \m)^T - I\right\|_{T_k} \geq \sum_{j \in [k]} \g_j + \frac{3ck h(k, \g, \d)}{4}.$$
\end{lemma}

We also require the following lemma, which shows that if a set of weights poorly approximates $\m$, then it is not in our convex set.
\begin{lemma}
  \label{lem:gmmkey}
  Suppose that (\ref{eqn:gmmsepconds2}) and (\ref{eqn:gmmsepconds3}) hold.
  Let $w \in S_{m,\ve}$ and set $\muhat = \sum_{i = 1}^m w_i X_i$ and $\Delta = \mu - \muhat$.
  Furthermore, suppose that $\|\Delta\|_2 \geq \Omega(h(k,\g,\d))$. 
  Then
  $$\left\|\sum_{i=1}^N w_i (X_i - \m)(X_i - \m)^T - I - \sum_{j \in [k]} \a_j (\m_j - \m)(\m_j - \m)^T\right\|_2  \geq  \Omega\left(\frac{\|\Delta\|_2^2}{\ve}\right).$$
 %Additionally,
 %$$\left\| \sum_{i \in \Sbad} w_i (X_i - \mu) (X_i - \mu)^T \right\| \geq \Omega \left( \frac{\| \Delta \|^2}{\epsilon } \right).$$
\end{lemma}

By contraposition, if a set of weights is in our set, then it provides a good approximation for $\m$:
\begin{corollary}
  \label{cor:gmmclose}
  Suppose that (\ref{eqn:gmmsepconds2}) and (\ref{eqn:gmmsepconds3}) hold.
  Let $w \in \mathcal{C}_{h(k,\g,\d)}$ for $\d = \Omega(\ve \log 1/\ve)$.
  Then $$ \|\Delta\|_2 \leq O(\ve \sqrt{\log 1/\ve}).$$
\end{corollary}

\subsubsection{Separation Oracle}
\label{sec:gmmoracle}

In this section, we provide a separation oracle for $\mathcal{C}_\eta$.
In particular, we have the following theorem:

\begin{theorem}
\label{thm:gmmsep}
Fix $\ve > 0$, and let $\delta = \Omega(\ve \log 1/\ve).$ 
Suppose that (\ref{eqn:gmmsepconds2}) and (\ref{eqn:gmmsepconds3}) hold.
Let $w^\ast$ denote the weights which are uniform on the uncorrupted points.
Then there is a constant $c$ and an algorithm such that:
\begin{enumerate}
\item (Completeness)
If $w = w^\ast$, then it outputs ``YES''.
\item (Soundness)
If $w \not\in \mathcal{C}_{ckh(k,\g,\d)}$, the algorithm outputs a hyperplane $\ell : \R^m \to \R$ such that $\ell(w) \geq 0$ but $\ell (w^\ast) < 0$.
\end{enumerate}
These two facts imply that the ellipsoid method with this separation oracle will terminate in $\poly (d, 1 / \epsilon)$ steps, and moreover, will with high probability output a $w'$ such that $\|w - w'\|_\infty \leq \ve / (N d k \log 1 / \ve)$ for some $w \in  \mathcal{C}_{ck h(k,\g,\delta)}$.
Moreover, it will do so in polynomially many iterations.
\end{theorem}
The proof is deferred to Section \ref{sec:gmmsepproof}.
\begin{algorithm}[htb]
\begin{algorithmic}[1]
\Function{SeparationOracleGMM}{$w$}
\State Let $\muhat = \sum_{i = 1}^N w_i X_i$.
\State For $i = 1, \ldots, N$, define $Y_i = X_i - \muhat$.
\State Let $M = \sum_{i = 1}^N w_i Y_i Y_i^T - I$.
\If{$\| M \|_{T_k} < \sum_{j \in [k]} \g_j + \frac{ck h(k,\g,\d)}{2}$}
	\State \textbf{return} ``YES".
\Else
        \State Let $\Lambda = \|M\|_{T_k}$.
        \State Let $U$ be a $d \times k$ matrix with orthonormal columns which span the top $k$ eigenvectors of $M$.
        \State \textbf{return} the hyperplane $\ell(w) = \Tr\left(U^T \left(\sum_{i = 1}^N w_i Y_i Y_i^T - I \right) U\right) - \Lambda > 0$
\EndIf
\EndFunction
\end{algorithmic}
\caption{Separation oracle sub-procedure for agnostically learning the span of the means of a GMM.}
\label{alg:sepgmm}
\end{algorithm}

After running this procedure, we technically do not have a set of weights in $\mathcal{C}_{ckh(k,\g,\d)}$.
But by the same argument as in Section \ref{sec:UnknownMeanConvex}, because the maximum distance between two points within any cluster is bounded, and we have the guarantee that $\| X_i - X_j \|^2 \leq O(d k \log 1 / \ve)$ for all $i, j$, we may assume we have a set of weights satisfying 
$$\left\| \sum_{i=1}^N w_i (X_i -\m) (X_i - \m)^T - I - \sum_{j \in [k]} \a_j (\m_j - \m) (\m_j - \m)^T \right\|_2 \leq 2ckh(k,\g,\d).$$

We require the following lemma, describing the accuracy of the empirical covariance matrix with the obtained weights.
\begin{lemma}
  \label{lem:gmmcovest}
  Let $\muhat = \sum_{i=1}^N w_i X_i$.
  After running the algorithm above, we have a vector $w$ such that
$$\left\| \sum_{i=1}^N w_i (X_i -\muhat) (X_i - \muhat)^T - I - \sum_{j \in [k]} \a_j (\m_j - \m) (\m_j - \m)^T \right\|_2 \leq 3ckh(k,\g,\d).$$
\end{lemma}
\begin{proof}
  By triangle inequality and Corollary \ref{cor:gmmclose},
  \begin{align*}
 &\left\| \sum_{i=1}^N w_i (X_i -\muhat) (X_i - \muhat)^T - I - \sum_{j \in [k]} \a_j (\m_j - \m) (\m_j - \m)^T \right\|_2 \\
 &\leq \left\| \sum_{i=1}^N w_i (X_i -\mu) (X_i - \mu)^T - I - \sum_{j \in [k]} \a_j (\m_j - \m) (\m_j - \m)^T \right\|_2 + \| \Delta \|_2^2 \\
 &\leq 2ckh(k,\g,\d) + O(\d) \leq 3ckh(k,\g,\d)
  \end{align*}

\end{proof}

\subsection{Spectral Clustering}
\label{sec:gmmspectralcluster}
Now that we have a good estimate of the true covariance matrix, we will perform spectral clustering while $\g$ is sufficiently large. 
We will adapt Lemma 6 from \cite{SOAJ14}, giving the following lemma:
\begin{lemma}
  \label{lem:gmmspeccluster}
  Given a weight vector $w$ as output by Algorithm \ref{alg:sepgmm}, if $\g \geq \Omega(\poly(k) \cdot\log 1/\ve)$, there exists an algorithm which produces a unit vector $v$ with the following guarantees:
  \begin{itemize}
    \item There exists a non-trival partition of $[k]$ into $S_0$ and $S_1$ such that $v^T\m_j > 0$ for all $j \in S_0$ and $v^T\m_j < 0$ for all $j \in S_1$;
    \item The probability of a sample being misclassified is at most $O(\poly(\ve/k))$, where a misclassification is defined as a sample $X$ generated from a component in $S_0$ having $v^TX < 0$, or a sample generated from a component in $S_1$ having $v^TX > 0$. 
  \end{itemize}
\end{lemma}
The algorithm will be as follows.
Let $v$ be the top eigenvector of 
$$\sum_{i=1}^N w_i (X_i -\muhat) (X_i - \muhat)^T - I.$$
For a sample $X$, cluster it based on the sign of $v^TX$. 
After performing this clustering, recursively perform our algorithm from the start on the two clusters.

The proof is very similar to that of Lemma 6 in \cite{SOAJ14}. 
Their main concentration lemma is Lemma 30, which states that they obtain a good estimate of the true covariance matrix, akin to our Lemma \ref{lem:gmmcovest}. 
Lemma 31 argues that the largest eigenvector of this estimate is highly correlated with the top eigenvector of the true covariance matrix.
Since $\g$ is large, this implies there is a large margin between the mean and the hyperplane.
However, by standard Gaussian tail bounds, the probability of a sample landing on the opposite side of this hyperplane is small.

We highlight the main difference between our approach and theirs.
For their clustering step, they require that no sample is misclustered with high probability.
As such, they may perform spectral clustering while $\g = \Omega\left(\poly(k) \cdot \log(d/\ve)\right)$.
We note that, in the next step of our algorithm, we will perform an exhaustive search.
This will result in an approximation which depends on the value of $\g$ at the start of the step, and as such, using the same approach as them would result in an overall approximation which depends logarithmically on the dimension.

We may avoid paying this cost by noting that our algorithm is agnostic.
They require that no sample is misclustered with high probability, while our algorithm tolerates that a $\poly(\ve/k)$-fraction of points are misclustered.
As such, we can continue spectral clustering until $\g = O\left(\poly(k) \cdot \log(1/\ve)\right)$.

\subsection{Exhaustive Search}
\label{sec:gmmsearch}
The final stage of the algorithm is when we know that all $\g_i$'s are sufficiently small.
We can directly apply the following lemma:
\begin{lemma}[Lemma 7 of \cite{SOAJ14}]
  Given a weight vector $w$ as output by Algorithm \ref{alg:sepgmm}, then the projection of $\frac{\m_j - \m}{\|\m_j - \m\|_2}$ onto the space orthogonal to the span of the top $k-1$ eigenvectors of 
  $$\left\|\sum_{i = 1}^N w_i (X_i - \muhat)(X_i - \muhat)^T - I\right\|_2 $$
  has magnitude at most
  $$O\left(\poly(k) \cdot \sqrt{h(k,\g,\d)}/\g_i^{1/2} \right) = O\left(\poly(k) \cdot \frac{\sqrt{\ve} \log (1/\ve)}{\g_i^{1/2}}\right).$$
\end{lemma}

At this point, our algorithm is identical to the exhaustive search of \cite{SOAJ14}.
We find the span of the top $k-1$ eigenvectors by considering the $(k-1)$-cube with side length $2\g$ centered at $\muhat$.
By taking an $\eta$-mesh over the points in this cube (for $\eta = \poly(\ve/dk)$ sufficiently small), we obtain a set of points $\tilde M$.
Via identical arguments as in the proof of Theorem 8 of \cite{SOAJ14}, for each $j \in [k]$, there exists some point $\tilde \m_j \in \tilde M$ such that
$$\|\tilde \m_j - \m_j\|_2 \leq O\left( \poly(k) \cdot \frac{\sqrt{\ve} \log(1/\ve)}{\sqrt{\a_j}}\right).$$
By taking a $k$-wise Cartesian product of this set, we are guaranteed to obtain a vector which has this guarantee simultaneously for all $\m_j$.

\subsection{Applying the Tournament Lemma}
\label{sec:gmmtournament}
In this section, we discuss details about how to apply our hypothesis selection algorithm.
First, in Section \ref{sec:gmmguess}, we describe how to guess the mixing weights and the variance of the components.
Then in Section \ref{sec:gmmprune}, we discuss how to ensure our hypotheses come from some fixed finite set, in order to deal with technicalities which arise when performing hypothesis selection with our adversary model.

\subsubsection{Guessing the Mixing Weights and Variance}
\label{sec:gmmguess}
The majority of our algorithm is focused on generating guesses for the means of the Gaussians.
In this section, we guess the remaining parameters: the mixing weights and the variance.
While most of these guessing arguments are standard, we emphasize that we reap an additional benefit because our algorithm is agnostic.
In particular, most algorithms must deal with error incurred due to misspecification of the parameters.
Since our algorithm is agnostic, we can pretend the misspecified parameter is the true one, at the cost of increasing the value of the agnostic parameter $\ve$.
If our misspecified parameters are accurate enough, the agnostic learning guarantee remains unchanged.

Guessing the mixing weights is fairly straightforward.
For some $\nu = \poly(\ve/k)$ sufficiently small, our algorithm generates a set of at most $(1/\nu)^k = \poly(k/\ve)^k$ possible mixing weights by guessing the values $\{0, \ve, \ve + \nu, \ve + 2\nu, \dots, 1 - \nu, 1\}$ for each $\a_j$.
Note that we may assume each weight is at least $\ve$, since components with weights less than this can be specified arbitrarily at a total cost of $O(k\ve)$ in total variation distance.

Next, we need to guess the variance $\s^2$ of the components.
To accomplish this, we will take $k+1$ samples (hoping to find only uncorrupted ones) and compute the minimum distance between any pair of them. 
Since we assume $k \ll 1/\ve$, we can repeatedly draw $k+1$ samples until we have the guarantee that at least one set is uncorrupted. 
If none of the $k+1$ samples are corrupted, then at least two of them came from the same component, and in our high-dimensional setting the distance between any pair of samples from the same component concentrates around $\sqrt{2d}\s$. 
After rescaling this distance, we can then multiplicatively enumerate around this value with granularity $\poly(\ve/dk)$ to get an estimate for $\s^2$ that is sufficiently good for our purposes.
Applying Corollary \ref{cor:kl-to-cov} bounds the cost of this misspecification by $O(\ve)$.
By rescaling the points, we may assume that $\s^2 = 1$.

\subsubsection{Pruning Our Hypotheses}
\label{sec:gmmprune}
In this section, we describe how to prune our set of hypotheses in order to apply Lemma \ref{tournamentLem}.
Recall that this lemma requires our hypotheses to come from some fixed finite set, rather than the potentially infinite set of GMM hypotheses.
We describe how to prune and discretize the set of hypotheses obtained during the rest of the algorithm to satisfy this condition.
For the purposes of this section, a hypothesis will be a $k$-tuple of $d$-dimensional points, corresponding only to the means of the components.
While the candidate mixing weights already come from a fixed finite set (so no further work is needed), the unknown variance must be handled similarly to the means.
The details for handling the variance are similar to (and simpler than) those for handling the means, and are omitted.

More precisely, this section will describe a procedure to generate a set of hypotheses $\mathcal{M}$, which is exponentially large in $k$ and $d$, efficiently searchable, and comes from a finite set of hypotheses which are fixed with respect to the true distribution. 
Then, given our set of hypotheses generated by the main algorithm (which is exponentially large in $k$ but polynomial in $d$), we iterate over this set, either replacing each hypothesis with a ``close'' hypothesis from $\mathcal{M}$ (i.e., one which is within $O(\ve)$ total variation distance), or discarding the hypothesis if none exists.
Finally, we run the tournament procedure of Lemma \ref{tournamentLem} on the resulting set of hypotheses.

At a high level, the approach will be as follows.
We will take a small set of samples, and remove any samples from this set which are clear outliers (due to having too few nearby neighbors).
This will give us a set of points, each of which are within a reasonable distance from some component mean.
Taking a union of balls around these samples will give us a space that is a subset of a union of (larger) balls centered at the component centers.
We take a discrete mesh over this space to obtain a fixed finite set of possible means, and round each hypothesis such that its means are within this set.

We start by taking $N = O(k \log(1/\tau) /\ve^2)$ samples, which is sufficient to ensure that the number of (uncorrupted) samples from component $j$ will be $(w_j \pm \Theta(\ve))N$ for all $j \in [k]$ with probability $1 - O(\tau)$.
Recall that we are assuming that $w_j = \Omega(\ve)$ for all $j$, as all other components may be defined arbitrarily at the cost of $O(k \ve)$ in total variation distance. 
This implies that even after corruption, each component has generated at least $\ve N$ uncorrupted samples.

By standard Gaussian concentration bounds, we know that if $N$ samples are taken from a Gaussian, the maximum distance between a sample and the Gaussian's mean will be at most $\zeta = O(\sqrt{d \log(N/\tau)})$ with probability $1 - \tau$. 
Assume this condition holds, and thus each component's mean will have at least $\ve N$ points within distance $\zeta$.
We prune our set of samples by removing any point with fewer than $\ve N$ other points at distance less than $2\zeta$. 
This will not remove any uncorrupted points, by the above assumption, and triangle inequality.
However, this will remove any corrupted points at distance at least $3\zeta$ from all component means, due to the fact that the adversary may only move an $\ve$-fraction of the points, and reverse triangle inequality.

Now, we consider the union of the balls of radius $3\zeta$ centered at each of the remaining points.
This set contains all of the component means, and is also a subset of the union of the balls of radius $6\zeta$ centered at the component means.
We discretize this set by taking its intersection with a lattice of side-length $\frac{\ve}{k\sqrt{d}}$.
We note that any two points in this discretization are at distance at most $\ve/k$. 
By a volume argument, the number of points in the intersection is at most $k \left(\frac{12\zeta k \sqrt{d} }{\ve}\right)^d$.
Each hypothesis will be described by the $k$-wise Cartesian product of these points, giving us a set $\mathcal{M}$ of at most $k^k \left(\frac{12\zeta k \sqrt{d} }{\ve}\right)^{kd}$ hypotheses.

Given a set of hypotheses $\mathcal{H}$ from the main algorithm, we prune it using $\mathcal{M}$ as a reference.
For each $h \in \mathcal{H}$, we see if there exists some $h' \in \mathcal{M}$ such that the means in $h$ are at distance at most $\ve/k$ from the corresponding means in $h'$.\footnote{We observe that the complexity of this step is polynomial in $d$ and $k$, not exponential, if one searches for the nearest lattice point in the sphere surrounding each unpruned sample, rather than performing a naive linear scan over the entire list.}
If such an $h'$ exists, we replace $h$ with $h'$ -- otherwise, $h$ is simply removed.
By Corollary \ref{cor:kl-to-means} and the triangle inequality, this replacement incurs a cost of $O(\ve)$ in total variation distance.
At this point, the conditions of Lemma \ref{tournamentLem} are satisfied and we may run this procedure to select a sufficiently accurate hypothesis.

\subsection{Handling Unequal Variances}
\label{sec:gmmdiffvar}
In this section, we describe the changes required to allow the algorithm to handle different variances for the Gaussians.
The main idea is to find the minimum variance of any component and perform clustering so we only have uncorrupted samples from Gaussians with variances within some known, polynomially-wide interval.
This allows us to grid within this interval in order to guess the variances, and the rest of the algorithm proceeds with minor changes.

The first step is to locate the minimum variance of any component.
Again using standard Gaussian concentration, in sufficiently high dimensions, if $N$ samples are taken from a Gaussian with variance $\s^2I$, the distance between any two samples will be concentrated around $\s(\sqrt{2d} - \Theta(d^{1/4}))$.
With this in hand, we use the following procedure to estimate the minimum variance.
For each sample $i$, record the distance to the $(\ve N + 1)$st closest sample.
We take the $(\ve N + 1)$st smallest of \emph{these} values, rescale it by $1/\sqrt{2d}$, and similar to before, guess around it using a multiplicative $(1 + \poly(\ve/kd))$ grid, which will give us an estimate $\hat \s_{min}^2$ for the smallest variance.
We note that discarding the smallest $\ve N$ fraction of the points prevents this statistic from being grossly corrupted by the adversary.
For the remainder of this section, assume that $\s_{min}^2$ is known exactly.

At this point, we partition the points into those that come from components with small variance, and those with large variance.
We will rely upon the following concentration inequality from \cite{SOAJ14}, which gives us the distance between samples from different components:
\begin{lemma}[Lemma 34 from \cite{SOAJ14}]
\label{lem:gmmconcdiffcomp}
Given $N$ samples from a collection of Gaussian distributions, with probability $1 - O(\tau)$, the following holds for every pair of samples $X, Y$:
$$\|X - Y\|_2^2 \in \left(d(\s_1^2 + \s_2^2) + \|\m_1 - \m_2\|_2^2\right)\left(1 \pm 4 \sqrt{\frac{\log \frac{N^2}{\tau}}{d}} \right), $$
where $X \sim \mathcal{N}(\m_1, \s_1^2 I)$ and $Y \sim \mathcal{N}(\m_2, \s_2^2 I)$.
\end{lemma}
Assume the event that this condition holds.
Now, let $H_\ell$ be the set of all points with at least $\ve N$ points at squared-distance at most $2\left(1 + \frac1k\right)^{\ell-1}\s_{\min}^2\left(1 + 4\sqrt{\frac{\log \frac{N^2}{\tau}}{d}}\right)$, for $\ell \in [k]$. 
Note that $H_\ell \subseteq H_{\ell+1}$.
Let $\ell^*$ be the minimum $\ell$ such that $H_\ell = H_{\ell + 1}$, or $k$ if no such $\ell$ exists, and partition the set of samples into $H_{\ell^*}$ and $\bar H_{\ell^*}$.
This partition will contain all samples from components with variance at most some threshold $t$, where $t \leq e \s_{\min}^2$ in $H_{\ell^*}$.
All samples from components with variance at least $t$ will fall into $\bar H_{\ell^*}$.
We continue running the algorithm with $H_{\ell^*}$, and begin the algorithm recursively on $\bar H_{\ell^*}$.\footnote{We require an additional guess of ``$k_1$ and $k_2$'': the split into how many components are within $H_{\ell^*}$ and $\bar H_{\ell^*}$ respectively.}

This procedure works due to the following argument.
When we compute $H_1$, we are guaranteed that it will contain all samples from components with variance $\s_{\min}^2$, by the upper bound in Lemma \ref{lem:gmmconcdiffcomp}.
However, it may also contain samples from other components -- in particular, those with variance at most $\gamma \s_{\min}^2$, for
$$\gamma \leq \left({1 + 16\sqrt{\frac{\log \frac{N^2}{\tau}}{d}}}\right)\bigg/\left({1 - 4\sqrt{\frac{\log \frac{N^2}{\tau}}{d}}}\right) \leq 1 + \frac{1}{k},$$
where the second inequality follows for $d$ sufficiently large.
Therefore, we compute $H_2$, which contains all samples from such components.
This is repeated for at most $k$ iterations, since if a set $H_{\ell + 1}$ is distinct from $H_{\ell}$, it must have added at least one component, and we have only $k$ components.
Note that $\left(1 + \frac1k\right)^k \leq e$, giving the upper bound on variances in $H_{\ell^*}$.

After this clustering step, the algorithm follows similarly to before.
The main difference is in the convex programming steps and concentration bounds.
For instance, before, we considered the set
$$\mathcal{C}_\eta = \left\{ w \in S_{N,\ve} : \left\| \sum_{i=1}^N w_i (X_i -\m) (X_i - \m)^T - \s^2 I - \sum_{j \in [k]} \a_j (\m_j - \m) (\m_j - \m)^T \right\|_2 \leq \eta \right\}.$$
Now, to reflect the different expression for the covariance of the GMM, we replace $\s^2 I$ with $\sum_{j \in [k]} \a_j \s_j^2 I$; for example:
$$\mathcal{C}_\eta = \left\{ w \in S_{N,\ve} : \left\| \sum_{i=1}^N w_i (X_i -\m) (X_i - \m)^T - \sum_{j \in [k]} \a_j \s_j^2 I - \sum_{j \in [k]} \a_j (\m_j - \m) (\m_j - \m)^T \right\|_2 \leq \eta \right\}.$$
We note that since all variances in each cluster are off by a factor of at most $e$, this will only affect our concentration and agnostic guarantees by a constant factor.

%!TEX root = ./main.tex

\section{Agnostically Learning Binary Product Distributions, via Filters} \label{sec:filterProduct}

In this section, we study the problem of agnostically learning a binary product distribution. 
Such a distribution is entirely determined by its coordinate-wise mean, which
we denote by the vector $p$, and our first goal is to estimate $p$ within $\ell_2$-distance $\widetilde{O}(\eps)$. 
%We can borrow many of the ideas that we sketched in earlier applications of the filtering approach. 
Recall that the approach for robustly learning the mean of an identity covariance Gaussian, sketched in the introduction,
%whose covariance is promised to be the identity 
was to compute the top absolute eigenvalue of a modified empirical covariance matrix. 
Our modification was crucially based on the promise that the covariance of the Gaussian is the identity. 
Here, it turns out that what we should do to modify the empirical covariance matrix 
is subtract off a diagonal matrix whose entries are $p_i^2$. 
These values seem challenging to directly estimate. 
Instead, we directly zero out the diagonal entries of the empirical covariance matrix. 
Then the filtering approach proceeds as before, and allows us to estimate $p$ 
within $\ell_2$-distance $\widetilde{O}(\eps)$, as we wanted. 
In the case when $p$ has no coordinates that are too biased towards either zero or one, 
our estimate is already $\widetilde{O}(\eps)$ close in total variation distance. 
We give an agnostic learning algorithm for this so-called balanced case (see Definition~\ref{def:balanced}) in Section~\ref{sec:bal-prod}. 

However, when $p$ has some very biased coordinates, this need not be the case. 
Each coordinate that is biased needs to be learned multiplicatively correctly. 
Nevertheless, we can use our estimate for $p$ that is close in $\ell_2$-distance 
as a starting point for handling binary product distributions that have imbalanced coordinates. 
Instead, we control the total variation distance via the $\chi^2$-distance between the mean vectors. 
Let $P$ and $Q$ be two product distributions whose means are $p$ and $q$ respectively. 
From Lemma~\ref{lem:prod-dtv-chi2}, it follows that
$$\dtv(P,Q)^2 \leq 4 \littlesum_i \frac{(p_i-q_i)^2}{q_i(1-q_i)} \;.$$
So, if our estimate $q$ is already close in $\ell_2$-distance to $p$, 
we can interpret the right hand side above as giving a renormalization of how 
we should measure the distance between $p$ and $q$ 
such that being close (in $\chi^2$-distance) implies that our estimate is close in total variation distance.
 We can then set up a corrected eigenvalue problem using our initial estimate $q$ as follows. 
 Let $\chi^2(v)_q = \littlesum_i v_i^2q_i(1-q_i)$. Then, we compute
$$\max_{\chi^2(v)_q=1}v^T\Sigma v \;,$$
where $\Sigma$ is the modified empirical covariance. 
Ultimately, we show that this yields an estimate that is $\widetilde{O}(\sqrt{\eps})$ close in total variation distance. See Section~\ref{sec:product} for further details.

\subsection{The Balanced Case} \label{sec:bal-prod}

The main result of this section is the following theorem:

\begin{theorem} \label{thm:binary-product-L2}
Let $P$ be a binary product distribution in $d$ dimensions and $\eps,\tau>0.$
Let $S$ be a multiset of $\Theta(d^4 \log(1/\tau)/\eps^2)$ independent samples from $P$
and $S'$ be a multiset obtained by arbitrarily changing an $\eps$-fraction of the points in $S.$
There exists a polynomial time algorithm that returns a product distribution $P'$
such that, with probability at least $1-\tau$,
we have $\|p-p'\|_2 = O(\eps\sqrt{\log(1/\eps)}),$
where $p$ and $p'$ are the mean vectors of $P$ and $P'$ respectively.
\end{theorem}

Note that Theorem~\ref{thm:binary-product-L2} applies to all binary product distributions,
and its performance guarantee relates the $\ell_2$-distance between the mean vectors
of the hypothesis $P'$ and the target product distribution $P$. If $P$ is balanced, i.e., it 
does not have coordinates that are too biased towards $0$ or $1$, 
this $\ell_2$-guarantee implies a similar total variation guarantee.
Formally, we have:

\begin{definition} \label{def:balanced}
For $0< c < 1/2,$ we say that a binary product distribution is \emph{$c$-balanced} if the expectation
of each coordinate is in $[c, 1-c].$
\end{definition}

For $c$-balanced binary product distributions,
we have the following corollary of Lemma~\ref{lem:prod-dtv-chi2}:
\begin{fact} \label{fact:balanced}
Let $P$ and $Q$ be $c$-balanced binary product distributions with mean vectors $p$ and $q.$
Then, we have that $\dtv(P,Q)=O\left(c^{-1/2} \cdot \|p-q\|_2 \right).$
\end{fact}

That is, for two $c$-balanced binary product distributions, where $c$ is a fixed constant,
the $\ell_2$-distance between their mean vectors is a good proxy for their total variation distance.
Using Fact~\ref{fact:balanced}, we obtain the following corollary of Theorem~\ref{thm:binary-product-L2}:

\begin{corollary} \label{cor:binary-product-balanced}
Let $P$ be a $c$-balanced binary product distribution in $d$ dimensions, 
where $c>0$ is a fixed constant, and $\epsilon,\tau>0$. 
Let $S$ be a multiset of $\Theta(d^4\log(1/\tau)/\eps^2)$ independent samples from $P$
and $S'$ be a multiset obtained by arbitrarily changing an $\eps$-fraction of the points in $S.$
There exists a polynomial time algorithm that returns a product distribution $P'$ such that with probability at least $1-\tau$,
$\dtv(P', P) = O(\eps\sqrt{\log(1/\eps)}/\sqrt{c}).$
\end{corollary}

\medskip

We start by defining a condition on the uncorrupted set of samples $S,$
under which our algorithm will succeed.

\begin{definition}[good set of samples] \label{def:good-balanced}
Let $P$ be an arbitrary distribution on $\{0, 1\}^d$ and $\eps>0.$
We say that a multiset $S$ of elements in $\{0, 1\}^d$ is {\em $\eps$-good with respect to $P$}
if for every affine function $L: \{0, 1\}^d \to \R$ we have
$|\Pr_{X \in_u S}(L(X) \ge 0) - \Pr_{X\sim P}(L(X) \ge 0)| \leq \eps/d.$
\end{definition}

The following simple lemma shows that a sufficiently large set of independent samples from $P$
is $\eps$-good (with respect to $P$) with high probability.

\begin{lemma} \label{lem:random-good}
Let $P$ be an arbitrary distribution on $\{0, 1\}^d$ and $\eps,\tau>0.$
If the multiset $S$ is obtained by taking  $\Omega((d^4+d^2\log (1/\tau))/\eps^2)$ independent samples from $P,$
it is $\eps$-good with respect to $P$ with probability at least $1-\tau .$
\end{lemma}
\begin{proof}
For a fixed affine function $L: \{0, 1\}^d \to \R,$  an application of the Chernoff bound yields that
after drawing $N$ samples from $P,$
we have that  $|\Pr_{X\in_u S}(L(X) \ge 0) - \Pr_{X\sim P}(L(X) \ge 0)| > \eps/d$
with probability at most $2 \exp(-N\eps^2/d^2).$
Since there are at most $2^{d^2}$ distinct linear threshold functions on $\{0, 1\}^d,$
by the union bound, the probability that there exists an $L$ satisfying the condition
$|\Pr_{X \in_u S}(L(X) \ge 0) - \Pr_{X\sim P}(L(X) \ge 0)| > \eps/d$ is at most
$2^{d^2+1} \exp(-N\eps^2/d^2),$ which is at most $\tau$ for $N=\Omega((d^4+d^2\log (1/\tau))/\eps^2).$
\end{proof}

Recall (see Definition~\ref{def:Delta-G}) that $\Delta(S,S')$ is the size of the symmetric difference of $S$ and $S'$ divided by the cardinality of $S.$

\medskip

Our agnostic learning algorithm establishing Theorem~\ref{thm:binary-product-L2}
is obtained by repeated application of the efficient procedure
whose performance guarantee is given in the following proposition:

\begin{proposition} \label{prop:filter-L2}
Let $P$ be a binary product distribution with mean vector $p$
and $\eps>0$ be sufficiently  small.
Let $S$ be $\eps$-good with respect to $P$, and
$S'$ be any multiset with $\Delta(S,S') \leq 2\eps.$
There exists a polynomial time algorithm \textsc{Filter-Balanced-Product}
that, given $S'$ and $\eps>0,$ returns one of the following:
\begin{itemize}
\item[(i)] A mean vector $p'$ such that $\|p-p'\|_2 = O(\eps\sqrt{\log(1/\eps)}).$
\item[(ii)] A multiset $S'' \subseteq S'$ such that $\Delta(S,S'') \leq \Delta(S,S') - 2\eps/d.$
\end{itemize}
\end{proposition}

We start by showing how Theorem~\ref{thm:binary-product-L2}  follows
easily from Proposition~\ref{prop:filter-L2}.

\begin{proof}[Proof of Theorem~\ref{thm:binary-product-L2}]
The proof of Theorem~\ref{thm:binary-product-L2} is very similar to that of Theorem~\ref{thm:filter-gaussian-mean}, however, we include it here for completeness.
By the definition of $\Delta(S, S'),$ since $S'$ has been obtained from $S$
by corrupting an $\eps$-fraction of the points in $S,$ we have that
$\Delta(S, S') \le 2\eps.$ By Lemma~\ref{lem:random-good}, the set $S$ of uncorrupted samples
is $\eps$-good with respect to $P$ with probability at least $1-\tau.$
We henceforth condition on this event.

Our algorithm iteratively applies the  \textsc{Filter-Balanced-Product} procedure of Proposition~\ref{prop:filter-L2} until it terminates
returning a mean vector $p'$ with $\|p-p'\|_2 = O(\eps \sqrt{\log(1/\eps)}).$
We claim that we need at most $d+1$ iterations for this to happen.
Indeed, the sequence of iterations results in a sequence of sets $S_0 = S', S_1', \ldots,$
such that $\Delta(S,S_i') \leq \Delta(S,S')  - i \cdot (2\eps/d).$
Thus, if the algorithm does not terminate in the first $d$ iterations, we have $S'_d = S,$
and in the next iteration we output the sample mean of $S$.
\end{proof}

\subsubsection{Algorithm \textsc{Filter-Balanced-Product}: Proof of Proposition~\ref{prop:filter-L2}}

In this section, we describe the efficient procedure establishing
Proposition~\ref{prop:filter-L2} followed by its proof of correctness.
Our algorithm \textsc{Filter-Balanced-Product} is very simple: 
We consider the empirical distribution defined by the (corrupted) sample multiset $S'.$
We calculate its mean vector $\mu^{S'}$ and covariance matrix $M$.
If the matrix $M$ has no large eigenvalues, we return $\mu^{S'}.$
Otherwise, we use the eigenvector $v^{\ast}$ corresponding to the maximum magnitude eigenvalue $\lambda^{\ast}$ of $M$
and the mean vector $\mu^{S'}$ to define a filter.
We zero out the diagonal elements of the covariance matrix for the following reason:
The diagonal elements could contribute up to $\Omega(1)$ to the spectral norm,
even without noise. This would prevent us from obtaining the desired error of $\widetilde{O}(\eps).$
Our efficient filtering procedure is presented in detailed pseudocode below.

\bigskip
\begin{algorithm}%[htb]
\begin{algorithmic}[1]
\Procedure{Filter-Balanced-Product}{$\eps, S'$}
\INPUT A multiset $S'$ such that there exists an $\eps$-good $S$ with $\Delta(S, S') \le 2\eps$
\OUTPUT Multiset $S''$ or mean vector $p'$ satisfying Proposition~\ref{prop:filter-L2}
\State Compute the sample mean $\mu^{S'}=\E_{X\in_u S'}[X]$ and the sample covariance $M$ with zeroed diagonal,
\State i.e., $M = (M_{i, j})_{1 \le i, j \le d}$ with $M_{i,j} = \E_{X\in_u S'}[(X_i-\mu^{S'}_i) (X_j-\mu^{S'}_j)]$, $i \ne j$, and $M_{i, i} = 0.$
\State  Compute approximations for the largest absolute eigenvalue of $M$, $\lambda^{\ast} := \|M\|_2,$
and the associated unit eigenvector $v^{\ast}.$

\If {$\|M\|_2 \leq O(\eps \log (1/\eps)),$}
 \textbf{return} $\mu^{S'}.$ \label{step:bal-small}
\EndIf
\State \label{step:bal-large}  Let $\delta := 3 \sqrt{ \eps  \|M\|_2}.$ Find $T>0$ such that
$$
\Pr_{X\in_u S'}(|v^{\ast} \cdot (X-\mu^{S'})|>T+\delta) > 8\exp(-T^2/2)+8\eps/d.
$$
\State \textbf{return} the multiset $S''=\{x\in S': |v^{\ast} \cdot (x-\mu^{S'}) | \leq T+\delta\}$.

\EndProcedure
\end{algorithmic}
\caption{Filter algorithm for a balanced binary product distribution}
\label{alg:balanced-product}
\end{algorithm}

%\begin{comment}
%\fbox{\parbox{6.1in}{
%{\bf Algorithm} \textsc{Filter-Balanced-Product}\\
%{\em Input:} $\eps>0$ and multiset $S'$ such that there exists an $\eps$-good $S$ with $\Delta(S, S') \le 2\eps$ \\
%{\em Output:}  Multiset $S''$ or mean vector $p'$ satisfying Proposition~\ref{prop:filter-L2}

%\vspace{0.2cm}

%\begin{enumerate}
%\item Compute the sample mean $\mu^{S'}=\E_{X\in_u S'}[X]$ and the sample covariance matrix $M$ with zeroed diagonal,
%i.e., $M = (M_{i, j})_{1 \le i, j \le d}$
%with $M_{i,j} = \E_{X\in_u S'}[(X_i-\mu^{S'}_i) (X_j-\mu^{S'}_j)]$, $i \ne j$, and $M_{i, i} = 0.$

%\item  Compute approximations for the largest absolute eigenvalue of $M$, $\lambda^{\ast} := \|M\|_2,$
%and the associated unit eigenvector $v^{\ast}.$

%\item \label{step:bal-small}
%If $\|M\|_2 \leq O(\eps \log (1/\eps)),$ return $\mu^{S'}.$

%\item \label{step:bal-large}  Let $\delta := 3 \sqrt{ \eps  \|M\|_2}.$ Let $T>0$ be such that
%$$\Pr_{X\in_u S'}(|v^{\ast} \cdot (X-\mu^{S'})|>T+\delta) > 8\exp(-T^2/2)+8\eps/d.$$
%Return the multiset $S''=\{x\in S': |v^{\ast} \cdot (x-\mu^{S'}) | \leq T+\delta\}$.
%\end{enumerate}

%}}
%\end{comment}

\bigskip

%This completes the description of the algorithm.

\noindent {\bf Tightness of our Analysis.}
We remark that the analysis of our filter-based algorithm is tight, and more generally
our bound of $O(\eps\sqrt{\log(1/\eps)})$ is a bottleneck for filter-based approaches.

More specifically, we note that our algorithm will never successfully add
points back to $S$ after they have been removed by the adversary.
Therefore, if an $\eps$-fraction of the points in $S$ are changed,
our algorithm may be able to remove these outliers from $S',$
but will not be able to replace them with their original values.
These changed values can alter the sample mean by as much as
$\Omega(\eps\sqrt{\log(1/\eps)}).$

To see this, consider the following example.
Let $P$ be the product distribution with mean $p,$
where $p_i=1/2$ for all $i.$
Set $\eps=2^{-(d-1)}.$
We draw a $\Theta(d^4\log(1/\tau)/\eps^2)$ size multiset $S$ which we assume is $\eps$-good.
The fraction of times the all-zero vector appears in $S$ is less than $2^{-(d-1)}.$
So, the adversary is allowed to corrupt all such zero-vectors.
More specifically, the adversary
replaces each occurrence of  the all-zero vector
with fresh samples from $P,$
repeating if any all-zero vector is drawn.
In effect, this procedure generates samples from
the distribution $\wt P,$
defined as $P$ conditioned on not being the all-zero vector.
Indeed, with high probability, the set $S'$ is $\eps$-good for $\wt P.$
So, with high probability, the mean of $S'$ in each coordinate is at least $1/2 + 2^{-(d+2)}.$
Thus, the $\ell_2$-distance between the mean vectors of $P$ and $\wt P$
is at least $\sqrt{d}2^{-(d+2)} = \Theta(\eps\sqrt{\log(1/\eps)}).$
Note that for any affine function $L,$ we have that
$\Pr_{X\in_u S'}(L(X) \geq 0) \leq \Pr_{X\in_u S}(L(X) \geq 0)/(1 - \eps) + 2\eps/d,$
which means that no such function can effectively distinguish between $S' \setminus S$ and $S,$
as would be required by a useful filter.

\medskip

The rest of this section is devoted to the proof of correctness of algorithm \textsc{Filter-Balanced-Product}.

\subsubsection{Setup and Basic Structural Lemmas} \label{ssec:L2-setup}

By definition, there exist disjoint multisets $L,E,$ of points in $\{0, 1\}^d$, where $L \subset S,$
such that $S' = (S\setminus L) \cup E.$ With this notation, we can write $\Delta(S,S')=\frac{|L|+|E|}{|S|}.$
Our assumption $\Delta(S,S') \le 2\eps$ is equivalent to $|L|+|E| \le 2\eps \cdot |S|,$ and the definition of $S'$
directly implies that $(1-2\eps)|S| \le  |S'| \le (1+2\eps) |S|.$ Throughout the proof, we assume that $\eps$
is a sufficiently small constant. %The assumption that $\eps \leq 1/2$ suffices for our purposes.
Our analysis will make essential use of the following matrices:
\begin{itemize}
\item $M_P$ denotes the matrix with $(i, j)$-entry $\E_{X\sim P}[(X_i - \mu^{S'}_i)(X_j - \mu^{S'}_j)],$ but $0$ on the diagonal.

\item $M_S$ denotes the matrix with $(i, j)$-entry $\E_{X\in_u S}[(X_i - \mu^{S'}_i)(X_j - \mu^{S'}_j)],$ but $0$ on the diagonal.

\item $M_E$ denotes the matrix with $(i, j)$-entry $\E_{X\in_u E}[(X_i - \mu^{S'}_i)(X_j - \mu^{S'}_j)].$

\item $M_L$ denotes the matrix with $(i, j)$-entry $\E_{X\in_u L}[(X_i - \mu^{S'}_i)(X_j - \mu^{S'}_j)].$
\end{itemize}
\noindent Our first claim follows from the Chernoff bound and the definition of a good set:
\begin{claim}\label{cernCor}
Let $w \in \R^d$ be any unit vector, then for any $T>0$,
$$
\Pr_{X\in_u S}(|w\cdot(X-\mu^{S'})| > T + \|\mu^{S'}-p\|_2) \leq 2 \exp(-T^2/2)+\eps/d.
$$
and
$$
\Pr_{X\sim P}(|w\cdot(X-\mu^{S'})| > T + \|\mu^{S'}-p\|_2) \leq 2 \exp(-T^2/2).
$$
\end{claim}
\begin{proof}
Since $S$ is $\eps$-good, the first inequality follows from the second one.
To prove the second inequality, it suffices to bound the probability that $|w\cdot (X-\mu^{S'}) - \E[w\cdot(X-\mu^{S'})]|> T$,
$X \sim P$, since the expectation in question is $w\cdot (p-\mu^{S'})$, whose absolute value is at most $\|\mu^{S'}-p\|_2$, by Cauchy-Schwarz.
Note that $w\cdot(X-\mu^{S'})$ is a sum of independent random variables $w_i(X_i-\mu^{S'}_i)$,
each supported on an interval of length $2|w_i|$. An application of the Chernoff bound completes the proof.
\end{proof}

The following sequence of lemmata bound from above the spectral norms of the associated matrices.
Our first simple lemma says that the (diagonally reduced) empirical covariance matrix $M_S$,
where $S$ is the set of uncorrupted samples drawn from the binary product distribution $P,$
is a good approximation to the matrix $M_{P},$ in spectral norm.

\begin{lemma}\label{operatorCloseLem}
If $S$ is $\eps$-good, $\|M_P - M_S\|_2 \leq O(\eps).$
\end{lemma}
\begin{proof}
It suffices to show that $ |(M_P)_{i,j} - (M_S)_{i,j}| \leq O(\eps/d)$
for all $i\neq j$. Then, we have that $$\|M_P - M_S\|_2 \leq \|M_P - M_S\|_F \leq O(\eps).$$
%For $i=j$, $(M_P)_{i,j} = (M_S)_{i,j}=0,$ so this holds. Otherwise,
Let $e_i$ denote the standard basis vector in the $i$-th direction in $\R^d.$
For $i \neq j$ we have:
\begin{align*}
(M_P)_{i,j} & = \E_{X \sim P}[(X_i - \mu^{S'}_i)(X_j - \mu^{S'}_j)] \\
& = \E_{X \sim P}[X_iX_j] - \mu^{S'}_i \E_{X \sim P}[X_j] - \mu^{S'}_j \E_{X \sim P}[X_i] + \mu^{S'}_j \mu^{S'}_i \\
& = \Pr_{X \sim P}((e_i+e_j) \cdot X \geq 2) - \mu^{S'}_i \Pr_{X \sim P}(e_j \cdot X \geq 1) - \mu^{S'}_j \Pr_{X \sim P}(e_i \cdot X \geq 1) + \mu^{S'}_j \mu^{S'}_i \;.
\end{align*}
A similar expression holds for $M_S$ except with probabilities for $X \in_u S.$
Since $S$ is $\eps$-good with respect to $P$, we have $|(M_P)_{i,j} - (M_S)_{i,j}| \leq \eps/d + \mu^{S'}_i  \eps/d + \mu^{S'}_j \eps/d \leq 3\eps/d$.
This completes the proof.
\end{proof}

As a simple consequence of the above lemma, we obtain the following:

\begin{claim} \label{M-norm-lc}
If $S$ is $\eps$-good, $\|M-(1/|S'|)(|S| M_P + |E| M_E - |L| M_L)\|_2 = O(\epsilon).$
\end{claim}
\begin{proof}
First note that we can write $|S'|M = |S| M_S + |E| M_E^0 - |L| M_L^0,$
where $M_E^0$ and $M_L^0$ are obtained from $M_E$ and $M_L$ by zeroing out the diagonal.
Observe that $|E|+|L|  =  O(\eps) |S'|.$ This follows from the assumption that
$\Delta(S, S') \le 2\eps$ and the definition of $S'.$
Now note that the matrices $M_E-M_E^0$ and $M_L-M_L^0$
are diagonal with entries at most $1,$ and thus have spectral norm at most $1.$
The claim now follows from Lemma \ref{operatorCloseLem}.
\end{proof}

Recall that if $\mu^{S'} = p$, $M_P$ would equal the (diagonally reduced) covariance matrix
of the product distribution $P,$ i.e., the identically zero matrix. The following simple lemma
bounds from above the spectral norm of $M_P$ by the $\ell_2^2$-norm
between the corresponding mean vectors:

\begin{lemma} \label{lem:MPi-norm}
We have that $\|M_P\|_2 \leq \|\mu^{S'}-p\|_2^2.$
\end{lemma}
\begin{proof}
Note that $(M_{P})_{i,j} = (\mu^{S'}_i-p_i)(\mu^{S'}_j-p_j)$ for $i\neq j$ and $0$ otherwise.
Therefore, $M_P$ is the difference of $(\mu^{S'}-p)(\mu^{S'}-p)^T$ and the diagonal matrix with entries $(\mu^{S'}_i-p_i)^2.$
This in turn implies that $$(\mu^{S'}-p)(\mu^{S'}-p)^T \succeq M_P \succeq \mathrm{Diag}(-(\mu^{S'}_i-p_i)^2) \;.$$
Note that both bounding matrices have spectral norm at most $\|\mu^{S'}-p\|_2^2$, hence so does $M_P.$
\end{proof}

The following lemma, bounding from above the spectral norm of $M_L,$
is the main structural result of this section.
This is the core result needed to establish that the subtractive error
cannot change the sample mean by much:

\begin{lemma}\label{lem:ML-bound}
We have that $\|M_L\|_2  =  O(\log(|S|/|L|)+\|\mu^{S'}-p\|_2^2+\eps \cdot |S|/|L| ),$ hence
$$(|L|/|S'|) \cdot \|M_L\|_2  =  O(\eps \log(1/\eps)+ \eps \|\mu^{S'}-p\|_2^2).$$
\end{lemma}
\begin{proof}
Since $L \subseteq S$, for any $x \in \{0, 1\}^d$, we have that
\begin{equation} \label{eqn:L-subset-S-prod}
|S| \cdot \Pr_{X \in_u S}(X= x) \geq |L| \cdot \Pr_{X \in_u L}(X= x) \;.
\end{equation}
Since $M_L$ is a symmetric matrix, we have $\|M_L\|_2 = \max_{\|v\|_2=1} |v^T M_L v|.$ So,
to bound $\|M_L\|_2$ it suffices to bound $|v^T M_L v|$ for unit vectors $v.$
By definition of $M_L,$
for any $v \in \R^d$ we have that
$$|v^T M_L v| = \E_{X \in_u L}[|v\cdot (X-\mu^{S'})|^2].$$
The RHS is in turn bounded from above as follows:
\begin{align*}
\E_{X \in_u L}[|v\cdot (X-\mu^{S'})|^2] & = 2\int_0^{\sqrt{d}} \Pr_{X \in_u L}\left(|v\cdot(X-\mu^{S'})|>T\right) \cdot T dT\\
& \leq 2 \int_0^{\sqrt{d}} \min \left\{ 1,  |S|/|L| \cdot \Pr_{X \in_u S}\left(|v\cdot(X-\mu^{S'})|>T\right)  \right\} TdT\\
& \ll \int_0^{4\sqrt{\log(|S|/|L|)}+\|\mu^{S'}-p\|_2}T dT \\
& + (|S|/|L|) \int_{4\sqrt{\log(|S|/|L|)}+\|\mu^{S'}-p\|_2}^{\sqrt{d}} \left( \exp(-(T-\|\mu^{S'}-p\|_2)^2/2)T+\eps T/d \right) dT\\
& \ll \log(|S|/|L|) + \|\mu^{S'}-p\|_2^2 + \eps \cdot |S|/|L| \;,
\end{align*}
where the second line follows from (\ref{eqn:L-subset-S-prod}) and the third line follows from Claim~\ref{cernCor}.
This establishes the first part of the lemma.

The bound $(|L|/|S|) \|M_L\|_2  =  O(\eps \log(1/\eps)+ \eps \|\mu^{S'}-p\|_2^2)$
follows from the previously established bound
using the monotonicity of the function $x \log (1/x),$ and the fact that $|L|/|S| \leq 2\eps.$
The observation $|S|/|S'| \leq 1+2\eps \leq 2$ completes the proof of the second part of the lemma.
\end{proof}

\noindent Claim~\ref{M-norm-lc} combined with
Lemmas~\ref{lem:MPi-norm} and~\ref{lem:ML-bound} and the triangle inequality yield the following:
\begin{corollary}\label{MApproxCor}
We have that
$\|M -(|E|/|S'|) M_E \|_2  =  O(\eps\log(1/\eps)+\|\mu^{S'}-p\|_2^2).$
\end{corollary}

We are now ready to analyze the two cases of the algorithm  \textsc{Filter-Balanced-Product}.

\subsubsection{The Case of Small Spectral Norm} \label{ssec:accurate-mean}
We start by analyzing the case where the mean vector $\mu^{S'}$ is returned.
This corresponds to the case that the spectral norm of $M$ is appropriately small,
namely $\|M\|_2 \leq O(\eps \log (1/\eps)).$
We start with the following simple claim:

\begin{claim} \label{claim:m}
Let $\mu^E, \mu^L$ be the mean vectors of $E$ and $L$ respectively.
Then, $\|\mu^E-\mu^{S'}\|_2^2 \leq \|M_E\|_2$ and $\|\mu^L-\mu^{S'}\|_2^2 \le \|M_L\|_2.$
\end{claim}
\begin{proof}
We prove the first inequality, the proof of the second being identical.
Note that $M_E$ is a symmetric matrix, so $\|M_E\|_2 = \max_{\|v\|_2=1} |v^T M_E v|.$
Moreover, for any vector $v$ we have that
$$v^T M_E v = \E_{X \in_u E}[|v\cdot (X-\mu^{S'})|^2] \geq |v\cdot (\mu^E-\mu^{S'})|^2.$$
Let $w=\mu^E-\mu^{S'}$ and take $v=w/\|w\|_2.$
We conclude that $\|M_E\|_2 \geq \|w\|_2^2,$ as desired.
\end{proof}

The following crucial lemma, bounding from above the distance $\|\mu^{S'}-p\|_2$ as a function of $\eps$
and $\|M\|_2,$ will be important for both this and the following subsections.

\begin{lemma} \label{lem:delta-distance}
We have that $\|\mu^{S'}-p\|_2 \leq 2\sqrt{\eps\| M\|_2} + O(\eps\sqrt{\log(1/\eps)}).$
\end{lemma}
\begin{proof}
First we observe that the mean vector $\mu^S$ of the uncorrupted sample set $S$ is close to $p.$
Since $S$ is $\eps$-good, this follows from the fact that
for any $i \in [d],$ we have $$|\mu^S_i - p_i| = |\Pr_{X \in_u S}[e_i \cdot X \geq 1]- \Pr_{X \sim P}[e_i \cdot X \geq 1]| \leq \eps/d.$$
Therefore, we get that $\|\mu^S-p\|_2 \leq \eps/\sqrt{d}.$

Consider $\mu^E$ and $\mu^L$, the mean vectors of $E$ and $L$, respectively.
By definition, we have that $$|S'| \mu^{S'} = |S| \mu^S + |E| \mu^E - |L| \mu^L \;,$$
and thus by the triangle inequality we obtain
$$\| \mu^{S'}-p \|_2  \leq  \|(|E|/|S'|)(\mu^E-p) -(|L|/|S'|) (\mu^L-p)\|_2+ \eps/\sqrt{d} \;.$$
Therefore, we have the following sequence of inequalities:
\begin{align*}
\|\mu^{S'}-p\|_2  & \leq (|E|/|S'|)  \cdot \|\mu^{E}- \mu^{S'}\|_2 + (|L|/|S'|) \cdot  \|\mu^{L}- \mu^{S'}\|_2 +  O(\eps) \cdot \|\mu^{S'}-p\|_2 + \eps/\sqrt{d}\\
& \leq (|E|/|S'|)  \cdot \sqrt{\|M_E\|_2}+  (|L|/|S'|) \cdot \sqrt{\|M_L\|_2}+  O(\eps) \cdot \|\mu^{S'}-p\|_2)) + \eps/\sqrt{d} \\
& \leq O(\eps\sqrt{\log(1/\eps)})+ (3/2) \sqrt{\epsilon\| M\|_2} + O(\sqrt{\eps}) \cdot \|\mu^{S'}-p\|_2\\
& \leq O(\eps\sqrt{\log(1/\eps)}))+  (3/2) \sqrt{\epsilon\| M\|_2} + \|\mu^{S'}-p\|_2/4  \;,
\end{align*}
where the first line follows from the triangle inequality, the second uses Claim~\ref{claim:m},
while the third uses Lemma~\ref{lem:ML-bound} and Corollary~\ref{MApproxCor}.
Finally, the last couple of lines assume that $\eps$ is sufficiently small.
The proof of Lemma~\ref{lem:delta-distance} is now complete.
\end{proof}

We can now deduce the correctness of Step~\ref{step:bal-small}  of the algorithm \textsc{Filter-Balanced-Product}, since for
 $\|M\|_2 \leq O(\eps \log (1/\eps)),$ Lemma~\ref{lem:delta-distance}
 directly implies that $\|\mu^{S'}-p\|_2  = O(\eps\sqrt{\log(1/\eps)}).$

\subsubsection{The Case of Large Spectral Norm} \label{ssec:filter-l2}
We next show the correctness of the algorithm \textsc{Filter-Balanced-Product} if it returns a
filter (rejecting an appropriate subset of $S'$) in Step~\ref{step:bal-large}.
This corresponds to the case that  $\|M\|_2  \geq C \eps \log(1/\eps),$
for a sufficiently large universal constant $C>0.$
We will show that the multiset $S'' \subset S'$ computed in Step~\ref{step:bal-large}
satisfies $\Delta(S, S'') \leq \Delta(S, S')  - 2\eps/d.$

We start by noting that, as a consequence of Lemma~\ref{lem:delta-distance}, we have the following:

\begin{claim} \label{claim:mean-l2-delta}
We have that $\|\mu^{S'}-p\|_2 \leq \delta: = 3 \sqrt{\eps \|M\|_2}.$
\end{claim}
\begin{proof}
By Lemma \ref{lem:delta-distance}, we have that
$\|\mu^{S'}-p\|_2 \leq 2\delta/3 + O(\eps\sqrt{\log(1/\eps)}).$
Recalling that $\|M\|_2  \geq C \eps \log(1/\eps),$
if $C>0$ is sufficiently large,
the term $O(\eps\sqrt{\log(1/\eps)})$ is at most $\delta/3.$
\end{proof}

By construction, $v^{\ast}$ is the unit eigenvector corresponding
to the maximum magnitude eigenvalue of $M.$
Thus, we have $(v^{\ast})^T M v^{\ast} = \|M\|_2  = \delta^2/(9\eps).$
We thus obtain that
\begin{equation} \label{eqn:var-E-v}
\E_{X\in_u E}[|v^{\ast} \cdot(X-\mu^{S'})|^2] = (v^{\ast})^T M_E v^{\ast} \ge \frac{\delta^2|S'|}{20\eps |E|} \;,
\end{equation}
where the equality holds by definition,
and the inequality follows from Corollary~\ref{MApproxCor}
and Claim~\ref{claim:mean-l2-delta} using the fact that $\eps$ is sufficiently small
and the constant $C$ is sufficiently large
(noting that the constant in the RHS of  Corollary~\ref{MApproxCor} does not depend on $C$).

We show that (\ref{eqn:var-E-v}) implies the existence of a $T>0$
with the properties specified in Step~\ref{step:bal-large} of the algorithm  \textsc{Filter-Balanced-Product}.
More specifically, we have the following crucial lemma:

\begin{lemma} \label{lem:T-exists}
If $\|M\|_2 \geq C \eps \log(1/\eps),$ for a sufficiently large constant $C>0,$
there exists a $T>0$ satisfying the property in Step~\ref{step:bal-large} of the algorithm  \textsc{Filter-Balanced-Product},
i.e., such that $$\Pr_{X\in_u S'}(|v^{\ast}\cdot (X-\mu^{S'})|>T+\delta) > 8\exp(-T^2/2)+8\eps/d \;.$$
\end{lemma}
\begin{proof}
Assume for the sake of contradiction that this is not the case,
i.e., that for all $T>0$ we have that
\begin{equation} \label{eqn:contradictionProd}
\Pr_{X\in_u S'}(|v^{\ast}\cdot (X-\mu^{S'})| \ge T+\delta)  \le 8\exp(-T^2/2)+8\eps/d \;.
\end{equation}
Since $E \subseteq S',$ for all $x \in \{0, 1\}^d$, we have that
$|S'|\Pr_{X\in_u S'}[X=x] \geq |E| \Pr_{Y\in_u E}[Y=x].$
This fact combined with (\ref{eqn:contradictionProd}) implies that for all $T>0$
\begin{equation} \label{eqn:contradiction2Prod}
\Pr_{Y\in_u E}(|v^{\ast}\cdot(Y-\mu^{S'})| \ge T+\delta) \ll (|S'|/|E|)(\exp(-T^2/2)+\eps/d) \;.
\end{equation}
Using (\ref{eqn:var-E-v}) and (\ref{eqn:contradiction2Prod}), we have the following sequence of
inequalities:
\begin{align*}
\delta^2|S'|/(\eps |E|) & \ll \E_{Y\in_u E}[|v^{\ast}\cdot(Y-\mu^{S'})|^2]\\
& = 2\int_0^\infty \Pr_{Y\in_u E}\left(|v^{\ast}\cdot(Y-\mu^{S'})| \ge T\right)\cdot T dT\\
& \ll  (|S'|/|E|)\int_0^{O(\sqrt{d})} \min \left\{ |E|/|S'|,\exp(-(T-\delta)^2/2)+\eps/d \right\}T dT\\
& \ll \int_0^{4\sqrt{\log(|S'|/|E|)}+\delta} Tdt +\int_{4\sqrt{\log(|S'|/|E|)}+\delta}^\infty (|S'|/|E|)\exp(-(T-\delta)^2/2)T dT + \int_0^{O(\sqrt{d})} \frac{\eps|S'|}{d |E|}T dT\\
& \ll \log(|S'|/|E|) +\delta^2 + \frac{\eps|S'|}{|E|} \;.
\end{align*}
This yields the desired contradiction recalling that the assumption $\|M\|_2 \geq C \eps \log(1/\eps)$ and the definition of $\delta$ imply that
$\delta \geq C' \eps\sqrt{\log(1/\eps)}$ for an appropriately large $C'>0.$
\end{proof}

The following simple claim completes the proof of Proposition~\ref{prop:filter-L2}:
\begin{claim} We have that
$
\Delta(S,S'') \leq \Delta(S,S') - 2\eps/d \;.
$
\end{claim}
\begin{proof}
Recall that $S' = (S\setminus L) \cup E,$ with $E$ and $L$ disjoint multisets such that $L \subset S.$
We can similarly write $S''=(S \setminus L') \cup E',$ with $L'\supseteq L$ and $E'\subset E.$
Since $$\Delta(S,S') - \Delta(S,S'')  = \frac{|E \setminus E'| - |L' \setminus L| }{|S|},$$
it suffices to show that $|E \setminus E'| \geq |L' \setminus L| + 2\eps|S|/d.$
Note that $|L' \setminus L|$ is the number of points rejected by the filter that lie in $S \cap S'.$
By Claim~\ref{cernCor} and Claim~\ref{claim:mean-l2-delta}, it follows that the fraction of elements $x \in S$
that are removed to produce $S''$ (i.e., satisfy $|v^{\ast}\cdot(x-\mu_{S'})|>T+\delta$) is at most $2\exp(-T^2/2) + \eps/d.$
Hence, it holds that $|L' \setminus L| \leq (2\exp(-T^2/2) + \eps/d) |S|.$
On the other hand, Step~\ref{step:bal-large} of the algorithm ensures that the fraction of elements of $S'$ that are rejected
by the filter is at least $8\exp(-T^2/2)+8\eps/d.$ Note that
$|E \setminus E'|$ is the number of points rejected by the filter that lie in $S' \setminus S.$
Therefore, we can write:
\begin{align*}
|E\setminus E'| & \geq (8\exp(-T^2/2)+8\eps/d)|S'| - (2\exp(-T^2/2) + \eps/d) |S| \\
				& \geq (8\exp(-T^2/2)+8\eps/d)|S|/2 - (2\exp(-T^2/2) + \eps/d) |S| \\
				& \geq  (2\exp(-T^2/2) + 3\eps/d) |S| \\
				& \geq |L' \setminus L| + 2\eps|S|/d \;,
\end{align*}
where the second line uses the fact that $|S'| \ge |S|/2$
and the last line uses the fact that $|L' \setminus L| / |S| \leq (2\exp(-T^2/2) + \eps/d).$
This completes the proof of the claim.
\end{proof}
%Therefore, a sample from $\widetilde P$ is coming from $E$ and is rejected by $A$ with probability at least $3\exp(-T^2/2)+4\epsilon/d$.
%From this it is easy to see that $\dtv(\widetilde P', P)\leq \eta^{1/2}\epsilon$. This completes our proof.

\subsection{Agnostically Learning Arbitrary Binary Product Distributions} \label{sec:product}

In this subsection, we build on the approach of the previous subsection to show the following:

\begin{theorem} \label{thm:binary-product-dtv}
Let $P$ be a binary product distribution in $d$ dimensions and $\eps,\tau > 0.$
There is a polynomial time algorithm that, given $\eps$ and
a set of $\Theta(d^6\log(1/\tau)/\eps^3)$ independent samples from $P,$
an $\eps$ fraction of which have been arbitrarily corrupted,
outputs (the mean vector of)
a binary product distribution $\widetilde P$
such that, with probability at least $1-\tau$,
$\dtv(P, \wt P) \leq O(\sqrt{\eps \log(1/\eps)}).$
\end{theorem}

%\new{Prose as to why previous approach fails as is: $\ell_2$-norm between mean vectors no longer a good proxy for total variation distance
%between two products. Need to use a different metric that is polynomially related to the variation distance and we can exploit algorithmically.
%Turns out that $\chi^2$-distance between means works. Need to change our definition of covariance for this to work.}

By Lemma~\ref{lem:prod-dtv-chi2}, the total variation distance
between two binary product distributions can be bounded from above 
by the square root by the $\chi^2$-distance between the corresponding means.
For the case of balanced product distributions, the $\chi^2$-distance 
and the $\ell_2$-distance are within a constant factor of each other.
Unfortunately, this does not hold in general, hence the guarantee of our previous algorithm 
is not sufficient to get a bound on the total variation distance.
Note, however that the $\chi^2$-distance and the $\ell_2$-distance can be related
by rescaling each coordinate by the standard deviation of the corresponding marginal.
When we rescale the covariance matrix in this way, we can use the top eigenvalue
and eigenvector as before, except that we obtain bounds that involve
the $\chi^2$ in place of $\ell_2$-distance.
The concentration bounds we obtain with this rescaling are somewhat weaker,
and as a result, our quantitative guarantees for the general case are correspondingly weaker than in the balanced case.
\new{ As in the filter algorithm for approximating the mean under second moment assumptions in \cite{DiakonikolasKKLMS17}, to handle this weaker concentration, we will choose a threshold at random, weighted towards larger thresholds instead of looking for a violation of a concentration inequality. This gives a filter that rejects more corrupted than uncorrupted samples in expectation and we will show that with high probability we still only throw away an $O(\eps)$ fraction of samples in the course of the algorithm.}

Similarly to the case of balanced product distributions,
we will require a notion of a ``good'' set for our distribution.
For technical reasons, the definition in this setting turns out to be
more complicated. Roughly speaking, this is to allow us to ignore coordinates
for which the small fraction of errors is sufficient to drastically change the sample mean.

\begin{definition}[good set of samples]
Let $P$ be a binary product distribution and $\eps,\eta>0.$
We say that a multiset $S$ of elements in $\{0, 1\}^d$ is
{\em $(\eps,\eta)$-good with respect to $P$} if
for every affine function $L: \{0, 1\}^d \to \R$
and every subset of coordinates $T \subseteq [d]$ satisfying $\sum_{i\in T} p_i(1-p_i)<\eta$
the following holds: Letting $S_T$ be the subset of points in $S$ that have their $i^{th}$ coordinate
equal to the most common value under $P$ for all $i\in T,$
and letting $P_T$ be the conditional distribution of $P$ under this condition, then
$$
|\Pr_{X\in_u S_T}(L(X) \geq 0) - \Pr_{X\sim P_T}(L(X) \geq 0)| \leq \eps^{3/2}/d^2 \;.
$$
\end{definition}

We note that a sufficiently large set of samples from $P$ will satisfy the above properties
with high probability:
\begin{lemma} \label{lem:random-good-dtv}
If $S$ is obtained by taking $\Omega(d^6 \log(1/\tau)/\eps^3)$ independent samples from $P,$
 it is $(\eps,1/5)$-good with respect to $P$ with probability at least $9/10.$
\end{lemma}
The proof of this lemma is deferred to Section \ref{sec:filterProductAppendix}.

%We recall our definition of $\Delta(S,S')$:
%\begin{definition}
%Given multisets $S$ and $S',$
%we let $\Delta(S,S')$ be the size of the symmetric difference of $S$ and $S'$ divided by the size of $S.$
%\end{definition}

We will also require a notion of the number of coordinates on which $S$ non-trivially depends:
\begin{definition}
For $S$ a multiset of elements in $\{0, 1\}^d,$
let $\supp(S)$ be the subset of $[d]$ consisting
of indices $i$ such that the $i^{th}$ coordinate of elements of $S$ is not constant.
\end{definition}

Similarly to the balanced case,
our algorithm is obtained by repeated application of an efficient filter procedure, whose precise guarantee is described below.
\begin{proposition} \label{prop:filter-dtv}
Let $P$ be a binary product distribution in $d$ dimensions and $\eps>0.$
Suppose that $S$ is an $(\epsilon,\eta)$-good multiset with respect to $P$ with $\eta > 10\eps$
and $S'$ be any multiset with $\Delta(S,S') \leq 20\eps.$
There exists a polynomial time algorithm which, given $\eps$ and $S'$, returns one of the following:
\begin{itemize}
\item[(i)] The mean vector of a product distribution $P'$ with $\dtv(P,P')=O(\sqrt{\eps\log(1/\eps)}).$
\item[(ii)] A multiset $S'' \subset S'$ of elements of $\{0, 1\}^d$ such that there exists a product distribution
$\widetilde P$ with mean $\wt p$ and a multiset $\widetilde S$
that is $(\eps, \eta-\|p-\wt p\|_1)$-good with respect to $\widetilde P$ such that
$$
\E[\Delta(\widetilde S,S'')]+ \|p-\wt p\|_1/6  \leq \Delta(S,S') \;.
$$
\end{itemize}
\end{proposition}

Our agnostic learning algorithm is then obtained by iterating this procedure.
\new{We can prove Theorem~\ref{thm:binary-product-dtv} given Proposition~\ref{prop:filter-dtv}.

\begin{proof}[Proof of Theorem~\ref{thm:binary-product-dtv}]
We draw $N = \Theta(d^6/\eps^3)$ samples forming a set $S$,
which is $(\eps,1/5)$-good with probability $9/10$ by Lemma \ref{lem:random-good-dtv}.
We condition on this event. The adversary corrupts an $\eps$-fraction of $S$
producing a set $S'$ with $\Delta(S,S') \leq 2 \eps.$
The iterations of the algorithm produce a sequence of sets $S_0=S, S_1, \ldots,S_k,$
where $S_i$ is $(\eps,\eta_i)$-good for some binary product distribution $P_i$
and some sets $S_i'.$ We note that $\Delta(S_i, S_i')$
is monotonically decreasing in expectation.
Since $|\mu^{P_i} - \mu^{P_{i+1}}| \leq \dtv(P_i,P_{i+1})$,
in the $i^{th}$ iteration, we have that $\E[\Delta(S_{i+1},S'_{i+1}) - \dtv(P_i,P_{i+1})] 
\leq \Delta(S_i, S'_i)$, as long as $\Delta(S_i,S'_i) \leq 20\eps.$

We need to show that the probability that we ever have $\Delta(S_i,S'_i) > 20\eps$ is small. 
Indeed we show that the probability that  
$\Delta(S_i, S'_i) + \sum_{j=0}^{i-1} \dtv(P_i, P_{i+1})$ is ever large is $1/10$.

We analyze the following procedure: We iteratively run \textsc{Filter-Product}. 
We stop if we output an approximation to the mean or if $\Delta(S_i,S'_i) + \sum_{j=0}^{i-1} \dtv(P_i,P_{i+1}) > 20 \eps |S|$. 
Proposition\ref{prop:filter-dtv} gives that 
$\E[\Delta(S_{i+1},S'_{i+1}) - \dtv(P_i,P_{i+1})/6] \leq \Delta(S_i.S'_i)$.
This expectation is conditioned on the state of the algorithm after previous iterations, which is determined by $S'_i$.
Thus, if we consider the random variables $X_i=\Delta(S_i,S'_i) + \sum_{j=0}^{i-1} \dtv(P_i,P_{i+1})/6$, 
then we have $\E[X_{i+1} | S'_i] \leq X_i$, 
i.e., the sequence $X_i$ is a sub-martingale with respect to $S'_i$. 
Using the convention that $S'_{i+1}=S'_i$, if we stop in less than $i$ iterations, 
since we must terminate $N$ iterations as every iteration removes at least one sample, 
the algorithm fails if and only if $|X_N| > 20 \eps$. 
By a simple induction or standard results on sub-martingales, we have 
$\E[X_N] \leq X_0$. Now $X_0 = \Delta(S_0,S'_0) \leq 2 \eps |S'_0|$. 
Thus, $\E[X_N] \leq 2\eps|S|$. By Markov's inequality, except with probability $1/10$, 
we have $X_N \leq 20 \eps |S|$.  
Therefore, the probability that we ever have $|X_i| > 20 \eps$ is at most $1/10$.
 
By a union bound, using Lemma \ref{lem:random-good-dtv}, 
$S_0$ is $(\eps,1/5)$-good and we have $|X_i| \leq 20 \eps$ with probability at least $4/5$. 
We assume that this holds. By induction, $S_i$ is $(\epsilon,1/5 -  \sum_{j=0}^{i-1} \dtv(P_i,P_{i+1}))$-good, 
and so is $(\epsilon, 1/5 - 100 \eps)$-good, which suffices since $1/5 - 100 \eps \geq 10 \eps$.

When it terminates, the algorithm outputs a product distribution $P'$ with 
$\dtv(P_k,P')=O(\sqrt{\eps\log(1/\eps)}).$ By the triangle inequality, we have that 
$$\dtv(P, P') \leq  \dtv(P_k, P') +  \sum_{j=0}^{k-1} \dtv(P_i, P_{i+1}))  \leq O(\sqrt{\eps\log(1/\eps)}) + 100 \eps \leq O(\sqrt{\eps\log(1/\eps)}) \;.$$
When $\tau \leq 1/5$, we will need to draw fresh $\eps$-corrupted samples 
and repeat this procedure $O(\log(1/\tau))$ times, 
and then one of the resulting output distributions is within total variation distance 
$O(\sqrt{\eps\log(1/\eps)})$ with probability at least $1-\tau/2$. 
Then we use the agnostic hypothesis selection procedure of Lemma \ref{tournamentLem}.
\end{proof}
}

\subsubsection{Algorithm \textsc{Filter-Product}:  Proof of  Proposition~\ref{prop:filter-dtv}}
In this section, we describe and analyze the efficient routine establishing
Proposition~\ref{prop:filter-dtv}.
%\new{Prose: Similar to balanced product. Explain how we get rid of biased coordinates and why we rescale.}
Our efficient filtering procedure is presented in detailed pseudocode below.

\bigskip
\begin{algorithm}[H]%[htb]
\begin{algorithmic}[1]
\Procedure{Filter-Product}{$\epsilon, S'$}
\INPUT $\eps>0$ and multiset $S'$ such that there exists an $\eps$-good $S$ with $\Delta(S, S') \le 2\eps$
\OUTPUT  Multiset $S''$ or mean vector $p'$ satisfying Proposition~\ref{prop:filter-dtv}
\State Compute the sample mean $\mu^{S'}=\E_{X\in_u S'}[X]$ and the sample covariance matrix $M$
\State i.e., $M = (M_{i, j})_{1 \le i, j \le d}$ with $M_{i,j} = \E_{X \in S'} [(X_i-\mu^{S'}_i) (X_j-\mu^{S'}_j)].$

\If {there exists $i \in [d]$ with $0<\mu^{S'}_i < \eps/d$ or $0<1 - \mu^{S'}_i < \eps/d,$}
\State let $S''$ be the subset of elements of $S'$  in which those coordinates take their most common value.
\State \textbf{return} $S''.$ \label{step:prod-biased}
\EndIf

%\Comment{comments don't work here}
%\Comment{ For the later steps we ignore any coordinates not in $\supp(S')$.}
{\em /*  For the later steps, we ignore any coordinates not in $\supp(S').$ */}

\State Compute approximations for the largest magnitude eigenvalue $\lambda'$ of $D M D,$
$\lambda' : = \| DMD\|_2,$ where $D= \mathrm{Diag}(1/\sqrt{\mu^{S'}_i(1-\mu^{S'}_i)}),$
and the associated unit eigenvector $v'.$

\If {$\|DMD\|_2< O(\log(1/\eps)),$} \textbf{return} $\mu^{S'}$ (re-inserting all coordinates affected by Step~\ref{step:prod-biased}). \label{step:prod-small}
\EndIf

\State  \label{step:prod-large}  Draw $Z$ from the distribution on $[0,1]$ with probability density function $2x$. 
\State Let $T= Z \max\{ |v^{\ast} \cdot (x - \mu^{S'})| : x \in S' \}$
	%\State Return the set $S'=\{x \in S :|v^{\ast} \cdot (X -\mu^S)| < T \}$.
%Let $\delta := 3 \sqrt{ \eps \|DMD\|_2}.$
%Find $T>0$ such that
%$$\Pr_{X\in_u S'}(|v^{\ast} \cdot (X-\mu^{S'})| > T+\delta) > 20/T^2+4\eps^{3/2}/d^2 \;,$$
where $v^{\ast} := Dv'.$
\State \textbf{return} the multiset $S'' = \{ x\in S':  |v^{\ast} \cdot (x-\mu^{S'}) | < T \} \;.$

\EndProcedure
\end{algorithmic}
\caption{Filter algorithm for an arbitrary binary product distribution}
\label{alg:arbritraryproduct}
\end{algorithm}

%\begin{comment}
%\fbox{\parbox{6.1in}{
%{\bf Algorithm} {\tt Filter-Product}\\
%{\em Input:} $\eps>0$ and multiset $S'$ such that there exists an $\eps$-good $S$ with $\Delta(S, S') \le 2\eps$ \\
%{\em Output:}  Multiset $S''$ or mean vector $p'$ satisfying Proposition~\ref{prop:filter-dtv}

%\vspace{0.2cm}

%\begin{enumerate}
%\item Compute the sample mean $\mu^{S'}=\E_{X\in_u S'}[X]$ and the sample covariance
%i.e., the matrix $M = (M_{i, j})_{1 \le i, j \le n}$ with $M_{i,j} = \E_{X \in S'} [(X_i-\mu^{S'}_i) (X_j-\mu^{S'}_j)].$

%\item \label{step:prod-biased}
%If there exists $i \in [d]$ with $0<\mu^{S'}_i < \eps/d$ or $0<1 - \mu^{S'}_i < \eps/d,$
%let $S''$ be the subset of elements of $S'$
%in which those coordinates take their most common value.
%Return $S''.$

%\vspace{0.15cm}

%{\em /*  For the later steps, we ignore any coordinates not in $\supp(S').$ */}

%\vspace{0.15cm}

%\item Compute approximations for the largest magnitude eigenvalue $\lambda'$ of $D M D,$
%$\lambda' : = \| DMD\|_2,$ where $D= \mathrm{Diag}(1/\sqrt{\mu^{S'}_i(1-\mu^{S'}_i)}),$
%and the associated unit eigenvector $v'.$

%\item  \label{step:prod-small}
%If $\|DMD\|_2< O(\log(1/\eps)),$ return $\mu^{S'}$ (re-inserting all coordinates affected by Step~\ref{step:prod-biased}).

%\item  \label{step:prod-large}
%Let $\delta := 3 \sqrt{ \eps \|DMD\|_2}.$
%Let $T>0$ such that
%$$\Pr_{X\in_u S'}(|v^{\ast} \cdot (X-\mu^{S'})| > T+\delta) > 20/T^2+4\eps^{3/2}/d^2 \;,$$
%where $v^{\ast} := Dv'.$
%Return the multiset $S'' = \{ x\in S':  |v^{\ast} \cdot (x-\mu^{S'}) | \leq T+\delta \} \;.$
%\end{enumerate}

%}}

%\bigskip
%\end{comment}
This completes the description of the algorithm. We now proceed to prove correctness.

\subsubsection{Chi-Squared Distance and Basic Reductions} \label{ssec:dtv-reductions}

As previously mentioned, our algorithm will use the $\chi^2$-distance between the mean vectors
as a proxy for the total variation distance between two binary product distributions.
Since the mean vector of the target distribution is not known to us, 
we will not be able to use the symmetric definition of the $\chi^2$-distance used in Lemma \ref{lem:prod-dtv-chi2}
We will instead require the following asymmetric version of the $\chi^2$-distance:

\begin{definition}
The $\chi^2$-distance of $x, y \in \R^d$ is defined by $\chi^2(x, y) \eqdef \sum_{i=1}^d \frac{(x_i - y_i)^2}{x_i(1-x_i)}.$
\end{definition}

The following fact follows directly from Lemma \ref{lem:prod-dtv-chi2}.

\begin{fact} \label{cor:dtv-chi-squared}
Let $P, Q$ be binary product distributions with mean vectors $p, q$ respectively.
%If $q=O(p)$ and $1-q=O(1-p)$, 
Then,  $\dtv(P, Q)  = O(\sqrt{\chi^2(p, q)}).$
\end{fact}

\new{There are two problems with using the $\chi^2$ 
distance between the mean vectors as a proxy for the total variation distance. 
The first is that the $\chi^2$-distance between the means 
is a very loose approximation of the total variational distance
when the means are close to $0$ or $1$ in some coordinate. 
To circumvent this obstacle, we remove such coordinates via an appropriate pre-processing 
in Step~\ref{step:prod-biased}. 
The second is that the above asymmetric notion of the $\chi^2$-distance may be quite far from the
symmetric definition. To overcome this issue, it suffices to have that $q_i=O(p_i)$ and $1-q_i=O(1-p_i).$
To ensure this condition is satisfied, we appropriately modify 
the target product distribution (that we aim to be close to). 
Next, we will show how we deal with these problems in detail.}

\medskip

Before we embark on a proof of the correctness of algorithm \textsc{Filter-Product},
we will make a few reductions that we will apply throughout.
First, we note that if some coordinate in Step~\ref{step:prod-biased} exists, then removing
the uncommon values of that coordinate increases $\Delta(S,S')$ by at most $\eps/d$
but decreases $|\supp(S')|$ by at least $1.$
We also note that, if $N$ is the set of coordinates outside of the support of $S'$,
the probability that an element in $S'$ has a coordinate in $N$
that does not take its constant value is $0.$
Note that this is at most $O(\eps)$ away from the probability
that an element taken from $P$ has this property,
and thus we can assume that $\sum_{i\in N} \min\{p_i,1-p_i\} = O(\eps).$
Therefore, after Step~\ref{step:prod-biased}, we can assume that  all coordinates $i$
have $\eps/d \leq p_i \leq 1-\eps/d.$

The next reduction will be slightly more complicated and depends on the following idea:
Suppose that there is a new product distribution $\wt{P}$ with mean $\wt{p}$
and an $(\eps,\eta-\|p-\wt p\|_1)$-good multiset $\wt{S}$ for $\wt{P}$ such that
$$
\Delta(\wt{S},S')+\|p-\wt p\|_1/5 \leq \Delta(S,S').
$$
Then, it suffices to show that our algorithm works
for $\wt{P}$ and $\wt{S}$ instead of $P$ and $S$
(note that the input to the algorithm, $S'$ and $\eps$ in the same in either case).
This is because the conditions imposed by the output in this case would be strictly stronger.
In particular, we may assume that $\mu^{S'}_i \geq p_i/3$ for all $i$:

\begin{lemma}
There is a product distribution $\wt{P}$ whose mean vector $\wt{p}$
satisfies $\mu^{S'}_i \geq \wt p_i/3$ and $1-\mu^{S'}_i \geq (1-\wt p_i/3)$ for all $i$,
and a set $\wt{S} \subseteq S$ that is $(\eps, \eta-\|p-\wt p\|_1)$-good for $\wt P$ and satisfies
$$\Delta(\wt{S},S')+\|p-\wt p\|_1/5 \leq \Delta(S,S').$$
\end{lemma}
\begin{proof}
If all coordinates $i$ have $\mu^{S'}_i \geq p_i/3$
and $1-\mu^{S'}_i \geq (1- p_i/3),$
then we can take $\wt{P}= P$ and $\wt{S}=S.$

Suppose that the $i^{th}$ coordinate has $\mu^{S'}_i < p_i/3.$
Let $\wt{P}$ be the product whose mean vector $\wt p$
has $\wt p_i=0$ and $\wt p_j = p_j$ for $j \neq i.$
Let $\wt{S}$ be obtained by removing from $S$
all of the entries with $1$ in the $i^{th}$-coordinate.
Then, we claim that $\wt{S}$ is $(\eps, \eta-p_i)$-good
for $\wt P$ and has $\Delta(\wt{S},S')+p_i/5 \leq \Delta(S,S').$
Note that here we have $\|p-\wt p\|_1 = p_i.$

First, we show that $\wt{S}$ is $(\eps, \eta-p_i)$-good for $\wt P.$
For any affine function $L(x)$ and set $T \subseteq [d]$ with $\sum_{j \in T} \wt p_j(1- \wt p_j) \leq \eta - p_i,$
we need to show that
$$
|\Pr_{X\in_u \wt S_T}(L(X)>0) - \Pr_{X\sim  \wt P_T}(L(X)>0)| \leq \eps^{3/2}/d^2 \;.
$$
Let $\wt T= T \cup \{ i \}.$
We may or may not have $i \in T$ but,
from the definition of $\wt p$,
$$\sum_{j \in T} \wt p_j(1- \wt p_j) = \sum_{j \in T \setminus \{i\}} \wt p_j(1- \wt p_j) = \sum_{j \in T \setminus \{i\}} p_j(1-p_j).$$
Thus,
$$\sum_{j \in \wt T}  p_j(1-p_j) = p_i(1-p_i) + \sum_{j \in T} \wt p_j(1- \wt p_j) \leq \eta - p_i + p_i(1-p_i) \leq \eta.$$
Since $S$ is good for $P$, we have that
$$|\Pr_{X\in_u S_{\wt T}}(L(X) \geq 0) - \Pr_{X \sim  P_{\wt T}}(L(X) \geq 0)| \leq \eps^{3/2}/d^2 \;.$$
Moreover, note that $S_{\wt T} = \wt S_T$ and $P_{\wt T}=\wt P_T.$
Thus,  $\wt{S}$ is $(\eps, \eta-p_i)$-good for $\wt P.$

Next, we show that $\Delta(\wt{S},S')+p_i/5 \leq \Delta(S,S').$
We write $S = \wt{S} \setminus \wt{L} \cup \wt E.$
We write $S_1, L_1, S'_1$ for the subset of $S, L, S'$
respectively,
where the $i^{th}$ coordinate is $1.$
Since $S$ is $(\eps,\eta)$-good for $P,$
we have that $|\mu^S_i -p_i| \leq  \eps^{3/2}/d^2.$
Recall that we are already assuming that $\wt p_i \geq \eps/d.$
Thus, $\mu^S_i \geq 29p_i/30.$
Therefore, we have that $|S_1| \geq 29 p_i |S|/30.$
On the other hand, we have that
$|S'_1| \leq p_i|S'|/3 \leq 11p_i|S|/30.$
Thus, $|L_1| = |S_1 \setminus S'_1|  \geq 18p_i |S| /30.$
This means that $p_i=O(\Delta(\wt{S},S')) = O(\eps).$
However, $\wt{E} = E \cup S'_1$ and $\wt{L} = L \setminus L_1$.
This gives
\begin{align*}
\Delta(\wt{S},S') & = \frac{|\wt{E}| + |\wt{L}|}{|\wt{S}|} \leq \frac{|E| + |S'_1| + |L| - |L_1|}{|\wt{S}|} \\
&\leq \frac{|E| + |L| - 7p_i/30}{|\wt{S}|}  = \frac{|E| + |L| - 7p_i|S|/30}{|S|(1-\mu^S_i)} \\
 &\leq \frac{\Delta(S,S') - 7p_i/30}{1-31p_i/30} = \Delta(S,S') - 7p_i/30 + O(\eps p_i) \\
 &\leq  \Delta(S,S') - p_i/5 \;.
 \end{align*}
Similarly, suppose that instead the $i^{th}$-coordinate
has $1-\mu^{S'}_i < (1-p_i)/3.$
Let $\wt{P}$ be the product whose mean vector
$\wt p$ has $\wt p_i=1$ and $\wt p_j = p_j$ for $j \neq i.$
Let $\wt{S}$ be obtained by removing from $S$
all of the entries with $0$ in the $i^{th}$-coordinate.
Then, by a similar proof we have that
$\wt{S}$ is $(\eps, \eta-(1-p_i))$-good for $\wt P$
and has $\Delta(\wt{S},S')+(1-p_i)/5 \leq \Delta(S,S').$
Note that here we have $\|p-\wt p\|_1 = 1 - p_i.$

By an easy induction, we can set all coordinates $i$
with $\mu^{S'}_i \geq \wt p_i/3$ and $1-\mu^{S'}_i \geq (1-\wt p_i/3)$
to $0$ or $1$ respectively, giving an $\wt S$ and $\wt P$
such that $\wt S$ is $(\eps, \eta - \|p-\wt p\|_1)$-good
for $\wt P$ and 
$$\Delta(\wt{S},S')+\|p-\wt p\|_1/5 \leq \Delta(S,S') \;,$$
as desired.
\end{proof}

\noindent In conclusion, throughout the rest of the proof we may and will assume that for all $i,$
\begin{itemize}
\item  $\eps/d \leq  \mu^{S'}_i \leq 1-\eps/d.$
\item $\mu^{S'}_i \geq p_i/3$ and $1-\mu^{S'}_i \geq (1-p_i)/3.$
\end{itemize}

\subsubsection{Setup and Basic Structural Lemmas} \label{ssec:dtv-setup}
As in the balanced case, we can write $S' = (S \setminus  L) \cup E$
for disjoint multisets $L$ and $E.$
Similarly, we define the following matrices:
\begin{itemize}
\item $M_P$ to be the matrix with $(i, j)$-entry $\E_{X\sim P}[(X_i - \mu^{S'}_i)(X_j - \mu^{S'}_j)],$

\item $M_S$ to be the matrix with $(i, j)$-entry $\E_{X\in_u S}[(X_i - \mu^{S'}_i)(X_j - \mu^{S'}_j)],$

\item $M_E$ to be the matrix with $(i, j)$-entry $\E_{X\in_u E}[(X_i - \mu^{S'}_i)(X_j - \mu^{S'}_j)],$ and

\item $M_L$ to be the matrix with $(i, j)$-entry $\E_{X\in_u L}[(X_i - \mu^{S'}_i)(X_j - \mu^{S'}_j)].$
\end{itemize}

\noindent Note that we no longer zero-out the diagonals of $M_P$ and $M_S$.
This will turn out to allow us to more naturally relate spectral properties of these matrices to the $\chi^2$-distance between the means.
We start with the following simple claim:

\begin{claim} \label{claim:Chebyshev}
 For any $v \in \R^d$ satisfying $\sum_{i=1}^d v_i^2 \mu^{S'}_i(1-\mu^{S'}_i) \leq 1,$
 the following statements hold:
  \begin{itemize}
 \item[(i)] $\Var_{X \sim P}[v \cdot X] \leq 9$ and
 $|v \cdot (p - \mu^{S'})| \leq \sqrt{\chi^2(\mu^{S'}, p)},$ and
 \item[(ii)] $ \Pr_{X \sim P}\left(|v \cdot X-\mu^{S'}| \geq T +  \sqrt{\chi^2(\mu^{S'}, p)}\right) \leq 9/T^2 \; .$
 \end{itemize}
\end{claim}
\begin{proof}
Recall that $p$ denotes the mean vector of the binary product $P.$
To show (i), we use the fact that $X_i \sim \mathrm{Ber}(p_i)$ and the $X_i$'s are independent.
This implies that
 $$\Var_{X \sim P}\left[\sum_{i=1}^d v_i X_i \right] = \sum_{i=1}^d v_i^2  \Var[X_i] =
 \sum_{i=1}^d v_i^2 p_i (1-p_i)  \leq  9 \sum_{i=1}^d v_i^2 \mu^{S'}_i (1-\mu^{S'}_i)\leq 9 \;,$$
 where we used that $p_i \leq 3 \mu^{S'}_i,$ $(1-p_i) \leq 3 (1-\mu^{S'}_i)$
 and the assumption in the claim statement.  For the second part of (i), note that
 $$|v \cdot (p - \mu^{S'})| = \left| \sum_{i=1}^d v_i \sqrt{\mu^{S'}_i(1-\mu^{S'}_i)} \cdot \frac{p_i - \mu^{S'}_i}{\sqrt{\mu^{S'}_i(1-\mu^{S'}_i)}} \right|
 \leq \sqrt{ \sum_{i=1}^d v_i^2 \mu^{S'}_i(1-\mu^{S'}_i)} \cdot \sqrt{\chi^2(\mu^{S'}, p)} \leq \sqrt{\chi^2(\mu^{S'}, p)}  \;,$$
 where the first inequality  is Cauchy-Schwarz, and the second follows from the assumption in the claim statement that
 $\sum_{i=1}^d v_i^2 \mu^{S'}_i(1-\mu^{S'}_i) \leq 1.$ This proves (i).

To prove (ii), we note that Chebyshev's inequality gives
$$\Pr_{X \sim P}(|v \cdot (X-p)| \geq T ) \leq \Var_{X \sim P} [v \cdot X]/T^2 \leq 9/T^2 \;,$$
where the second inequality follows from (i).
To complete the proof note the inequality
$$|v \cdot (X - \mu^{S'})| \geq T + \sqrt{\chi^2(\mu^{S'}, p)}$$
implies that
$$|v \cdot (X - p)| \geq |v \cdot (X - \mu^{S'})| - |v \cdot (p - \mu^{S'})| \geq T \;,$$
where we used the triangle inequality and the second part of (i).
\end{proof}

Let $\Cov[S]$ denote the sample covariance matrix with respect to $S,$
and $\Cov[P]$ denote the covariance matrix of $P.$
We will need the following lemma:

\begin{lemma} \label{lem:good-moments}
We have the following:
\begin{itemize}
\item[(i)] $\left|\sqrt{\chi^2(\mu^{S'}, \mu^S)} - \sqrt{\chi^2(\mu^{S'}, p)} \right| \leq O(\eps/d),$ and
\item[(ii)] $\|D\left(\Cov[S] -\Cov[P] \right)D\|_2 \leq O(\sqrt{\eps}) \;.$
\end{itemize}
\end{lemma}
\begin{proof}
For (i): Since $S$ is good, for any $i \in [d],$ we have
$$|\mu^{S}_i-p_i| = \left| \Pr_{X \in_u S} (e_i \cdot X \geq 1)- \Pr_{X \sim P} (e_i \cdot X \geq 1) \right| \leq \eps^{3/2}/d^2 \;.$$
Therefore, by the triangle inequality we get
$$\left|\sqrt{\chi^2(\mu^{S'}, \mu^S)}-  \sqrt{\chi^2(\mu^{S'}, p)} \right|
\leq \sqrt{\sum_{i=1}^d \frac{(\mu^S_i - p_i)^2}{\mu^{S'}_i(1-\mu^{S'}_i)}} \leq \sqrt{ \frac{d \cdot (\eps^{3}/d^4)} {\eps/(2d)}} \leq O(\eps/d) \;,$$
where the second inequality uses the fact that  $\mu^{S'}_i(1-\mu^{S'}_i) \geq \eps/(2d).$

For (ii): Since $S$ is good, for any $i, j \in [d],$ we have
$$\left| \E_{X \in_u S}[X_iX_j- p_ip_j] \right| =
\left| \Pr_{X \in_u S} [(e_i+e_j) \cdot X \geq 1] - \Pr_{X \sim P} [(e_i+e_j) \cdot X \geq 1] \right|
\leq \eps^{3/2}/d^2 \;.$$
Combined with the bound $|\mu^{S}_i-p_i| \leq \eps^{3/2}/d^2$ above,
this gives
$$|\Cov[S]_{i,j} - \Cov[P]_{i,j}| \leq O(\eps^{3/2}/d^2) \;.$$
We thus obtain
$$\|\Cov[S] -\Cov[P] \|_2 \leq \|(\Cov[S] -\Cov[P])\|_F \leq O(\eps^{3/2}/d) \;.$$
Note that $\|D\|_2 = \max_i \left( 1/\sqrt{\mu^{S'}_i(1-\mu^{S'}_i)} \right) \leq \sqrt{2d/\eps}.$
Therefore,
$$\|D \left(\Cov[S] -\Cov[P] \right)D\|_2 \leq O(\sqrt{\eps}) \;.$$
\end{proof}

Combining Claim~\ref{claim:Chebyshev} and Lemma~\ref{lem:good-moments}
we obtain:

\begin{corollary} \label{cor:Chebyshev}
 For any $v \in \R^d$ with $\sum_{i=1}^d v_i^2 \mu^{S'}_i(1-\mu^{S'}_i) \leq 1,$ we have:
 \begin{itemize}
 \item[(i)] $\Var_{X \in_u S}[v \cdot X] \leq 10$ and $|v \cdot (\mu^{S} - \mu^{S'})| \leq \sqrt{\chi^2(\mu^{S'}, p)} + O(\eps/d),$
 and
 \item[(ii)] $\Pr_{X \in_u S}\left( |v \cdot X-\mu^{S'}| \geq T +  \sqrt{\chi^2(\mu^{S'}, p)} \right) \leq 9/T^2 +\eps^{3/2}/d^2 \;.$
 \end{itemize}
\end{corollary}
\begin{proof}
We have that
\begin{align*}
|\Var_{X \in_u S} [v \cdot X] - \Var_{Y \sim P}[v \cdot Y]|
&= v^T \left( \Cov[S] -\Cov[P] \right) v  \\
& \leq  \|D^{-1}v\|_2^2 \cdot \|D \left( \Cov[S] -\Cov[P] \right) D\|_2 \\
& \leq O(\sqrt{\eps}) \\
& \leq 1 \;,
\end{align*}
where the second line uses Lemma~\ref{lem:good-moments} (ii), and the assumption
$\|D^{-1}v\|_2^2 = \sum_{i=1}^d v_i^2 \mu^{S'}_i(1-\mu^{S'}_i) \leq 1,$
and the third line holds for small enough $\eps.$
Thus, using Claim~\ref{claim:Chebyshev} (i), we get that
$$\Var_{X \in_u S}[v \cdot X] \leq \Var_{Y \sim P}[v \cdot Y]+1 \leq 10 \;.$$
By the Cauchy-Schwarz inequality and Lemma~\ref{lem:good-moments},  we get
$$|v \cdot (\mu^{S} - \mu^{S'})| \leq \sqrt{\chi^2(\mu^{S'}, \mu^S)}  \leq \sqrt{\chi^2(\mu^{S'}, p)} + O(\eps/d) \;.$$
This proves (i).

Part (ii) follows directly from Claim \ref{claim:Chebyshev} (ii) and the assumption that $S$ is good for $P.$
\end{proof}

\new{
\begin{lemma} \label{lem:M-S-close}
We have that $\|D (M_S - M_P) D \|_2 \leq O(\sqrt{\eps})$.
\end{lemma}
\begin{proof}
We can show that $|(M_S)_{i,j} - (M_P)_{i,j} | \leq O(\epsilon^{3/2}/d^2)$ 
for all $i, j \in [d]$, by expanding the LHS in terms of the differences of linear threshold functions 
on $S$ and $P$ in the same way as the proof of Lemma~\ref{lem:good-moments}. 
Thus, 
$$\|M_S - M_P\|_2^2 \leq \|M_S - M_P\|_F^2 \leq \sum_{i,j} |(M_S)_{i,j} - (M_P)_{i,j} |^2 \leq O(\eps^3/d^2) \;.$$
Finally note that $\|D\|_2 = \max_i \left( 1/\sqrt{\mu^{S'}_i(1-\mu^{S'}_i)} \right) \leq \sqrt{2d/\eps}$, 
and so 
$$\|D (M_S - M_P) D \|_2 \leq \|D\|_2^2  \|M_S - M_P\|_2 \leq 2d/\eps \cdot O(\eps^{3/2}/d) = O(\sqrt{\eps}) \;.$$ 
\end{proof}

\noindent Combining the above we obtain:

\begin{corollary}
We have that $\left\| D (|S'|M-|S|M_P - |E| M_E + |L| M_L) D \right\|_{2}  =  O(|S'| \cdot \sqrt{\eps}) \;.$
\end{corollary}
\begin{proof}
This follows from Lemma \ref{lem:M-S-close} combined with the fact that
$|S'|M = |S|M_S +|E|M_E-|L|M_L$ and the observation $|S| \leq |S'|/(1-2\eps) \leq 2 |S'|.$
\end{proof}
}

\noindent We have the following lemma:

\begin{lemma} \label{lem:M-P-small}
We have that $\|D M_{P} D \|_2 \leq 9 + \chi^2(\mu^{S'}, p).$
\end{lemma}
\begin{proof}
Note that $M_{P} = (\mu^{S'}-p)(\mu^{S'}-p)^T + \mathrm{Diag}(p_i(1-p_i)).$
For any $v'$ with $\|v'_2\| \leq 1,$ the vector $v=Dv'$ satisfies
$\sum_{i=1}^d v^2_i \mu^{S'}_i (1-\mu^{S'}_i) \leq 1.$
Therefore, we can write
$$v'^T D M_{P} D v' =  v^T M_{P} v  = (v \cdot (\mu^{S'}-p))^2 + v^T \mathrm{Diag}(p_i(1-p_i)) v \;.$$
 Using Claim \ref{claim:Chebyshev} (i), we get
 $$(v \cdot (\mu^{S'}-p))^2 \leq \chi^2(\mu^{S'}, p)$$ and
 $$|v^T \mathrm{Diag}(p_i(1-p_i)) v| = |\Var_{X \sim P}(v \cdot (X-p))| \leq 9 \;.$$
 This completes the proof.
 \end{proof}

\noindent The following crucial lemma bounds from above the contribution to the error from $L$:

\begin{lemma}\label{MLBoundLem-2}
The spectral norm $\|D M_L D\|_2  =  O(|S'|/|L| + \chi^2(\mu^{S'}, p)).$
\end{lemma}
\begin{proof}
Similarly, we need to bound from above the quantity $|v'^T D M_L D v'|$ for all $v' \in \R^d$ with $\|v'\|_2 \leq 1$.
Note that $|v'^T D M_L D v'| = |v^T M_L v| = \E_{X\in_u L}[|v \cdot (X-\mu^{S'})|^2],$
where the vector $v=Dv'$ satisfies $\sum_{i=1}^d v^2_i \mu^{S'}_i (1-\mu^{S'}_i) \leq 1.$
The latter expectation is bounded from above as follows:
\begin{align*}
\E_{X\in_u L}[(v \cdot (X-\mu^{S'}))^2] & \leq 2\E_{X \in_u L}[(v \cdot (X-p))^2] + 2(v \cdot (\mu^{S'} - p))^2 \\
& \leq   2\E_{X\in_u L}[(v \cdot (X-p))^2] + 2\chi^2(\mu^{S'}, p)\\
& \leq (2|S|/|L|) \cdot  \E_{X\in_u S}[(v \cdot (X-p))^2] + 2\chi^2(\mu^{S'}, p) \\
& \leq 20|S|/|L| + 2\chi^2(\mu^{S'}, p) \\
&\leq 21|S'|/|L| + 2\chi^2(\mu^{S'}, p)  \;,
\end{align*}
where the first line uses the triangle inequality,
the second line uses Claim \ref{claim:Chebyshev} (i),
the third line follows from the fact that $L \subseteq S$,
the fourth line uses Corollary~\ref{cor:Chebyshev} (i),
and the last line uses the fact that $\eps$ is small enough.
\end{proof}

The above lemmata and the triangle inequality yield the following:
\begin{corollary}\label{MApproxCor-2}
We have that $\left\| D \left( M -(|E|/|S'|) M_E \right) D  \right\|_2  =  O(1+\chi^2(\mu^{S'}, p)) \;.$
\end{corollary}

We are now ready to analyze the two cases of the algorithm  \textsc{Filter-Product}.

\subsubsection{The Case of Small Spectral Norm} \label{ssec:accurate-mean-dtv}
We start by considering the case where the vector $\mu^{S'}$ is returned.
It suffices to show that in this case $\dtv(P,P') = O(\sqrt{\eps \log(1/\eps)}).$

Let $N$ be the set of coordinates not in $\supp(S').$
We note that only an $\eps$ fraction of the points in $S$
could have that any coordinate in $N$ does not have its most common value.
Therefore, at most a $2\eps$ fraction of samples from $P$ have this property.
Hence, the contribution to the variation distance coming from these coordinates is $O(\eps).$
So, it suffices to consider only the coordinates not in $N$
and show that $\dtv(P_{\overline{N}},P'_{\overline{N}})= O(\sqrt{\eps \log(1/\eps)}).$
Thus, we may assume for the sake the analysis below that $N=\emptyset.$

We begin by bounding various $\chi^2$-distances
by the spectral norm of appropriate matrices.

\begin{lemma} \label{lem:mel-2}
Let $\mu^E,  \mu^L$ be the mean vector of $E$ and $L$, respectively.
Then, $\chi^2(\mu^{S'}, \mu^E) \leq \|D M_E D\|_2$ and
$\chi^2(\mu^{S'},\mu^L) \le \| D M_L D\|_2.$
\end{lemma}
\begin{proof}
We prove the first inequality, the proof of the second being very similar.

Note that for any vector $v,$
$v^T M_E v = \E_{X\in_u E}[|v\cdot (X-\mu^{S'})|^2]\geq |v\cdot (\mu^E-\mu^{S'})|^2.$
Let $v \in \R^d$ be the vector defined by
$$v_i=\frac{\mu^E_i-\mu^{S'}_i}{\mu^{S'}_i(1-\mu^{S'}_i)\sqrt{\chi^2(\mu^{S'}, \mu^E)}} \;.$$
We have that
$$\|D^{-1}v\|_2^2 = \sum_{i=1}^d v_i^2 \mu^{S'}_i (1-\mu^{S'}_i)
= \frac{1}{\chi^2(\mu^{S'}, \mu^E)} \sum_{i=1}^d \frac{(\mu^E_i-\mu^{S'}_i)^2}{\mu^{S'}_i(1-\mu^{S'}_i)}=1.$$
Therefore,
$$\|D M_E D\|_2 \geq v^T M_E v \geq |v\cdot (\mu^E-\mu^{S'})|^2 = \chi^2(\mu^{S'}, \mu^E) \;.$$
\end{proof}

We can now prove that the output in Step~\ref{step:prod-small} has the desired guarantee:

\begin{lemma} \label{lem:delta-distance-2}
We have that $\sqrt{\chi^2(\mu^{S'},p)} \leq 2\sqrt{\eps \|D M D\|_2} + O(\sqrt{\eps}).$
\end{lemma}
\begin{proof}
Since $S'=(S\setminus L)\cup E,$
we have that $|S'| \mu^{S'} = |S| \mu^S + |E| \mu^E - |L| \mu^L.$ Recalling that $L, E$ are disjoint,
the latter implies that
\begin{equation} \label{eq:chi-squared-jensen}
(|S|/|S'|) \sqrt{\chi^2(\mu^{S'}, \mu^{S})} \leq (|E|/|S'|) \sqrt{\chi^2(\mu^{S'},\mu^E)}
+ (|L|/|S'|)\sqrt{\chi^2(\mu^{S'},\mu^L)} \;.
\end{equation}

First note that, by Lemma \ref{lem:good-moments},
$|\sqrt{\chi^2(\mu^{S'}, \mu^S)} - \sqrt{\chi^2(\mu^{S'}, p)}| \leq O(\eps/d).$
Lemma \ref{lem:mel-2} and Corollary \ref{MApproxCor-2} give that
$$(|E|/|S'|)^2 \chi^2(\mu^{S'}, \mu^E) \leq  (|E|/|S'|)^2 \|D M_E D\|_2 + O(\eps) \leq  (|E|/|S'|) \|D M D \|_2 + O(\eps(1+\chi^2(\mu^{S'}, p))) \;.$$
Thus,  $$(|E|/|S'|) \sqrt{\chi^2(\mu^{S'}, \mu^E)} \leq \sqrt{(|E|/|S'|) \|D M D\|_2} + \sqrt{\eps}\cdot O(1+\sqrt{\chi^2(\mu^{S'}, p))} \;.$$
Lemmas \ref{MLBoundLem-2} and \ref{lem:mel-2} give that
$$(|L|/|S'|)^2 \chi^2(\mu^{S'},\mu^L) \le (|L|/|S'|)^2 \|D M_L D\|_2 \leq O((|L|/|S'|)^2 \chi^2(\mu^{S'}, p) +\eps) \;.$$
Thus,
$$(|L|/|S'|) \sqrt{\chi^2(\mu^{S'},\mu^L)}\leq  O((|L|/|S'|) \sqrt{\chi^2(\mu^{S'}, p)}) +O(\sqrt{\eps}) \;.$$

Substituting these into (\ref{eq:chi-squared-jensen}), yields
$$(|S|/|S'|) \sqrt{\chi^2(\mu^{S'},p)} \leq \sqrt{(|E|/|S'|) \|D M D\|_2} + O\left(\sqrt{\eps} \left(1+\sqrt{\chi^2(\mu^{S'}, p)}\right)\right) \; .$$
For $\eps$ sufficiently small, we have that the $\sqrt{\chi^2(\mu^{S'},p)}$ terms satisfy
$$(|S|/|S'|) - O(\sqrt{\eps}) \geq 1 - 2\eps - O(\sqrt{\eps}) \geq\frac{1}{\sqrt{2}}.$$
Recalling that $|E|/|S'| \leq \Delta(S,S')|S|/|S'| \leq (5/2) \eps$, we now have:
$$\sqrt{\chi^2(\mu^{S'},p)} \leq (5/2)\sqrt{\eps \|D M D\|_2} + O(\sqrt{\eps}) \;,$$
as required.
\end{proof}

\begin{corollary} \label{cor:delta-2} Let  $\delta := 3\sqrt{\eps|\lambda|}$.
For some universal constant $C$, if $\delta \leq C \sqrt{\eps \log(1/\eps)},$
then $\sqrt{\chi^2(\mu^{S'},p)}  \leq O(\sqrt{\eps\log(1/\eps)}).$
Otherwise, we have $\sqrt{\chi^2(\mu^{S'},p)} \leq \delta.$
\end {corollary}
\begin{proof}
By Lemma \ref{lem:delta-distance-2}, we have that
$$\sqrt{\chi^2(\mu^{S'},p)} \leq \frac{5}{6}\delta + O(\sqrt{\eps}) \; .$$
If $C$ is sufficiently large,  when $\delta > C \sqrt{\eps\log(1/\eps)},$
this $O(\sqrt{\eps})$ is at most $C\sqrt{\eps\log(1/\eps)}/6.$
\end{proof}

\begin{claim} \label{clm:correct}
If the algorithm terminates at Step~\ref{step:prod-small}, then
we have $\dtv(P, P') \leq O(\sqrt{\eps \log(1/\eps}),$
where $P'$ is the product distribution with mean vector $\mu^{S'}$.
\end{claim}
\begin{proof}
By Corollary \ref{cor:delta-2}, we have that
$\sqrt{\chi^2(\mu^{S'},p)}  \leq O(\sqrt{\eps \log(1/\eps)}).$
Thus, by  Corollary~\ref{cor:dtv-chi-squared}, the total variation distance
between the product distributions with means $p$ and $\mu^{S'}$
is $O(\sqrt{\eps\log(1/\eps)}).$
\end{proof}

\subsubsection{The Case of Large Spectral Norm} \label{ssec:filter-dtv}

We next need to show the correctness of the algorithm if it returns a filter,
If we reach this step, then we have $\|DMD\|_2 = \Omega(1),$
indeed $|v'DMDv'^T|=\Omega(1),$
and  by Corollary~\ref{cor:delta-2}, it follows that
$\sqrt{\chi^2(\mu^{S'},p)} \leq \delta$ where $\delta := 3 \sqrt{ \eps \|DMD\|_2}.$

Since $\|v'\|_2=1,$ $Dv'$ satisfies
$\sum_{i=1}^d (Dv')_i^2 \mu^{S'}_i (1-\mu^{S'}_i) = \sum_{i=1}^m v_i'^2 = 1.$
Thus, we can apply Corollary~\ref{cor:Chebyshev} to it.

\new{
%Something
\begin{lemma} \label{lem:expected-good} 
We have
$\E_Z[\Delta(S,S'')] \leq \Delta(S,S')$.
\end{lemma}
\begin{proof}
Let $a= \max_{x \in S'} |v^{\ast} \cdot x - \mu^{S'}|$. 
Firstly, we look at the expected number of samples we reject:
\begin{align*}
\E_Z[|S''|]-|S'| & = \E_Z\left[|S'|\Pr_{X \in_u S'}[|X-\mu^{S'}| \geq a Z]\right] \\
& = |S'| \int_{0}^{1} \Pr_{X \in_u S'}\left[|v^{\ast} \cdot(X-\mu^{S'})| \geq a x \right] 2x dx \\
& = |S'| \int_{0}^{a} \Pr_{X \in_u S'}\left[|v^{\ast} \cdot(X-\mu^{S'})| \geq T \right] (2T /a)  dT \\
& = |S'| \E_{X \in_u S'} \left[(v^{\ast} \cdot(X-\mu^{S'}))^2\right]/a \\
&=  (|S'|/a) \cdot v^{\ast T} M v^{\ast}  = (|S'|/a) \lambda' \;.
\end{align*}
Next, we look at the expected number of false positive samples we reject. 
If we write $S''=S \cup L' \setminus E'$ for disjoint multisets $L'$ and $E'$, 
then these are the elements of $L' \setminus L$. We have:
\begin{align*}
\E_Z[|L'|]-|L| & = \E_Z\left[(|S|-|L|)\Pr_{X \in_u S \setminus L}\left[|X-\mu^{S'}| \geq T\right] \right] \\
& \leq \E_Z\left[|S|\Pr_{X \in_u S}[|v^{\ast} \cdot(X-\mu^{S'})| \geq a Z] \right] \\
& = |S| \int_{0}^{1} \Pr_{X \in_u S}[|v^{\ast} \cdot(X-\mu^{S'})| \geq a x] 2x \; dx \\
& = |S| \int_{0}^{a} \Pr_{X \in_u S}[|v^{\ast} \cdot(X-\mu^{S'})| \geq T ] (2T /a) \; dT \\
& \leq |S| \int_{0}^{\infty} \Pr_{X \in_u S}[|v^{\ast} \cdot(X-\mu^{S'}))|\geq T ] (2T /a) \; dT  \\
& = |S| \E_{X \in_u S} \left[(v^{\ast} \cdot(X-\mu^{S'})))^2 \right]/a \\ 
&= (|S'|/a) \cdot v^{\ast T} M_S v^{\ast} = (|S'|/a) \cdot v'^{ T} D M_S  D v' \\
& \leq (|S'|/a) \cdot \|D M_S D \|_2  \\ 
& \leq (|S'|/a) \cdot \|D M_P D \|_2 + (|S'|/a) \cdot \|D (M_P - M_S) D \|_2 \\
& \leq (|S'|/a) \cdot (\sqrt{\eps} +  9 + \chi^2(\mu^{S'}, p)) \\
& \leq (|S'|/a) \cdot O(1+\delta^2) \leq (|S'|/a) \cdot O(1+ \eps \lambda') \;,
\end{align*}
where the penultimate line uses Lemmas \ref{lem:M-S-close} and \ref{lem:M-P-small}.
When $\lambda'$ is at least a sufficiently large constant, 
$\lambda'$ is bigger than $2 \cdot O(1+ \eps \lambda')$,
and so $\E_Z[S'']-S' \geq 2(\E_Z[L']-L)$. Now consider that 
$|S''|=|S|+|E'|-|L'|=|S'| - |E|+|E'|+|L|-|L'|$, and thus 
$|S''|-|S'|=|E|-|E'| + |L'|-|L|$. 
This yields that $|E|- \E_Z[|E'|] \geq (\E_Z[L']-L)$, 
which can be rearranged to $\E_Z[|E'| + |L'|] \leq |E| + |L|$ 
or in other terms $\E_Z[\Delta(S,S'')] \leq \Delta(S,S')$.
\end{proof}
}

%The following simple claim completes the proof of Proposition~\ref{prop:filter-dtv},
%and thus proves Theorem~\ref{thm:binary-product-dtv}.

%!TEX root = ./main.tex

\section{Agnostically Learning Mixtures of Two Balanced Binary Products, via Filters}\label{sec:filter-mixtures}

In this section, we study the problem of agnostically learning 
a mixture of two balanced binary product distributions. 
Let $p$ and $q$ be the coordinate-wise means of the two product distributions. 
Let $u =\frac{p}{2} - \frac{q}{2}$. 
Then, when there is no noise, the empirical covariance matrix is
$\Sigma = u u^T + D$,
where $D$ is a diagonal matrix whose entries are $\frac{p_i + q_i}{2} - \frac{(p_i- q_i)^2}{4}$. 
Thus, it can already have a large eigenvalue. 
Now in the presence of corruptions it turns out that we can construct a filter 
when the {\em second} absolute eigenvalue is also large. 
When it is the case that only the top absolute eigenvalue is large, 
we know that both $p$ and $q$ are close to 1-dimensional affine subspace (a.k.a. line) $\{ \mu + c v: c \in \R \} $, 
where $\mu$ is the empirical mean and $v$ is the top eigenvector. 
And by performing a grid search over $c$, we will find a good candidate hypothesis. 

Unfortunately, bounds on the top absolute eigenvalue 
do not translate as well into bounds on the total variation distance of our estimate to the true distribution, 
as they did in all previous cases (e.g., if the top absolute eigenvalue is small in the case of learning 
the mean of a Gaussian with identity covariance, 
we can just use the empirical mean, etc). 
In fact, an eigenvalue $\lambda$ could just mean that $p$ and $q$ differ 
by $\sqrt{\lambda}$ along the direction $v$. 
However, we can proceed by zeroing out the diagonals. 
If $uu^T$ has any large value along the diagonal, 
this operation can itself produce large eigenvalues. 
So, this strategy only works when $\|u\|_\infty$ is appropriately bounded. 
When $\| u \|_\infty$ is large, there is a separate strategy to deal with large entries in $u$ 
by guessing a coordinate whose value is large and conditioning on it, 
and once again setting up a modified eigenvalue problem. 
%We defer the details to the full version of our paper. 
Our overall algorithm then follows from balancing all of these different cases, and we describe the technical components in more detail in the next subsection. 

\subsection{The Full Algorithm}

This section is devoted to the proof of the following theorem:

\begin{theorem} \label{thm:product-mixtures-main}
Let $\Pi$ be a mixture of two $c$-balanced binary product distributions in $d$ dimensions.
Given $\eps > 0$ and $\poly(d,1/\eps)\log(1/\tau)$ independent samples from $\Pi$, 
an $\eps$-fraction of which have been arbitrarily corrupted,
there is a polynomial time algorithm that, with probability at least $1-\tau$, 
outputs a mixture of two binary product distributions $\Pi'$ such that 
$\dtv(\Pi,\Pi') = O(\eps^{1/6}/\sqrt{c})$.
\end{theorem}

Recall that our overall approach is based on two strategies that succeed under different assumptions. 
Our first algorithm (Section~\ref{ssec:prod-mix-anchor}) assumes that there exists 
a coordinate in which the means of the two component product distributions differ by a substantial amount. 
Under this assumption, we can use the empirical mean vectors conditioned on this coordinate being $0$ and $1$. 
We show that the difference between these conditional mean vectors 
is almost parallel to the difference between the mean vectors of the product distributions. 
Considering eigenvectors perpendicular to this difference will prove a critical part of the analysis of this case.
Our second algorithm (Section~\ref{ssec:prod-mix-close}) succeeds under the assumption 
that the mean vectors of the two product distributions are close in all coordinates. 
This assumption allows us to zero out the diagonal of the covariance matrix without introducing too much error.

Both these algorithms give an iterative procedure that produces filters 
which improve the sample set until they produce an output. 
We note that these algorithms essentially only produce a line in $\R^d$ 
such that both mean vectors of the target product distributions 
are guaranteed to be close to this line in $\ell_2$-distance. 
The assumption that our product distributions are balanced implies that $\Pi$ is close 
in variation distance to some mixture of two products whose mean vectors lie exactly on the given line. 
Given this line, we can exhaustively compare $\Pi$ to a polynomial number of such mixtures 
and run a tournament to find one that is sufficiently close.

%For technical reasons, both of our aforementioned algorithms will require that the mixing weights 
%are bounded away from $0$ and $1$. If the mixing weights are very close to these extreme values, 
%$\Pi$ will be close to a single product distribution, and we will be able to learn it to small error 
%using our algorithm (from the previous section) for learning a single balanced binary product, 
%treating contributions from the other component product as additional noise.

We note that together these algorithms will cover all possible cases. 
Our final algorithm  runs all of these procedures in parallel, 
obtaining a polynomial number of candidate hypothesis distributions, such that at least one is sufficiently close to $\Pi$. 
We then run the tournament described by Lemma~\ref{tournamentLem} in order to find a particular candidate that is sufficiently close to the target. 
To ensure that all the distributions returned are in some finite set $\mathcal{M}$, 
we round each of the probabilities of each of the products  to the nearest multiple of $\eps/d$, and similarly 
round the mixing weight to the nearest multiple of $\eps$.  This introduces at most $O(\eps)$ additional error.

\begin{algorithm}[h]
\begin{algorithmic}[1]
\Procedure{LearnProductMixture}{$\epsilon, \tau , S'$}
\INPUT a set of $\poly(d, 1/\eps) \log(1/\tau)$ samples of which an $\eps$-fraction have been corrupted.
\OUTPUT a mixture of two balanced binary products that is $O(\eps^{1/6})$-close to the target

\State Run the procedure \textsc{Filter-Balanced-Product}($2\eps^{1/6}, S'_1$) for up to $d+1$ iterations 
on a set $S'_1$ of corrupted samples of size $\Theta(d^4\log(1/\tau)/\eps^{1/3})$.

\For{each $1 \leq i^{\ast} \leq d$}
\State Run the procedure \textsc{Filter-Product-Mixture-Anchor}($i^{\ast}, \eps, S'_{2,i^{\ast}}$) for up to $d+1$ iterations 
on a set $S'_{2,i^{\ast}}$ of corrupted samples of size $\Theta(d^4\log(1/\tau)/\eps^{13/6})$.
\EndFor

\State Run the procedure \textsc{Filter-Product-Mixture-Close}($\eps, S'_3, \delta := \eps^{1/6}$) for up to $d+1$ iterations 
on a set $S'_3$ of corrupted samples of size $\Theta(d^4\log(1/\tau)/\eps^{13/6})$.

\State Run a tournament among all mixtures output by any of the previous steps. Output the winner.

\EndProcedure
\end{algorithmic}
\caption{Filter algorithm for agnostically learning a mixture of two balanced products}
\label{alg:product-mixture}

\end{algorithm}

\subsection{Mixtures of Products Whose Means Differ Significantly in One Coordinate}  \label{ssec:prod-mix-anchor}

We will use the following notation. 
Let $\Pi$ be a mixture of two $c$-balanced binary product distributions.
We will write $\Pi$ as $\alpha P +(1-\alpha) Q,$ where $P, Q$ are binary product distributions
with mean vectors $p, q$, and $\alpha \in [0, 1]$.
In this subsection, we prove the following theorem:

\begin{theorem} \label{thm:product-mixture-anchor}
Let $\Pi = \alpha P +(1-\alpha) Q$ be a mixture of two $c$-balanced binary product distributions in $d$ dimensions,
with $\eps^{1/6}\le \alpha \le 1-\eps^{1/6},$ such that there exists $1 \leq i^{\ast} \leq d$
with $p_{i^{\ast}} \geq q_{i^{\ast}} + \eps^{1/6}.$
There is an algorithm that, 
given $i^{\ast}$, $\eps > 0$ and $\Theta(d^4\log(1/\tau)/\eps^3)$ independent samples from $\Pi$, 
an $\eps$-fraction of which have been arbitrarily corrupted, 
runs in polynomial time and, with probability at least $1-\tau$, 
outputs a set $R$ of candidate hypotheses such that 
there exists $\Pi'\in R$ satisfying $\dtv(\Pi,\Pi') = O(\eps^{1/6}/\sqrt{c})$.
\end{theorem}

For simplicity of the analysis, we will assume without loss of generality that $i^{\ast} = d,$ unless otherwise specified.
First, we determine some conditions under which our sample set will be sufficient.
We start by recalling our condition of a good set for a balanced binary product distribution:

\begin{definition}\label{basicGoodDefn}
Let $P$ be a binary product distribution in $d$ dimensions and let $\epsilon>0$. 
We say that a multiset $S$ of elements of $\{0,1\}^d$ is \emph{$\eps$-good} with respect to $P$ if for every affine function $L : \R^d \to \R$
it holds
$$
|\Pr_{X\in_u S}(L(X)>0) - \Pr_{X\sim P}(L(X)>0)| \leq \eps/d.
$$
\end{definition}

We will also need this to hold after conditioning on the last coordinate.

\begin{definition}
Let $P$ be a binary product distribution in $d$ dimensions and let $\eps>0$. 
We say that a multiset $S$ of elements of $\{0,1\}^d$ is \emph{$(\eps, i)$-good} with respect to $P$ 
if $S$ is $\eps$-good with respect to $P$, and $S^j  \eqdef \{x\in S: x_i=j\}$ 
is $\eps$-good for the restriction of $P$ to $x_i=j$, for $j \in \{0,1\}$.
\end{definition}

Finally, we define the notion of a good set for a mixture of two balanced products.

\begin{definition}
Let $\Pi=\alpha P+(1-\alpha)Q$ be a mixture of two binary product distributions. 
We say that a multiset $S$ of elements of $\{0,1\}^d$ is \emph{$(\eps,i)$-good} with respect to $\Pi$ 
if we can write $S=S_P\cup S_Q$, where $S_P$ is $(\eps,i)$-good with respect to $P$, 
$S_Q$ is $(\eps,i)$-good with respect to $Q$, and $|\frac{|S_P|}{|S|}-\alpha| \leq \eps/d^2$.
\end{definition}

We now show that taking random samples from $\Pi$ produces such a set with high probability.
\begin{lemma}\label{lem:mixture-anchor-good}
Let $\Pi = \alpha P +(1-\alpha) Q$ be a mixture of binary product distributions, 
where $P, Q$ are binary product distributions with mean vectors $p, q$. 
Let $S$ be a set obtained by taking $\Omega(d^4\log(1/\tau)/\eps^{13/6})$ independent samples from $\Pi$. 
Then, with probability at least $1-\tau$, $S$ is $(\eps,i)$-good with respect to $\Pi$ for all $i\in [d]$.
\end{lemma}
The proof of this lemma is deferred to Section \ref{sec:filter-mixtures-appendix}.

%As we are going to be treating the last coordinate very differently, we will need the following definition
%\begin{definition}
%For $v\in \R^d$ let $v_{-d}\in \R^{d-1}$ be obtained by removing the last coordinate of $v$.
%\end{definition}

We claim that given a good set with an $\eps$-fraction of its entries corrupted, 
we can still determine $\Pi$ from it. In particular, this is achieved by iterating the following proposition.
\begin{proposition}
Let $\Pi=\alpha P+(1-\alpha)Q$ be a mixture of two $c$-balanced binary products, with $p_d \geq q_d +\eps^{1/6}$ and $\eps^{1/6}<\alpha<1-\eps^{1/6}$. 
Let $S$ be an $(\eps, d)$-good multiset for $\Pi$, and let $S'$ be any multiset with $\Delta(S,S') \leq 2\eps$. 
There exists an algorithm which, given %$p_d,q_d,\alpha,$ and
$S'$ and $\eps>0$, runs in polynomial time and returns either a multiset $S''$ with $\Delta(S,S'') \leq \Delta(S,S')  - 2\eps/d$, 
or returns a list of mixtures of two binary products $\mathcal{S}$ 
such that there exists a $\Pi' \in \mathcal{S}$ with $\dtv(\Pi,\Pi')=O(\eps^{1/6}/\sqrt{c})$.
\end{proposition}

We note that iteratively applying this algorithm until it outputs a set $R$ of mixtures gives Theorem~\ref{thm:product-mixture-anchor}.

\medskip

\noindent {\bf Notation.} All vectors in this section should be assumed to be over the first $d-1$ coordinates only. 
We will write $p_{-d}$ and $q_{-d}$ for the first $d-1$ coordinates of $p$ and $q$, 
but for other vectors we will use the similar notation to that used elsewhere to denote $(d-1)$-dimensional vectors.

%Before we begin the analysis, we note that changing the values of $\alpha, p_d,$ or $q_d$ by $O(\eps/d)$ only changes $\Pi$ by $O(\eps/d)$ in variational distance. Thus, it suffices to consider the case where $\frac{|S_P|}{|S|}=\alpha$, $\frac{|(S_P)_1|}{|S_P|} = p_d$, and $\frac{|(S_Q)_1|}{|S_Q|} = q_d$. This will simplify the analysis somewhat.

The algorithm, written in terms of $i^{\ast}$ instead of $d$ for generality, is as follows:

\begin{algorithm}[H]
\begin{algorithmic}[1]
\Procedure{Filter-Product-Mixture-Anchor}{$i^{\ast},  \epsilon,  S'$}
\State Let $\mu$ be the sample mean of $S'$ without the $i^{\ast}$ coordinate. 
Let $\Sigma$ be the sample covariance matrix of $S'$ without the  $i^{\ast}$ row and column. 
%_{-d}$, i.e. the sample covariance matrix of $S'$ with the last row and column removed.

\State Let $S'_0$ and $S'_1$ be the subsets of $S'$ with a $0$ or $1$ in their $i^{\ast}$ coordinates, respectively.

\State Let $\mu^{(j)}$ be the sample mean of $S'_j$ without the $i^{\ast}$ coordinate.

\State Let $u=\mu^{(1)}- \mu^{(0)}$.
Compute the unit vector $v^{\ast} \in\R^{d-1}$ with  $v^{\ast} \cdot u=0$
that maximizes $v^T {\Sigma} v$ and let $\lambda = v^{\ast T}  \Sigma v^{\ast}$.

{\em /* Note that $v^{\ast}$ is the unit vector maximizing the quadratic form $v^T \Sigma v$ over the subspace $u\cdot v=0$, 
and thus can be approximated using standard eigenvalue computations.*/}

\If {$\lambda \leq \gamma$}

{\em /* $\gamma$ is some absolute constant to be determined in the course of the analysis*/}

\State Let $L$ be the set of points $\mu+i (\eps^{1/6}/\|u\|_2) u$ truncated to be in $[c,1-c]^d$ for $i \in \Z$ with $|i| \leq 1+\sqrt{d}/\eps^{1/6}$.
\State \textbf{return} the set of distributions $\Pi'=\alpha'P'+(1-\alpha')Q'$ with the means of $P'$ and $Q'$, 
$p', q'$ with $p'_{-i^{\ast}}, q'_{-i^{\ast}} \in L$ and $p'_{i^{\ast}},q'_{i^{\ast}} \in [c,1-c], \alpha' \in [0,1],$ multiples of $\eps^{1/6}$. \label{step:anchor-small-ev}
\EndIf

%then output $R$, the set of mixtures of products $\Pi'=\alpha'P'+(1-\alpha')Q'$ where $\alpha'$ is a multiple of $\eps^{1/6}$ with $\eps^{1/6} \leq \alpha \leq 1- %\eps^{1/6}$ and mean vectors(p,q)$ of the form $p=\tilde \mu + a u$,
%$q=\tilde \mu + bu$ for $a,b$ which are multiples of $\eps/\|u\|_2$ with $|a|,|b| \leq 1 + \sqrt{d}$.

\State  Let $\delta =  C(\eps^{1/6}\sqrt{\lambda}+ \eps^{2/3}\log(1/\eps))$ for a sufficiently large constant $C$.

\State Find a real number $T>0$ such that
$$
\Pr_{X \in_u S'} \left (|v^{\ast} \cdot (X_{-i^{\ast}}- \mu)|>T+\delta \right) > 8\exp(-T^2/2)+8\eps/d \;.
$$
%\If {no such $T$ exists} \textbf{fail.} \EndIf

\State \textbf{return} the set $S''=\{x\in S':|v\cdot (x_{-i^{\ast}}- \mu) | \leq T+\delta\}$.

\EndProcedure
\end{algorithmic}
\caption{Filter algorithm for a mixture of two binary products whose means differ significantly in some coordinate}
\label{alg:anchor-point}
\end{algorithm}

We now proceed to prove correctness.
We note that given $S=S_P\cup S_Q$, we can write
$$
S' = S'_P \cup S'_Q \cup E.
$$
where $S'_P \subseteq S_P$, $S'_Q \subseteq S_Q$ and $E$ is disjoint 
from $S_P \setminus S'_P$ and $S_Q \setminus S'_Q$. 
Thus, we have $$\Delta(S,S')=\frac{|S_P \setminus S'_P|+|S_Q \setminus S'_Q|+|E|}{|S|} \;.$$
We use the notation $\mu^{S_P}, \mu^{S'_P}, \mu^E \in \R^{d-1}$ etc., 
for the means of $S_P$, $S'_P$, $E$, etc., excluding the last coordinate.
%We let $\beta =\frac{|E|}{|S|}$, $\gamma_P=\frac{|L_P|}{|S_P|}$, $\gamma_Q=\frac{|L_Q|}{|S_Q|}$. We note that $\Delta(S,S')= \alpha\gamma_P+(1-\alpha)\gamma_Q+\beta < \eps$. We also let $(1+\alpha)=\frac{|S|}{|S'|}$.

We next need some basic lemmata relating the means of some of these distributions.

\begin{lemma} \label{lem:ML-restate}
Let $P$ be a binary product distribution with mean vector $p$. 
Let $S$ be an $\eps$-good multiset for $P$ in the sense of 
Definition~\ref{basicGoodDefn}. 
Let $\tilde{S}$ be a subset of $S$ with $|S|-|\tilde{S}|=O(\eps|S|)$. 
Let $\mu^{\tilde{S}}$ be the mean of $\tilde{S}$. 
Then, $\|p-\mu^{\tilde S}\|_2 \leq O(\eps \sqrt{\log(1/\eps)})$.
\end{lemma}
\begin{proof}
Since $S$ is $\eps$-good, $\|\mu^S-p\|_2 \leq \eps/\sqrt{d}$.
Let $L=S \setminus \tilde S$. We can apply appropriate lemmata from Section~\ref{sec:bal-prod}. 
Note that Lemma \ref{lem:ML-bound} and Claim~\ref{claim:m}, 
only depend on $\mu^{S'}$ as far as it appears in the definition of $M_L,$ 
and we may treat it as a parameter that we will set to $p$. 
By Lemma \ref{lem:ML-bound} with $\mu^{S'} :=p$, we have $\|\E_{X \in_u L}[(X - p)(X - p)^T]\|_2 \leq O\left(\log(|S|/|L|) + \eps |S|/|L| \right)$. 
By Claim~\ref{claim:m} again with $\mu^{S'} :=p$, it follows that
$(|L|/|S|) \|\mu^{L} - p\|_2 \leq O(\eps \sqrt{\log(1/\eps)})$. 
Since $|S|\mu^S = |\tilde S|\mu^{\tilde S} + |L| \mu^L$, 
we have $\mu^S - \mu^{\tilde S} =-(|L|/|\tilde{S}|) (\mu^L - \mu^S)$ 
and so 
$$\|\mu^S - \mu^{\tilde S}\|_2 \leq (|L|/|\tilde{S}|) \|\mu^L - \mu^S\|_2 \leq O(\eps^2/\sqrt{d}) + O(1+\eps)(|L|/|S|) \|\mu^L - p\|_2 \leq O(\eps \sqrt{\log(1/\eps)}).$$
By the triangle inequality, $\|p - \mu^{\tilde S}\|_2 \leq \eps/\sqrt{d} +  O(\eps \sqrt{\log(1/\eps)}) = O(\eps \sqrt{\log(1/\eps)})$.
\end{proof}

We next show that $\mu^{(1)}-\mu^{(0)}$ is approximately parallel to $p_{-d}-q_{-d}$. 
Note that if we had $S=S'$ and $\mu^{S_P}=p_{-d}, \mu^{S_Q} = q_{-d}$, 
then  $\mu^{(1)}-\mu^{(0)}$ would be a multiple of $p_{-d}-q_{-d}$. 
Since $S$ is $\eps$-good, we can bound the error introduced by $\mu^{S_P}-p$, 
$\mu^{S_Q}-q_{-d}$ and Lemma~\ref{lem:ML-restate} allow us to bound the error 
in taking $\mu^{S'_P}, \mu^{S'_Q}$ instead of $p_{-d}, q_{-d}$. 
However, we still have terms in the conditional means of $E$:

\begin{lemma} \label{lem:almost-parallel} 
For some scalars $a=O(\eps)$, $b^0=O(|E^0|/|S'|)$, $b^1 = O(|E^1|/|S'|)$, we have
$$\|(1-\mu_{d})\mu_{d}u - (\alpha(1-\alpha)(p_{d}-q_{d})+a)(p_{-d}-q_{-d}) - b^0(\mu^{E^0}-\mu)-b^1(\mu^{E^1}-\mu)\|_2 \leq  O(\eps\log(1/\eps)) \;,$$
where $E^j$ is the subset of $E$ with last entry $j$, $\mu^{E^j}$ is the mean of $E^j$ 
with $d^{th}$ coordinate removed.
\end{lemma}
\begin{proof}

%First we consider the size of $\mu_d$. For this, we note that the sample mean of $S'$ differs from the sample mean of $S$ by $O(\eps)$ in each %coordinate, and therefore,
%$$
%\mu_d = \alpha p_d+(1-\alpha)q_d +O(\eps), \ \ \ \ \ \ \ \ \ \ 1-\mu_d = \alpha(1-p_d)+(1-\alpha)(1-q_d)+O(\eps).
%$$

Let $S_P'^j, S_Q'^j, E^j, S'^j$ denote the subset of the appropriate set 
in which the last coordinate is $j$. Let $\mu^{{S'_P}^j}, \mu^{{S'_Q}^j}, \mu^{E^j}$ 
denote the means of $S_P'^j, S_Q'^j$, and $E^j$ with the last entry truncated, respectively.
%Let $\gamma_P^j$ and $\gamma_Q^j$ denote $\frac{|L_P^j|}{|S_P|}$ and $\frac{|L_Q^j|}{|S_Q|}$, respectively. Let $\beta^j=\frac{|E^j|}{|S|}$.

We note that
$$
S'^j = S'_P \cup S'_Q\cup E^j.
$$
Taking the means of the subsets of $S'^j$, we find that
$$
|S'^j |\mu^{(j)}  = |S_{\tilde P^j}|\mu^{{S'_P}^j} + |S_Q'^j|\mu^{{S'_Q}^j} + |E^j| \mu^{E^j}.
$$
%Note that by goodness of the appropriate distributions, we have for all $i$, $|(\mu^{S_P}^j)_i-p_i| \leq \eps/d$. Thus we have $\|\mu^{S_P}^j-p_{-d}\|_2,\|\mu^{S_Q}^j-q_{-d}\|%_2=O(\eps/\sqrt{d})$ and $||S_P^j|/|S_P| -p_d| \leq \eps/d$, $||S_Q^j|/|S_Q| -q_d| \leq \eps/d$.
Therefore, using this and Lemma \ref{lem:ML-restate}, we have that
$$
|S'^j|\mu^{(j)}  = |S_P'^j|p_{-d} + | S_Q'^j|q_{-d} + |E^j| \mu^{E^j} + O(\eps\log(1/\eps)|S^j|),
$$
where $O(\eps)$ denotes a vector of $\ell_2$-norm $O(\eps)$.

Thus, we have
\begin{align}
|S'^0||S'^1| (\mu^{(1)}-\mu^{(0)}) & = (|S'^0||S_P'^1| - |S'^1|| S_P'^0|) p_{-d} \nonumber \\
									& + (|S'^0|| S_Q'^1| - |S'^1|| S_Q'^0|) q_{-d} \nonumber \\
									& + |E^1||S'^0| \mu^{E^1} - |E^0| |S'^1| \mu^{E^0} + O(\eps \log(1/\eps)(|S^1||S'^0|+|S^0||S'^1|)) \;. \label{eq:full}
\end{align}
Since $|S'^j| = |S_P'^j| + |S_Q'^j| + |E^j|$, we have:
\begin{align*}
0 = |S'^0||S'^1| - |S'^1||S'^0| & = (|S'^0||S_P'^1| - |S'^1||S_P'^0|) \\
& + (|S'^0||S_Q'^1| - |S'^1||S_Q'^0|)  + |E^1| |S'^0| - |E^0||S'^1|  \;.
\end{align*}
Thus, the sum of the coefficients of the $p_{-d}$ and $q_{-d}$ terms in 
Equation (\ref{eq:full}) is $|E^0||S'^1| - |E^1| |S'^0|$, 
which is bounded in absolute value by $|E||S'| \leq O(\eps|S|^2)$. 
Meanwhile, the $p_{-d}$ coefficient of Equation (\ref{eq:full})  has:
\begin{align*}
& |S'^0||S_P'^1| - |S'^1||S_P'^0| \\
 =~& |S'^0||S_P^1|-|S'^1||S_P^0|+O(\eps|S'|^2) = |S'^0||S|\alpha p_d-|S'^1||S|\alpha(1-p_d) + O(\eps|S'|^2) \\
=~& |S^0||S|\alpha p_d-|S^1||S|\alpha(1-p_d) + O(\eps|S'|^2) \\
=~& ((\alpha(1-p_d)+(1-\alpha)(1-q_d))\alpha p_d - (\alpha p_d+(1-\alpha)q_d)\alpha(1-p_d) + O(\eps))|S'|^2 \\
=~& (\alpha(1-\alpha)(1-q_d)p_d - \alpha(1-\alpha)q_d(1-p_d) + O(\eps))|S'|^2 = (\alpha(1-\alpha)(p_d-q_d)+O(\eps))|S'|^2 \;.
\end{align*}
Noting that $(|E^1| |S'^0| - |E^0||S'^1|)\alpha=O(\eps|S'|^2)$ and $(|E^1| |S'^0| - |E^0||S'^1|)(1-\alpha)=O(\eps|S'|^2)$, 
we can write Equation (\ref{eq:full}) as:
 \begin{align*}
|S'^0||S'^1| (\mu^{(1)}-\mu^{(0)}) & = (\alpha(1-\alpha)(p_d-q_d)+O(\eps))|S'|^2(p_{-d} -q_{-d}) \\
									& + (|E^1| |S'^0| - |E^0||S'^1|)(\alpha p_{-d}+(1-\alpha)q_{-d}) \\
									& + |E^1||S'^0| \mu^{E^1} - |E^0||S'^1| \mu^{E^0} + O(\eps\log(1/\eps)|S'|^2) \;.
\end{align*}
We write $\mu^\Pi = \alpha p_{-d}+(1-\alpha)q_{-d}$ and so, dividing by $|S'|^2$ and recalling that $|E|/|S'| \leq O(\eps)$, we get
\begin{align}
\mu_d(1-\mu_d)(\mu^{(1)}-\mu^{(0)}) & = (\alpha(1-\alpha)(p_d-q_d)+O(\eps))(p_{-d} -q_{-d}) + O(|E^1|/|S'|)(\mu^{E^1} - \mu^\Pi) \nonumber \\
& + O(|E^0|/|S'|)(\mu^{E^0}-\mu^\Pi)  + O(\eps\log(1/\eps)) \;. \label{eq:almost}
\end{align}
If $\mu^\Pi = \mu$, then we would be done. 
So, we must bound the error introduced by making this substitution. 
We can express $\mu$ as
\begin{align*}
|S'|\mu & =|S'_P|\mu^{S'_P} +|S'_Q|\mu^{S'_Q}  +|E|\mu^E \\
&= |S|\mu^\Pi+O(\eps|S|)(p_{-d}-q_{-d}) + O(\eps\log(1/\eps)|S'|) + |E^1| \mu^{E^1} +|E^0| \mu^{E^0} \;,
\end{align*}
and so
$$|S|(\mu^\Pi-\mu) = O(\eps|S|)(p_{-d}-q_{-d}) + O(\eps\log(1/\eps)|S|) + |E^1| (\mu^{E^1}-\mu) +|E^0| (\mu^{E^0}-\mu) \;.$$
Thus, we have
$$\mu^\Pi = \mu + O(\eps)(p_{-d}-q_{-d}) + O(\eps\log(1/\eps)) + O(|E^1|/|S'|) (\mu^{E^1}-\mu) +O(|E^0|/|S'|) (\mu^{E^0}-\mu) \;.$$
Substituting this into Equation (\ref{eq:almost}), gives the lemma.
\end{proof}

We now show that, for any vector $v$ perpendicular to $u$, 
if the variance of $S'$ in the $v$-direction is small, then 
$v\cdot p_{-d}$ and $v\cdot q_{-d}$ are both approximately $v\cdot \mu$.

\begin{lemma}\label{lem:mean-dist-bound}
For any $v$ with $\|v\|_2=1$, $v\cdot u=0$, 
we have that $|v\cdot (p_{-d}-\mu)| \leq \delta$ and $|v\cdot (q_{-d}-\mu)| \leq \delta$ for
$\delta  := C (\eps^{1/6}\|\Sigma\|_2+ \eps^{2/3}\log(1/\eps))$
for a sufficiently large constant $C$ as defined in the algorithm.
\end{lemma}
\begin{proof}
We begin by noting that
\begin{align*}
v^T \Sigma v & = \Var_{X\in_u S'}(v\cdot X) = \E_{X\in_u S'}[|v\cdot (X-\mu)|^2]\\
& \geq (|E^j|/|S'|) \E_{X\in_u E^j}[|v\cdot (X-\mu)|^2]\\
& \geq (|E^j|/|S'|)  |v\cdot(\mu^{E^j}-\mu)|^2 \;.
\end{align*}
Next, since $v\cdot u=0$, we have by Lemma~\ref{lem:almost-parallel} that
\begin{align*}
& |v\cdot (p_{-d}-q_{-d})| \\
\leq~&\frac{1}{\alpha(1-\alpha)(p_d-q_d)} \cdot \left(O(|E^0|/|S'|)v\cdot (\mu^{E^0}-\mu)+O(|E^1|/|S'|)v\cdot (\mu^{E^1}-\mu) +O(\eps\log(1/\eps))\|v\|_2 \right)\\
=~& O\left(\frac{1}{\alpha(1-\alpha)(p_d-q_d)}\right) \left( \sqrt{\eps (v^T \Sigma v)} + \eps\log(1/\eps) \right).
\end{align*}
However, we have that $|S'|\mu=|S'_P| \mu^{S'_p} + |S'_Q| \mu^{S'_q} + |E| \mu^E + |S'|O(\eps\log(1/\eps))$, 
and so 
$$(|S'|-|E| )(\mu-\mu^{S'_p}) = |S'_Q| (\mu^{S'_Q}-\mu^{S'_P}) + |E| (\mu^E-\mu) +|S'| O(\eps\log(1/\eps)) \;.$$ 
Now, we have:
$$\mu - p_{-d} = (1-\alpha + O(\eps))(q_{-d}-p_{-d}) + O(|E|/|S'|) (\mu^E-\mu) + O(\eps\log(1/\eps)) \;.$$
Thus,
$$
|v\cdot(p_{-d}-\mu)| = O(v\cdot (p_{-d}-q_{-d})) + O(|E|/|S'|) (v\cdot(\mu^E-\mu)-v\cdot(\mu-p_{-d})) + O(\eps\log(1/\eps)).
$$
Therefore, 
$$
|v\cdot(p_{-d}-\mu)| =  O\left(\frac{1}{\alpha(1-\alpha)(p_d-q_d)}\right) \left(\sqrt{\eps (v^T \Sigma v)} + \eps\log(1/\eps) \right) \;.
$$
Inserting our assumptions that $\alpha(1-\alpha) \geq \eps^{1/6}/2$, and $p_d - q_d \geq \eps^{1/6}$ gives
$$
|v\cdot(p_{-d}-\mu)| =  O(\eps^{1/6}\sqrt{ \|\Sigma\|_2} + \eps^{2/3}\log(1/\eps)) \leq \delta \;,
$$
when $C$ is sufficiently large.

The other claim follows symmetrically.
\end{proof}

We can now show that if we return $R,$ 
some distribution returned is close to $\Pi$. 
First, we show that there are points on $L$ close to $p_{-d}$ and $q_{-d}$.

\begin{lemma} \label{lem:close-line}
There are $c,d \in \R$ such that $\tilde p=\mu+cu$ 
and $\tilde q=\mu+du$ have $\|\tilde p-p_{-d}\|_2, \|\tilde q-q_{-d}\|_2 \leq \delta$.
\end{lemma}
\begin{proof}
If we take the $c$ that minimizes $\|\tilde p-p_{-d}\|_2$, 
then $u\cdot(\tilde p-p_{-d}) =0$. Thus, we can apply Lemma \ref{lem:mean-dist-bound}, 
giving that $|(\tilde p-p_{-d}) \cdot (p_{-d}- \mu)| \leq \|\tilde p-p_{-d}\|_2 \delta$.

However, $\tilde p-\mu=cu$ so we have $(\tilde p-p_{-d}) \cdot (\tilde p-\mu)=0$ 
and thus $$\|\tilde p-p\|_2^2 = |(\tilde p-p_{-d}) \cdot (p_{-d}- \mu)| \leq \|\tilde p-p_{-d}\|_2 \delta.$$
Therefore, $\|\tilde p-p_{-d}\|_2 \leq \delta.$
\end{proof}

It is clear that even discretizing $c$ and $d$, we can still find such a pair that satisfies this condition.
\begin{lemma} \label{lem:discrete-close-line}
There are $p', q' \in L$ such that $\|p_{-d}-p'\|_2, \|q_{-d}-q'\|_2 \leq \delta +O(\eps^{1/6}) \;.$
\end{lemma}
\begin{proof}
By Lemma \ref{lem:close-line}, there exist points 
$\tilde p=\mu+(a/\|u\|_2)u$ and $\tilde q=\mu+(b/\|u\|_2)u$ 
with $a,b \in \R$ that have $\|\tilde p-p_{-d}\|_2, \|\tilde q-q_{-d}\|_2 \leq \delta.$

Letting $i\eps^{1/6}$ be the nearest integer multiple of $\eps^{1/6}$ to $a$, 
we have that $p' := \mu + i (\eps^{1/6}/\|u\|_2)u$ has 
$\|p_{-d}-p'\|_2 \leq \|\tilde p_{-d}-p\|_2 + \|p'-\tilde{p}\|_2 \leq  \delta + \eps^{1/6}.$

Note that we have $\|p_{-d}-\tilde p\|_2 \leq \|p_{-d}-\mu\|_2 \leq \sqrt{d}\|p_{-d}-\mu\|_\infty \leq \sqrt{d}$, which implies that $a \leq \sqrt{d}$. 
Thus, $|i| \leq 1 + \sqrt{d}/\eps^{1/6}$. If $p' \notin [c,1-c]$, then replacing any coordinates less than $c$ 
with $c$ and more than $1-c$ with $1-c$ 
can only decrease the distance to $p$ since $p \in [c,1-c]^d$. 
Thus, there is a point $p' \in L$ with $\|p_{-d}-p'\|_2 \leq \delta +O(\eps^{1/6})$.

Similarly, we show that there is a $q' \in L$ such that $\|q-q'\|_2 \leq \delta +O(\eps^{1/6})$.
\end{proof}

\begin{corollary} 
If the algorithm terminates at Step \ref{step:anchor-small-ev}, then there is a $\Pi' \in R$ with with $\dtv(\Pi',\Pi) = O(\eps^{1/6}/\sqrt{c})$.
\end{corollary}
\begin{proof}
By Lemma~\ref{lem:discrete-close-line}, there exists $\tilde p$, $\tilde q \in L$ 
such that $\|p_{-d}-\tilde p\|_2, \|q_{-d}-\tilde q\|_2 \leq \delta +O(\eps^{1/6})$. 
But now there is a distribution $\Pi' \in R,$ 
where $\Pi'=\alpha'P'+(1-\alpha')Q'$ for binary products $P'$ and $Q',$ 
whose mean vectors are $p', q'$ and with $|\alpha'-\alpha| \leq \eps^{1/6}$, 
$\|p'_{-d}-p_{-d}\|_2, \|q'_{-d}-q_{-d}\|_2 \leq O(\eps^{1/6})$ 
and $|p'_d-p_d|, |q'_d-q_d|=O(\eps^{1/6})$. 
Note that this implies that $\|p'-p\|_2, \|q'-q\|_2=O(\eps^{1/6}).$

Since $P$ and $Q$ are $c$-balanced, 
we have $\dtv(P,P') \leq O(\|p-p'\|_2/\sqrt{c}) \leq O(\eps^{1/6}/\sqrt{c})$ 
and $$\dtv(Q,Q') \leq O(\|q-q'\|_2/\sqrt{c}) \leq O(\eps^{1/6}/\sqrt{c}).$$
Thus, $\dtv(\Pi',\Pi) \leq \delta+O(\eps^{1/6}/\sqrt{c})$.
Since we terminated in Step \ref{step:anchor-small-ev}, $\lambda \leq O(1)$, 
and so $\delta =  C(\eps^{1/6}\sqrt{\lambda}+\eps^{2/3}\log(1/\eps)) = O(\eps^{1/6}).$
\end{proof}

Now, we are ready to analyze the second part of our algorithm. 
The basic idea will be to show that if $\lambda$ is large, 
then a large fraction of the variance in the $v$-direction is due to points in $E$.

\begin{lemma} 
If $\lambda \geq \Omega(1)$, then 
$$\Var_{X \in_u S'}[v^{\ast} \cdot X] \ll \frac{|E| \E_{Y\in_u E}[|v^{\ast} \cdot(Y- \mu)|^2]}{|S'|(\alpha(1-\alpha)(p_d-q_d))^2}.$$
\end{lemma}
\begin{proof}
We have that
\begin{align*}
|S|\Var_{X \in_u S'}[v^{\ast} \cdot X] & = |S'_P|\left(\Var_{X \in_u S'_P}[v^{\ast} \cdot X] +|v^{\ast} \cdot(\mu^{S'_P}-\mu)|^2\right) \\
 & + |S'_Q|\left(\Var_{X \in_u S'_Q}[v^{\ast} \cdot X] + |v^{\ast} \cdot (\mu^{S'_Q}-\mu)|^2\right) \\
& + |E| \E_{X\in_u E}[|v^{\ast} \cdot(X-\mu)|^2]] \;.
\end{align*}
Since $S_P$ and $S_Q$ are $\eps$-good, we have that
\begin{align*}
\Var_{X \in_u S'_P}[v^{\ast} \cdot X] & = \E_{X \in_u S'_P}[(v\cdot X - v^{\ast} \cdot \mu^{S'_P})^2 ] \\
& \leq (|S_P|/|S'_P|) \E_{X \in_u S_P}[(v^{\ast} \cdot X - v^{\ast} \cdot \mu^{S'_P})^2 ] \\
& = (|S_P|/|S'_P|) \left(\Var_{X \in_u S_P}[v^{\ast} \cdot X] + (v^{\ast} \cdot (\mu^{S_P} - \mu^{S'_P}))^2\right) \\
& \leq (|S_P|/|S'_P|) \left(\Var_{X \sim P}[v^{\ast} \cdot X] + (v^{\ast} \cdot (p_{-d} - \mu^{S'_P}) +O(\eps\sqrt{\log(1/\eps)}))^2\right) \\
& \leq (1+O(\eps/\alpha)) \cdot (O(1) + O(\eps \sqrt{\log(1/\eps)})^2) \leq O(1) \;,
\end{align*}
and similarly,
$$
\Var_{X \in_u S'_Q}[v^{\ast} \cdot X] = O(1) \;.
$$

Thus, we have:
$$|S'|\Var_{X \in_u S'}[v^{\ast} \cdot X] \leq |E|\E_{X\in_u E}[|v^{\ast} \cdot(X-\mu)|^2] + 
O(1+|v^{\ast} \cdot(p_{-d}-\mu)|^2+|v^{\ast} \cdot(q_{-d}-\mu)|^2)|S'| \;.$$
By Lemma \ref{lem:almost-parallel}, we have
\begin{align*}
|v\cdot(p_{-d}-\mu)|, |v^{\ast} \cdot(q_{-d}-\mu)| & \leq O(1/(\alpha(1-\alpha)(p_{d}-q_{d})))\\
& \ \ \cdot(O(|E_0|/|S'|)|v^{\ast} \cdot(\mu^{E^0}-\mu)| +O(|E^1|/|S'|)|v^{\ast} \cdot(\mu^{E^1}-\mu)|+O(\eps\log(1/\eps)))\\
& \leq \sqrt{(|E|/|S'|) \E_{Y\in_u E}[|v^{\ast} \cdot(Y- \mu)|^2]} + O(\eps\log(1/\eps)) \;.
\end{align*}
However,
$$
\lambda = \Var_{X \in_u S'}[v^{\ast} \cdot X] \ll \frac{|E| \E_{Y\in_u E}[|v^{\ast} \cdot(Y- \mu)|^2]}{|S|(\alpha(1-\alpha)(p_d-q_d))^2} + O(1) \;.
$$
Since $\lambda$ is larger than a sufficiently large constant, this completes the proof.
\end{proof}

We next show that the threshold $T>0$ required by our algorithm exists.
\begin{lemma} 
If $\lambda \geq \Omega(1)$, there exists a $T>0$ such that
$$\Pr_{X \in_u S'} \left( |v^{\ast} \cdot (X-\mu)|>T+\delta \right) > 8\exp(-T^2/2)+8\eps/d \;.$$
\end{lemma}
\begin{proof}
Assume for the sake of contradiction that this is not the case, i.e., that for all $T>0$ we have that
$$
\Pr_{X\in_u S'}(|v^{\ast} \cdot (X- \mu)| \ge T+\delta)  \le 8\exp(-T^2/2)+8\epsilon/d \;.
$$
Using the fact that $E\subset S'$,
this implies that for all $T>0$
$$|E|\Pr_{Y\in_u E}(|v^{\ast} \cdot(Y- \mu)|>T+\delta) \ll |S'|(\exp(-T^2/2)+\epsilon/d) \;.$$
Therefore, we have that
\begin{align*}
\E_{Y\in_u E}[|v^{\ast} \cdot (Y-\mu)|^2] & \ll \delta^2 + \E_{Y\in_u E}[\min(0, |v^{\ast} \cdot (Y-\mu)|-\delta)^2]\\
& \ll \delta^2 + \int_0^{\sqrt{d}}\Pr_{Y\in_u E}(|v^{\ast} \cdot (Y-\mu)|>T+\delta) T dT\\
& \ll \delta^2+ \int_0^{\sqrt{d}} (\epsilon/d) T dT + \int_0^{2\sqrt{\log(|S'|/|E|)}} T dT \\ & \phantom{\ll\delta^2}+\int_{2\sqrt{\log(|S'|/|E|)}}^\infty (|S'|/|E|)\exp(-T^2/2)T dT \\\
& \ll \delta^2 +\epsilon + \log(|S'|/|E|) \;.
\end{align*}
On the other hand, we know that
$$
\E_{Y\in_u E}[|v^{\ast} \cdot (Y-\mu)|^2] \gg (\alpha(1-\alpha)(p_d-q_d))^2 \lambda|S'|/|E| \gg \log(|S'|/|E|) \;.
$$
Combining with the above we find that
$$
%\delta^2 = O(\epsilon \lambda/(\alpha(1-\alpha)(p_d-q_d))^2 ) \gg (\alpha(1-\alpha)(p_d-q_d))^2\lambda|S'|/|E|.
\delta^2 = O(\epsilon^{1/3} \lambda) \gg (\alpha(1-\alpha)(p_d-q_d))^2\lambda|S'|/|E| \;.
$$
Or in other words,
$$
%\epsilon^{2} \geq \epsilon|E|/|S'| \gg (\alpha(1-\alpha)(p_d-q_d))^4 \geq \epsilon^{6/5},
\epsilon^{4/3} \geq \epsilon^{1/3}|E|/|S'| \gg (\alpha(1-\alpha)(p_d-q_d))^2 \geq \eps^{2/3} \;,
$$
which provides a contradiction.
\end{proof}

Finally, we show that $S''$ is closer to $S$ than $S'$ was.
\begin{claim} 
If the algorithm returns $S''$ then $\Delta(S,S'')\leq \Delta(S,S') - 2\eps/d$.
\end{claim}
\begin{proof}
Since $S'' \subset S$, we can write 
$S'' = S''_P \cup S''_Q \cup E''$ 
for $S''_P \subseteq S'_P$, $S''_Q \subseteq S_Q$ 
and $E'' \subset E$, where $E''$ has disjoint support 
from $S''_P \setminus S_P$ and $S''_Q \setminus S_Q$. 
Thus, we need to show that
$$|E'' \setminus E| \geq 2\eps|S|/d + |S'_P \setminus S''_P| + |S'_Q \setminus S''_Q| \;.$$
We have that
\begin{align*}
|S' \setminus S''| & = \Pr_{X\in_u S'}(|v\cdot (X- \mu)| \ge T+\delta) |S'| \\
& \geq (8\exp(-T^2/2)+8\epsilon/d)|S'| \geq (4\exp(-T^2/2)+4\epsilon/d)|S| \;.
\end{align*}
By Hoeffding's inequality, we have that
$$\Pr_{X\sim P}(|v^{\ast} \cdot (X- p_{-d})| \ge T) \leq 2\exp(-T^2/2) \;.$$
By Lemma \ref{lem:mean-dist-bound}, we have that $|v^{\ast} \cdot(\mu-p_{-d})| \leq \delta$ 
and so
$$\Pr_{X\sim P}(|v^{\ast} \cdot (X- \mu)| \ge T+\delta) \leq 2\exp(-T^2/2) \;.$$
Since $S$ is $(\eps, d)$-good, we have
$$\Pr_{X\in_u S_P}(|v^{\ast} \cdot (X- \mu)| \ge T+\delta) \leq 2\exp(-T^2/2) + \eps/d \;.$$
We get the same bound for $X \in_u S_Q$, and so
\begin{align*}
& \Pr_{X \in_u S}(|v^{\ast} \cdot (X- \mu)| \ge T+\delta)  \\
= ~&(|S_P|/|S|)\Pr_{X\in_u S_P}(|v^{\ast} \cdot (X- \mu)| \ge T+\delta) + (|S_Q|/|S|)\Pr_{X\in_u S_Q}(|v^{\ast} \cdot (X- \mu)| \ge T+\delta)\\
 \leq ~&2\exp(-T^2/2) + \eps/d \;.
\end{align*}
Since $L''_P \cup L''_Q \subseteq S$ 
but any $x \in (S'_P \setminus S''_P) \cup (S'_Q \setminus S''_Q)$ 
has $v^{\ast} \cdot (x-\mu) \geq T + \delta$, 
we have that
\begin{align*}
|S'_P \setminus S''_P| + |S'_Q \setminus S''_Q| & \leq \Pr_{X \in_u S}(|v^{\ast} \cdot (X- \mu)| \ge T+\delta)|S| \\
& \leq (2\exp(-T^2/2) + \eps/d)|S| \;.
\end{align*}
Finally, we have that:
\begin{align*}
|E\setminus E'| & = |S' \setminus S''| - |S'_P \setminus S''_P| - |S'_Q \setminus S''_Q| \\
				& \geq \left( 4\exp(-T^2/2)+\frac{4 \eps}{d}\right)|S| - \left( 2\exp(-T^2/2) + \frac{\eps}{d}\right) |S| \\
				& \geq  \left( 2\exp(-T^2/2) + \frac{3 \eps}{d} \right) |S| \\
				& \geq |S'_P \setminus S''_P| + |S'_Q \setminus S''_Q| + \frac{2 \eps}{d} \;,
\end{align*}
which completes the proof.
\end{proof}

\subsection{Mixtures of Products Whose Means Are Close in Every Coordinate} \label{ssec:prod-mix-close}

In this section, we prove the following theorem:
\begin{theorem}\label{mixCloseThm}
Let $\eps,\tau>0$ and let $\Pi = \alpha P+(1-\alpha)Q$ be a $d$-dimensional 
mixture of two $c$-balanced product distributions $P$ and $Q$ whose means $p$ and $q$ satisfy $\|p-q\|_\infty \leq \delta,$ 
for $\delta \geq \sqrt{\eps \log(1/\eps)}$, and $c \leq p_i, q_i \leq 1-c$ for $i \in [d].$ 
Let $S$ be a multiset of $\Omega(d^4\log(1/\tau)/(\eps^{2}\delta))$ independent samples from $\Pi$. 
Let $S'$ be obtained by adversarially changing an $\eps$-fraction of the points in $S$. 
There exists an algorithm that runs in polynomial time and, with probability at least $1-\tau$,
returns a set of distributions $R$ such that some $\Pi' \in R$ has $\dtv(\Pi,\Pi') \leq O(\delta/\sqrt{c})$.
\end{theorem}

We will assume without loss of generality that $\alpha\leq 1/2.$ 
We may also assume that $\alpha>10\delta \geq 10\epsilon$ since otherwise, 
we can make use of our algorithm for learning a single product distribution. 

%\begin{definition}
%Let $P$ be a binary product distribution on $\{0,1\}^d$ and let $\epsilon>0$. 
%We say that a multiset $S$ of elements of $\{0,1\}^d$ is $\epsilon$-good with respect to $P$ if for every affine function $L$
%$$
%|\Pr_{x\in_u S}(L(x)>0) - \Pr_{x\sim P}(L(x)>0)| \leq \eps/6d \;.
%$$
%\end{definition}

In this context, we require the following slightly different definition of a good set:

\begin{definition}
Let $S$ be a multiset in $\{0,1\}^d$. 
We say that $S$ is \emph{$\eps$-good} for the mixture $\Pi$ 
if there exists a partition $S=S_P\cup S_Q$ 
such that $\left| \frac{|S_P|}{|S|} - \alpha\right| \leq \epsilon$ 
and that $S_P$ and $S_Q$ are $\eps/6$-good 
for the component product distributions $P$ and $Q$, respectively.
\end{definition}

\begin{lemma}
\label{lem:mixture-close-good}
If $\Pi$ has mixing weights $\delta \leq \alpha \leq 1-\delta$, 
with probability at least $1-\tau$, 
a set $S$ of $\Omega(d^4\log 1/\tau/(\eps^2\delta))$ samples drawn from $\Pi$ is good for $\Pi$.
\end{lemma}
The proof of this lemma is in Section \ref{sec:filter-mixtures-appendix}.

Our theorem will follow from the following proposition:

\begin{proposition}
Let $\Pi$ be as above and $S$ be a good multiset for $\Pi$. 
Let $S'$ be any multiset with $\Delta(S,S')\leq 2\epsilon$. 
There exists a polynomial time algorithm that, given $S'$, $\eps>0$ and $\delta$,
returns either a multiset $S''$ with $\Delta(S,S'') \leq \Delta(S,S')-2\eps/d$ 
or a set of parameters of binary product distributions of size $O(d/(\eps \delta^2))$ 
which contains the parameters of a $\Pi'$ with $\dtv(\Pi,\Pi') \leq  O(\delta/\sqrt{c})$.
\end{proposition}

Before we present the algorithm, 
we give one final piece of notation. 
For $S$ a set of points, we let $\Cov(S)$ be the sample covariance matrix of $S$ 
and $\Cov_0(S)$ be the sample covariance matrix with zeroed out diagonal. 
Our algorithm is presented in detailed pseudocode in Algorithm~\ref{alg:mixture-product-close}.

\begin{algorithm}[htb]
\begin{algorithmic}[1]
\Procedure{Filter-Product-Mixture-Close}{$\epsilon, S',\delta$}
\State Compute $\mu$, the sample mean of $S'$, and $\Cov_0(S').$ 
Let $C$ be a sufficiently large constant.
\If {$\Cov_0(S')$ has at most one eigenvector with absolute eigenvalue more than $C\delta^2$}

\State Let $v^{\ast}$ be the unit eigenvector of $\Cov_0(S')$ with largest absolute eigenvalue. 
\State Let $L$ be the set of points $\mu+i \delta v^{\ast}$ truncated to be in $[c,1-c]^d$, 
for $i \in \Z$ with $|i| \leq 1+\sqrt{d}/\delta$. %where $v$ is the unit eigenvector of $\Cov_0(S')$ with largest eigenvalue.
\State \textbf{return} the set of distributions of the form $\Pi'=\alpha' P'+(1-\alpha') Q'$ 
with the means of $P'$ and $Q'$ in $L$ and $\alpha'$ is a multiple of $\eps$ in $[10\eps,1/2]$. \label{step:close-cover}
\EndIf
\State Let $v^{\ast}$ and $u^{\ast}$ be orthogonal eigenvectors with eigenvalues more than $C\delta^2$. %Let $\pi:\R^d\rightarrow \R^2$ be the projection onto the subspace spanned by $v$ and $u$.
\State Find a number $t \geq 1 + 2\sqrt{\log(1/\eps)}$ and $\theta$ a multiple of $\delta^2/d$ such that $r=(\cos \theta)u^{\ast} +(\sin \theta)v^{\ast}$ satisfies
$$
\Pr_{X\in_u S'}\left(\Pr_{Y\in_u S'}\left(|r\cdot(X-Y)|<t \right)<2\epsilon\right) > 12\exp(-t^2/4) + 3\eps/d \;. 
$$ \label{step:close-threshold}
\State \textbf{return} the set
$S''=\{x\in S' \mid \Pr_{Y \in_u S'}(|r \cdot (x-Y)|<t)\geq 2\epsilon\} \;.$ \label{step:close-filter}
\EndProcedure
\end{algorithmic}
\caption{Filter algorithm for mixture of two binary products whose means are close in every coordinate}
\label{alg:mixture-product-close}
\end{algorithm}

%\begin{comment}
%\begin{enumerate}
%\item Compute $\mu$, the sample mean of $S'$, and $\Cov_0(S').$ Let $C$ be a sufficiently large constant.
%\item If $\Cov_0(S')$ has at most one eigenvector with eigenvalue with absolute value more than $C\delta^2$ for a sufficiently large constant $C$:
%\begin{enumerate}
%\item Let $v$ be the eigenvector of $\Cov_0(S')$ with largest eigenvalue. Let $L$, the set of points $\mu+i \delta v$ 
%truncated to be in $[c,1-c]^d$ for $i \in \Z$ with $|i| \leq 1+\sqrt{d}/\delta$, where $v$ is the unit eigenvector of $\Cov_0(S')$ with largest eigenvalue.
%\item Return the set of distributions of the form $\Pi'=\alpha' P'+(1-\alpha') Q$ with the means of $P$ and $Q$ in $L$, 
%and $\alpha'$ is a multiple of $\eps$ in $[10\eps,1/2]$.
%\end{enumerate}
%\item Otherwise, let $v$ and $u$ be orthogonal eigenvectors with eigenvalues more than $C\delta^2$. %Let $\pi:\R^d\rightarrow \R^2$ be the projection onto the subspace %spanned by $v$ and $u$.
%\item Find a number $t \geq 1 + 2\sqrt{\log(1/\eps)}$  and $\theta$ a multiple of $\delta^2/d$ such that $r=(\cos \theta)u +(\sin \theta)v$ has
%$$
%\Pr_{x\in_u S'}(\Pr_{y\in_u S'}(|r\cdot(x-y)|<t)<2\epsilon) > 12\exp(-t^2/4) + 3\eps/d.
%$$
%\item Return the set
%$$S''=\{x\in S': \Pr_{y\in_u S'}(|r.(x-y)|<t)\geq 2\epsilon\}.$$
%\end{enumerate}
%\end{comment}

To analyze this algorithm, we begin with a few preliminaries. 
Firstly, we recall that $S=S_P\cup S_Q$. 
We can write $S'=S'_P\cup S'_Q\cup E$, 
where $S'_P \subset S_P$, $S'_Q\subset S_Q$, 
and $$|S|\Delta(S,S') = |S_P\setminus S'_P|+|S_Q\setminus S'_Q|+|E| \;.$$ 
Let $\mu^{S'_P}$ and $\mu^{S'_Q}$ be the sample means of $S'_P$ and $S'_Q$, respectively.

\begin{lemma} \label{lem:ML-restate-2}
We have that
$
\alpha \|p-\mu^{S'_P}\|_2, (1-\alpha)\|q-\mu^{S'_Q}\|_2 = O(\epsilon \sqrt{\log(1/\epsilon)}) \;.
$
\end{lemma}
\begin{proof}
This follows from Lemma \ref{lem:ML-restate}.
\end{proof}

\noindent We will require that the matrix $\Cov_0(S')$ is close to being PSD. 
The proof of this fact is rather technical and we defer it to the appendix.
\begin{lemma} \label{lem:T}
Let $T$ be the multiset obtained from $S'$ by replacing all points of $S'_P$ 
with copies of $\mu^{S'_P}$ and all points of $S'_Q$ with copies of $\mu^{S'_Q}$. 
Then,
$\| \Cov_0(S') - \Cov(T) \|_2 = O(\delta^2) \;.$
\end{lemma}

We are now prepared to show that 
the first return condition outputs a correct answer. 
We begin by showing that vectors $u$ with large inner products 
with $\mu^{S'_P}-\mu$ or $\mu^{S'_Q}-\mu$ correspond to large eigenvectors of $\Cov_0(S')$.
\begin{lemma} \label{lem:thing}
For $u\in \R^d$, we have 
$$\alpha(u\cdot(\mu^{S'_P}-\mu))^2 + (1-\alpha)(u\cdot(\mu^{S'_P}-\mu))^2 \leq 2u^T \Cov_0(S') u+O(\delta^2)\|u\|_2^2.$$
\end{lemma}
\begin{proof}
Using Lemma \ref{lem:T}, we have 
$u^T \Cov_0(S') u = \Var_{X\in_u T}(u\cdot X) + O(\delta^2)\|u\|_2^2$. 
From the definition of $T$ it follows that
\begin{align*}
\Var_{X\in_u T}(u\cdot X) & \geq \left(\frac{|S'_P|}{|S'|} \right)(u\cdot(\mu^{S'_P}-\mu))^2 + \left(\frac{|S'_Q|}{|S'|} \right)(u\cdot(\mu^{S'_Q}-\mu))^2 
+ \frac{|E|}{|S'|}\Var_{X\in_u E}(u\cdot X) \\
& \geq (\alpha-2\eps) (u\cdot(\mu^{S'_P}-\mu))^2 + (1-\alpha-2\eps) (u\cdot(\mu^{S'_Q}-\mu))^2 \\
& \geq \alpha/2 \cdot (u\cdot(\mu^{S'_P}-\mu))^2 + (1-\alpha)/2 \cdot (u\cdot(\mu^{S'_Q}-\mu))^2 \;. \tag*{(since $\alpha,1-\alpha \geq 4\eps$)}
\end{align*}
\end{proof}

Next, we show that, if there is only one large eigenvalue of $\Cov_0(S')$, 
the means in question are both close to a given line.

\begin{lemma} \label{lem:l2-at-last}
There are $\tilde p, \tilde q \in L$ 
such that $\|p-\tilde p\|_2 \leq O(\delta/\sqrt{\alpha})$ 
and $\|q-\tilde q\|_2 \leq O(\delta/\sqrt{1-\alpha}).$
\end{lemma}
\begin{proof}
Let $p'=\mu+av^{\ast}$, $q'=\mu+bv^{\ast}$ with $a,b \in \R$ 
be the closest points to $p$ and $q$ on the line 
$\mu + cv^{\ast}$, for $c \in \R$. 
Then, $v^{\ast} \cdot (p'-p)=0$ and since $v^{\ast}$ is the only eigenvector 
of the symmetric matrix $\Cov_0(S')$ with eigenvalue more 
than $C(\delta^2 + \epsilon\log(1/\epsilon)),$
we have that 
$$(p'-p)^T \Cov_0(S') (p'-p) \leq C(\delta^2 + \epsilon\sqrt{\log(1/\epsilon)}) \|p'-p\|_2^2 \;.$$ 
We thus obtain:
\begin{align*}
\|p'-p\|_2^4 & = (p'-p) \cdot (p-\mu)^2 \\
                     &\leq 2(p'-p) \cdot (p-\mu^{S_P})^2 + 2(p'-p) \cdot (p-\mu^{S_P})^2 \\
	            & \leq O(\eps^2\log(1/\eps)/\alpha^2)\|p'-p\|_2^2 + (4/\alpha) \cdot (p'-p)^T \Cov_0(S') (p'-p)^T +O(\delta^2/\alpha)\|p'-p\|_2^2 \\
	            & \leq  O((\delta^2+\epsilon\log(1/\epsilon))/\alpha) \|p'-p\|_2^2 \tag*{(since $\alpha \geq \eps$)} \\
		  & \leq O(\delta^2/\alpha) \|p'-p\|_2^2 \;,
\end{align*}
where the second line uses Lemmas~\ref{lem:ML-restate-2} and~\ref{lem:thing}. 
We thus have that $\|p'-p\|_2 \leq O(\delta/\sqrt{\alpha})$. 
Letting $i\delta$ be the nearest integer multiple to $a$, 
we have that $\tilde p := \mu + i \delta v^{\ast}$ has 
$$\|p-\tilde p\|_2 \leq \|p'-p\|_2 + \|p'-\tilde{p}\|_2 \leq  O(\delta/\sqrt{\alpha}).$$
Note that we have $\|p-p'\|_2 \leq \|p-\mu\|_2 \leq \sqrt{d}\|p-\mu\|_\infty \leq \sqrt{d}$. 
So, $a \leq \sqrt{d}/\delta$. 
Thus, $|i| \leq 1 + \sqrt{d}/\delta.$ 
If $\tilde p \notin [c,1-c],$ 
then replacing any coordinates less than $c$ with $c$ 
and more than $1-c$ with $1-c$ can only decrease the distance to $p$, 
since $p \in [c,1-c]^d$.

Similarly, we show that there is a $\tilde q \in L$ 
such that $\|q-\tilde q\|_2 \leq O(\delta/\sqrt{1-\alpha})$, which completes the proof.
\end{proof}

\begin{corollary}
If the algorithm outputs a set of distributions in Step \ref{step:close-cover}, 
then one of those distributions has $\dtv(\Pi',\Pi) \leq O(\delta/\sqrt{c})$.
\end{corollary}
\begin{proof}
There is a distribution in the set $\Pi=\alpha' P'+(1-\alpha')Q'$, 
where $|\alpha-\alpha'| \leq \eps$ and the means of $P'$ and $Q'$ 
are $\tilde p$ and $\tilde q$ as in Lemma~\ref{lem:l2-at-last}. 
Then, we have $\dtv(P,P') \leq \|p - \tilde p\|/\sqrt{c} \leq O(\delta/\sqrt{\alpha c})$ 
and $\dtv(Q,Q') \leq \|p - \tilde p\|/\sqrt{c} \leq O(\delta/\sqrt{(1-\alpha})c)$. 
Thus, we have
$$\dtv(\Pi',\Pi) \leq O(\eps) + \alpha\dtv(P,P')+(1-\alpha)\dtv(Q,Q')  \leq O(\eps) +O((\sqrt{\alpha}+\sqrt{1-\alpha})\delta/\sqrt{c}) \leq O(\delta/\sqrt{c}) \;.$$
\end{proof}

Next, we analyze the second case of the algorithm. 
We must show that Step \ref{step:close-threshold} will find an $r$ and $t$. 
First, we claim that there is a $\theta$ which makes $r$ 
nearly perpendicular to $\mu^{S'_p}-\mu^{S'_Q}$.

\begin{lemma} 
There exists a $r=(\cos \theta)u^{\ast} + (\sin \theta)v^{\ast}$, 
with $\theta$ a multiple of $\delta^2/d$,
that has $$|r \cdot (\mu^{S'_P}-\mu^{S'_Q})| \leq \delta^2/\sqrt{d}.$$
\end{lemma}
\begin{proof}
Let $z=(\mu^{S'_P}-\mu^{S'_Q})$. 
If $u^{\ast} \cdot z = 0$, then $\theta = 0$ suffices. 
Otherwise, we take $\theta' = \cot^{-1}  (\frac{ v^{\ast} \cdot z}{u^{\ast} \cdot z}).$ 
Then, let $\theta$ be the nearest multiple of $\delta^2/d$ to $\theta'$. 
Note that $|\cos \theta - \cos \theta'|,|\sin \theta - \sin \theta'| \leq |\theta-\theta'|$ 
and $|u^{\ast} \cdot z|,|v^{\ast} \cdot z| \leq \sqrt{\|z\|_2} \leq \sqrt{d}$. Then, we have
\begin{align*}
|r \cdot z |& = |(\cos \theta)(u^{\ast} \cdot z) + (\sin \theta) (v^{\ast} \cdot z)| \\
&\leq |(\cos \theta')(u^{\ast} \cdot z) + (\sin \theta') (v^{\ast} \cdot z)| + |\theta -\theta'|\sqrt{d} \\
&= |\sin \theta'| |u^{\ast} \cdot z + (\cot \theta') (v^{\ast} \cdot z)| +  |\theta -\theta'|\sqrt{d} \\
&\leq  0 + \delta^2/\sqrt{d} \;.
\end{align*}
\end{proof}

%Note that such an $r$ exists. If $u \cdot (\mu^{S'_P}-\mu^{S'_Q}) = 0$ then we can take $r=u$. Otherwise $r': = (u \cdot (\mu^{S'_P}-\mu^{S'_Q})) v - (v \cdot (\mu^{S'_P}-\mu^{S'_Q})) u$ satisfies $r' \cdot (\mu^{S'_P}-\mu^{S'_Q}) = 0$ and $r' \neq 0$ and so we can take $r=r'/\|r'\|_2$.

We now need to show that for this $r$, Step \ref{step:close-threshold} will find a $t$. 
For this $r$, $r\cdot\mu^{S'_P}$ and $r\cdot\mu^{S'_Q}$ are close. 
We need to show that $E$ contains many elements $x$ 
whose $r\cdot x$ is far from these. We can express this in terms of $T$:
\begin{lemma} \label{lem:T-fail}
Let $r$ be a unit vector in $r\in \langle u^{\ast},v^{\ast}\rangle$ 
with $|r \cdot (\mu^{S'_P}-\mu^{S'_Q})| \leq \delta^2/\sqrt{d}.$ 
Then, there is a $t > 1$ such that
$$
\Pr_{X\in_u T}(r \cdot(X-\mu^{S'_P})>2t) > 12\exp(-(t-1)^2/4) + \frac{3 \eps}{d} \;.
$$
\end{lemma}
\begin{proof}
First, we wish to show that $\E_{X\in_u E}[(r\cdot(X-\mu^{S'_P}))^2]$ is large.

Since $r \in  \mathrm{span} (u^{\ast}, v^{\ast})$, $|r^T \Cov_0(S') r| \geq C \delta^2$. 
By Lemma~\ref{lem:T}, we have that 
$$\Var_{X \in_u T}(r \cdot X) = r^T \Cov(T) r \geq r^T \Cov_0(S') r -O(\delta^2) \geq (C-O(1))\delta^2 \geq (C/2)\delta^2 \;,$$
for sufficiently large $C$ and we also have that $r^T \Cov_0(S') r$ is positive.

We note that
\begin{align}
r^T \Cov(T) r & = \Var(r\cdot T)\nonumber\\
& = (|E|/|S'|)\Var_{X\in_u E}(r\cdot X) + O(\alpha)(r\cdot (\mu-\mu^{S_P'}))^2 + O(1-\alpha)(r\cdot (\mu-\mu^{S_Q'}))^2 \nonumber\\ 
& \ \ \ \ \  +(|E|/|S'|)(r\cdot (\mu-\mu^{E}))^2\nonumber\\
& = (|E|/|S'|) \left(\Var_{X\in_u E}(r\cdot X)+(r\cdot (\mu-\mu^{E}))^2\right) +O(\delta^2) \;.\label{covTEqn}
\end{align}
Now,
$$
\E_{X\in_u E}[(r\cdot(X-\mu^{S'_P}))^2] = \Var_{X\in_u E}(r\cdot X)+(r\cdot (\mu^{S_P'}-\mu^{E}))^2 \;.
$$
We also have that
\begin{align*}
|S'|(r\cdot \mu) & = (|S'|-|E|)(r\cdot \mu^{S_P'}) + |S_Q'|(r\cdot (\mu^{S_P'}-\mu^{S_Q'})) + |E|(r\cdot \mu^E)\\
& = (|S'|-|E|)(r\cdot \mu^{S_P'}) + |E|(r\cdot \mu^E) + |S'|O(\delta^2) \;.
\end{align*}
Thus,
$$
(|S'|-|E|)(r\cdot (\mu - \mu^E)) = (|S'|-|E|)(r\cdot (\mu^{S_P'}-\mu^E))+|S'|O(\delta^2) \;,
$$
or
$$
(r\cdot (\mu - \mu^E)) = (r\cdot (\mu^{S_P'}-\mu^E))+O(\delta^2) \;.
$$
This implies that
$$
(r\cdot (\mu^{S_P'}-\mu^E))^2 \geq (r\cdot (\mu - \mu^E))^2/2 - O(\delta^4) \;.
$$
Substituting into \eqref{covTEqn}, we have
$$
(|E|/|S'|)\E_{X\in_u E}[(r\cdot(X-\mu^{S'_P}))^2]  = (|E|/|S'|)[\Var_{X\in_u E}[r\cdot X]+(r\cdot (\mu^{S_P'}-\mu^{E}))^2] -O(\delta^4) \gg C/2\delta^2.
$$
Thus, for $C$ sufficiently large,
$$
\E_{X\in_u E}[(r\cdot(X-\mu^{S'_P}))^2] \gg \delta^2/\eps.
$$

Suppose for a contradiction that this lemma does not hold. 
Then, since $E \subset T$, we have 
$$\Pr_{X\in_u E}\left(r \cdot(X-\mu^{S'_P})>2t)\right) \leq (|S'|/|E|)12\exp(-t^2/2) + \frac{3 \eps}{d} \;.$$ 
Thus, we have 
$$\Pr_{X\in_u E}(r \cdot(X-\mu^{S'_P})>t)) \leq (|S'|/|E|)12\exp(-(t-1)^2/4) + \frac{3 \eps}{d} \;,$$
and we can write
\begin{align*}
|S'| \delta^2 & \ll |E|  \E_{X \in_u E}[(r.X-r.\mu^{S_P})^2] \\
& = |E| \int_0^{\sqrt{d}} \Pr_{X\in_u E}(r \cdot(X-\mu^{S'_P})>t)) t dt \\
& \ll |E| \int_0^{1+\sqrt{\log(|S'|/|E|)/2}} t dt + |S'| \int^\infty_{1+\sqrt{\log(|S'|/|E|)/2}} \exp(-(t-1)^2/4) t dt + \int_0^{\sqrt n} \eps/d t dt \\
& \ll |E| \log(|S'|/|E|) + |E| + |S'|(|E|/|S|) + \eps \\
& \leq |S'| \cdot O(\eps \log(1/\eps)) \;.
\end{align*}
Since we assumed that $\delta^2 \geq \Omega(\eps \log(1/\eps)$, 
this is a contradiction.

\end{proof}

\noindent To get a similar result for $S'$, 
we first need to show that $S'_P$ 
and $S'_Q$ are suitably concentrated about their means:
\begin{lemma} \label{lem:Chernoff-r}
If  $t \geq 1$,
$$
(1-|E|/|S'|)\Pr_{X \in_u S'_P \cup S'_Q}\left(r \cdot(X-\mu^{S'_P}) > t \right) \leq \frac{5}{4} \exp(-(t-1)^2/2) + \frac{\eps}{5d}.
$$
If $t \geq 1 + \sqrt{2\log 6/\eps}$, this is strictly less than $2\eps/3.$
\end{lemma}
\begin{proof}
\begin{align*}
\Pr_{X \in_u S'_P} (r \cdot(X-\mu^{S'_P}) \leq t) 
& \leq (|S_P|/|S'_P|) \Pr_{X \in_u S_P} (r \cdot(X-\mu^{S'_P}) \leq t) \\
& \leq \left( 1 + \frac{O(\eps)}{1-\alpha} \right) \cdot  \left(\Pr_{X \sim P} (r \cdot(X-\mu^{S'_P}) \leq t) + \frac{\eps}{12d} \right)\\
& = \left( 1 + \frac{O(\eps)}{1-\alpha} \right) \cdot  \left(\Pr_{X \sim P} \left(r \cdot(X-p) \leq t- (r \cdot (\mu^{S'_P} - p))\right) + \frac{\eps}{6d} \right) \\
& \leq \left( 1 + \frac{O(\eps)}{1-\alpha} \right) \cdot \left( 2\exp(-(t-1/2)^2/2) + \frac{\eps}{6d}\right) \;. \tag*{(using Lemma \ref{lem:ML-restate-2} and Hoeffding's inequality)} \\
\end{align*}
Similarly, 
$$\Pr_{X \sim S'_Q} (r \cdot(X-\mu^{S'_Q}) \leq t) \leq \left( 1 + \frac{O(\eps)}{1-\alpha} \right) \cdot \left( 2\exp(-(t-1/2)^2/2) + \frac{\eps}{6d}\right) \; .$$
Since $|r\cdot(\mu^{S_Q}-\mu^{S_P})| \leq \delta^2/\sqrt{d} \leq 1/2,$ 
we have 
$$\Pr_{X \sim S'_Q} \left(r \cdot(X-\mu^{S'_Q}) \leq t\right) \leq (\left( 1 + \frac{O(\eps)}{1-\alpha} \right) \cdot \left(2\exp(-(t-1)^2/2) + \frac{\eps}{6d} \right) \;.$$
Noting that $1-(|S'_P|+|S'_Q|)/|S'| = |E|/|S'| \geq 4\eps/3$, we have
\begin{align*}
& (1-|E|/|S'|) \Pr_{X \in_u S'_P \cup S'_Q}\left(r \cdot(X-\mu^{S'_P}) > t \right) \\
& = (|S'_P|/|S'|) \Pr_{X \sim S'_P} \left( r \cdot(X-\mu^{S'_P}) > t \right) + (|S'_Q|/|S'|) \Pr_{X \sim S'_Q} \left( r \cdot(X-\mu^{S'_P}) > t \right) \\
& = ((\alpha+O(\eps)) \left(1 + \left( 1 + \frac{O(\eps)}{\alpha} \right) \right)+ (1-\alpha+O(\eps)) \left(1 + \left( 1 + \frac{O(\eps)}{1-\alpha} \right) \right)\cdot \left(2\exp(-(t-1)^2/2) + \frac{\eps}{6d}\right) \\
& \leq (1+O(\eps)) \cdot \left(2\exp(-(t-1)^2/2) + \frac{\eps}{6d}\right) \\
& \leq \frac{5}{2} \exp(-(t-1)^2/2) + \frac{\eps}{5d} \;,
\end{align*}
for $\eps$ sufficiently small.
If $t \geq 1 + \sqrt{2\log 6/\eps}$,
this expression is $(5/2)(\eps/6) + \eps/5d \leq 2\eps/3$.
\end{proof}

Now we can finally show that a $t$ exists for this $r$, so Step \ref{step:close-threshold} will succeed:
\begin{lemma} \label{lem:PQ-fail}
There is a $t \geq 1 + 2\sqrt{\log (9/\eps)}$ such that
$$\Pr_{X\in_u S'}\left(\Pr_{Y\in_u S'}\left(r\cdot(X-Y)>t\right)<2\eps \right) > 12\exp(-(t-1)^2/4) + \frac{3\eps}{d} \; .$$
\end{lemma}
\begin{proof}
By Lemma \ref{lem:T-fail}, there exists a $t \geq 1$ such that
$$\Pr_{X\in_u T}\left(r \cdot(X-\mu^{S'_P})>2t)\right) > 12\exp(-(t-1)^2/4) + \frac{3\eps}{d} \; .$$
Using the definition of $T$, the points when $x=\mu^{S'_P}$ or $x=\mu^{S'_Q}$ 
do not contribute to this probability so all points in $T$ 
that satisfy $r \cdot(x-\mu^{S'_P})>2t$ come from $E$. 
Since $E \subset S'$ and $|S'|=|T|$, we have
\begin{equation} \label{eq:far-r}
\Pr_{X\in_u S'}\left(r \cdot(X-\mu^{S'_P})>2t \right) \geq \Pr_{X\in_u T}\left(r \cdot(X-\mu^{S'_P})>2t\right) > 12\exp(-(t-1)^2/4) + \frac{3\eps}{d} \; .
\end{equation}
Noting that $|E|/|S'| \leq 4\eps/3$, 
all except a $4\eps/3$ fraction of points $x \in T$ have 
$r \cdot(x-\mu^{S'_P})=O(\delta^2)$. 
So, $4\eps/3 \geq 12\exp(-(t-1)^2/4)$. 
Therefore, $t \geq 1 + 2\sqrt{\log (9/\eps)}$.

Thus, by Lemma \ref{lem:Chernoff-r}, 
we have $(1-|E|/|S'|)\Pr_{X \in_u S'_P \cup S'_Q}\left( r \cdot(X-\mu^{S'_P}) > t \right) < 2\eps/3$. 
Again, using that
$|E|/|S'| \leq 4\eps/3$, we have that
$$
\Pr_{X \in_u S'} \left( r \cdot(X-\mu^{S'_P}) > t \right) < 2\eps \;.
$$
Consequently, if $x$ satisfies $r \cdot(x-\mu^{S'_P})>2t$, 
then it satisfies $\Pr_{Y \in_u S'} \left( r\cdot(x-Y) \leq t \right) <2\epsilon$. 
Substituting this condition into Equation (\ref{eq:far-r}) gives the lemma.
\end{proof}

Again we need to show that any filter does not remove too many points of $S$. 
We need to show this for an arbitrary $r$, not just one nearly parallel to $\mu^{S'_P}-\mu^{S'_Q}$.
\begin{lemma} \label{lem:conc-close-1d}
For any unit vector $r'$ and $t \geq 2\sqrt{\log(1/\eps)}$, we have
$$
(1-|E|/|S'|)\Pr_{X\in_u S'_P \cup S'_Q}\left(\Pr_{Y\in_u S'}\left(r'\cdot(X-Y) \leq t \right)<2\eps \right) \leq 3\exp(-t^2/4)+ \frac{\eps}{4 d} \;.
$$
\end{lemma}
\begin{proof}
Using Hoeffding's inequality, we have
\begin{align}
|S'_P|\Pr_{X \in_u S'_P} \left( r \cdot (p-X) > t/2 \right) & \leq |S_P| \Pr_{X \in_u S_P} \left( |r' \cdot (X-p)| > t/2 \right) \nonumber \\
& \leq |S_P| \left(\Pr_{X \sim P}(|r' \cdot (X-p)| > t/2) + \frac{\eps}{6 d} \right) \nonumber \\
& \leq |S_P| \left(2\exp(-t^2/4)+ \frac{\eps}{6 d}\right) \;. \label{eq:half-t}
\end{align}

Every point $x$ with $|r'\cdot(x-p)| \leq t/2$ has $|r'\cdot(x-y) \leq t|$ 
for all $y$ with $|r'\cdot(y-p)| \leq t/2$. Thus, for $x$ with $|r'\cdot(x-p)| \leq t/2$, we have
$$\Pr_{Y \in_u S'}(r'\cdot(x-Y) \leq t) \geq \frac{|S_P|}{|S'|} -\frac{|S_P|}{|S'|}  \left(2\exp(-t^2/4)+ \frac{\eps}{6d} \right) \;.$$
When $t \geq  2\sqrt{\log(1/\eps)}$, we have
\[
\frac{|S_P|}{|S'|} \left(2\exp(-t^2/4)+ \frac{3\eps}{d} \right) \leq (1+2\eps) \cdot \left( 2 \eps + \frac{\eps}{6d} \right) \leq 3 \eps \; .
\]
Also, we have
\[
\frac{|S_P|}{|S'|} \leq \frac{(\alpha-\eps/6)|S|}{|S|(1-2\eps)} \leq \alpha - 3\eps \leq 7 \eps \; .
\]
Thus, we have $\Pr_{Y \in_u S'}(r\cdot(x-Y) \leq t) \geq 4 \eps > 2 \eps.$

But inequality (\ref{eq:half-t}) gives a bound on the number of $x$ in $S_P$ 
that do not satisfy this condition. That is,
$$|S'_P|\Pr_{X\in_u S'_P}\left(\Pr_{Y\in_u S'}\left(r'\cdot(X-Y) \leq t \right)<2\eps \right) \leq   |S_P| \left(2\exp(-t^2/4)+ \frac{\eps}{6d} \right) \;.$$
Similarly, every point $x$ with $|r' \cdot (x-q)| \leq t/2$ has
$$\Pr_{y\in_u S'}(r' \cdot(x-y) \leq t) > 2\eps$$
and 
\[
|S'_Q|\Pr_{X \in_u S'_Q} (r' \cdot (X-p) > t/2) \leq \left( 2\exp(-t^2/4)+ \frac{\eps}{6d} \right) \; .
\] Thus,
$$|S'_Q| \Pr_{X\in_u S'_Q}\left(\Pr_{Y\in_u S'}(r' \cdot (X-Y) \leq t)<2\eps \right) \leq   |S_Q| \left(2\exp(-t^2/4)+ \frac{\eps}{6d} \right) \;.$$
Summing these gives
$$(|S'_P|+|S'_Q|)\Pr_{X\in_u S'_P \cup S'_Q}\left(\Pr_{Y\in_u S'}(r'\cdot(X-Y) \leq t)<2\eps \right) \leq |S| \left(2\exp(-t^2/4)+ \frac{\eps}{6d} \right) \;.$$
Dividing by $|S'|$ and noting that $|S| \leq (1+2\eps)|S'| \leq (3/2)|S'|$ completes the proof.
\end{proof}

Now, we can show that the filter improves $\Delta(S,S''),$ 
and such that the algorithm is correct in the filter case.
\begin{claim}
If we reach Step \ref{step:close-filter} and return $S''$, 
then $\Delta(S,S'') \leq \Delta(S,S')-2\eps/d.$
\end{claim}
\begin{proof}
We can write $S''=S''_P \cup S''_Q \cup E''$, 
where $E''$ has disjoint support from $S_P \setminus S''_P$ 
and $S_Q \setminus S''_Q$. 
Note that, since we have $S'' \subset S'$, 
we can define these sets such that $S''_P \subseteq S'_P$, 
$S''_Q \subseteq S'_Q$ and $E'' \subseteq E$. 
We assume that we do.
Now we have that 
$$\Delta(S,S') - \Delta(S,S'') = \frac{|E'' \setminus E'| - |S''_P \setminus S'_P| - |S''_Q \setminus S'_Q|}{|S|} \;.$$ 
Therefore, 
$$\Delta(S,S') - \Delta(S,S'') = \frac{|S'' \setminus S'| - 2(|S''_P \setminus S'_P| + |S''_Q \setminus S'_Q|)}{|S|} \;.$$
In Step \ref{step:close-threshold}, we found a vector $r$ and $t \geq 1 + 2\sqrt{\log(1/\eps)}$ such that
$$
\Pr_{X\in_u S'}\left(\Pr_{Y\in_u S'}(|r \cdot (X-Y)|<t)<2\epsilon\right) > 12\exp(-(t-1)^2/4) + \frac{3 \eps}{d} \;.
$$
Then in Step \ref{step:close-filter}, we remove at least a $12\exp(-t^2/4) + 3\eps/d$ fraction of points. 
That is, 
$$|S'' \setminus S'| \geq \left( 12\exp(-t^2/4) + \frac{3 \eps}{d} \right)|S'| \;.$$
The fact that $t \geq 1 + 2\sqrt{\log(1/\eps)}$ allows us to use 
Lemma \ref{lem:conc-close-1d}, with $r'=r$, yielding that:
$$
(1-|E|/|S'|)\Pr_{X\in_u S'_P \cup S'_Q}\left(\Pr_{Y\in_u S'}\left(r\cdot(X-Y) \leq t-1\right)<2\epsilon\right) \leq 3\exp(-(t-1)^2/4)+ \frac{ \eps}{4d} \;.
$$
This implies that
$$
(1-|E|/|S'|)\Pr_{X\in_u S'_P \cup S'_Q}\left(\Pr_{Y\in_u S'}\left(r\cdot(X-Y) < t\right)<2\epsilon\right) \leq 3\exp(-(t-1)^2/4)+ \frac{ \eps}{4d} \;.
$$
Thus, 
$$|S''_P \setminus S'_P| + |S''_Q \setminus S'_Q| \leq \left( 3\exp(-(t-1)^2/4)+ \frac{ \eps}{4d}\right)|S'| \;,$$ 
and we have
\begin{eqnarray*}
\Delta(S,S') - \Delta(S,S'') &\geq& \left(12\exp(-(t-1)^2/4) + 3\eps/d - 2\left(3\exp(-(t-1)^2/4)+ \frac{ \eps}{4d} \right)\right)|S'|/|S| \\
                                             &\geq& \frac{2 \eps}{d} \;,
\end{eqnarray*}
since $|S'| \geq |S|(1 - \Delta(S,S')) \geq (1-2\eps)|S| \geq 5|S|/6$.
\end{proof}

%!TEX root = ./main.tex

\section*{Acknowledgements}
J.L. would like to thank Michael B. Cohen and Samuel B. Hopkins for some very helpful discussions. We thank the reviewers for their detailed feedback, and Lili Su for pointing out an error in the proof of Lemma~\ref{lem:random-good-gaussian-mean}.

\bibliographystyle{alpha}
\bibliography{allrefs}
\appendix
%!TEX root = ./main.tex

%\section{Deferred Proofs from Section \ref{sec:preliminaries}} \label{sec:preliminariesAppendix}

%!TEX root = ./main.tex

\section{Deferred Proofs from Section \ref{sec:sepGaussian}}
\label{sec:concAppendix}
This section contains deferred proofs of several concentration inequalities.

\begin{prevproof}{Lemma}{lem:union-bound}
Recall that for any $J \subseteq [N]$, we let $w^J \in \R^N$ be the vector which is given by $w^J_i = \frac{1}{|J|}$ for $i \in J$ and $w_i^J = 0$ otherwise. By convexity, it suffices to show that 
\[\Pr \left[\exists J : |J| = (1 - \eps) N , \mbox{ and } \left \| \sum_{i = 1}^N w_i^J Y_i Y_i^\top - (1 - \eps) I \right\|_2 \geq \delta_1 \right] \leq \tau \; .\]
For any fixed $w^J$ we have
\begin{align*}
\sum_{i = 1}^n w^J_i Y_i Y_i^\top - I &= \frac{1}{(1 - \eps) N} \sum_{i \in J} Y_i Y_i^\top - I \\
&= \frac{1}{(1 - \eps) N} \sum_{i = 1}^N Y_i Y_i^\top - \frac{1}{1 - 2\eps} I \\
&~~~~~- \left( \frac{1}{(1 - \eps) N} \sum_{i \not\in J} Y_i Y_i^\top - \left(\frac{1 }{1 - \eps} - 1 \right) I \right) \; .
\end{align*}
Therefore, by the triangle inequality, we have
\begin{align*}
\left\| \sum_{i = 1}^N w^I_i Y_i Y_i^\top - (1 - \eps) I \right\|_2 &\leq \left\| \frac{1}{(1 - \eps) N} \sum_{i = 1}^N Y_i Y_i^\top - \frac{1}{1 - \eps} I \right\|_2 \\
&~~~~+ \left\| \frac{1}{(1 - \eps) N} \sum_{i \not\in J} Y_i Y_i^\top - \left(\frac{1 }{1 - \eps} - 1 \right) I \right\|_2 \; .
\end{align*}

Observe that the first term on the right hand side does not depend on the choice of $J$.
Let $E_1$ denote the event that 
\begin{equation}
\label{eq:mean-conc-e1}
\left\| \frac{1}{(1 - \eps) N} \sum_{i = 1}^N Y_i Y_i^\top - \frac{1}{1 - \eps} I \right\|_2 \leq \delta_1 \; .
\end{equation}
By Lemma~\ref{lem:vershynin}, this happens with probability $1 - \tau$ so long as
\[N = \Omega \left( \frac{d + \log (1 / \tau)}{\delta_1^2} \right) \; .\]
For any $J \subset [n]$ so that $|J| = (1 - \eps) n$, let $E_2 (J)$ denote the event that 
\[
\left\| \frac{1}{(1 - \eps) N} \sum_{i \not\in J} Y_i Y_i^\top - \left(\frac{1 }{1 - \eps} - 1 \right) I \right\|_2 \leq \delta_1 \; .
\]
Fix any such $J$.
By multiplying both sides by $\rho = (1 - \eps) / \ve$, the event $E_2 (J)$ is equivalent to the event that
\[\left\| \frac{1}{\ve N} \sum_{i \not\in J} Y_i Y_i^\top - I \right\|_2 > \rho \delta_1 \; .\]

Let $A,B$ be as in Lemma \ref{lem:vershynin}.
Observe that $\rho \delta_1 = \Omega (\log 1 / \eps) \geq 1$ for $\eps$ sufficiently small.
Then, by Lemma~\ref{lem:vershynin}, we have that for any fixed $J$,
\begin{align*}
\Pr \left[ \left\| \frac{1}{\ve N} \sum_{i \not\in J} Y_i Y_i^\top - I \right\|_2 > \rho \delta_1 \right] \leq 4 \exp \left(A d - B \eps N \rho \delta_1 \right) \; .
\end{align*}
Let $H(\eps)$ denote the binary entropy function.
We now have
\begin{align*}
&\Pr \left[ \left( \bigcap_{J: |J| = (1 - \eps) N} E_2 (J) \right)^c ~ \right] \\
&~~~~\stackrel{(a)}{\leq} 4 \exp \left(\log \binom{N}{\eps N} + A d - B \eps N \rho \delta_1 \right) \\
&~~~~\stackrel{(b)}{\leq} 4 \exp \left(N H(\eps) + A d - B \eps N \rho \delta_1 \right) \\
&~~~~\stackrel{(c)}{\leq} 4 \exp \left( \eps N ( O(\log 1 / \eps) - N \rho ) + Ad \right) \\
&~~~~\stackrel{(d)}{\leq} 4 \exp \left( -\eps N / 2 + Ad \right) \stackrel{(e)}{\leq} O(\tau) \; ,
\end{align*}
as claimed, where (a) follows by a union bound over all sets $J$ of size $(1 - \eps) N$, (b) follows from the bound $\log \binom{n}{\eps n} \leq \eps H(\eps)$, (c) follows since $H(\eps) = O(\eps \log 1 / \eps)$ as $\eps \to 0$, (d) follows from our choice of $\delta_1$, and (e) follows from our choice of $n$.
This completes the proof.
\end{prevproof}

\begin{prevproof}{Theorem}{thm:fourth-order}
We first recall Isserlis' theorem, which we will require in this proof.
\begin{theorem}[Isserlis' theorem]
Let $a_1, \ldots, a_k \in \R^d$ be fixed vectors.
Then if $X \sim \normal (0, I)$, we have
\[\E \left[ \prod_{i = 1}^k \langle a_i, X \rangle \right] = \sum \prod \langle a_i, a_j \rangle \; ,\]
where the $\sum \prod$ is over all matchings of $\{1, \ldots, k\}$.
\end{theorem}

Let $v = A^\flat \in \Ssym$.
We will show that
\[\langle v, M v \rangle = 2 v^T \left(\Sigma^{\otimes 2} \right) v + v^T \left(\Sigma^\flat \right) \left( \Sigma^\flat \right)^T v \; . \]
Since $M$ is a symmetric operator on $\R^{d^2}$, its quadratic form uniquely identifies it and this suffices to prove the claim.

Since $A$ is symmetric, it has a eigenvalue expansion $A = \sum_{i = 1}^d \lambda_i u_i u_i^T$, which immediately implies that $v = \sum_{i = 1}^d \lambda_i u_i \otimes u_i$.
Let $X \sim \normal (0, \Sigma)$.
We compute the quadratic form:
\begin{align*}
\langle v, M v \rangle &= \sum_{i, j = 1}^d \lambda_i \lambda_j \langle u_i \otimes u_i, \E [(X \otimes X) (X \otimes X)^T] u_j \otimes u_j \rangle \\
&=  \sum_{i, j = 1}^d \lambda_i \lambda_j \E \left[ \langle u_i \otimes u_i, (X \otimes X) (X \otimes X)^T u_j \otimes u_j \rangle \right] \\
&= \sum_{i, j = 1}^d \lambda_i \lambda_j \E \left[ \langle u_i, X \rangle^2 \langle u_j, X \rangle^2 \right] \\
&= \sum_{i, j = 1}^d \lambda_i \lambda_j  \E \left[ \langle B^T u_i, Y \rangle^2 \langle B^T u_j, Y \rangle^2 \right] \\
&= \sum_{i, j = 1}^d \lambda_i \lambda_j \left( \langle B^T u_i, B^T u_i \rangle \langle B^T u_j, B^T u_j \rangle + 2 \langle B^T u_i, B^T u_j\rangle^2 \right) \; ,
\end{align*}
where the last line follows by invoking Isserlis's theorem.
We now manage both sums individually.
We have
\begin{align*}
\sum_{i, j = 1}^d \lambda_i \lambda_j  \langle B^T u_i, B^T u_i \rangle \langle B^T u_j, B^T u_j \rangle &= \left( \sum_{i = 1}^d \lambda_i u_i^T \Sigma u_i \right)^2 \\
&= \left( \sum_{i = 1}^d \lambda_i \left( u_i \otimes u_i \right)^T \left( \Sigma^\flat \right) \right)^2 \\
&= v^T  \left(\Sigma^\flat \right) \left( \Sigma^\flat \right)^T v \; ,
\end{align*}
and
\begin{align*}
\sum_{i, j = 1}^d \lambda_i \lambda_j \langle B^T u_i, B^T u_j\rangle^2 &= \sum_{i, j} \lambda_i \lambda_j \langle (B^T u_i)^{\otimes 2}, (B^T u_j)^{\otimes 2} \rangle \\
&= \sum_{i, j = 1}^d \lambda_i \lambda_j \langle (B^T \otimes B^T) u_i \otimes u_i, (B^T \otimes B^T) u_j \otimes u_j \rangle \\
&= \sum_{i, j = 1}^d \lambda_i \lambda_j (u_i \otimes u_i) \Sigma^{\otimes 2} (u_j \otimes u_j) \\
&= v^T \Sigma^{\otimes 2} v \; .
\end{align*}
\end{prevproof}

\begin{prevproof}{Corollary}{cor:unknown-covariance-deviation}
Let $\mathfrak{S}_m = \{S \subseteq [N]: |S| = m \}$ denote the set of subsets of $[N]$ of size $m$.
The same Bernstein-style analysis as in the proof of Lemma~\ref{lem:union-bound} yields that there exist universal constants $A, B$ so that:
\begin{align*}
&\Pr \left[\exists T \in \mathfrak{S}_m: \left\| \frac{1}{m} \sum_{i \in T} X_i X_I^\top - I \right \|_F \geq O \left( \delta_2 \frac{N}{m} \right) \right] \\
&~~~~\leq 4 \exp \left( \log \binom{N}{m} + A d^2 - B \delta_2 N \right) \; .
\end{align*}
Thus, union bounding over all $m \in \{1, \ldots, \eps N\}$ yields that
\begin{align*}
&\Pr \left[\exists T~\mbox{s.t.} |T| \leq \eps N:  \left\| \frac{1}{|T|} \sum_{i \in T} X_i X_I^\top - I \right \|_F \geq O \left( \delta_2 \frac{N}{|T|} \right) \right] \\
&~~~~\leq 4 \exp \left( \log (\eps N) + \log \binom{N}{\eps N} + A d^2 - B \delta_2 n \right) \leq \tau \; ,
\end{align*}
by the same manipulations as in the proof of Lemma~\ref{lem:union-bound}.
\end{prevproof}

\subsection{Proof of Theorem \ref{thm:fourth-moment-union-bound}}
This follows immediately from Lemmas \ref{GoodSamplesLemma} and \ref{ACloseLem}.

\section{Deferred Proofs from Section \ref{sec:filterGaussian}} \label{sec:filterGaussianAppendix}

\subsection{Proof of Lemma \ref{lem:random-good-gaussian-mean}}
\begin{prevproof}{Lemma}{lem:random-good-gaussian-mean}
Let $N = \Omega( (d/\eps^2) \poly\log(d/\eps\tau))$ be the number of samples drawn from $G$.
For (i), the probability that a coordinate of a sample is at least $\sqrt{2\nu\log(Nd/3\tau)}$ 
is at most $\tau/3dN$ by Fact \ref{fact:tail-bound}. By a union bound, 
the probability that all coordinates of all samples are smaller than $\sqrt{2\nu\log(Nd/3\tau)}$ 
is at least $1-\tau/3$. In this case, $\|x\|_2 \leq \sqrt{ 2\nu d \log(Nd/3\tau)} = O(\sqrt{d \nu \log(N\nu/\tau)})$.

%By Fact \ref{fact:gaussians-are-not-large}, the first property holds with probability at least $1-\tau/3$. 
After translating by $\mu^G$, we note that (iii) follows 
immediately from Lemma~\ref{lem:mean1} and (iv) follows from Theorem 5.50 of \cite{Vershynin}, 
as long as $N =\Omega(\nu^4 d\log(1/\tau)/\eps^2)$, with probability at least $1-\tau/3$. 
It remains to show that, conditioned on (i), (ii) holds with probability at least $1-\tau/3$.

To simplify some expressions, let $\delta := \eps/(\log(d \log d/\eps\tau))$ and $R=C\sqrt{d\log(|S|/\tau)}$.
We need to show that for all unit vectors $v$ and all $0 \leq T \leq R$ that
\begin{equation} \label{eq:grail}
\left| \Pr_{X\in_u S}[|v \cdot (X-\mu^G)| > T] - \Pr_{X\sim G}[|v \cdot (X-\mu^G) > T \ge 0] \right| \leq \frac{\delta}{T^2} \;.
\end{equation}

Firstly, we show that for all unit vectors $v$ and $T >0$
$$
\left| \Pr_{X\in_u S}[|v \cdot (X-\mu^G)| > T] - \Pr_{X\sim G}[|v \cdot (X-\mu^G)| > T \ge 0] \right| \leq \frac{\delta}{10 \nu \ln(1/\delta)}
$$
with probability at least $1-\tau/6$. Since the VC-dimension of the set of all halfspaces is $d+1$, 
this follows from the VC inequality~\cite{DL:01}, 
since we have more than $\Omega(d/(\delta/(10 \nu \log(1/\delta))^2)$ samples. 
We thus only need to consider the case when $T \geq \sqrt{10 \nu \ln(1/\delta)}$.

\begin{lemma} For any fixed unit vector $v$ and $T > \sqrt{10 \nu\ln(1/\delta)}$, 
except with probability $\exp(-N\delta/(6C\nu))$, we have that
$$\Pr_{X\in_u S}[|v \cdot (X-\mu^G)| > T] \leq \frac{\delta}{ C T^2} \;,$$
where $C=8$.
\end{lemma}
\begin{proof}
Let $E$ be the event that $|v \cdot (X-\mu^G)| > T$. Since $G$ is sub-gaussian, 
Fact \ref{fact:tail-bound} yields that $\Pr_G[E]= \Pr_{Y \sim G}[|v \cdot (X-\mu^G)| > T] \leq \exp(-T^2/(2\nu))$. 
Note that, thanks to our assumption on $T$, we have that $T \leq \exp(T^2/(4\nu))/2C$, 
and therefore $T^2\Pr_G[E] \leq \exp(-T^2/(4\nu))/2C \leq \delta/2C$.

Consider $\E_S[\exp(t^2/ (3\nu) \cdot  N \Pr_S[E])]$.
Each individual sample $X_i$ for $1 \leq i \leq N$, is an independent copy of $Y \sim G$, and hence:
\begin{align*}
\E_S\left[\exp \left( \frac{T^2}{3\nu} \cdot  N \Pr_{S}[E] \right) \right] 
& = \E_S \left[ \exp \left( \frac{T^2}{3\nu} \right) \cdot \sum_{i=1}^n 1_{X_i \in E}) \right] \\
&= \prod_{i=1}^N \E_{X_i}\left[ \exp \left( \frac{T^2}{3\nu} \right) \cdot \sum_{i=1}^n 1_{X_i \in E})\right] \\
& = \left(\exp \left( \frac{T^2}{3\nu} \right) \Pr_G [G] + 1 \right)^N \\
&\stackrel{(a)}{\leq}  \left( \exp \left( \frac{T^2}{6 \nu} \right) + 1 \right)^N \\
&\stackrel{(b)}{\leq} (1 + \delta^{5/3})^N \\
&\stackrel{(c)}{\leq} \exp (N \delta^{5/3}) \; ,
\end{align*}
where (a) follows from sub-gaussianity, (b) follows from our choice of $T$, and (c) comes from the fact that $1 + x \leq e^x$ for all $x$.

Hence, by Markov's inequality, we have
\begin{align*}
\Pr \left[ \Pr_S [E] \geq \frac{\delta}{C T^2} \right] &\leq \exp \left( N \delta^{5/3} - \frac{\delta N}{3C} \right) \\
&= \exp (N \delta (\delta^{2/3} - 1 / (3C))) \; .
\end{align*}
Thus, if $\delta$ is a sufficiently small constant and $C$ is sufficiently large, this yields the desired bound.
% Note that if  $\Pr_S[E] > \frac{\delta}{C\nu T^2}$, 
% then $\exp(T^2/3\nu \cdot  N \Pr_{S}[E]) =\exp(N\delta/3C\nu)$. 
% By Markov's inequality, this happens with probability at most $\exp(-N\delta/6C\nu)$.
\end{proof}

Now let $\mathcal{C}$ be a $1/2$-cover in Euclidean distance 
for the set of unit vectors of size $2^{O(d)}$. 
By a union bound, for all $v' \in \mathcal{C}$ and $T'$ a power of 2 between $\sqrt{4\nu\ln(1/\delta)}$ and $R$, we have that
$$\Pr_{X\in_u S}[|v' \cdot (X-\mu^G)| > T'] \leq \frac{\delta}{ 8 T^2} $$
except with probability 
$$2^{O(d)} \log(R)\exp(-N\delta/6C\nu) = \exp\left(O(d) + \log \log R -N\delta/6C\nu\right) \leq \tau/6 \;.$$
However, for any unit vector $v$ and $\sqrt{4\nu\ln(1/\delta)} \leq T \leq R$, there is a $v' \in \mathcal{C}$ 
and such a $T'$ such that for all $x \in \R^d$, we have $|v \cdot (X-\mu^G)| \geq |v' \cdot (X-\mu^G)|/2$, 
and so $|v' \cdot (X-\mu^G)| > 2T'$ implies  $|v' \cdot (X-\mu^G)| > T.$

Then, by a union bound, (\ref{eq:grail}) holds simultaneously for all unit vectors $v$ and all $0 \leq T \leq R$,  
with probability a least $1-\tau/3$. This completes the proof.
\end{prevproof}

\subsection{Proof of Lemma \ref{lem:evenp}}
\begin{prevproof}{Lemma}{lem:evenp}
Note that an even polynomial has no degree-$1$ terms. 
Thus, we may write $p(x)= \sum_i p_{i,i} x_i^2 + \sum_{i > j} p_{i,j} x_i x_j + p_o$. 
Taking $(P_2)_{i,i} = p_{i,i}$ and $(P_2')_{i,j} = (P_2')_{j,i} = \frac12 p_{i,j}$, for $i>j$, 
gives that $p(x)=x^T P_2' x + p_0$. 
Taking $P_2 = \Sigma^{1/2} P_2' \Sigma^{1/2}$, 
we have $p(x)=(\Sigma^{-1/2}x)^T P_2 (\Sigma^{-1/2}x) + p_0$, for a $d \times d$ symmetric matrix $P_2$ and $p_0 \in \R$.

Let $P_2=U^T \Lambda U$, where $U$ is orthogonal and $\Lambda$ is diagonal 
be an eigen-decomposition of the symmetric matrix $P_2$. 
Then, $p(x)=(U\Sigma^{-1/2}x)^T P_2 (U\Sigma^{-1/2}x)$. 
Let $X \sim G$ and $Y=U \Sigma^{-1/2} X$. 
Then, $Y \sim \normalpdf(0,I)$ and $p(X) = \sum_i \lambda_i Y_i^2 + p_0$ 
for independent Gaussians $Y_i$. Thus, $p(X)$ follows a generalized $\chi^2$-distribution.

Thus, we have
$$\E[p(X)]=\E\left[\sum_i \lambda_i Y_i^2 + p_0\right] = p_0+\sum_i \lambda_i=p_0+\tr(P_2) \;,$$
and
$$\Var[p(X)]= \Var\left[\sum_i \lambda_i Y_i^2 + p_0\right] = \sum_i \lambda_i^2 = \|P_F\|_2 \;.$$

%Where should we take the concentration inequality from? This is from Laurent and Massart which I think has slightly better constants than just applying Bernstein's inequality for sub-exponential random variables:
\begin{lemma}[cf. Lemma 1 from \cite{LaurentMassart}]
Let $Z= \sum_i a_i Y_i^2$, where $Y_i$ are independent random variables 
distributed as $\normalpdf(0,1)$. 
Let $a$ be the vector with coordinates $a_i$. Then, 
 $$\Pr(Z \geq 2\|a\|_2\sqrt{x} + 2 \|a\|_\infty x) \leq \exp(-x) \;.$$
\end{lemma}

We thus have:
$$
\Pr\left(\sum_i \lambda_i (Y_i^2-1) > 2\sqrt{(\sum_i \lambda_i^2) t} + 2(\max_i \lambda_i) t \right) \leq e^{-t} \;.
$$
Noting that $\tr(P_2)=\sum_i \lambda_i$,$\sum_i \lambda_i^2=\|P_2\|_F$ 
and $\max_i \lambda_i =\|P_2\|_2 \leq \|P_2\|$, for $\mu_p = \E[p(X)]$ we have:
$$
\Pr(p(X) - \mu_p >  2\|P_2\|_F(\sqrt{ t} + t)) \leq e^{-t} \;.
$$
Noting that $2\sqrt{a}=1+a-(1-\sqrt{a})^2 \leq 1+a$ for $a > 0$, we have
$$
\Pr(p(X)-\mu_p >  \|P_2\|_F(3t+1)) \leq e^{-t} \;.
$$
Applying this for $-p(x)$ instead of $p(x)$ and putting these together, we get
$$
\Pr(|p(X)-\mu_p| >  \|P_2\|_F(3t+1)) \leq 2e^{-t} \;.
$$
Substituting $t=T/3\|P_2\|_F-1/3$, and $ 2 \|P_2\|_F^2 = \Var_{X \sim G}(p(X))$ gives:
$$\Pr(|p(X) - \E_{X \sim G}[p(X)]| \geq T) \leq 2 e^{1/3-2T/3\Var_{X \sim G}[p(X)]} \; . $$

The final property is a consequence of the following anti-concentration inequality:

\begin{theorem}[\cite{CW01}] 
Let $p:\R^d \rightarrow \R$ be a degree-$d$ polynomial. Then, for $X \sim \normalpdf(0,I)$, we have
$$\Pr(|p(X)| \leq \eps \sqrt{\E[p(X)^2]} \leq O(d \eps^{1/d}) \;.$$
\end{theorem}
This completes the proof.
\end{prevproof}

\subsection{Proof of Lemma \ref{GoodSamplesLemma}}
\begin{prevproof}{Lemma}{GoodSamplesLemma}
\new{Firstly, we note that it suffices to prove this for the case $\Sigma=I$, 
since for $X \sim \normal (0,\Sigma)$, $Y=\Sigma^{-1/2} X$ is distributed as $\normal (0,I)$, 
and all the conditions transform to those for $G=\normal (0,I)$ under this transformation.}

Condition \ref{farPoints} follows by standard concentration bounds on $\|x\|_2^2$.
Condition \ref{covariance} follows by estimating the entry-wise error between $\Cov(S)$ and $I$. 
\new{These two conditions hold by Lemma \ref{lem:random-good-gaussian-mean}, 
since they follow from (i), (iii), and (iv) of $(\eps,\tau)$ goodness in the sense of Definition \ref{def:good-set}.}

Condition \ref{cocovariance} is slightly more involved. 
Let $\{p_i\}$ be an orthonormal basis for the set of even, degree-$2$, mean-$0$ polynomials with respect to $G$. 
Define the matrix $M_{i,j} = \E_{x\in_u S}[p_i(x)p_j(x)]-\delta_{i,j}$. This condition is equivalent to $\|M\|_2 = O(\eps)$. 
Thus, it suffices to show that for every $v$ with $\|v\|_2 = 1$ that $v^TMv = O(\eps)$. 
It actually suffices to consider a cover of such $v$'s. Note that this cover will be of size $2^{O(d^2)}$. 
For each $v$, let $p_v = \sum_i v_i p_i$. We need to show that $\Var(p_v(S)) = 1 +O(\eps)$. 
We can show this happens with probability $1-\new{\tau} 2^{-\Omega(d^2)}$, and thus it holds for all $v$ in our cover by a union bound.

Condition \ref{tails} is substantially the most difficult of these conditions to prove. 
Naively, we would want to find a cover of all possible $p$ and all possible $T$, 
and bound the probability that the desired condition fails. Unfortunately, the best a priori bound on $\Pr(|p(G)| > T)$ 
are on the order of $\exp(-T)$. As our cover would need to be of size $2^{d^2}$ or so, 
to make this work with $T=d$, we would require on the order of $d^3$ samples in order to make this argument work.

However, we will note that this argument is sufficient to cover the case of $T<10\log(1/\eps)\log^2(d/\eps)$.

Fortunately, most such polynomials $p$ satisfy much better tail bounds. 
Note that any even, mean zero polynomial $p$ can be written in the form $p(x) = x^T A x - \tr(A)$ for some matrix $A$. 
We call $A$ the associated matrix to $p$. We note by the Hanson-Wright inequality that $\Pr(|p(G)| > T) = \exp(-\Omega(\min((T/\|A\|_F)^2,T/\|A\|_2))).$ 
Therefore, the tail bounds above are only as bad as described when $A$ has a single large eigenvalue. 
To take advantage of this, we will need to break $p$ into parts based on the size of its eigenvalues. We begin with a definition:

\begin{definition}
Let $\mathcal{P}_{k}$ be the set of even, mean-$0$, degree-$2$ polynomials, such that the associated matrix $A$ satisfies:
\begin{enumerate}
\item $\rank(A)\leq k$
\item $\|A\|_2 \leq 1/\sqrt{k}$.
\end{enumerate}
\end{definition}
Note that for $p\in \mathcal{P}_k$ that $|p(x)| \leq |x|^2/\sqrt{k} + \sqrt{k}$. 

Importantly, any polynomial can be written in terms of these sets.
\begin{lemma}
Let $p$ be an even, degree-$2$ polynomial with $\E[p(G)] = 0, \Var(p(G))=1$. 
Then if $t=\lfloor \log_2(d) \rfloor$, it is possible to write $p=2(p_1+p_2+\ldots+p_{2^t}+p_d)$ where $p_k\in \mathcal{P}_k$.
\end{lemma}
\begin{proof}
Let $A$ be the associated matrix to $p$. 
Note that $\|A\|_F = \Var{p} = 1$. Let $A_{k}$ be the matrix corresponding to the top $k$ eigenvalues of $A$. 
We now let $p_1$ be the polynomial associated to $A_1/2$, $p_2$ be associated to $(A_2-A_1)/2$, 
$p_4$ be associated to $(A_4-A_2)/2$, and so on. It is clear that $p=2(p_1+p_2+\ldots+p_{2^t}+p_d)$. 
It is also clear that the matrix associated to $p_k$ has rank at most $k$. 
If the matrix associated to $p_k$ had an eigenvalue more than $1/\sqrt{k}$, 
it would need to be the case that the $k/2^{nd}$ largest eigenvalue of $A$ had size at least $2/\sqrt{k}$. 
This is impossible since the sum of the squares of the eigenvalues of $A$ is at most $1$.

This completes our proof.
\end{proof}

We will also need covers of each of these sets $\mathcal{P}_k$. 
\new{We will assume that condition (\ref{farPoints}) holds, 
i.e., that $\|x\|_2 \leq \sqrt{R}$, 
where $R=O(d\log(d/\eps\tau))$. Under this condition, 
$p(x)$ cannot be too large and this affects how small a variance polynomial we can ignore.}
\begin{lemma}
For each $k$, there exists a set $\mathcal{C}_k\subset \mathcal{P}_k$ such that
\begin{enumerate}
\item For each $p\in \mathcal{P}_k$ there exists a $q\in \mathcal{C}_k$ such that $\Var(p(G)-q(G)) \leq 1/R^2d^2$.
\item $|\mathcal{C}_k| = 2^{O(dk\log R)}.$
\end{enumerate}
\end{lemma}
\begin{proof}
We note that any such $p$ is associated to a matrix $A$ of the form $A = \sum_{i=1}^k \lambda_i v_i v_i^T$, 
for $\lambda_i \in [0,1/\sqrt{k}]$ and $v_i$ orthonormal. It suffices to let $q$ correspond to the matrix 
$A' = \sum_{i=1}^k \mu_i w_i w_i^T$ for with $|\lambda_i -\mu_i| < 1/R^2d^3$ and $|v_i-w_i| < 1/R^2d^3$ for all $i$. 
It is easy to let $\mu_i$ and $w_i$ range over covers of the interval and the sphere with appropriate errors. 
This gives a set of possible $q$'s of size $2^{O(dk\log R)}$ as desired. 
Unfortunately, some of these $q$ will not be in $\mathcal{P}_k$ as they will have eigenvalues that are too large. 
However, this is easily fixed by replacing each such $q$ by the closest element of $\mathcal{P}_k$. 
This completes our proof.
\end{proof}

We next will show that these covers are sufficient to express any polynomial.
\begin{lemma}
Let $p$ be an even degree-$2$ polynomial with $\E[p(G)]=0$ and $\Var(p(G))=1$. 
It is possible to write $p$ as a sum of $O(\log(d))$ elements of some $\mathcal{C}_k$ plus another polynomial of variance at most $O(1/R^2)$.
\end{lemma}
\begin{proof}
Combining the above two lemmata we have that any such $p$ can be written as
$$
p = (q_1 + p_1) + (q_2 + p_2) + \ldots (q_{2^t}+p_{2^t}) + (q_d+p_d) = q_1+q_2+\ldots+q^{2^t}+q^d + p' \;,
$$
where $q_k$ above is in $\mathcal{C}_k$ and $\Var[p_k(G)] < 1/R^2d^2$. 
Thus, $p'=p_1+p_2+\ldots+p_{2^t}+p_d$ has $\Var[p'(G)] \leq O(1/R^2)$. 
This completes the proof.
\end{proof}
The key observation now is that if $|p(x)| \geq T$ for $\|x\|_2 \leq \sqrt{d/\eps}$, 
then writing $p=q_1+q_2+q_4+\ldots+q_d+p'$ as above, it must be the case that 
$|q_k(x)| > (T-1)/(2\log(d))$ for some $k$. Therefore, to prove our main result, 
it suffices to show that, with high probability over the choice of $S$, 
for any $T\geq 10\log(1/\eps)\log^2(d/\eps)$ and any $q\in \mathcal{C}_k$ 
for some $k$, 
that $\Pr_{x\in_u S}(|q(x)| > T/(2\log(d))) < \eps / (2 T^2 \log^2(T) \log(d))$. 
Equivalently, it suffices to show that for $T\geq 10 \log(1/\eps)\log(d/\eps)$ it holds 
$\Pr_{x\in_u S}(|q(x)| > T/(2\log(d))) < \eps / (2 T^2 \log^2(T) \log^2(d))$. 
Note that this holds automatically for $T>R$, as $p(x)$ cannot possibly be that large 
for $\|x\|_2 \leq \sqrt{R}$. Furthermore, note that losing a constant factor in the probability, 
it suffices to show this only for $T$ a power of $2$.

Therefore, it suffices to show for every $k\leq d$, every $q\in \mathcal{C}_k$ 
and every $R/\sqrt{k} \gg T \gg \log(1/\eps)\log R$ that 
with probability at least $1-\new{\tau} 2^{-\Omega(dk\log R)}$ over the choice of $S$ we have that 
$\Pr_{x\in_u S}(|q(x)|> T) \ll \eps/(T^2 \log^4(R))$. However, by the Hanson-Wright inequality, 
we have that 
$$\Pr(|q(G)| > T) = \exp(-\Omega(\min(T^2,T\sqrt{k}))) < (\eps/(T^2 \log^4 R))^2 \;.$$ 
Therefore, by Chernoff bounds, the probability that more than a $\eps/(T^2 \log^4 R)$-fraction of the elements of $S$ satisfy this property is at most
\begin{align*}
\exp(-\Omega(\min(T^2,T\sqrt{k}))|S|\eps/(T^2 \log^4 R)) & = \exp(-\Omega(|S|\eps/(\log^4 R)\min(1,\sqrt{k}/T)))\\
& \leq \exp(-\Omega(|S|k\eps^2/R(\log^4 R)))\\
& \leq \exp(-\Omega(|S|k\eps/d(\log(d/\eps\tau))(\log^4(d/\log (1/\eps \tau)))))\\
& \leq \tau \exp(-\Omega(dk\log(d/\eps))) \;,
\end{align*}
as desired.

This completes our proof.

\end{prevproof}

%!TEX root = ./main.tex

\section{Deferred Proofs from Section \ref{sec:sepGMM}}
\label{sec:sepGMMAppendix}

\begin{prevproof}{Theorem}{thm:naive-cluster}
The first two properties follow directly from (\ref{eqn:gmmsepconds1}).
We now show the third property.
Suppose this does not happen, that is, there are $j, j'$ such that $\ell = \ell(j) = \ell (j')$ such that $\| \mu_j - \mu_{j'} \|_2^2 \geq \Omega (d k \log k / \ve)$.
That means that by (\ref{eqn:gmmsepconds1}) there is some sequence of clusters $S_1, \ldots, S_t$ such that $S_i \cap S_{i + 1} \neq \emptyset$ for each $i$, $|S_i| \geq 4 \ve N$ for each $i$, and moreover, there is a $X_i \in S_1$ such that $\| X_i - \mu_1 \|_2^2 \leq O(d \log k / \ve)$ and an $X_{i'} \in S_t$ such that $\| X_{i'} - \mu_2 \|_2^2 \leq O(d \log k / \ve)$.
But by (\ref{eqn:gmmsepconds1}), we know that each $S_i$ contains an point $X_{i''}$ such that $\| X_{i''} - \mu_{r_i} \|_2^2 \leq O(d \log k / \ve)$.
In particular, by the triangle inequality, this means that if $\| \mu_{r_i} - \mu_{r_{i + 1}} \|_2^2 \leq O(d \log k / \ve)$ for all $i = 1, \ldots, t - 1$, and we can set $\mu_{r_1} = \mu_j$ and $\mu_{r_t} = \mu_{j'}$.

Construct an auxiliary graph on $k$ vertices, where we put an edge between nodes $r_i$ and $r_{i + 1}$.
By the above, there must be a path from $j$ to $j'$ in this graph.
Since this graph has $k$ nodes, there must be a path of length at most $k$ from $j$ to $j'$; moreover, by the above, we know that this implies that $\| \mu_j - \mu_{j'} \|_2^2 \leq O(k d \log k / \ve)$. 

Finally, the fourth property follows from the same argument as the proof of the third.
\end{prevproof}

\begin{prevproof}{Lemma}{lem:gmmtklb}
  Let $C = \sum_{i = 1}^N w_i (X_i - \m)(X_i - \m)^T - I$.
  Let $v$ be the top eigenvector of $$\sum_{i=1}^N w_i (X_i -\m) (X_i - \m)^T - I - \sum_{j \in [k]} \a_j (\m_j - \m) (\m_j - \m)^T$$
  Observe that by (\ref{eqn:gmmsepconds2}), we have
  \begin{align*}  
  \sum_{i=1}^N w_i (X_i -\m) (X_i - \m)^T &\succeq \sum_{i \in \Sgood} w_i (X_i - \mu) (X_i - \mu)^T \\
  &\succeq w_g ( I + Q) - f(k, \gamma, \delta_1) I \\
  &\succeq ( 1 - \ve) ( I + Q) - f(k, \gamma, \delta_1) I \; , 
  \end{align*}
  and so in particular
  \[
  \sum_{i=1}^N w_i (X_i -\m) (X_i - \m)^T - ( I + Q) \succeq -\ve ( I + Q) - f(k, \gamma, \delta_1) I \; .
  \]
  Therefore, for any unit vector $u \in \R^d$, we must have
  \[
  u^T \left( \sum_{i=1}^N w_i (X_i -\m) (X_i - \m)^T - ( I + Q) \right) u \geq -\ve u^T ( I + Q) u^T - f(k, \gamma, \delta_1) \geq - \frac{c}{2} h(k, \gamma, \delta) \; .
  \]
  In particular, since $\left| v^T  \left( \sum_{i=1}^N w_i (X_i -\m) (X_i - \m)^T - ( I + Q) \right) v \right| \geq c k h(k, \gamma, \delta)$, we must have 
  \[
  v^T \left( \sum_{i=1}^N w_i (X_i -\m) (X_i - \m)^T - ( I + Q) \right) v > 0,
  \] 
  and hence
  \[
  v^T  \left( \sum_{i=1}^N w_i (X_i -\m) (X_i - \m)^T - ( I + Q) \right) v \geq c k h(k, \gamma, \delta) \; .
  \]
	
  Let $U = [v, u_1, \ldots, u_{d - 1}]$ be an $d \times k$  matrix with orthonormal columns, where the columns span the set of vectors $\left\{(\m_j - \m)\ :\ j \in [k] \right\} \cup \{v\}$.
  We note the rank of this set is at most $k$ due to the definition of $\m$. 

  Using the dual characterization of the Schatten top-$k$ norm, we have that 
  $$\left\| C \right\|_{T_k} = \max_{X \in \mathbb{R}^{d \times k}} \Tr (X^T C X) \geq \Tr (U^T C U).  $$
Observe that since $\mbox{span} (Q) \subseteq \mbox{span}(U)$, we have
\begin{align*}
\| C \|_{T_k} \geq \Tr \left( U^T   C U \right) &= \Tr \left( U^T  \left( \sum_{i=1}^N w_i (X_i -\m) (X_i - \m)^T - ( I + Q) \right) U \right) + \| Q \|_{T_k}  \\
&= \Tr \left( U^T (C - Q) U \right) + \sum_{j \in [k]} \gamma_j \\
&= v^T (C - Q) v + \sum_{i = 1}^{k - 1}u_i^T (C - Q) u_i + \sum_{j \in [k]} \gamma_j \\
&\geq c k h(k, \gamma, \delta) - (k - 1) \frac{c}{2} h(k, \gamma, \delta) + \sum_{j \in [k]} \gamma_j  \\
&\geq \frac{c}{2} k h(k, \gamma, \delta) + \sum_{j \in [k]} \gamma_j \; ,
 \end{align*}
 as claimed.

%  First, note that $v^T C v \geq ckh(k,\g,\d)$, by the definition of $\mathcal{C}_{ckh(k,\g,\d)}$.
%
%  Next, consider the vector $u_j$, which is the unit vector in the direction of $\m_j - \m$.
%  We know that
%  \begin{align*}
%    u_j^T\left(\sum_{i = 1}^N w_i (X_i - \m)(X_i - \m)^T - I\right)u_j &\geq u_j^T\left(\sum_{i \in \Sgood } w_i (X_i - \m)(X_i - \m)^T - w_g I\right)u_j\\
%                                                                   &\geq u_j^T\left(w_gQ - f(k,\g,\d_1)I  \right)u_j \\
%                                                                   &\geq \g_j - 4\ve \g_j - f(k,\g,\d) \\
%                                                                   &\geq \g_j - \frac{c h(k,\g,\d)}{4}
%  \end{align*}
%  where the second inequality follows by (\ref{eqn:gmmsepconds2}).
%
%  Summing over all these terms gives
%  $$\left\|C \right\|_{T_k} \geq \sum_{j \in [k]} \g_j + \frac{3ck h(k,\g,\d)}{4},$$
%  as desired.

\end{prevproof}

\begin{prevproof}{Lemma}{lem:gmmkey}
  By Fact~\ref{fact:weights} and (\ref{eqn:gmmsepconds3}) we have $\|\sum_{i = \Sgood} \frac{w_i}{w_g}X_i - \m\|_2 \leq k^{1/2}\d_2$.
  Now, by the triangle inequality, we have
  $$\left \| \sum_{i \in \Sbad} w_i (X_i - \mu) \right \|_2 \geq \|\Delta\|_2 - \left \| \sum_{i \in \Sgood} w_i (X_i - \mu) - w_g \mu \right \|_2 \geq \Omega (\|\Delta\|_2).$$  
  Using the fact that variance is nonnegative we have
  \[
    \sum_{i \in \Sbad} \frac{w_i}{w_{b}} (X_i - \mu) (X_i - \mu)^T \succeq \left( \sum_{i \in \Sbad} \frac{w_i}{w_{b}} \left( X_i - \mu \right) \right) \left( \sum_{i \in \Sbad} \frac{w_i}{w_{b}} \left( X_i - \mu \right) \right)^T \; ,
  \]
  and therefore 
  \[
    \left\| \sum_{i \in \Sbad} w_i (X_i - \mu) (X_i - \mu)^T \right\|_2 \geq \Omega \left( \frac{\| \Delta \|_2^2}{w_b} \right) \geq \Omega \left( \frac{\| \Delta \|_2^2}{\epsilon} \right).
  \]

  On the other hand, 
  \begin{align*}
    \left\| \sum_{i \in \Sgood} w_i (X_i - \mu) (X_i - \mu)^T - I \right\|_2 &\leq \left\| \sum_{i \in \Sgood} w_i (X_i - \mu) (X_i - \mu)^T - w_g I \right\|_2 + w_b \leq f(k, \g, \delta_1) + w_b .
  \end{align*}
  where in the last inequality we have used Fact~\ref{fact:weights} and (\ref{eqn:gmmsepconds2}). 
  Hence altogether this implies that
  \begin{align*}
     \left\| \sum_{i = 1}^N w_i (X_i - \mu) (X_i - \mu)^T - I \right\|_2 &\geq \Omega \left( \frac{\|\Delta\|_2^2}{\ve} \right) -w_b - f(k,\g,\delta_1) \geq \Omega \left( \frac{\|\Delta\|_2^2}{\ve} \right) \; ,
  \end{align*}
  as claimed.
\end{prevproof}

\subsection{Proof of Theorem \ref{thm:gmmsep}}
\label{sec:gmmsepproof}
Once more, let $\Delta = \m - \hat \m$ and expand the formula for $M$:
\begin{align*}
  \sum_{i = 1}^N w_i Y_i Y_i^T - I &= \sum_{i = 1}^N w_i (X_i - \mu + \Delta) (X_i - \mu + \Delta)^T - I \\
                                   &= \sum_{i = 1}^N w_i (X_i - \mu) (X_i - \mu)^T - I + \sum_{i = 1}^N w_i (X_i - \mu) \Delta^T + \Delta \sum_{i = 1}^N  w_i (X_i - \mu)^T + \Delta \Delta^T \\
                                   &= \sum_{i = 1}^N  w_i (X_i - \mu) (X_i - \mu)^T - I - \Delta \Delta^T \; .
\end{align*}

We start by proving completeness.
\begin{claim}
  \label{clm:gmmcomplete}
Suppose that $w = w^*$.
Then $\|M\|_{T_k} \leq \sum_{i \in [k]} \g_j + \frac{ck h(k, \g, \d_1)}{2}$.
\end{claim}
\begin{proof}
  $w^\ast$ are the weights that are uniform on the uncorrupted points.
  Because $\Sbad \leq 2\ve N$, we have that $w^\ast \in S_{N,\ve}$.
  Using (\ref{eqn:gmmsepconds2}), this implies that $w^* \in \mathcal{C}_{ f(k, \g, \d_1)}$.
  By Corollary \ref{cor:gmmclose}, $\|\Delta\|_2 \leq O(\ve \sqrt{\log 1/\ve})$.
  \begin{align*}
    &\hphantom{\leq } \left\| \sum_{i = 1}^m  w^{\ast}_i (X_i - \mu) (X_i - \mu)^T - I - \Delta \Delta^T \right\|_{T_k} \\
    &\leq \left\| \sum_{i = 1}^N  w^{\ast}_i (X_i - \mu) (X_i - \mu)^T - I  - \sum_{j \in [k]} \a_j (\m_j - \m) (\m_j - \m)^T\right\|_{T_k} + \left\|\sum_{j \in [k]} \a_i (\m_j - \m) (\m_j - \m)^T \right\|_{T_k} + \| \Delta \Delta^T \|_2 \\
    &\leq kf(k,\g,\delta_1) + \sum_{j \in [k]} \g_j + O(\ve^2 \log 1/\ve) \\
    &< \sum_{j \in [k]} \g_j + \frac{ck h(k,\g, \d)}{2} \; .
  \end{align*}
\end{proof}
\begin{claim}
  \label{clm:gmmcomplete2}
  Suppose that $w \not \in \mathcal{C}_{ck h(k,\g,\d)}$.
  Then $\|M\|_{T_k} > \sum_{i \in [k]} \g_j + \frac{ck h(k, \g, \d_1)}{2}$.
\end{claim}
\begin{proof}
  We split into two cases.
  In the first case, $\|\Delta\|_2^2 \leq \frac{ck h(k,\g,\d)}{10}$. 
  By Lemma \ref{lem:gmmtklb}, we have that
  $$\left\| \sum_{i=1}^N w_i (X_i - \m)(X_i - \m)^T - I \right\|_{T_k} \geq \sum_{j \in[k]} \g_j + \frac{3ck h(k,\g,\d)}{4}.$$
  By the triangle inequality,
  $$\|M \|_{T_k} \geq \sum_{j \in [k]} \g_j + \frac{3ck h(k,\g,\d)}{4} - \|\Delta\|_2^2 \geq \sum_{i \in [k]} \g_j + \frac{ck h(k, \g, \d)}{2},$$
  as desired.

  In the other case, $\|\Delta\|_2^2 \geq \frac{ck h(k,\g,\d)}{10}$.
  Recall that $Q = \sum_{j \in [k]} \a_j (\m_j - \m) (\m_j -\m)^T$ from (\ref{eq:gmmcorr}). 
  Write $M$ as follows:
  \begin{align*}
    M &=  \sum_{i = 1}^N w_i (X_i - \mu) (X_i - \mu)^T - I - \Delta \Delta^T \\
      &= \left(\sum_{i \in \Sgood} w_i (X_i - \mu) (X_i - \mu)^T - w_g I - w_g Q \right) + w_g Q + \sum_{i \in \Sbad} w_i (X_i - \mu) (X_i - \mu)^T - w_b I - \Delta \Delta^T  \\
  \end{align*}
  Now taking the Schatten top-$k$ norm of $M$, we have
\begin{align}
  &\hphantom{=} \left\|\left(\sum_{i \in \Sgood} w_i (X_i - \mu) (X_i - \mu)^T - w_g I - w_g Q \right) + w_g Q + \sum_{i \in \Sbad} w_i (X_i - \mu) (X_i - \mu)^T - w_b I - \Delta \Delta^T\right\|_{T_k} \nonumber \\
  &\geq \left\| w_g Q + \sum_{i \in \Sbad} w_i (X_i - \mu) (X_i - \mu)^T \right\|_{T_k} - \left\|\sum_{i \in \Sgood} w_i (X_i - \mu) (X_i - \mu)^T - w_g I - w_g Q \right\|_{T_k} - \|w_b I\|_2 - \left\| \Delta \Delta^T \right\|_2\nonumber \\
  &\geq \left\| w_g Q + \sum_{i \in \Sbad} w_i (X_i - \mu) (X_i - \mu)^T \right\|_{T_k} - k f(k,\g,\d_1) - 4\ve - \|\Delta\|_2^2\nonumber \\
  &\geq \left(\sum_{j \in [k]} \g_j - 4\ve k \g\right) + \left\|\sum_{i \in \Sbad} w_i (X_i - \mu) (X_i - \mu)^T \right\|_{T_k} - k f(k,\g,\d_1) - 4\ve - \|\Delta\|_2^2\nonumber \\
  &\geq \sum_{j \in [k]} \g_j + \Omega\left(\frac{\|\Delta\|_2^2}{\ve}\right) \label{eq:gmmbigl}\\
  &\geq \sum_{j \in [k]} \g_j + \frac{ck h(k, \g, \d)}{2}. \nonumber
\end{align}
The first inequality is the triangle inequality, the second is by (\ref{eqn:gmmsepconds2}) and Fact \ref{fact:weights}, the third is because the summed matrices are positive semidefinite, the fourth follows from Lemma \ref{lem:gmmkey}, and the fifth holds for all $c$ sufficiently large.
\end{proof}
By construction, we have that $\ell(w) \geq 0$.
It remains to show that $\ell(w^\ast) < 0$. 
\begin{align*}
&\left\|\frac{1}{|\Sgood|}\sum_{i \in \Sgood} (X_i - \hat \m) (X_i - \hat \m)^T - I \right\|_{T_k}  \\
=&\left\|\frac{1}{|\Sgood|}\sum_{i \in \Sgood} (X_i - \m + \Delta) (X_i - \m + \Delta)^T - I \right\|_{T_k} \\
\leq &\left \| \frac{1}{|\Sgood|}\sum_{i \in \Sgood} (X_i - \m) (X_i - \m)^T - I  - \sum_{j \in [k]} \a_j (\m_j - \m)(\m_j - \m)^T \right\|_{T_k} \\ 
+ &\left\|\sum_{j \in [k]} \a_j (\m_j - \m)(\m_j - m)^T \right\|_{T_k} + 2 \|\Delta\|_2 \left\|\frac{1}{|\Sgood|}\sum_{i \in \Sgood} (X_i - \m ) \right\|_{T_k} + \|\Delta\|_2^2 \\
\leq &k f(k,\g,\d_1) + \sum_{j \in [k]} \g_j + 2 k^{1/2}\d_2 \|\Delta\|_2 + \|\Delta\|_2^2
\end{align*}

Therefore, 
$$\ell(w^*) \leq kf(k, \g, \d) + \sum_{j \in [k]} \g_j  + 2k^{1/2}\d \|\Delta\|_2 + \|\Delta\|_2^2 - \Lambda.$$

If $\|\Delta\|_2^2 \leq \frac{ckh(k,\g,\d)}{10}$, then
$$\ell(w^*) \leq \sum_{j \in [k]} \g_j  + kh(k, \g, \d)  + \frac{2k\d \sqrt{ch(k,\g,\d)}}{\sqrt{10}}+ \frac{ckh(k,\g,\d)}{10}- \Lambda.$$
We wish to show that 
$$\frac{2k\d \sqrt{ch(k,\g,\d)}}{\sqrt{10}} \leq \frac{ckh(k,\g,\d)}{10},$$
or equivalently, that
$$\d \leq \frac{\sqrt{ch(k,\g,\d)}}{2 \sqrt{10}}.$$
But this is true for $c$ sufficiently large, as $\sqrt{h(k,\g,\d)} \geq \sqrt{\d}$.
Therefore,
$$\ell(w^*) \leq \sum_{j \in [k]} \g_j  +  \frac{(c + 5)kh(k,\g,\d)}{5}- \Lambda \leq 0,$$
where the second inequality holds since $\Lambda > \sum_{j \in [k]} \g_j + \frac{ck h(k,\g, \d)}{2}$.

On the other hand, consider when $\|\Delta\|_2^2 \geq \frac{ckh(k,\g,\d)}{10}$.
By (\ref{eq:gmmbigl}), we know that 
$$\Lambda \geq \sum_{j \in [k]} \g_j + \Omega\left(\frac{\|\Delta\|_2^2}{\ve}\right).$$
Then we know
$$\ell(w^*) \leq kf(k, \g, \d)  + 2k^{1/2}\d \|\Delta\|_2 + \|\Delta\|_2^2 - \Omega\left(\frac{\|\Delta\|^2_2}{\ve}\right).$$
The first and third terms are immediately dominated by $\Omega\left(\frac{\|\Delta\|^2_2}{\ve}\right)$, it remains to show that
$$k^{1/2} \d \|\Delta\|_2 = O\left(\frac{\|\Delta\|^2_2}{\ve}\right).$$
Or equivalently,
$k^{1/2} \d \ve = O\left(\|\Delta\|_2\right).$
This follows since $$\|\Delta\|_2 \geq O(\sqrt{h(k,\g,\d)}) \geq O(\sqrt{k \d^2}) = O(k^{1/2} \d \ve)$$
Therefore in this case as well, $\ell(w^*) < 0$, as desired.

%!TEX root = ./main.tex

\section{Deferred Proofs from Section \ref{sec:filterProduct}} \label{sec:filterProductAppendix}
\begin{prevproof}{Lemma}{lem:random-good-dtv}
By Lemma~\ref{lem:random-good} applied with $\eps' :=\epsilon^{3/2}/10d$ in place of $\eps,$
since we have $\Omega(d^4\log(1/\tau)/\eps'^2)$ samples from $P,$
with probability at least $1-\tau,$
the set $S$ is such that for all affine functions $L,$
it holds $|\Pr_{X\in_u S}(L(X) \geq 0) - \Pr_{X\sim P}(L(X) \geq 0)| \leq \eps'/d.$
We henceforth condition on this event.

Let $C_T$ be the event that all coordinates in $T$ take their most common value.
For a single coordinate $i,$ the probability that it does not take its most common value,
$\min\{p_i,1-p_i\},$ satisfies
$$\min\{p_i, 1-p_i\} =p_i(1-p_i)/\max\{p_i,1-p_i\} \leq 2p_i(1-p_i).$$
Thus, by a union bound, we have that $\Pr_{P}(C_T) \geq 3/5.$
Let $\#_T(x)$ be the number of coordinates of $x$ in $T$ which do not have their most common value,
and observe that  $\#_T(x)$ is an affine function of $x.$
Noting that for $x \in \{0, 1\}^d$, we have that $1-\#_T(x) > 0$ if and only if $C_T$ holds for $x,$
it follows that $|\Pr_S(C_T) - \Pr_{P}[C_T]| \leq \eps'/d.$ Hence, we have that $\Pr_S(C_T) \geq 1/2.$

For any affine function $L(x),$
let $$L_T(x)=L(x) - \#_T(x) \cdot \max_{y \in \{0, 1\}^d} L(y).$$
Note that for $x \in \{0, 1\}^d$, we have that $L_T(x) > 0$ if and only if $L(x) > 0$ and $C_T$ holds for $x.$
Therefore, we can write
\begin{align*}
\left| \Pr_{X \in_u S}(L(X) >0)-\Pr_{X \sim P}(L(X)>0) \right|
& = \left|\frac{\Pr_{X \in_u S}(L_T(X) > 0)}{\Pr_{X \in_u S}(C_T)} - \frac{\Pr_{X \sim P}(L_T(X) > 0)}{\Pr_{X \sim P}(C_T)} \right|\\
& = \frac{\left| \Pr_{X \in_u S}(L_T(X) > 0)\Pr_{X \sim P}(C_T) -\Pr_{X \sim P}(L_T(X) > 0)\Pr_{X \in_u S}(C_T) \right|}{\Pr_{X \in_u S}(C_T)\Pr_{X \sim P}(C_T)} \\
& \leq (10/3) \cdot \Big(\Pr_{X \in_u S}(L_T(X) > 0) \cdot \big( \Pr_{X \sim P}(C_T) - \Pr_{X \in_u S}(C_T) \big) \\
& - \Pr_{X \in_u S}(C_T) \big(\Pr_{X \sim P}(L_T(X) > 0)-\Pr_{X \in_u S}(L_T(X) > 0)\big) \Big) \\
& \leq (10/3) \cdot 2 \eps'/d \leq \eps^{3/2}/d^2 \;.
\end{align*}
This completes the proof of Lemma~\ref{lem:random-good-dtv}.
\end{prevproof}
%!TEX root = ./main.tex

\section{Deferred Proofs from Section \ref{sec:filter-mixtures}} \label{sec:filter-mixtures-appendix}

\begin{prevproof}{Lemma}{lem:mixture-anchor-good}
Let $S_P \subseteq S$ be the set of samples drawn from $P$ and $S_Q \subseteq S$ be the set of samples drawn from $Q$.
Firstly, we note that by a Chernoff bound, $\left| |S_P|/|S| - \alpha \right| \leq O(\eps/d^2)$ with probability $1-\tau/3$.
Assuming this holds, it follows that $|S_P| \geq (\alpha/2)|S| \geq (\eps^{1/6}/2) |S| = \Omega(d^4\log(1/\tau)/\eps^2)$. 
Similarly, $|S_Q| \geq (1-\alpha)|S|/2 \geq  \Omega(d^4\log(1/\tau)/\eps^2)$.
%\begin{comment}
%Note that with high probability, at least $\Omega(d^4/\eps^2)$ samples are drawn from each of $P$ and $Q$ with each possible bit in the $i^{th}$ coordinate. Note that the %distribution of these samples is distributed like the product distribution $P$ or $Q$ conditioned on the appropriate value of the $i^{th}$ coordinate. By Lemma \ref{lem:random-%good}, we have that each of these multisets are $\eps$-good for the appropriate product distribution with probability $19/20$, and therefore with probability at least $3/4$ the %fraction of points from each of these sets is approximately correct, and all of these four subsets are $\eps$-good for the appropriate product distribution. This implies that $S$ is $%(\eps,i)$-good with respect to $\Pi$.
%This is false. We need $S_P$ to be \eps-good for $P$ as well. There is an argument that does not involve the ridiculous calculation below but it is more subtle.
%\end{comment}
%\begin{comment}
By Lemma \ref{lem:random-good} applied with $\eps' := (c^2/4) \cdot \eps$ in place of $\eps$, 
since we have $\Omega((d^4+d^2\log(\tau))/\eps'^2)$ samples, 
with probability $1-\tau/3$, the set $S_P$ is $\eps'$-good for $P$, 
i.e., it satisfies that for all affine functions $L$, $|\Pr_{X \in_u S_P}(L(X)>0) - \Pr_{X\sim P}(L(X)>0)| \leq \eps'/d.$ 
We show that assuming $S$ is $\eps'$-good, it is $(\eps,i)$ good for each $1 \leq i \leq d$.

%By similar reasoning, we have that with high probably, at least a $c/2$-fraction of the elements in $S_P$ have $i^{th}$-coordinate $0$. Therefore, %since these samples are distributed as $\Omega(d^4/\eps^2)$ samples from the restriction of $P$ to $x_i=0$, we have by Lemma \ref{lem:random-good}

Note that $\Pr_{X \sim P}[X_i=1] = p_i \geq c$ and $\Pr_{X \sim P}[X_i=0] =1-p_i \geq c$.
Since $|\Pr_{X \sim P}[X_i=1] - \Pr_{X \in_u S_P}[X_i=1]| \leq  c^2 \eps/(4d)$, it follows that 
$\Pr_{X \in_u S}[X_i=1] \geq c/2$.
For any affine function $L$, define $L^{(0)}(x) := L(x) - x_i (\max_y |L(y)|)$ and $L^{(1)}(x) := L(x) - (1-x_i) (\max_y |L(y)|)$. 
Then, we have the following:
\begin{align*} 
& \left|\Pr_{X \in_u S_P^1}\left(L(X) > 0\right)-\Pr_{X \sim P}\left(L(X) > 0 \mid X_i=1\right) \right| = \\
& =\left|\frac{\Pr_{X \in_u S_P^1}\left(L^{1}(X) > 0\right)}{\Pr_{X \in_u S_P^1}\left(X_i > 0\right)}- \frac{\Pr_{X \sim P}\left(L^{1}(X) > 0\right)}{p_i}\right| \\
& \leq (2/c^2)\left(\Pr_{X \in_u S_P^1}\left(L^{1}(X) > 0\right) p_i - \Pr_{X \sim P}\left(L^{1}(X) > 0\right) \Pr_{X \in_u S_1}\left(X_i > 0\right) \right) \\
& \leq (2/c^2)\left(p_i\left(\Pr_{X \in_u S_P^1}\left(L^{1}(X) > 0\right) - \Pr_{X \sim P}\left(L^{1}(X) > 0\right)\right)  -
 \Pr_{X \sim P}\left(L^{1}(X) > 0\right)\left(\Pr_{X \in_u S_P^1}\left(X_i > 0\right) - p_i\right)\right) \\
& \leq 2/c^2 \cdot 2\eps'/d \leq \eps/d \;.
\end{align*}
Similarly, we obtain that 
$$\left| \Pr_{X \in_u S_1}\left(L(X) > 0\right)-\Pr_{X \sim \Pi}\left(L(X) > 0\right) \right| \leq \eps/d \;.$$
So, we have that $S_P$ is $(\eps,i)$ good for $P$ for all $1 \leq i \leq d$ with probability $1-\tau/3$ . 
Similarly, $S_Q$ is $(\eps,i)$ good for $Q$ for all $1 \leq i \leq d$ with probability $1-\tau/3$. 
Thus, we have that $||S_P|/|S| - \alpha| \leq \eps/d^2$, $S_P$ is $(\eps,i)$ good for $P$ and $S_Q$ is $(\eps,i)$ good for $Q$ 
for all $1 \leq i \leq d$ with probability $1-\tau$. 
That is, $S$ is $(\eps,i)$-good for $\Pi$ for all $1 \leq i \leq d$ with probability  at least $1-\tau$.
%\end{comment}
\end{prevproof}

\begin{prevproof}{Lemma}{lem:mixture-close-good}
Let $S_P \subseteq S$ be the set of samples drawn from $P$ 
and $S_Q \subseteq S$ be the set of samples drawn from $Q$.
Firstly, we note that by a Chernoff bound, 
$||S_P|/|S| - \alpha| \leq O(\eps/d^2)$ with probability at least $1-\tau/3$. 
Assuming this holds, $|S_P| \geq (\alpha/2)|S| \geq \delta|S| = \Omega(d^4\log(1/\tau)/\eps^2)$. 
Similarly, $|S_Q| \geq (1-\alpha)|S|/2 \geq  \Omega(d^4\log(1/\tau)/\eps^2)$.

By Lemma~\ref{lem:random-good} applied with $\eps' := \eps/6$, 
since we have $\Omega(d^4\log(1/\tau)/\eps'^2)$ samples, 
with probability at least $1-\tau/3$, 
the set $S_P$ is $\eps$-good for $P$. 
Similarly,  with probability at least $1-\tau/3$, 
the set $S_Q$ is $\eps$-good for $Q$.
Thus, with probability $1-\tau$, 
we have that $\left| \frac{|S_P|}{|S|} - \alpha\right| \leq \eps$ 
and that $S_P$ and $S_Q$ are $\eps$-good for $P$ and $Q$ respectively.
\end{prevproof}

\begin{proof}[Proof of Lemma~\ref{lem:T}]
Noting that the mean of $T$ is $\mu$ and $|T|=|S'|$, we have:
\begin{align}
|S'|\Cov(S')  & = |S'_P| \E_{X \in_u S'_P}[(X-\mu)(X-\mu)^T] + |S'_Q| \E_{X \in_u S'_Q}[(X-\mu)(X-\mu)^T] \nonumber \\ 
& \ \ \ + |E| \E_{X \in_u E}[(X-\mu)(X-\mu)^T] \nonumber \\
& = |S'_P|\left(\Cov(S'_P)+(\mu^{S'_P}-\mu)(\mu^{S'_P}-\mu)^T\right) + |S'_Q|\left(\Cov(S'_P)+(\mu^{S'_Q}-\mu)(\mu^{S'_Q}-\mu)^T\right) \nonumber \\
& \ \ \ + |E| \E_{X \in_u E}[(X-\mu)(X-\mu)^T] \nonumber \\
& = |S'_P|\Cov(S'_P) + |S'_Q|\Cov(S'_Q) + |S'|\Cov(T) \;. \label{eq:T}
\end{align}
Since $P$ and $Q$ are product distributions,  
$\Cov(S'_P)$ and $\Cov(S'_Q)$ can have large diagonal elements but small off-diagonal ones. 
On the other hand, we bound the elements on the diagonal of $\Cov(T)$, but $\|\Cov(T)\|_2$ 
may still be large due to off-diagonal elements.

By the triangle inequality, and Equation (\ref{eq:T}) with zeroed diagonal, we have:
\begin{align}
\|\Cov_0(S')-\Cov(T)\|_2 & \leq \|\Cov_0(S') - \Cov_0(T)\|_2+\|\Cov_0(T)-\Cov(T)\|_2 \nonumber \\
& \leq \left( \frac{|S'_P|}{|S'|}\right)\|\Cov_0(S'_P)\|_2 + \left( \frac{|S'_Q|}{|S'|}\right) \|\Cov_0(S'_Q)\|_2 + \|\Cov_0(T)-\Cov(T)\|_2 \;.
\label{eq:tri-cov}
\end{align}
We will bound each of these terms separately.

Note that $\Cov_0(T)-\Cov(T)$ is a diagonal matrix 
and its non-zero entries are $$(\Cov_0(T)-\Cov(T))_{i,i} = \Var_{X \in_u T}[X_i].$$
Since the mean of $T$ is $\mu$, for all $i$, 
we have that $\Var_{X \in_u T}[X_i] \leq \E_{X \in_u T}[ \|X - \mu\|_\infty^2]$. 
We seek to bound the RHS from above.

Note that $\mu$ satisfies 
$|S'|\mu = |S'_P| \mu^{S'_P} + |S'_Q| \mu^{S'_Q} + |E|\mu^{E}$. 
Since $|S'|-|E|=|S'_P| + |S'_Q|,$ 
we have $(|S'|-|E|)(\mu - \mu^{S'_P}) = |S'_Q|(\mu^{S'_Q}-\mu^{S'_P}) + |E|(\mu^{E} - \mu).$ 
Using that $|S'|-|E|=(1+O(\eps))|S|$, $|S_Q'| = (1-\alpha)|S| - O(\eps)$, $|E| \leq O(\eps)|S|$, we have
$$\|\mu - \mu^{S'_P}\|_\infty \leq (1-\alpha + O(\eps))\|\mu^{S'_Q}-\mu^{S'_P}\|_\infty + O(\eps) \;.$$
Similarly,
$$\|\mu - \mu^{S'_Q}\|_\infty \leq (\alpha + O(\eps))\|\mu^{S'_Q}-\mu^{S'_P}\|_\infty + O(\eps) \;.$$
Since $S$ is $\eps$-good for $\Pi$, it follows that
$\|\mu^{S_P} -p\|_\infty \leq \eps/d$ 
and $\|\mu^{S_Q} -q\|_\infty \leq \eps/d.$ 
Also, 
$$\||S_P|\mu^{S_P}-|S'_P|\mu^{S'_P}\|_\infty \leq |S_P|-|S'_P| \leq O(\eps)|S| \;.$$
Thus, 
$$\|\mu^{S_P} - \mu^{S'_P}\|_\infty \leq  \|\mu^{S_P} - (|S'_P|/|S_P|)\mu^{S'_P}\|_\infty + (|S_P|-|S'_P|)/|S_P| \leq O(\eps)|S|/|S_P| \leq O(\alpha\eps) \;.$$ 
Similarly, we show that 
$$\|\mu^{S_Q} - \mu^{S'_Q}\|_\infty \leq O((1-\alpha)\eps).$$ 
Finally, $\|p-q\|_\infty \leq \delta.$ 
Thus, by the triangle inequality, we get
$$\|\mu^{S'_Q}-\mu^{S'_P}\|_\infty \leq O(\alpha\eps) + \eps/d + \delta + \eps/d + O((1-\alpha)\eps) \leq \delta + O(\eps) \;.$$
We have the following sequence of inequalities:
\begin{align*}
|S'| \Var_{X \in_u T}[X_i] & \leq |S'|\E_{X \in_u T}[ \|X - \mu\|_\infty^2] \\
				&=|S'_P|\|\mu - \mu^{S'_P}\|_\infty^2 + |S'_Q|\|\mu - \mu^{S'_Q}\|_\infty^2 \\
				&\ \ \ +  |E|\E_{X \in_u T}[ \|X - \mu\|_\infty^2] \\
				& \leq (|S'_P| + |S'_Q|)(\|\mu^{S'_Q}-\mu^{S'_P}\|_\infty + O(\eps))^2 + |E| \\
				& \leq (\delta^2+O(\eps))|S'| \;.
\end{align*}
Thus, 
$$\|\Cov_0(T)-\Cov(T)\|_2 = \max_i (\Cov_0(T)-\Cov(T))_{i,i} = \max_i \Var(T_i) \leq O(\delta^2 + \epsilon).$$
It remains to bound the 
$\left( \frac{|S'_P|}{|S'|}\right)\|\Cov_0(S'_P)\|_2 + \left( \frac{|S'_Q|}{|S'|}\right) \|\Cov_0(S'_Q)\|_2$ terms  in (\ref{eq:tri-cov}).
To analyze the first of these terms, note that $\Cov_0(P)= \mathbf{0}.$ 
We have that
\begin{align*}
\|\Cov_0(S'_P)\|_2 & = \|\Cov_0(S'_P)-\Cov(P)+\textrm{Diag}(\Var_{X \sim P}(X_i)) \|_2 \\ & \leq \|\Cov(S'_P) -\Cov(P)\|_2
+ \max_i(|\Var_{X \in_u S'_P}(X_i)-\Var_{X \sim P}(X_i)|) \;.
\end{align*}
Noting that 
$$|\Var_{X \in_u S'_P}(X_i)-\Var_{X \sim P}(X_i)|= e_i^T (\Cov(S'_P) -\Cov(P))e_i,$$ 
we have that
$$
\max_i(|\Var_{X \in_u S'_P}(X_i)-\Var_{X \sim P}(X_i)|) \leq  \|\Cov(S'_P) -\Cov(P)\|_2 \;,
$$
and so 
$$\|\Cov_0(S'_P)\|_2 \leq 2 \|\Cov(S'_P) -\Cov(P)\|_2.$$
By the triangle inequality, 
$$\|\Cov(S'_P) -\Cov(P)\|_2 \leq \|\Cov(S'_P) -\Cov(S_P)\|_2 + \|\Cov(S_P) -\Cov(P)\|_2 \;.$$ 
Note that since $S$ is good, the $(i,j)$-th entry of $\Cov(S_P) -\Cov(P)$ 
has absolute value at most $\eps/d.$ 
Thus, 
$$\|\Cov(S_P) -\Cov(P)\|_2 \leq \|\Cov(S_P) -\Cov(P)\|_F \leq \eps \;,$$
which gives 
$$\|\Cov_0(S'_P)\|_2 \leq 2 \|\Cov(S'_P)-\Cov(S_P)\|_2 + O(\eps).$$ 
We have
$$
\|\Cov(S'_P)-\Cov(S_P)\|_2  = \sup_{\|v\|_2=1}\left(|\Var_{X\in_u S'_P}(v\cdot X)-\Var_{X \in_u S_P}(v\cdot X) | \right) \;.
$$
Since $S'_P \subseteq S_P$,
\begin{align*}
|S'_P| \Var_{X\in_u S'_P}(v\cdot X) & \leq |S_P| \E_{X \in_u S_P}[(v \cdot X-\mu^{S'_P}] \\
& \leq |S_P| (\Var_{X\in_u S_P}(v\cdot X) + \|\mu^{S'_P} - \mu^{S_P}\|_2^2) \\
& \leq (1+O(\eps/\alpha))|S'_P| \cdot (\Var_{X\in_u S_P}(v\cdot X) + O(\eps^2\log(1/\eps)/\alpha^2)) \;.
\end{align*}
Thus,
\begin{align*}
|\Var_{X\in_u S'_P}(v\cdot X)-\Var_{X \in_u S_P}(v\cdot X)| & \leq O(\eps/\alpha)\Var_{X\in_u S_P}(v\cdot X) + O(\eps^2\log(1/\eps)/\alpha^2)) \\
& \leq  O(\eps/\alpha)\Var_{X\sim P}(v\cdot X) + O(\eps^2\log(1/\eps)/\alpha^2)) \tag*{(since $\|\Cov(S_P) -\Cov(P)\|_2 \leq \eps$)} \\
& \leq O(\eps/\alpha) + O(\eps^2\log(1/\eps)/\alpha^2))  \tag*{(since $\|\Cov(P)\|_2 \leq 1$)} \\
& \leq O(\eps\log(1/\eps)/\alpha) \;. \tag*{(since $\alpha \geq \eps$)}
\end{align*}
Thus, we have that 
$$\|\Cov_0(S'_P)\|_2 \leq 2 \cdot O(\eps\log(1/\eps)/\alpha) + O(\eps) \leq O(\eps\log(1/\eps)/\alpha).$$
Therefore, combining the above we have that
$$
\left( \frac{|S'_P|}{|S'|}\right)\|\Cov_0(S'_P)\|_2 = (\alpha+O(\eps)) \|\Cov_0(S'_P)\|_2 = O(\epsilon\log(1/\epsilon)) \;.
$$
A similar argument shows
$$
\left( \frac{|S'_Q|}{|S'|}\right)\|\Cov_0(S'_Q)\|_2 = O(\epsilon\log(1/\epsilon)).
$$
Combining this with the above gives that 
$$\| \Cov_0(S') - \Cov(T) \|_2 = O(\delta^2 + \epsilon\log(1/\epsilon)).$$
By the assumption on $\delta$ in Theorem~\ref{mixCloseThm}, $\delta^2 = \Omega(\epsilon\log(1/\epsilon))$, 
and the proof is complete.
\end{proof}

\end{document}